%% file: dissertacao_fisica.tex
\PassOptionsToPackage{english}{babel}
\documentclass[
	12pt,				
	oneside,			
	a4paper,			
	English				
	]{abntex2}
\usepackage[T1]{fontenc}		
\usepackage[scaled]{helvet}

\usepackage[utf8]{inputenc}		
\usepackage{setspace}			
\usepackage{longtable}			
\usepackage{threeparttable}		
\usepackage{indentfirst}		
\usepackage{color}				
\usepackage{graphicx}			
\DeclareGraphicsExtensions{.eps,.png,.jpg,.pdf}
\usepackage{epstopdf}			
\usepackage{microtype} 			
\usepackage{caption,subcaption,rotating}	
\usepackage{multirow,array,booktabs}		
\usepackage{hyperref}
\usepackage{setspace}
\hypersetup{
			pdffitwindow=true,
			colorlinks=true,
			linkcolor=black,
			citecolor=black,
			urlcolor=black}
\usepackage{amsmath}
\usepackage{tabularx}
\usepackage{multirow}
\usepackage[table,xcdraw]{xcolor}
\usepackage{booktabs}
\usepackage{textcomp}
\usepackage{nomencl}
\usepackage{eufrak}
\usepackage{varwidth}
\usepackage{amsthm}
\usepackage{enumitem}
\usepackage{pdfpages}
\usepackage{bbm}
\usepackage{mathpazo}
\usepackage{amssymb}
\usepackage{titlesec}

\titleformat{\chapter}
{\normalfont\Large\bfseries}{\thechapter}{.5em}{\vspace{.1ex}}[]
\titlespacing*{\chapter}      
{0pt}{0pt}{15pt}
\titleformat{\section}
{\normalfont\large\bfseries}{\thesection}{.5em}{\vspace{.1ex}}[]
\titlespacing*{\chapter}      
{0pt}{0pt}{15pt}
\counterwithin{figure}{chapter}
\counterwithin{table}{chapter}

\newcommand{\aref}[1]{\hyperref[#1]{Apêndice~\ref{#1}}} 
\newcommand{\quadroname}{Quadro}
\newcommand{\listofquadrosname}{Lista de quadros}
\newfloat[chapter]{quadro}{loq}{\quadroname}
\newlistof{listofquadros}{loq}{\listofquadrosname}
\newlistentry{quadro}{loq}{0}

\counterwithout{quadro}{chapter}

\newcommand{\ket}[1]{\mbox{$ | #1 \rangle $}}
\newcommand{\bra}[1]{\mbox{$ \langle #1 | $}}
\counterwithin{quadro}{chapter}
\begin{document}
\selectlanguage{English}
\frenchspacing 

\pretextual
\newtheorem{teoremas}{Theorem}
[chapter]
\newtheorem{definicoes}{Definition}
[chapter]
\newtheorem{lemas}{Lemma}
[chapter]
\newtheorem{corolarios}{Corollary}
[chapter]
\newtheorem{remark}{Remark}
[chapter]
\newtheorem{postuladothermo}{Thermodynamics Postulate}
\newtheorem{postuladoq}{QM Postulate}
\newtheorem{postuladoq2}{QM Postulate (Density Matrix)}
\newtheorem{principle}{Principle}
\input{capa}
\input{rosto}


\newpage
\begin{epigrafe}
	\vspace*{\fill}
\begin{flushright}
	\textit{"Order is manifestly maintained in the universe..."\\
		James Prescott Joule}\\
\end{flushright}
\begin{flushright}
	\textit{"It appears to me, extremely difficult, if not quite impossible, \\
		to form any distinct idea of anything, capable of being excited and communicated,\\
		 in the manner the heat was excited and communicated in these experiments,\\
		  except it be motion."\\
		Count Rumford}
\end{flushright}
\begin{flushright}
	\textit{"[The] conviction that every mathematical problem \\
		can be solved is a powerful	incentive to us as we work. \\
		We hear within us the perpetual call: There is the problem.\\
		Seek its solution. You can find it by pure thinking,\\
		for in mathematics there is no ignorabimus!" \\
		David Hilbert}
\end{flushright}
\end{epigrafe}
\input{abstract}
\mbox{}
\thispagestyle{empty}
\newpage
\input{listas}
\textual

\input{intro_fisica}

\input{fundamentacao_fisica}

\input{fundamentacao_fisica2}

\input{descricaodomodelo_fisica}

%
%
\input{osciladordissipativo}

\input{consideracoesfisica}

\postextual
\bibliographystyle{unsrt}
{\normalsize \bibliography{referencias_080618}}

\apendices

\input{apendice_caldirolakanai}

\input{apendice_workclassical}

\input{apendice_workquantum}

\end{document}

%% file: capa.tex
\begin{center}
\textbf{\foreignlanguage{brazil}{FEDERAL UNIVERSITY OF PARANÁ} \\
\foreignlanguage{brazil}{PHYSICS GRADUATE PROGRAM}} \\
\vspace{6cm}

\textbf{THALES AUGUSTO BARBOSA PINTO SILVA}
\vspace{3.5cm}

\large{\textbf{A DEFINITION OF QUANTUM MECHANICAL WORK}}
\vspace{3cm}

\normalsize
\textbf{ \foreignlanguage{brazil}{DISSERTATION}}
\vfill

\textbf{CURITIBA} \\
\textbf{2018}
\end{center}

%% file: rosto.tex
\begin{center}
THALES AUGUSTO BARBOSA PINTO SILVA \\
\vspace{8.5cm}

A DEFINITION OF QUANTUM MECHANICAL WORK
\normalsize
\vspace{2cm}
\end{center}

\begin{flushright}
\begin{minipage}[9.6cm]{8cm}

\foreignlanguage{brazil}{Master's degree dissertation accepted by the Federal University of Paraná.}

\vspace{0.5cm}
\noindent

Supervisor: Prof. Renato Moreira Angelo$^{1}$ \linebreak
\\
Chair: Prof. Frederico Borges de Brito$^{2}$,\linebreak
Prof. Wilson Marques Junior$^{1}$. \linebreak
\\
(1) Department of Physics, Federal University of Paraná, Curitiba, Brazil.
\\
(2) Physics Institute of São Carlos, University of São Paulo, São Carlos, Brazil.
\end{minipage}
\end{flushright}

\vfill

\begin{center}
\textbf{CURITIBA}\\ 
\textbf{2018}
\end{center}

%% file: abstract.tex
\chapter*{ABSTRACT}

\noindent Quantum Mechanics can be seen as a mathematical framework that describes experimental results associated with microscopic systems. On the other hand, the theory of Classical Thermodynamics (along with Statistical Physics) has been used to characterize macroscopic systems in a general way, whereby mean quantities are considered and the connections among them are formally described by state equations. In order to relate these theories and build up a more general one, recent works have been developed to connect the fundamental ideas of Thermodynamics with quantum principles. This theoretical framework is sometimes called Quantum Thermodynamics. Some key concepts associated with this recent theory are work and heat, which are very well established in the scope of Classical
Thermodynamics and compose the energy conservation principle expressed by the first law of Thermodynamics. Widely accepted definitions of heat and work within the context of Quantum Thermodynamics were introduced by Alicki in his 1979 seminal work.  Although such definitions can be shown to directly satisfy the  first law of Thermodynamics and have been successfully applied to many contexts, there seems to be no deep foundational justification for them. In fact, alternative definitions have been proposed with basis on analogies with Classical Thermodynamics. In the present dissertation, a definition of quantum mechanical work is introduced which preserves the mathematical structure of the classical concept of work without, however, in any way invoking the notion of trajectory. By use of Gaussian states and the Caldirola-Kanai model, a case study is conducted through which the proposed quantum work is compared with Alicki's definition, both in quantum and semiclassical regimes, showing promising results. Conceptual inadequacies of Alicki's model are found in the classical limit and possible interpretations are discussed for the presently introduced notion of work. Finally, the new definition is investigated in comparison with a classical-statistical approach for superposition and mixed states.  \\[1.5cm]  
\noindent \textbf{Keywords:} Work. Quantum Thermodynamics. Caldirola-Kanai model.

%% file: listas.tex




\newpage

\tableofcontents*
\pagebreak
\setcounter{table}{0}

%% file: intro_fisica.tex
\chapter{INTRODUCTION}
\label{cap:int}

\setcounter{page}{10}
\par The derivation of an accurate framework for the physical description of macroscopic systems always challenged the scientific community. An important step in this direction was the establishment of Classical Mechanics. Regarding the equations of motion, precise results were obtained for few components macroscopic systems, and it was expected that for systems with large number of constituents, this precision would also be verified \cite{laplace2012pierre}. Such belief was proved true for simple macroscopic systems, such as symmetrical rigid bodies \cite{arnol2013mathematical,goldstein2011classical,kleppner2013introduction}. However, for systems involving more intricate relations among its constituents and  internal degrees of freedom not of simple description, a formal proof and verification with experimental data were not possible through purely classical arguments. Such lack was justified, mainly, due to operational limitations \cite{schwabl2006statistical,tolman1938principles,landau1968statistical}: 
\begin{enumerate}[label=(\roman*)]
	\item it is necessary a large amount of calculations (and therefore computation) in order to determine the dynamical variables at each instant of time for each element composing the system and 
	\item the precise knowledge of the initial condition may be virtually impossible to determine, depending on the number of degrees of freedom of the system.
\end{enumerate}
Thermodynamics then arose as a form to work around those limitations, where few variables are needed to describe macroscopic multi particles system behavior \cite{callen1985thermodynamics,moran2010fundamentals,gurtin2010mechanics}. Such empirical model has been developed since the mid eighteenth century, in parallel with Classical Mechanics \cite{gemmer20044}.
\par Classical Thermodynamics, so to speak, has been applied in different industrial sectors, since its birth, such as combustion engines, turbines, boilers and compressors, for instance \cite{moran2010fundamentals}. Nowadays, its tools, under the scope of continuum thermodynamics, provide useful prediction for the dynamics of new materials \cite{gurtin2010mechanics,truesdell2004non,silva2015elastoplasticity,EUeHPAD,macosko1994rheology}. In the general standard form in which the theory is founded, its main hypotheses are treated, roughly, considering systems in equilibrium or quasi-equilibrium. From the laws of Thermodynamics and its underlying hypotheses, state equations can be deduced correlating a few measurable macroscopic variables, which gives to Thermodynamics considerable predictive power to deal with macroscopic systems. Thermodynamics has then established itself as a fundamental branch of Physics.
\par By the end of the eighteenth century, the scientific community looked for a mechanical basis for Thermodynamics \cite{gemmer20044}. Under this perspective, it was reasonable to think that any macroscopic system, then believed to be composed of particles, should be describable by the dynamics of its constituents. Firstly, Thermodynamics was tried to be derived entirely from the realistic paradigm settled by Classical Mechanics. Such approach was pursued by Boltzmann who made significant efforts in this direction \cite{gemmer20044}. However, as Boltzmann himself verified in subsequent works, such coupling could only be achieved within the use of statistical arguments. As a consequence Maxwell, Boltzmann, Gibbs, Poincaré and others proposed a statistical point of view to study the dynamics of many body systems \cite{arnol1968ergodic}. Such approaches gave rise to the theory which is now called Statistical Physics \cite{schwabl2006statistical,tolman1938principles,landau1968statistical,kubo1957statistical,reif2009fundamentals,chandler1987introduction,schrodinger1989statistical}.  
\par The  introduction of statistical concepts such as equilibrium ensembles (micro-canonical, canonical and grand-canonical, for instance) and the corresponding distributions, together with the ergodic hypothesis, provided to Statistical Physics powerful tools to connect the experimentally verified results of Thermodynamics with the Classical Mechanics equations of motion. It was therefore established a connection between microscopic properties of the constituents of a system (micro-states) with macroscopic variables (macro-states). Questions with respect to out-of-equilibrium systems were treated, in the relatively recent work of Jarzinsky \cite{jarzynski1997nonequilibrium}: the so called fluctuation theorems \cite{campisi2011colloquium} were developed in other to restrict non-equilibrium variables in a general way. Despite the broad applicability of the Statistical Physics, some fundamental question remain open: Until which size of a system, would its thermodynamical quantities provide accurate results? What features a system must have for the equilibrium and ergodic hypotheses to be valid? From about which size of the system must statistical considerations be taken into account? In order to answer these questions, first it was essential to establish a theory that could successfully account for experimental data related to microscopic systems,  a task not accomplished by Classical Mechanics \cite{sakurai2017modern,eisberg1979fisica}. In order to fulfill such theoretical vacancy, the Quantum Mechanics was born and developed ever since, providing accurate results. The Statistical Physics would therefore be extended in a way such that the Quantum Mechanics could be accommodated.
\par  The Quantum Statistics is frequently considered within the Von Neumann density matrix formalism and presents precise results. As well-known example of its application, it can be mention the ideal paramagnetic spin-$\frac{1}{2}$ system \cite{salinas1997introduccao}. Such formalism also provided great advances in Information Theory \cite{nielsen2010quantum,bilobran2015measure,goold2016role}. As a consequence of the development of Quantum Statistics in equilibrium systems, it was possible to establish a connection of Quantum Mechanics with thermodynamical variables for macroscopic systems in equilibrium, analogous to the classical case \cite{schwabl2006statistical,tolman1938principles,landau1968statistical,kubo1957statistical,reif2009fundamentals,chandler1987introduction}. With respect to non-equilibrium cases, the fluctuation theorems, mentioned above for the classical cases, were extended to the realm of Quantum Mechanics \cite{campisi2011colloquium,aaberg2018fully,perarnau2017no,hanggi2015other,alhambra2016fluctuating,jarzynski2017stochastic}. It can be argued that the couple Quantum Mechanics and Statistical Physics has provided a robust formal basis for the description of a large class of physical systems. However, the questions raised in the previous paragraph can still be made. It is by no means clear what size a system must have for the hypotheses of Statistical Physics to remain applicable. In addition, one may ask whether the thermodynamical laws will hold for few-particle systems. Since Quantum Mechanics has provided more successful descriptions for microscopic few-particle systems, the attempts to answer those questions are considered within the scope of Quantum Mechanics. Recently, a new area treating few-particle system, relating Thermodynamics with Quantum Mechanics, was born: some authors called it Quantum Thermodynamics (QTh) \cite{gemmer20044,valente2017quantum,youssef2009quantum,alipour2016correlations}.
\par QTh is frequently considered within the scope of open quantum systems (OQS), where the dynamics of a system is investigated upon interaction with an ideally large environment \cite{gemmer20044}. Recently, general results were obtained in order to approximate Quantum Mechanics to statistical ensemble considerations \cite{popescu2006entanglement,tasaki1998quantum}: by regarding high dimensional weakly interacting quantum systems, its parts could, to a good approximation, behave as statistical ensembles.  
\par Dissipation effects and the evolution of an open subsystem in a thermal bath have been studied, under the master equation formalism \cite{louisell1973quantum,gardiner2004quantum,wen2004quantum}. The bath was considered as having multiple degrees of freedom, weak interaction or/and no memory, such that Born-Markov approximations could be considered \cite{karlewski2014time}. However, these approximations become inappropriate for systems involving a small number of parts.  In these cases, non-Markovian regimes and strong-interacting models are to be employed  \cite{tan2011non,tu2008non,zhang2012general}. QTh has also been connected to Information Theory, where several results have been reported in connection with entanglement, reality, locality, correlations, Maxwell's demon, among others \cite{nielsen2010quantum,maruyama2009colloquium,micadei2017reversing}.
\par Traditional concepts of classical Thermodynamics have been revisited in the domain of QTh. In Ref. \cite{kliesch2014locality} the locality of temperature has been analyzed, by questioning until which size of a subsystem, the equilibrium temperature offers good agreement with a measure of it. Also, temperature has been investigated as a dynamical variable \cite{rugh1997dynamical}. Heat and work, by their turn, are the thermodynamical concepts more often invoked in QTh discussions \cite{campisi2011colloquium,aaberg2018fully,perarnau2017no,hanggi2015other,alhambra2016fluctuating,jarzynski2017stochastic,valente2017quantum,youssef2009quantum,alipour2016correlations,allahverdyan2005work,tonner2005autonomous,lorch2018optimal,bender2000quantum,alicki1979quantum}. The present dissertation aims at proposing a quantum mechanical notion for the work imparted on a single quantum particle. 
\section{\textbf{Problem definition}}
\par  In Classical Mechanics, work has a very well established definition \cite{arnol2013mathematical,goldstein2011classical,kleppner2013introduction}. The work $\mathcal{W}_{cl}$ performed by a force $\underline{f}$ on a particle which undergoes a displacement $d\underline{x}$ in a path $c$ is given by
\begin{equation}
\mathcal{W}_{cl}=\int_{c}\underline{f}\cdot\underline{dx}.\label{eq:wcl1}
\end{equation} 
Throughout this dissertation, the notation $\underline{\left(\cdot\right)}$ is considered to denote finite-dimensional real vectors, i.e.,  $\underline{v}$ denotes a vector in $\mathbb{R}^{N}$, with $N>0$ being a natural number. In Eq. \eqref{eq:wcl1}, $\underline{f}$ and $\underline{x}$ stands for \emph{three-dimensional} vectors referring to force and position, respectively.  
\par The direct adaptation of the expression \eqref{eq:wcl1} into the quantum domain was not conducted, perhaps because there is no trivial quantum counterpart for the classical notions of trajectory, force, and displacement \cite{sakurai2017modern,perarnau2017no}. Therefore, alternative approaches have been considered in the quantum realm. 
\par A definition of work, considered by a relatively great part of the scientific community, is the one first introduced by Alicki \cite{alicki1979quantum}. In his paper, Alicki based his definition considering that the work applied on a system must be associated with its Hamiltonian time rate; the heat, on the other hand, is related with time changes in the density matrix. As a consequence, it was established a quantum form for the first law of thermodynamics. Considering the Hamiltonian $H$ and the system density matrix $\rho$ one has that the total energy $E=\mathrm{Tr}\left[\rho(t)H(t)\right]$ changes in time can be divided as 
\begin{equation}
	\displaystyle\Delta E(t_{f}-t_{i})=\underbrace{\displaystyle\int_{t_{i}}^{t_{f}}\mathrm{Tr}\left[\frac{\partial\rho(t)}{\partial t}H(t)\right]dt}_{\displaystyle \mathcal{Q}_{ak}(t_{f}-t_{i})}+\underbrace{\displaystyle\int_{t_{i}}^{t_{f}}\mathrm{Tr}\left[\rho(t)\frac{\partial H(t)}{\partial t}\right]dt}_{\displaystyle \mathcal{W}_{ak}(t_{f}-t_{i})},\label{eq:alickifirstime}
\end{equation}
where $\mathcal{Q}_{ak}(t_{f}-t_{i})$ and $\mathcal{W}_{ak}(t_{f}-t_{i})$ are the heat and work that are imposed on the system in the time interval $t_{f}-t_{i}$ and $\mathrm{Tr}$ is the trace operation. Alicki's approach is widely adopted, especially for derivations of the aforementioned fluctuation theorems. By use of those theorems, experimental results have been obtained which provide conceptual support for the theory  \cite{campisi2011colloquium,ribeiro2016quantum}. However, Alicki's definition has also been criticized by some authors \cite{goold2016role,weimer2008local}. It is remarked here that this definition is not unique, that is, it is just a particular form into which the total energy change can be decomposed. There is, \emph{a priori}, no fundamental principle forcing one to make the particular identifications suggested in Eq. \eqref{eq:alickifirstime}, although the argument seems intuitive for classic examples \cite{schwabl2006statistical,reif2009fundamentals,bender2000quantum}. In addition, it is no clear whether this approach can indeed furnish the correct classical limit.
\par There is another definition related with work that is frequently mentioned under the scope of autonomous quantum systems \cite{tonner2005autonomous,lorch2018optimal,weimer2008local}. For a given source of energy, the flows into the system is considered to be exclusively due to work when there is thermal isolation and exclusively due to heat when no form of work takes place. Such definition is based on arguments that resembles Classical Thermodynamics and provide the possibility of development of theoretical thermal machines using quantum mechanical resources. However, since no explicit definition of work or heat was made, one might argue, with basis on quantum fluctuations, that it should not be discarded that some classical forms of work could be viewed, in the quantum domain, as heat. Likewise, heat transfer could somehow manifests itself as a form of quantum work. 
\par The approach to be adopted in the present work aims at developing a definition of work in a quantum mechanical form, departing from the classical well-known form \eqref{eq:wcl1}. Here, the problem of the absence of trajectories in Quantum Mechanics is put aside by considering the following argument. If the classical trajectory $\underline{x}(t)$ of a particle of mass $m$ is a continuous function of $t$, $\underline{\dot{x}}(t)$ is the velocity of the particle, and  $\underline{f}(t) = m\underline{\ddot{x}}(t)$ is the resultant force, then the work \eqref{eq:wcl1} imparted on the particle can be re-written as

\begin{equation}
\mathcal{W}_{cl}=\int_{c(t)}\underline{f}\cdot\underline{\dot{x}}dt=\int_{c(t)} m\underline{\ddot{x}}\cdot\underline{\dot{x}}dt.\label{eq:trabclass2}
\end{equation} 

Is this form more susceptible to a quantum generalization than \eqref{eq:wcl1}? The definition to be formalized in the text departs from the answer to this question. 
\par  By analogy with the classical definition \eqref{eq:trabclass2}, a mechanical formula for work is introduced, using the Heisenberg picture of Quantum Mechanics and usual methods of expectation value calculation. As an immediate consequence, the obtained quantity is shown not to be a state function, since it depends on a time interval. Also, using the uncertainty principle an interesting aspect underlying the proposed definition for generic systems is discussed. As a case study, the definition of work proposed is calculated for the Caldirola-Kanai Hamiltonian, which is often used to describe the dynamics of a damped oscillator \cite{kanai1948quantization,caldirola1941forze}. A comparison between the results regarding the proposed definition with the Alick model is provided and the adequacy of these proposals in correctly describing the classical limit is discussed.
\par The conceptual grounding on which the work is based, is presented in chapters \ref{cap:fundamentacao} and \ref{cap:fundamentacao2}, where the main hypotheses, results, and definitions of Thermodynamics and Statistical Physics are briefly reviewed. A summary of the fundamental aspects of Quantum Mechanics and QTh is then presented, with particular emphasis to notions related to work and heat.
\par The definition of work proposed in this dissertation is introduced in chapter \ref{cap:definicaotrab}, where some properties of the defined notion are discussed and preliminary results exposed for generic systems.
\par The central case study of this dissertation is presented in chapter  \ref{cap:kanaicaldirola}, where the Caldirola-Kanai Hamiltonian is discussed and simulations of the proposed definition of work are shown. Most importantly, the results regarding the proposed definition are compared with Alick's proposal, and the adequacy of both are assessed in reproducing the classical predictions. 
\par Finally, in chapter \ref{cap:consid}, some remarks are made considering the main aspects related to the results, highlighting the advantages, limitations, and properties of the proposed definition.

%% file: fundamentacao_fisica.tex
\chapter{THEORETICAL BACKGROUND: CLASSICAL PERSPECTIVE}
\label{cap:fundamentacao}
The introduction of the concept of work was first established in the domain of Classical Mechanics and then introduced in the fundamental laws of Thermodynamics. The connection between both theories was then made through statistical assumptions. Throughout this chapter, it is discussed the concept of work and the related concepts under a somewhat historical sequence. First, it is analyzed how work is defined in Classical Mechanics. Then it is investigated how it was introduced in Thermodynamics and related concepts are discussed. Finally, the theory of Statistical Physics, aligned with Classical Mechanics, is considered and some properties to be used in subsequent chapter are discussed/reviewed.
\section{\textbf{Classical Mechanics and energy}}
\label{sec:classical}
\par In a classical system, a particle is considered to be a compact system with mass $m$. In standard analysis, it is not taking into account internal degrees of freedom associated with the particle: for it does not rotate around itself, as in the case of rigid bodies motion, for instance. The motion of this particle can be described \cite{arnol2013mathematical} by a mapping $\underline{x}:I\rightarrow \mathbb{R}^{3}$, where the interval $I\in \mathbb{R}$ is an representation of a time interval. The position of the particle, identified by this motion at an instant $t$ (a point in the interval $I$), is therefore just $\underline{x}(t)$. 
\par There may be forces (gravitational or Coulombian for instance) acting on such a particle. Their representation, in the mathematical structure of Classical Mechanics, can be described by a vector field\footnote{The field is here considered to be continuous: this consideration is sufficient to prove that the integral in Eq. \eqref{eq:wcl2} exists \cite{arnol2013mathematical}.} $\underline{f}:\mathbb{R}^{3}\rightarrow\mathbb{R}^{3}$. The work $\mathcal{W}_{cl}$\footnote{The subindex $\left(\cdot\right)_{cl}$ here is considered for classical systems, to be differentiate from quantum ones in subsequent chapters.} done on the particle,  by the external force $\underline{f}$, on a curve $c$ of finite length, is defined as
\begin{equation}
	\mathcal{W}_{cl}= \int_{c}\underline{f}\cdot\underline{dx}.\label{eq:wcl2}
\end{equation}
Newton's second law establishes a connection between the resultant force $\underline{f}_{R}$(the sum of all the forces $\underline{f}_{1}$,$\underline{f}_{2},\cdots,\underline{f}_{N}$ applied on a system) with its acceleration\footnote{Here the mass is considered to be constant since the system is composed of only one particle, which does not change its mass during time (non-relativistic regime).}:
\begin{equation}
	\underline{f}_{R}=\sum_{i=1}^{n}\underline{f}_{i}=m\underline{\ddot{x}}.
\end{equation}
Therefore, the work done by all the forces on the particle $\mathcal{W}_{cl}^{R}$ is computed as
\begin{equation}
	\mathcal{W}_{cl}^{R}=\int_{c}\underline{f}_{R}\cdot\underline{dx}=\int_{c}m\underline{\ddot{x}}\cdot\underline{dx}=m\int_{t_{i}}^{t_{f}}\underline{\ddot{x}}\cdot\underline{\dot{x}}dt=\left.\frac{m\dot{x}^{2}}{2}\right|_{t_{i}}^{t_{f}},
\end{equation}
where $t_{i}$ and $t_{f}$ are the time before and after the particle travels through the path $c$. The scalar $\frac{m\dot{x}^{2}}{2}$ is defined as the \emph{kinetic energy} $K_{cl}$ of the particle. The result given above can therefore be shortly stated as
\begin{equation}
\mathcal{W}_{cl}^{R}\left(t_{f}-t_{i}\right)=K_{cl}(t_{f})-K_{cl}(t_{i}),\label{eq:teotrabk}
\end{equation}
that is, \emph{the total work done on a particle equals the change in its kinetic energy}. The result \eqref{eq:teotrabk} may, for some class of problems, be used to determine the motion features of a particle  in an easier manner as one would do by solving Newton's second law equation directly (or by using Lagrange-Euler or Hamilton equations). 
\par It is important to remark that the definition of kinetic energy can be seen as an intrinsic particle property: knowing particle's velocity and mass, its kinetic energy can be determined without any knowledge of the configuration of the rest of the system with which the particle interacts. Therefore it can be stated, in these terms, that the kinetic energy is the energy that ultimately "belongs to the particle", that is, if one aims to define an energy that somehow is intern or belongs exclusively to the particle, then the strongest candidate should be the kinetic energy. 
\par The result \eqref{eq:teotrabk} is not restrictive for a system of only one particle: considering the generalized position vector $\underline{x}_{T}=\left(\underline{x}_{1},\underline{x}_{2},...,\underline{x}_{N}\right)$ and generalized resultant forces $\underline{f}_{T}=\left(\underline{f}_{R,1},\underline{f}_{R,2},...,\underline{f}_{R,N}\right)$ of a system of $N$ particles, with $\underline{x}_{i}$ and $\underline{f}_{R,i}$ being the position and the resultant force associated with the $i$-th particle, it can be proved \cite{arnol2013mathematical} that 
\begin{equation}
	K_{T}(t_{f})-K_{T}(t_{i})=\sum_{i=1}^{N}K_{cl,i}(t_{f})-\sum_{i=1}^{N}K_{cl,i}(t_{i})=\int_{\underline{x}_{T}(t_{i})}^{\underline{x}_{T}(t_{f})}\underline{f}_{T}\cdot\underline{dx}_{T}
\end{equation}
\par An important class of problems in physics involves conservative systems, which make reference to another kind of energy, \emph{viz.} the potential energy.
\begin{definicoes}
	A system is called conservative if the forces depend only on the location of a point in the system $\left(\underline{f}_{T}=\underline{f}_{T}(\underline{x}_{T})\right)$ and if the work of $\underline{f}_{T}$ along any path depends only on the initial and final points of the path.
	\label{def:conservative}
\end{definicoes} 
\par Definition \ref{def:conservative} can not be overestimated: as will be seen, it gives a primitive intuition related with energy in a mechanical system. In particular, it reveals a formal link with the concept of potential energy\footnote{All theorems stated in this section were proved at Chapter 10 of \cite{arnol2013mathematical}.}.
\begin{teoremas}
	For a system to be conservative it is necessary and sufficient that a potential energy $U(\underline{x}_{T})$ exists such that \footnote{The derivatives associated with $\nabla_{\underline{x}_{T}}$ are considered with respect to the the motion $\underline{x}_{T}$, a notation that will be maintained throughout the present work.}
	\begin{equation}
	\underline{f}_{T}=-\nabla_{\underline{x}_{T}}U.\label{eq:potentialforce}
	\end{equation}
	\label{th:potentialenergy}
\end{teoremas}
\par Theorem \ref{th:potentialenergy} permits to define the potential energy $U$, which plays a prominent role in the total energy.
\begin{definicoes}
	If there are forces acting on the system such that there is a potential $U$ which satisfies Eq. \eqref{eq:potentialforce}, then the mechanical energy associated with this system is defined as
	\begin{equation}
		E=K_{cl}+U. \label{eq:totenergy}
	\end{equation}
\end{definicoes}
This definition does not clarify, on its own, the relevance of the notion of energy. This task is accomplished by the following result.
\begin{teoremas}
	The mechanical energy of a conservative system is preserved under the motion: $E(t_{f})=E(t_{i})$. 
	\label{th:conservation}
\end{teoremas}
\par Theorem \ref{th:conservation} states that, if only conservative forces are suppose to act on the system, then there is a measurable scalar property, called mechanical energy, that is preserved. The premise of the theorem, \emph{viz.} that the forces must be conservative, seems, at a first sight, too restrictive. To highlight the generality of the theorem another result should be recalled.
\begin{teoremas}
	If there are only forces of interaction between the particles composing a system, and they depend only on the distance between the particles, then the system is conservative. 
	\label{th:interactiondistanceforces}
\end{teoremas}
The forces treated in the scope of classical physics are mostly distance dependent (Coulombian and Gravitational, for instance). For interaction forces of such kind, the system is considered conservative. However, there are cases in which, though the interaction forces are only distance-dependent, there are external forces that yields a non-conservative character to the system. This is the case, for example, of drag forces acting on a rigid body or a thermal coupling with a reservoir. It is thus established the following definition, connected with the recognition of dissipative effects.
\begin{definicoes}
	A decrease $E(t_{i})-E(t_{f})$ in the mechanical energy is called an increase in the non-mechanical\footnote{The terminology "mechanical" and "non-mechanical" is justified from a purely \emph{historical} point of view, since one cannot always infer the precise nature of energy dissipation. Also, there are some conservative systems, in which the potential energy $U$ is due to a non-mechanical interaction (Coulombian, for example) and contributes to the mechanical energy (see Eq. \eqref{eq:totenergy}).} energy $E_{d}$: 
	\begin{equation}
		-\Delta E=E(t_{i})-E(t_{f})=E_{d}(t_{f})-E_{d}(t_{i})=\Delta E_{d}.
	\end{equation}
	\label{def:dissipative}
\end{definicoes}
\par The definition \ref{def:dissipative} expresses the idea that the total energy is conserved, i.e.,the sum of energies $E+E_{d}$ is conserved. It is generally considered that some mechanism extracts mechanical energy in the form $E_{d}$ of a non-conservative system. The statement that the energy will actually flow to another system and will not disappear, is intrinsically contained in the term $E_{d}$ and is a widely accepted physical \emph{principle}. As mentioned by Callen \cite{callen1985thermodynamics}, at the beginning of the chapter where he introduces internal energy, \emph{"The development of the principle of conservation of energy has been one of the most significant achievements in the evolution of physics. The present form of the principle was not discovered in one magnificent stroke of insight but has been slowly and laboriously developed over two and a half centuries. The first recognition of a conservation principle, by Leibnitz in 1693, referred only to the sum of the kinetic energy ($\frac{m v^{2}}{2}$) and the potential energy ($mgh$) of a simple mechanical mass point in the terrestrial gravitational field. As additional types of systems were considered, the established form of the conservation principle repeatedly failed, but in each case it was found possible to revive it by the addition of a new mathematical term - a 'new kind of energy'. Thus consideration of charged systems necessitated the addition of the Coulomb interaction energy ($\frac{Q_{1}Q_{2}}{r}$) and eventually of the energy of the electromagnetic field. In 1905 Einstein extended the principle to the relativistic region, adding such terms as the relativistic rest-mass energy. In the 1930’s Enrico Fermi postulated the existence of a new particle, called the neutrino, solely for the purpose of retaining the energy conservation principle in nuclear reactions. Contemporary research in nuclear physics seeks the form of interaction between nucleons within a nucleus in order that the conservation principle may be formulated explicitly at the sub-nuclear level. Despite the fact that unsolved problems of this type remain, the energy conservation principle is now accepted as one of the most fundamental, general, and significant principles of physical theory"}. The connection that will be established between Thermodynamics and mechanics, in order to define work, starts from the very intuition of energy provided in the previous discussion. In order to make such connection explicit, in what follows a simple example of mechanical system is discussed. 
\subsection{Sliding block with a spring}
\par Consider a system composed of a block attached to a table by a spring, as schematically depicted in Fig. \ref{fig:slidingnofriction}.
\begin{figure}[!h]
	\begin{center}
		\includegraphics[angle=0, scale=0.5]{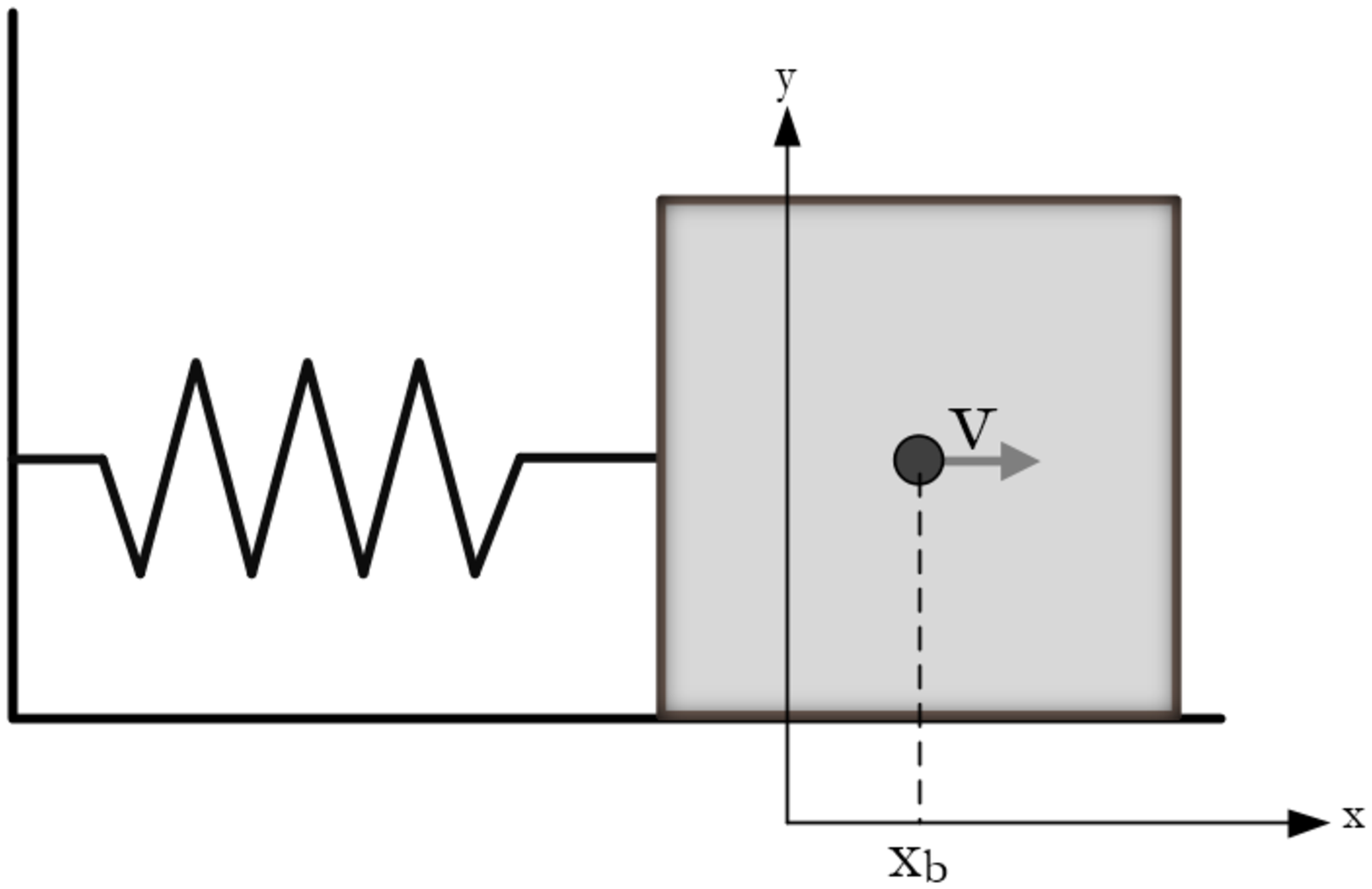}
		\caption{Sliding block with center of mass position $x_{b}$ and velocity $v$ on a frictionless table: the energy flows from the spring to the block and vice-versa.}
		\label{fig:slidingnofriction}
	\end{center}
\end{figure} 
The following aspects are assumed:
\begin{itemize}
	\item The spring has negligible mass; 
	\item The table and the block have no friction with each other;
	\item The spring potential is modeled as $U=\frac{k}{2} x_{b}^{2}$, where $x_{b}$ is the displacement of the center of mass of the block with respect to its equilibrium position;
	\item The table is considered as an inertial frame;
	\item The table's mass $m_{t}$ is very large compared to the block's $m_{b}$, that is, $\frac{m_{b}}{m_{t}}\ll 1$;
	\item The block is initially with speed $v_{0}>0$ at its equilibrium position ($x_{b}=0$).
\end{itemize}  
From these conditions and Hooke's law, it can be deduced that the movement of the center of mass of the block will be sinusoidal. Since the system is conservative, the mechanical energy of the spring-block system will be preserved. In terms of energy flow, the total mechanical energy of the block, initially kinetic, will flow to the spring until the motion reach its maximum amplitude, where all the energy will be stored as potential one. Then the spring potential energy will flow back to the block in a kinetic energy form, and this energetic cycle is repeated virtually infinite times. It can be noted that the problem was solved for the center of mass of the block. What happens with its constituents? An intrinsic hypothesis related is that the block is a rigid body and, although the particles have their own degrees of freedom in microscopic scale, they do not affect the center of mass motion significantly: these details becomes irrelevant for solving the problem. On the other hand, when dissipative or/and thermal effects are considered, the same is not necessarily true. To discuss this idea, we consider the same system as above but let the interface block-table be rough. The model with friction included is illustrated in Fig.  \ref{fig:slidingfriction}.
\begin{figure}[h]
	\begin{center}
		\begin{subfigure}[h]{0.4\textwidth}
			\includegraphics[angle=0, scale=0.3]{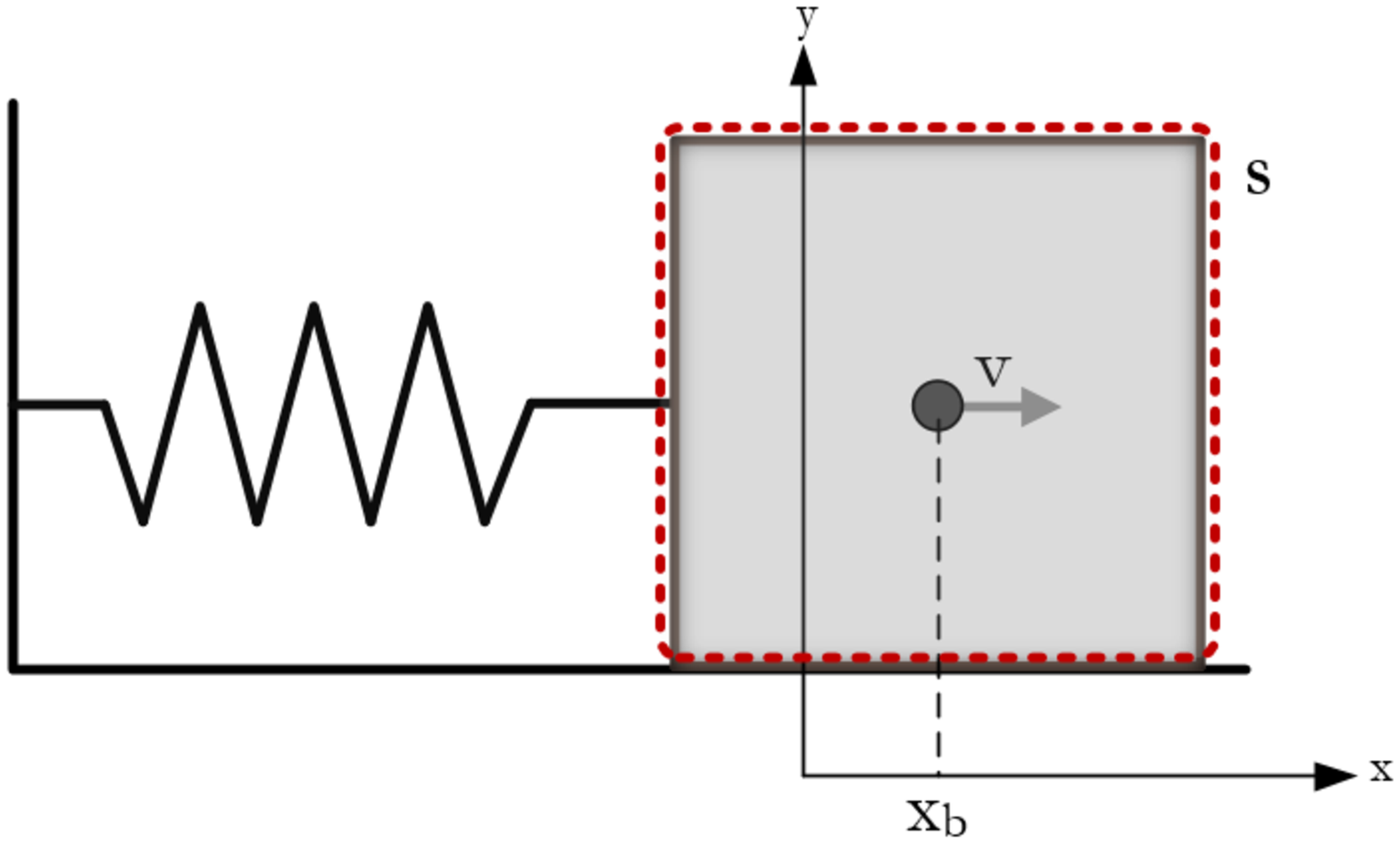} 
			\caption{S: block.}
			\label{fig:slidingfrictiona}
		\end{subfigure}
		\begin{subfigure}[h]{0.4\textwidth}
			\includegraphics[angle=0, scale=0.3]{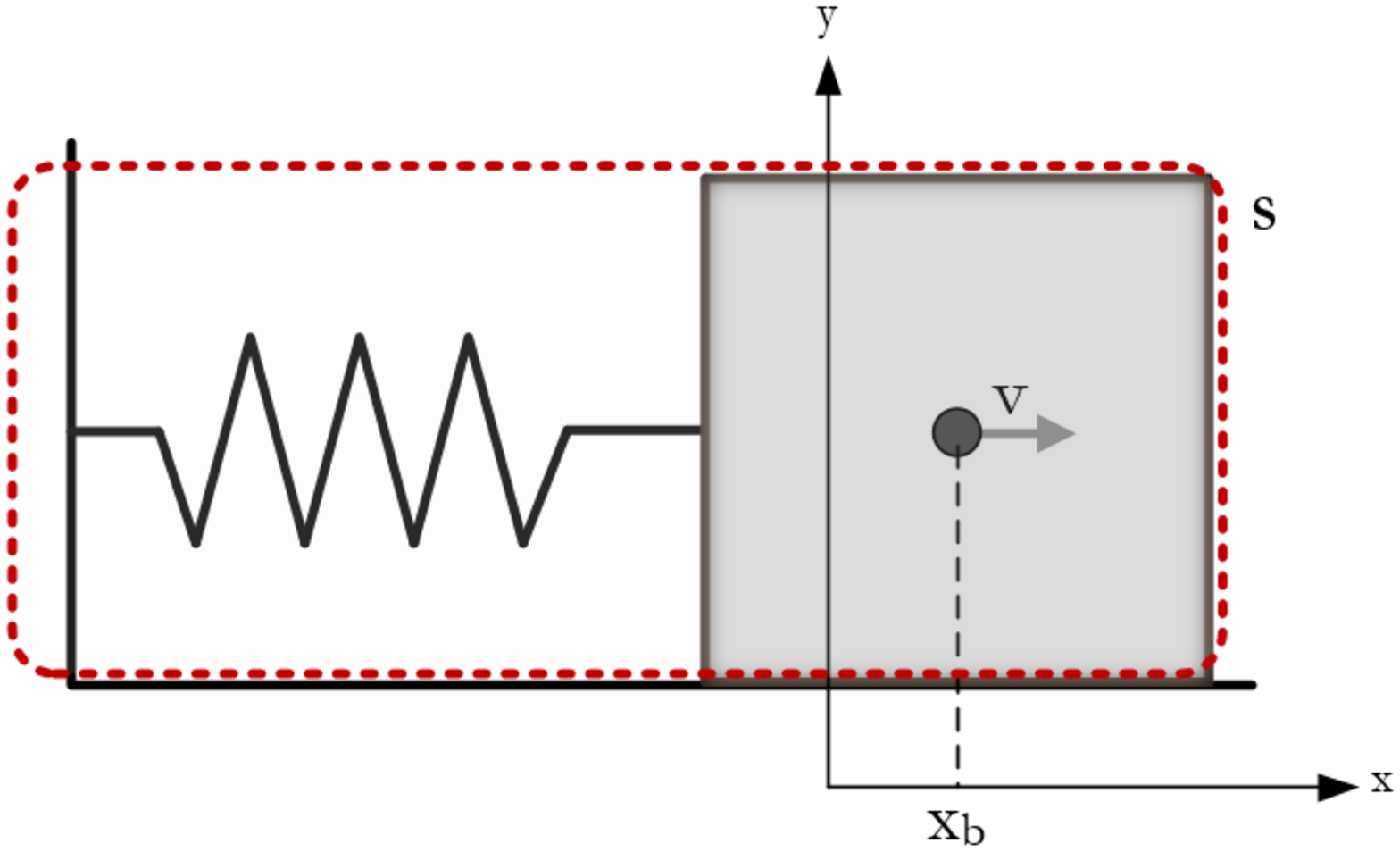}
			\caption{S: block+spring.}
			\label{fig:slidingfrictionb}
		\end{subfigure}
		\begin{subfigure}[h]{0.4\textwidth}
			\includegraphics[angle=0, scale=0.3]{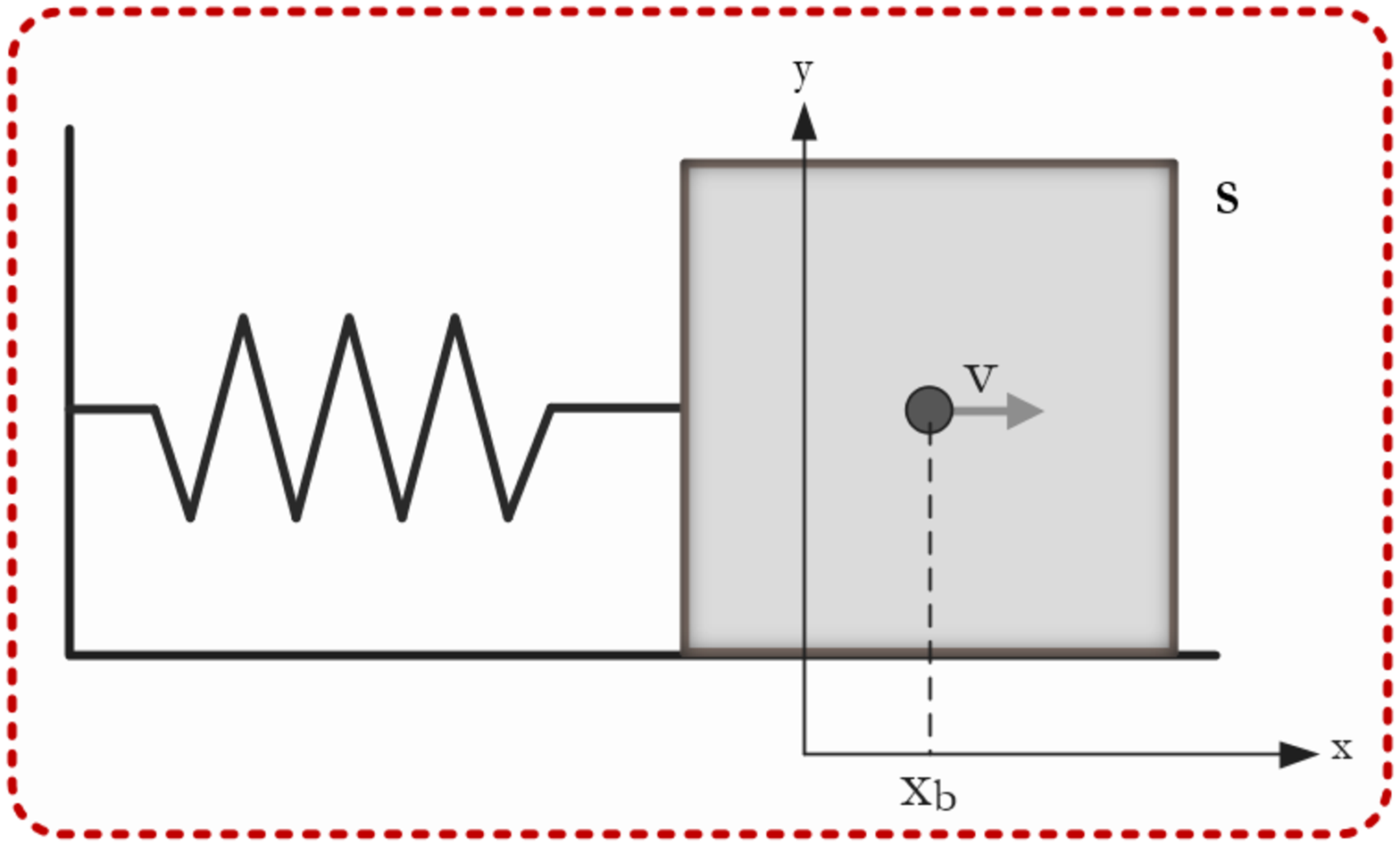}
			\caption{S: block+spring+table.}
			\label{fig:slidingfrictionc}
		\end{subfigure}
		\caption{Sliding block on table with equilibrium position fixed in $x=0$ and speed $v$, on a table with friction effects. The system (and therefore the internal energy) can be considered differently.}
		\label{fig:slidingfriction}
	\end{center}
\end{figure}
\par The spring-block system is not considered conservative due to the friction term: the mechanical energy decreases at an exponential rate\footnote{For a system in which the friction coefficient is considered constant.}. The kinetic energy of the block, which at the beginning of the motion equals to the total mechanical energy of the system, is converted into potential energy, stored in the spring, and posteriorly disappears due to friction\cite{kleppner2013introduction,david1988fundamentals}. What happened to the mechanical energy of the system? 
\par The definition of "system" ("S" represented in the figures) determines the form in which energy is qualified, whether it is internal energy or external work. Since the energy of the whole system (table, spring, block) is supposed to be conserved, then independently of what is assigned as the system "S", the changes in its internal energy $E_{\text{int}}$ must equal the work applied by external systems $\mathcal{W}_{\text{ext}}$ in it, that is
\begin{equation}
	\Delta E_{\text{int}}= \mathcal{W}_{\text{ext}}.\label{eq:intext}
\end{equation}
When the table and the spring are regarded as being outside the system "S", then all the internal energy is kinetic energy and can be described by a macroscopic degree of freedom $x_{b}$. It changes due to external work. This situation is depicted in Fig. \ref{fig:slidingfrictiona} and can be equated as
 \begin{equation}
 	\Delta E_{\text{int}}=\Delta \mathcal{W}_{\text{ext}}\Rightarrow \Delta K=\mathcal{W}_{\text{spring}}+\mathcal{W}_{\text{friction}},
 \end{equation}
where $\mathcal{W}_{\text{spring}}$ is the work done by the spring and $W_{\text{friction}}$ by friction forces. When only the table is outside the system, as in Fig. \ref{fig:slidingfrictionb}, the friction contributes with the external work, for which the description is already addressed by Classical Mechanics (although the friction force can be seen as an effective description).  In this case, it can be said that the friction "removes" internal energy of the system and
\begin{equation}
\Delta E_{\text{int}}=\Delta \mathcal{W}_{\text{ext}}\Rightarrow \Delta K+\Delta U_{\text{spring}}=\mathcal{W}_{\text{friction}},
\end{equation}
where the potential energy of the spring $U_{\text{spring}}=\frac{kx_{b}^{2}}{2}$ is now regarded as part of the internal energy. On the other hand, regarding the table as part of the system, represented in Fig. \ref{fig:slidingfrictionc}, results in
\begin{equation}
\Delta E_{\text{int}}=0\Rightarrow \Delta K+\Delta U_{\text{spring}}+\Delta E_{\text{friction}}=0\label{eq:wspring}
\end{equation}
and the friction effects are treated as internal energy variations $\Delta E_{\text{friction}}$, without a description in terms of potential energy, since the friction force is nonconservative. In other words, the internal energy now includes a nonconservative macroscopically inaccessible energy which is called frequently as thermal energy \cite{david1988fundamentals}. Then Eq. \eqref{eq:wspring} can be rewritten as
\begin{equation}
 \Delta E_{\text{int}}=0\Rightarrow \Delta K+\Delta U_{\text{spring}}+\Delta E_{\text{thermal}}=0\label{eq:wspring2}
\end{equation}
being $E_{\text{thermal}}$ the thermal energy. This sort of energy is frequently associated with a microscopic form of energy. 
\par The constituent particles of the system have movements that, due to operational limitations, cannot be accurately tracked. In fact, Classical Mechanics does not apply to the atomic scale (even if it did, one could not compute $10^{23}$ trajectories) and Quantum Mechanics not even accommodate the notion of trajectory. The thermodynamical perspective is therefore invoked, where a "coarseness process" (to be described later) is made so that the motion of each particle is considered to be random. This randomness is frequently associated with thermal energy \cite{callen1985thermodynamics}. Hence, the kinetic energy related to the center of mass of the block diminishes while their constituents' kinetic energy increases: this is effectively verified by measuring the temperature increase at the bottom of the block, for instance. In other words, the "organized" macroscopically-accessible kinetic energy carried by the block center of mass is pulverized into infinitely many "random" microscopically-untrackable tiny amounts of kinetic energy carried by point masses. This generic perspective will be imported to the discussion around the definition of work to be proposed in the chapter \ref{cap:definicaotrab}.  
\par In summary, what is generally called "dissipation" actually is, in classical mechanical terms, the conversion of the energy carried by a macroscopic degree of freedom into internal microscopic ones. For this apparently random dynamics, Classical Mechanics tells a story in terms of work and mechanical energy. Such interpretation is not restrictive to the system above and can be adopted for other systems. Thermodynamics approaches this microscopic scenario by means of a macroscopic mimic, as put by Callen \cite{callen1985thermodynamics}: \emph{"Thermodynamics, in contrast (to Mechanics and Electromagnetism), is concerned with the macroscopic consequences of the myriads of atomic coordinates that, by virtue of the coarseness of macroscopic observations, do not appear explicitly in a macroscopic description of a system"}. In the example treated in this section, the "macroscopic description" refers to the analysis of the kinetic energy of the block center of mass, the spring potential energy, and the lost of mechanical energy. The Thermodynamics (with a Statistical Physics theoretic background), on the other hand, provides a description of the microscopic, thermal, effects related to the system, in terms of macroscopic observable properties. 
 
\section{\textbf{Classical Thermodynamics}}
Thermodynamics provides good results for general macroscopic systems. Due to definitions of physical quantities like work, heat, temperature, entropy, and other thermodynamics concepts, it was possible to describe and invent a vast number of mechanisms, machines, engines that are applied in different industrial and home sectors \cite{moran2010fundamentals}. Due to the great achievements of the theory with regards to macroscopic phenomena, it is an aim for some scientists, engineers and mathematicians to describe what are the limits in which the theory may be applied: what is the minimum number of particles in which the above concepts may be used to thermodynamically describe a system? Is there a minimum volume (size) that the system must have for those concepts to be applied? These questions are currently treated at the scope of QT, which will be treated latter in this dissertation. However, to understand and to be able to answer those questions, it is essential to be able to understand in a broader manner the foundation of thermodynamics, that gives the concepts of heat, work, temperature, and others, physical meaning in order that prediction and understanding could be achieved. With this purpose in mind the main features related with the concept of work and heat are treated in this section. First, it is considered the equilibrium hypothesis and the intrinsic coarseness related with classical Thermodynamics

\subsection{Thermodynamical Equilibrium}
\par The equilibrium hypothesis is crucial to the development of classical Thermodynamics \cite{schwabl2006statistical,callen1985thermodynamics}. According to Callen, equilibrium can be stated as follows: \emph{"in all systems there is a tendency to evolve toward states in which the properties are determined by previously applied external influences. Such simple terminal states are, by definition, time independent. They are called equilibrium states"}. From the knowledge given by the atomic theory, the condition that all variables related to the constituents of a material are unchanged throughout time is virtually impossible to be achieved. However, the states mentioned in Callen's definition are not particle states. Rather, they are "macro states", those related with macroscopic properties of the system (as energy, linear momentum, among others), which are presumed not to significantly vary upon interaction with the external world (environment).  The success of this approach, which provides accurate predictions and satisfactory explanations for experiments, derives from the fact that macroscopic measurements are extremely slow and coarse, when compared with the atomic scales of time and length \cite{callen1985thermodynamics}. Therefore, although each constituent variable of a macroscopic system cannot be considered time-independent, the quantity resulting from averaging over all particles does not significantly vary with time, this meaning that significant fluctuations around mean values cannot be detected. It follows from this that the system as whole can be described by its essentially time-independent macroscopic properties, the so-called thermodynamic coordinates. As good candidates for such coordinates, one usually employs those variables subject to conservation principles, such as energy, linear and angular momentum. Since it is common to adopt the system center of mass as reference frame, the remaining significant quantities turn out to be energy, volume, and mole numbers. These ideas are summed up in the following Postulate \cite{callen1985thermodynamics}.
\begin{postuladothermo}
	There exist particular states (called equilibrium states) of simple systems that, macroscopically, are characterized completely by the internal energy $U_{i}$, the volume $V$, and the mole numbers $N_{1}$, $N_{2}$,...$N_{r}$ of the chemical components.
	\label{postthermo1}
\end{postuladothermo}
\par At first sight, one may consider that the microscopic variables of motion do not contribute effectively for thermodynamical properties, since one measures essentially macroscopic properties. However, as it will be shown in the Statistical Physics scope, the very concept of temperature (and pressure as well) has a tight relation with microscopic configuration of motion. However such a connection is established via random motion considerations. 
\par As one may verify from the Postulate \ref{postthermo1}, the definition of internal energy is fundamental for the description of the states considered in Thermodynamics. As it will be seen, such property can be defined by path-independent function written in terms of thermodynamic coordinates which the system passes through. To make this point, we firstly need to discuss energy measurements.
\subsection{Internal energy, work and heat}
\label{subsec:measureenergy}
\par How can the internal energy of a macroscopic system be measured? The general form $\Delta E_{\text{int}}=\mathcal{W}_{\text{ext}}$ of Eq. \eqref{eq:intext} answers this question: by tracking a macroscopic degree of freedom which is able to reveal the external work imparted on the system\footnote{Measuring $\mathcal{W}_{\text{ext}}$ gives the change in the internal energy, not an absolute value $E_{\text{int}}$}.
\par As a concrete instance, consider a container within which the inner constituents remain ideally isolated from the surroundings (adiabatic walls). Assume that some kind of work is done on the system  (for instance, by moving a wall or applying a torque in a paddle shaft). Since the system is thermally isolated, it is expected that the work done be converted into some other kind of energy, internal to the system. Therefore, if one assumes a value of internal energy $U_{fi}$ to a specific state, called the energy of fiducial state, one can measure the energy of the same system after some work $\mathcal{W}$ has been done, $U_{i}=U_{fi}+\mathcal{W}$, where $U_{i}$ is then called the thermodynamic internal energy of the system. As a fundamental idea, necessary for this measurement, is that the walls are adiabatic. If such imposition were not made, it would not be possible to established directly the connection of work with the variation of energy. Also, notice that the hypothesis of energy conservation had to be implicitly considered. Of course, the accuracy of the results depends on how precisely the hypotheses are satisfied, for there always is some leak of heat in realistic walls. 
\par It is worth noticing that energy considerations provide the sight of a fundamental distinction between Classical Mechanics and Thermodynamics. If, by adopting a Classical Mechanics perspective, one exclusively looked at the energy transported by the center of mass of the system, then one would wrongly state that the energy that the system received via external work would had been lost. In contrast, by admitting that the energy flowed to the interior, one can assert that the energy has just been transformed into unsearchable degrees of freedom. This is, in fact, assumed by Thermodynamics from the very beginning; it is not part of its concerns to furnish a detailed microscopic view of the universe. 
\par Another important aspect is related with the fact that the system was assumed to be "isolated", that is, no heat flux was allowed. But how can it be guaranteed that the mechanism that provides work does not provide some sort of heat? Presumably, the aforementioned experiment would lead to different results if the paddle that stirs the system were at a temperature much higher than that of the system. How can heat be distinguished from work? 
\par The distinction operationally emerges in the statement of the \emph{first law of Thermodynamics}: \emph{"the heat flux to a system in any process (at constant mole numbers) is simply the difference in internal energy between the final and initial states, diminished by the work done in that process"}. Mathematically,
\begin{equation}
	\delta \mathcal{Q}=dU_{i}-\delta \mathcal{W},\label{eq:firstlaw}
\end{equation}
where $\delta \mathcal{Q}$ and $\delta \mathcal{W}$ are infinitesimal amounts of energy under the form of heat and work, respectively, whereas $dU_{i}$ accounts for the resulting change in the internal energy. Here the symbol $\delta$ denotes an inexact differential, whose meaning is object of a vast discussion in the literature \cite{schwabl2006statistical,callen1985thermodynamics,chandler1987introduction,pathria1996}. One might argue, at a first sight, that the first law provides a definition for heat. However, this position cannot be maintained because the very notion of internal energy, as constructed above, demands the absence of heat transfer, which makes the argument be logically cyclic.
\par Despite these conceptual difficulties, it is still possible to obtain some information about heat in practical situations by looking at changes in the temperature of the system. This, however, does not yield a clear mechanical picture or general definition for heat. One could try a better understanding by considering the classical theory of heat transfer, where conduction, convection, and irradiation appears as basic forms of heat transfer \cite{incropera1999fundamentos}. However, once again one may ask for the basic mechanisms behind these  phenomena, in which case no trivial answer comes as well. 
\par An irradiation process is considered, for instance. From a macroscopic perspective, as radiation enters the system no net work occurs;  in this case, it is commonly said that heat enters the system. On the other hand, from a microscopic viewpoint, an atom receives both energy and momentum, with two important consequences. One, the center of mass of the atom gets a kickback (its kinetic energy changes), which can be thought of as deriving from work. Two, the volume of the atom increases, which in a semiclassical description is explained by the excitation of the atom (see Fig. \ref{fig:photonabsor} for the Bohr description of a hydrogen atom, where an electron jumps to a larger orbit upon absorption of a photon).  If this volume increase is understood to be due to work, as is usually done at macroscopic level, then it follows that, from a microscopic perspective, irradiation can be accounted solely in terms of work. It is fair then to ask whether heat can actually be define under an atomic prescription. This discussion will be recovered in the QT context, in the next chapter.
\begin{figure}[!h]
	\begin{center}
		\includegraphics[angle=0, scale=0.45]{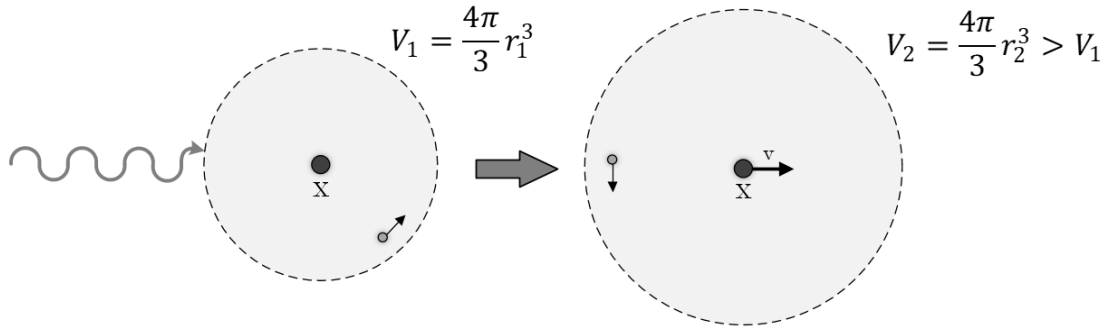}
		\caption{Schematic representation of the absorption of a photon by a hydrogen atom according to Bohr's model. Before the absorption, the mean volume of the atom is $V_{1}=\frac{4\pi r_{1}^{3}}{3}$. Upon absorption, the electron jumps to an outer orbit, which makes the atom volume increase to $V_{2}=\frac{4\pi r_{2}^{3}}{3}$. }
		\label{fig:photonabsor}
	\end{center}
\end{figure}
\par A most explicit distinction of work and heat is made qualitatively, in the scope of classical Thermodynamics. As Callen \cite{callen1985thermodynamics} writes \emph{"But it is equally possible to transfer energy via the hidden atomic modes of motion as well as via those that happen to be macroscopically observable. An energy transfer via the hidden atomic modes is called heat"}. From such perspective, there is an unknowable or unmeasurable atomic motion, that can transfer energy from a system to another. Such form of energy transfer is called heat. It can be concluded, therefore, that what is called heat is nothing but energy transfer mediated by microscopic work performed at scales inaccessible to macroscopic monitoring. In the following subsections, some aspects related to Thermodynamics general structure are succinctly treated for the sake of completeness only. 
\subsection{Thermodynamics problem, entropy and temperature }
\par The problem that Thermodynamics is concerned with can be summarized in the following form \cite{callen1985thermodynamics}: \emph{"the single, all-encompassing problem of Thermodynamics is the determination of the equilibrium state that eventually results after the removal of internal constraints in a closed, composite system"}. It can be regarded the example of a container with two gases, divided by an adiabatic, impermeable wall in equilibrium. Each one of the parts of the system is considered as being at equilibrium and there is no \emph{internal} constraint in each part of the container. However, the two gases are separated by an impermeable adiabatic wall. After the removal of the wall, the system with the two gases will be now considered as the new closed system. The Thermodynamics problem is to determine which will be the equilibrium state to be achieved terminally. In order to determine, from all the possible equilibrium states, what is the terminal one, it is imposed the following Postulates.
\begin{postuladothermo}
	"There exists a function (called the entropy $S$) of the extensive parameters of any composite system, defined for all equilibrium states and having the following property: the values assumed by the extensive parameters in the absence of an internal constraint are those that maximize the entropy over the manifold of constrained equilibrium states".\label{post:entropy}
\end{postuladothermo}
\begin{postuladothermo}
	"The entropy of a classical composite system is additive over the constituent subsystems. The entropy is continuous and differentiable and is a monotonically increasing function of the energy".
\end{postuladothermo}
\par It is important to remark that this does not necessarily imply that Thermodynamics does not involve out of equilibrium processes, but that only terminal, equilibrium states will be analyzed in Thermodynamics.
\par Within these Postulates, the definition of temperature emerges: by maximizing the entropy of a two component ($1$ and $2$) system, in which each component interacts with each other by transferring energy, there is a quantity $\left.\frac{\partial U_{i}}{\partial S}\right|_{V,N_{1},...,N_{r}}$ that equals for both system, i.e. $\left.\frac{\partial U_{i}^{1}}{\partial S_{1}}\right|_{V_{1},N_{1}^{1},...,N_{r}^{1}}=\left.\frac{\partial U_{i}^{2}}{\partial S_{2}}\right|_{V_{2},N_{1}^{2},...,N_{r}^{2}}$. Such definition is a manifestation of the sometimes called zeroth law of Thermodynamics, which stipulates \emph{"the existence of a common parameter for two or more physical systems in mutual equilibrium"} \cite{pathria1996}. This maximizing process is also made in the scope of Statistical Physics and therefore is a result that intersects both Thermodynamics and Statistical Physics. 
\par Finally, one has the following Postulate:
\begin{postuladothermo}
	"The entropy of any classical system vanishes in the state for which $\left.\frac{\partial U_{i}}{\partial S}\right|_{V,N_{1},...,N_{r}}=0$".
\end{postuladothermo}
\par The Postulate is considered as an extension of the so-called Nernst Postulate \cite{callen1985thermodynamics}(or theorem \cite{schwabl2006statistical,schrodinger1989statistical}) or third law. Since there is no mention of the above Postulate in the rest of the present text, it is no further commented about the third law. In the sequence, it is treated the quasi-static considerations and its importance to the establishment of Thermodynamics.
\subsection{Second law}
\par The second law can be stated mainly in three forms: the Clausius definition, the Kelvin one and another based on a mathematical inequality \cite{schwabl2006statistical,callen1985thermodynamics,moran2010fundamentals}. Since they are all equivalent, only the mathematical inequality will be treated here. For a system under any thermodynamic process it holds that 
\begin{equation}
TdS\geq \delta \mathcal{Q}.
\end{equation}
A necessary (but not sufficient) condition for the equality to occur is that the system must pass through a quasi-static process\footnote{Accordingly with Schwabl \cite{schwabl2006statistical}, \emph{"a quasistatic process takes place slowly with respect to the characteristic relaxation time of the system, i.e. the time within which process is the one which the system passes from a non equilibrium state to an equilibrium state, so that the system remains in equilibrium at each moment during such a process."}}. In cases where the process is not quasi-static \emph{"turbulent flows and temperature fluctuations take place, leading to the irreversible production of heat"}\cite{schwabl2006statistical}. A sufficient condition for the equality to hold is that the system must be under a quasi-static reversible process.  For such processes, heat is directly  connected with entropy: $TdS=\delta \mathcal{Q}$. As will be seen in the Statistical Physics treatment of entropy, entropy is linked with the number of microstates which the system can access. Therefore, in this case, the heat is connected with the microscopic properties of the material. The work, on the other hand, is frequently associated with macroscopic measurable properties. In fact, according with Chandler \cite{chandler1987introduction} \emph{"the work term has the general form $\delta \mathcal{W}=\underline{f}\cdot \underline{dX}$, where $\underline{f}$ is the applied 'force', and $\underline{X}$ stands for a mechanical extensive variable"}. The extensibility can be related with the size of the system, which is a macroscopically inferable notion. From such statements, it can be concluded that heat is connected with microscopic properties and work with macroscopic ones. This recovers to the discussion made on subsection \ref{subsec:measureenergy}.
\subsection{Summary of Thermodynamics' scenario}
\par The main aspects treated above can be put as follows:
\begin{itemize}
	\item \emph{Equilibrium states}: terminal time independent states toward which a system evolves, in which the properties are determined by previously applied external influences; depend on the macroscopically measurable variables of the system such as internal energy, number of constituents, volume and others;
	\item \emph{Entropy}:  an additive function defined for all equilibrium states whose maximization determines the equilibrium state;
	\item \emph{Thermal equilibrium:} $\left.\frac{\partial U_{i}^{1}}{\partial S_{1}}\right|_{V_{1},N_{1}^{1},...,N_{r}^{1}}=\left.\frac{\partial U_{i}^{2}}{\partial S_{2}}\right|_{V_{2},N_{1}^{2},...,N_{r}^{2}}\rightarrow T_{1}=T_{2}$;
	\item \emph{Heat transfer}: qualitatively, an energy transfer via hidden atomic modes;
	\item \emph{Work transfer}: An energy transfer of the form $\underline{f}\cdot \underline{dX}$, where $\underline{f}$ is the applied "force", and $\underline{X}$ stands for a mechanical extensive variable; defined in agreement with other branches of physics, like Classical Mechanics, Electrodynamics, among others;
	\item \emph{First law of Thermodynamics:} $\delta \mathcal{Q}=dU_{i}-\delta \mathcal{W}$;
	\item \emph{Second law of Thermodynamics:} $TdS\geq \delta \mathcal{Q}$;
	\item \emph{Third law of Thermodynamics:} $S\rightarrow 0$ as $T\rightarrow 0$.
\end{itemize}
Thermodynamics can be viewed, in essence, as an empirical theory. Interestingly, it is nevertheless possible to support it with mechanical principles supplemented with statistical aspects.
\section{\textbf{Statistical Physics}}
\par A macroscopic system is composed of many internal degrees of freedom. In the case of standard examples within the scope of Thermodynamics, the number of degrees of freedom is expected to have the order of magnitude of the Avogadro number. However, it is currently impossible to solve all equations of motion of a system with so many degrees of freedom. A different perspective is the statistical one: infinitely many copies of the system are considered (ensemble) that obey some restriction, as for instance a fixed value of energy or number of constituents. 
\par The statistical description is useful not only in cases where one cannot compute all the equations; it can be used also in cases where there is not a precise knowledge of some conditions around the problem. For example, in Classical Mechanics, from the equations of motion, one may determine the position and momentum and other properties in every instant of time, given that the initial condition is known. However, from an experimental point of view, the initial conditions have an irreducible operational uncertainty, due to the apparatus imprecision or some external influence (the measurer, for instance). Therefore one may, for example, determine the \emph{mean values} and \emph{variance} associated with those initial conditions. Such errors naturally limit the predictions for such experiments. In particular, such uncertainty will propagate in time and, therefore, blur the knowledge about the future of the system. That does not mean necessarily that Classical Mechanics is not precise, but only that the initial conditions are not determined with full precision. In order to work around such problem, one may consider a distribution associated with the initial condition. From such distribution, one may compute expectation values and variances at all times. In taking into account such experimental uncertainties in the initial conditions (which are critical for chaotic systems), this approach provides a "fairer" (though more limited) predictive power to Classical Mechanics. A fundamental tool in this scenario is the establishment of an "equation of motion" for the probability distribution itself. This point is discussed next.
\subsection{Liouville Theorem}
\par A phase space is defined as the space spanned by the $6N$ variables (coordinates and momenta) of a classical system of $N$ particles \cite{schwabl2006statistical}. Therefore, the various states of a system, may be mathematically expressed by points in this space \cite{landau1968statistical}. The Liouville theorem \cite{schwabl2006statistical,tolman1938principles,de2011topicos} states that any statistical distribution in the phase state, of a closed system subjected to the Hamilton equations of motion, has the property:
\begin{equation}
	\frac{d \rho_{cl}}{dt}=0.
\end{equation} 
Note that this does not imply that $\frac{\partial \rho_{cl}}{\partial t}=0$. This theorem has the following consequences:
\begin{itemize}
	\item The volume elements in phase space do not change, that is $d\underline{q}(t)d\underline{p}(t)=d\underline{q}(0)d\underline{p}(0)$\footnote{Here  $d\underline{q}(t)=dq_{x}^{1}dq_{y}^{1}dq_{z}^{1}dq_{x}^{2}...dq_{z}^{N}$ and $d\underline{p}(t)=dp_{x}^{1}dp_{y}^{1}...dp_{z}^{N}$ is the differential volume of the phase space of a system of $N$ particles.};
	\item $\rho_{cl}(\underline{q}(\underline{q}_{0},\underline{p}_{0},t),\underline{p}(\underline{q}_{0},\underline{p}_{0},t),t)=\rho_{cl}(\underline{q}_{0},\underline{p}_{0},0)$, i.e., the Hamiltonian flow of trajectories implies that the local density will not change from the point $\left(\underline{q}_{0},\underline{p}_{0}\right)$ to $\left(\underline{q}(\underline{q}_{0} ,\underline{p}_{0},t),\right.$ $\left.\underline{p}(\underline{q}_{0},\underline{p}_{0},t)\right)$;
	\item By considering the Poisson bracket notation\footnote{$\left\{u,v\right\}\equiv \sum_{i=1} \left[\frac{\partial u}{\partial p_{i}}\frac{\partial v}{\partial q_{i}}-\frac{\partial u}{\partial q_{i}}\frac{\partial v}{\partial p_{i}}\right]$.}, it can be proved that 
	\begin{equation}
		\frac{\partial \rho}{\partial t}=-\left\{\rho,H_{cl}\right\},\label{eq:poissonliou}
	\end{equation} 
	where $H_{cl}$ is the Hamiltonian that describes the system dynamics.
\end{itemize}
This approach allows for a statistical treatment of scenarios involving subjective ignorance about initial conditions, where each phase-space point follows a deterministic trajectory while $\rho$ introduces a probability distribution for these points.
\subsection{Subjective ignorance in Classical Mechanics}
\label{subsec:subjective}
\par A unidimensional system is considered, for the sake of simplicity. Although the system is treated under the scope of Classical Mechanics, it is assumed that there can be a subjective ignorance, that is, the observer cannot determine with certainty the position and momentum at the beginning of a test made on the system. This sort of ignorance may occur because the experiment apparatus does not provide a full precision for the measurement, for instance. It is fair to suppose, however, that the observer can determine the initial position and momentum mean values, by making measurements many times on copies of the same system. The observer can thus obtain 
\begin{equation}
\langle q_{0}\rangle= q_{0}^{'}\,\,\,\,\,\,\,\,\,\, \text{and}\,\,\,\,\,\,\,\,\,\,\langle p_{0}\rangle= p_{0}^{'}.
\end{equation}
as mean values.
Also, the observer can determine the initial variance associated with momentum and position:
\begin{equation}
\left\langle\left(q_{0}-q_{0}^{'}\right)^{2}\right\rangle = \sigma_{q,0}^{2}\,\,\,\,\,\,\,\,\,\, \text{and}\,\,\,\,\,\,\,\,\,\, \left\langle\left(p_{0}-p_{0}^{'}\right)^{2}\right\rangle= \sigma_{p,0}^{2}.
\end{equation} 
A statistical distribution that satisfies the equations above, commonly used in statistical analysis \cite{szczepinski2002error}, is the classical Gaussian distribution $\rho_{gcl}(q_{0},p_{0},0)$, defined as
\begin{equation}
\rho_{gcl}(q_{0},p_{0},0)=\frac{1}{2\pi \sigma_{p,0}\sigma_{q,0}}\mathrm{e}^{-\frac{1}{2}\left(\left(\frac{q_{0}-q_{0}^{'}}{\sigma_{q,0}}\right)^{2}+\left(\frac{p_{0}-p_{0}^{'}}{\sigma_{p,0}}\right)^{2}\right)}\equiv \mathcal{G}\left[q_{0}^{'},p_{0}^{'},\sigma_{q,0},\sigma_{p,0}\right].
\end{equation}
where the notation $\mathcal{G}\left[ q_{0}^{'},p_{0}^{'},\sigma_{q,0},\sigma_{p,0}\right]$ stands for the Gaussian distribution centered in $(q_{0}^{'},p_{0}^{'})$, with respective variances $\sigma_{q,0}$ and $\sigma_{p,0}$. 
For this distribution one can check that
\begin{align}
&\int_{-\infty}^{\infty}\int_{-\infty}^{\infty}\rho_{gcl}(q_{0},p_{0},0)dq_{0}dp_{0}=1,\\
&\int_{-\infty}^{\infty}\int_{-\infty}^{\infty}\rho_{gcl}(q_{0},p_{0},0)p_{0}^{2}dq_{0}dp_{0}=p_{0}^{'2}+\sigma_{p,0}^{2},\\
&\int_{-\infty}^{\infty}\int_{-\infty}^{\infty}\rho_{gcl}(q_{0},p_{0},0)q_{0}^{2}dq_{0}dp_{0}=q_{0}^{'2}+\sigma_{q,0}^{2}.
\end{align}
For future proposes, it is also interesting to consider a scenario where the observer has a relatively high certainty about two distinct phase-space points. This can be treated by statistical mixture of two Gaussian distributions, with different symmetric (for the sake of simplicity) centers $(q_{0}^{'},p_{0}^{'})$ and $(-q_{0}^{'},-p_{0}^{'})$. Considering the same variance for both distributions, i.e., $\sigma_{q,0}^{2}$ and $\sigma_{p,0}^{2}$, the mixed distribution is, therefore,
\begin{equation}
\rho_{mgcl}(t,q_{0},p_{0})=\frac{1}{2}\left(\mathcal{G}\left[ q_{0}^{'},p_{0}^{'},\sigma_{q,0},\sigma_{p,0}\right]+\mathcal{G}\left[ -q_{0}^{'},-p_{0}^{'},\sigma_{q,0},\sigma_{p,0}\right]\right)
\end{equation}
such that
\begin{align}
&\int_{-\infty}^{\infty}\int_{-\infty}^{\infty}\rho_{mgcl}(t,q_{0},p_{0})dq_{0}dp_{0}=1,\\
&\int_{-\infty}^{\infty}\int_{-\infty}^{\infty}\rho_{mgcl}(t,q_{0},p_{0})p_{0}dq_{0}dp_{0}=0,\\
&\int_{-\infty}^{\infty}\int_{-\infty}^{\infty}\rho_{mgcl}(t,q_{0},p_{0})p_{0}^{2}dq_{0}dp_{0}=p_{0}^{'2}+\sigma_{p,0}^{2},\\
&\int_{-\infty}^{\infty}\int_{-\infty}^{\infty}\rho_{mgcl}(t,q_{0},p_{0})q_{0}dq_{0}dp_{0}=0,\\
&\int_{-\infty}^{\infty}\int_{-\infty}^{\infty}\rho_{mgcl}(t,q_{0},p_{0})q_{0}^{2}dq_{0}dp_{0}=q_{0}^{'2}+\sigma_{q,0}^{2},\\
&\int_{-\infty}^{\infty}\int_{-\infty}^{\infty}\rho_{mgcl}(t,q_{0},p_{0})q_{0}p_{0}dq_{0}dp_{0}=q_{0}^{'}p_{0}^{'}.
\end{align}
\par By solving the Hamilton equations of motion one may obtain time-dependent functions for whatever dynamic variables, such as position $q(t)\equiv q(q_{0},p_{0},t)$, momentum, $p(t)\equiv p(q_{0},p_{0},t)$, velocity $\dot{q}(t)\equiv \dot{q}(q_{0},p_{0},t)$, momentum derivative $\dot{p}(t)\equiv \dot{p}(q_{0},p_{0},t)$, among others. Considering, a functional dependent on the functions related with the dynamical variables, say $f(t)\equiv f(q_{0},p_{0},t)\equiv f(q(t),p(t),\dot{q}(t),\dot{p}(t)...)$, it follows, from the Liouville theorem, that
\begin{equation}
\begin{array}{rl}
\displaystyle\langle f(t)\rangle&\displaystyle=\int_{-\infty}^{\infty}\int_{-\infty}^{\infty}\underbrace{\rho_{cl}(q(t),p(t),t)}_{\rho_{gcl}(q_{0},p_{0},0)}\underbrace{f(t)}_{f(q_{0},p_{0},t)}\underbrace{dq(t)dp(t)}_{dq_{0}dp_{0}}\\
&\displaystyle=\int_{-\infty}^{\infty}\int_{-\infty}^{\infty}\rho_{gcl}(q_{0},p_{0},0)f(q_{0},p_{0},t)dq_{0}dp_{0}.
\end{array}
\end{equation}
Therefore, by solving the equations of motion for an arbitrary initial condition $\left(q_{0},p_{0}\right)$, determining the expression $f(q_{0},p_{0},t)$ and integrating in the initial phase-space volume with the initial distribution is the same as determining $\rho_{cl}(q(t),p(t),t)$\footnote{Here $\rho_{cl}$ can assume $\rho_{gcl}$ for the Gaussian distribution or $\rho_{mgcl}$ for the mixed distribution.} and integrating over the phase-space volume at time $t$. This identity will be used throughout the work. 
\subsection{Statistical approach to Thermodynamics}
\label{subsec:stathermo}
\par The statistical description of a large system connects its microscopic effects with macroscopic results, as already mentioned. Fundamental for the development of this link, are the definitions of macrostate and microstate. A microstate in a Classical Mechanics system is a point in the phase space and, therefore, is defined by all the coordinates and momenta. Macroestate is the one defined by Thermodynamics, being characterized by few macroscopic variables (energy, volume and others\footnote{See Postulate \ref{postthermo1}.}). The definition of ensemble explicitly makes the connection between the microstates and the macrostate of a macroscopic system: an statistical ensemble is the collection of all the microstates which can represent the same macrostate, weighted by their frequency of occurrence. The equilibrium ensembles constitute a class commonly used in the scope of Thermodynamics. In the following, the main hypothesis considered in Statistical Physics textbooks in association with equilibrium ensembles are summarized.
\begin{itemize}
	\item \emph{Statistical equilibrium}: Some authors \cite{schwabl2006statistical,reif2009fundamentals} consider that $\frac{\partial \rho }{\partial t}=0$ defines the notion of statistical equilibrium. Reif \cite{reif2009fundamentals} states that it also characterizes the thermodynamical equilibrium. Landau and Lifshitz \cite{landau1968statistical} states that a macroscopic system is in statistical, thermodynamic, equilibrium if any macroscopic subsystem\footnote{A macroscopic part of the system that is small when compared with the whole system.} of it has its physical quantities  to a high degree of accuracy equal to their mean values.
	\item \emph{Equiprobability a priori and micro-canonical ensemble}: The equiprobabilty hypothesis treated by Tolman \cite{tolman1938principles} asserts that the microstates compatible with the macrostate of a physical system, when in thermodynamic equilibrium, are equiprobable, i.e. have the same weight. Such hypothesis is explicitly considered in order to define the micro-canonical ensemble. Consider that the energy of  system is assured to lie in the interval $\left[E,E+\Delta_{E}\right]$. Then the micro-canonical ensemble is composed of equally weighted microstates that lie in the region of the phase space such that the Hamiltonian of the system $H_{cl}(\underline{q},\underline{p})$ satisfies the condition $E\leq H_{cl}(\underline{q},\underline{p})\leq E+\Delta_{E}$. Such region is called by some authors as the \emph{energy shell}.  
	\item \emph{Ergodic hypothesis}: The (sometimes called quasi-)Ergodic hypothesis states that the point in the phase space corresponding to a system in equilibrium, will eventually pass arbitrarily close to any given microstate that belongs to the respective equilibrium ensemble \cite{gemmer20044,arnol1968ergodic,lebowitz1973modern}. Therefore, for sufficiently long times as compared with relevant microscopic time scales of the system, time averaging is equivalent to averaging in the statistical ensemble. Such a hypothesis is considered also by Landau and Lifshitz indirectly\cite{landau1968statistical}.  
	\item \emph{Entropy}: It can assume the form \cite{schwabl2006statistical}: $S=-k_{B}\int_{\Gamma}\frac{d\underline{p}d\underline{q}}{\hbar^{6N}}\rho_{cl}\left(\underline{p},\underline{q}\right)\ln\left[\rho_{cl}\left(\underline{p},\underline{q}\right) \right]$, where $k_{B}$ is frequently considered as the Boltzmann constant. It has more than one interpretation. In the case of the micro-canonical ensemble, it assumes the Boltzmann proposition for entropy $S=k_{B}\ln\left[\hat{\Omega}(E)\right]$, where $\hat{\Omega}(E)$ is the so-called energy shell volume. As mentioned by Landau and Lifshitz, and from the Thermodynamics Postulate \ref{post:entropy}, the direction of a process will be dictated by the maximization of the entropy, i.e. the process will tend to an equilibrium macrostate, at which the entropy is maximum. 
	\item \emph{Temperature}: It is frequently obtained considering instances involving an isolated system, with total energy E, with two partitions of fixed volumes that can transfer energy to each other and the interaction energy is small compared with the inner energy of each partition. In the most probable configuration \cite{schwabl2006statistical}, in which the energy of each partition is represented by $\hat{E}_{1}$ and $E-\hat{E}_{1}$ , one finds that $\left.\frac{\partial S_{b,1}(E_{1})}{\partial E_{1}}\right|_{\hat{E}_{1}}=\left.\frac{\partial S_{b,2}(E_{2})}{\partial E_{2}}\right|_{E-\hat{E}_{1}}$. Therefore, from the conditions of thermal equilibrium between two systems (zeroth law), the temperature emerges as $T^{-1}=\frac{\partial S}{\partial E}$. Note that such definition is made under the assumption that the system as a whole (composed by the two partitions) is in equilibrium.
	\item \emph{Canonical Ensemble}: The canonical ensemble is sometimes of simpler treatment and computation and, therefore, is frequently used for a system that energetically interacts with another system, called reservoir, at temperature $T$. The distribution assumes the form $\rho_{c}=\frac{\mathrm{e}^{-\frac{H_{cl}}{kT}}}{\mathcal{Z}}$ where $\mathcal{Z}$ is the partition function $\mathcal{Z}=\int_{\underline{q},\underline{p}}\mathrm{e}^{-\frac{H_{cl}}{kT}}\frac{d\underline{p}d\underline{q}}{\hbar^{6N}}$ and $H_{cl}$ is the Hamiltonian of the system. 
\end{itemize} 
\par In the context of QT problems, one aims at determining the dynamics and the thermodynamical properties in atomic time-space regime. However, an important open question for physics is the length- and time-scale domain to which thermodynamics laws still apply. This ambition is believed by some authors to be accomplished by studies in the QT scope. 
\par To finalize this brief review of Classical Thermodynamics and Statistical Physics, in what follows it is presented the Statistical Physics accounts for work and heat. 
\subsection{External parameters in Thermodynamics systems}
\label{subsec:externalparam}
\par  The present discussion almost entirelly follows the textbooks of Reif and Schwabl \cite{schwabl2006statistical,reif2009fundamentals}. Reif defines external parameters as \emph{"some macroscopically measurable independent parameters which are known to affect the equations of motion (i.e., appear in the Hamiltonian) of this system"}. By considering this definition he then distinguishes two types of interaction. First, for the so-called "purely thermal interaction", the external parameters remain fixed. Reif then defines heat as \emph{"the mean energy transferred from one system to the other as a result of purely thermal interaction"}, i.e., the energy transfer where the macroscopically measurable parameters, associated with the Hamiltonian, does not change. On the other hand, the author defines that a thermally isolated system \emph{"cannot interact thermally with any other system"}. For systems that are thermally isolated, there is then the \emph{"purely mechanical interaction"}. The \emph{"macroscopic work"}, as called by Reif, is the energy transfer between systems when they interact through a purely mechanical interaction. 
\par An example that illustrates the definition of macroscopic work given above can be considered \cite{schwabl2006statistical}: the external parameter is assumed to be the volume $V$, i.e. the Hamiltonian has a dependence on $V$, $H_{cl}\equiv H_{cl}(V)$. When this Hamiltonian dependence is considered, the energy shell volume $\hat{\Omega}(E)$ will\footnote{For more details, the reader is referred to \cite{schwabl2006statistical}.} also depends on $V$, $\hat{\Omega}(E)\equiv\hat{\Omega}(E,V)$. The entropy, considering the micro-canonical ensemble, will also depends on such parameter $S=S_{B}=k\ln \left(\hat{\Omega}(E,V)\right)$. Therefore, it can be proved that \cite{schwabl2006statistical}
\begin{equation}
	dS=\frac{1}{T}\left(dE-\left\langle\frac{\partial H_{cl}}{\partial V}\right\rangle dV\right).
\end{equation}
For a system confined in walls of volume $V$, the pressure $P$ is determined to be $P=-\left\langle\frac{\partial H_{cl}}{\partial V}\right\rangle$, such that 
\begin{equation}
	dE=TdS-PdV.
\end{equation}
The above equation, considering the work done on the confined system as $\delta \mathcal{W}=-PdV$, can be seen as a differential form of the first law of Thermodynamics for reversible processes, i.e., where $\delta \mathcal{Q}=TdS$. This example therefore evidences the purposes of the definitions given by Reif: the transfer energy $\delta \mathcal{W}=-PdV$ refers to the purely mechanical contribution to the total energy transfer and $\delta \mathcal{Q}=TdS$ to the purely thermal one. In fact, the connection between work and an external parameter changes is also considered by Jarzinsky in his well-known equality connecting non-equilibrium properties with equilibrium ones \cite{jarzynski1997nonequilibrium}. This seems as a reasonable \emph{definition} of macroscopic work and heat and is exported to the quantum realm by Alicki \cite{alicki1979quantum}, to be treated in the next chapter.

%% file: fundamentacao_fisica2.tex
\chapter{THEORETICAL BACKGROUND: QUANTUM THEORY}
\label{cap:fundamentacao2}
\par The formalism to be employed throughout the results presented in the following chapters is stated here. In analogy with the previous chapter, Quantum Mechanics (QM) is firstly presented in its textbook form, where pure states are considered, and then quantum statistics is introduced. It is important to emphasize that the postulates to be presented were taken from Ref. \cite{nielsen2010quantum}.
\par In the scope of pure states, the Heisenberg picture, to be adopted in the formulation of the work definition stated at the next chapter, is explored in section \ref{sec: Heisenberg}. Then, some discussions related to the interface between classical and quantum perspectives are described. The Gaussian states are considered and the Ehrenfest theorem is discussed, where the idea of "macroscopically measurable properties" stated at previous chapter is renewed.  
\par The attention is turned to quantum statistics in section \ref{sec:quantumstat}, where some postulates are set. There some concepts related to Thermodynamics are succinctly considered and a summary of the main ideas are treated.
\par QTh is finally discussed at section \ref{sec:QT}, where the definitions of work approached in the literature are described in general terms.

\section{\textbf{Quantum Mechanics postulates}}
\par The QM canvas is set by the first postulate:
\begin{postuladoq}
	Associated to any isolated physical system is a complex complete vector space with inner product (that is, a Hilbert space) known as the state space of the system. The system is completely described by its state vector, which is a unit vector in the system state space.
	\label{postqt1}
\end{postuladoq}
The condition that the state vector is a unit one is formally written, accordingly with Dirac bracket notation as $\langle\psi|\psi \rangle=1$. A set $S$ of unit vectors is said to be orthonormal, if for every $\ket{\psi_{i}},\ket{\psi_{j}}\in S$, then $\langle\psi_{i}|\psi_{j} \rangle=\delta_{ij}$ where $\delta_{ij}$ represents the Kronecker delta.  Note that this postulate does not mention anything related with the physics of the state vector: it just states that it must be an unit vector of a Hilbert space. Since such space is linear, a linear combination of its elements also belongs to it. Therefore the description of a state $\ket{\psi}$ can be a linear combination of other states vector, i.e.
\begin{equation}
	\ket{\psi}=\sum_{k}c_{k}\ket{\psi_{k}}.
\end{equation}
In the language of QM, the state $\ket{\psi}$ is a superposition of the states $\left\{ \ket{\psi_{k}}\right\}$.
\par The temporal evolution of a state vector is given as follows.
\begin{postuladoq}
	The evolution of a closed quantum system is described by a unitary transformation. That is, the state $\ket{\psi}$ of the system at time $t_{1}$ is related to the state $\ket{\psi^{'}}$ of the system at time $t_{2}$ by a unitary operator $\mathcal{U}$ which depends only on	the times $t_{1}$ and $t_{2}$,
	\begin{equation}
		\ket{\psi^{'}}=\mathcal{U}\ket{\psi}.
	\end{equation} \label{postqt2}
\end{postuladoq}
\par The postulate asserts that if a system is supposed to be closed, then its evolution is unitary. In this statement, however, one may question: what unitary operator should be used?  First, for the vast class of systems treated in QM, the evolution is continuous in time. In this case, the state must satisfy the \emph{Schrödinger equation},
\begin{equation}
	i\hbar \frac{\partial \ket{\psi}}{\partial t}=H\ket{\psi}
\end{equation}
where $\hbar$ is the reduced Planck constant, $i^{2}=-1$, and $H$ is the Hamiltonian operator. The unitary-operator form of Schrödinger equation is,
\begin{equation}
		i\hbar \frac{\partial \mathcal{U}}{\partial t}=H\mathcal{U}.\label{eq:schrun}
\end{equation}
For cases in which $H$ is time-independent, it can be checked that, for two different instant of time $t_{0}$ and  $t$, the respective states $\ket{\psi (t_{0})}$ and  $\ket{\psi (t)}$ are related by $\ket{\psi (t)}=\mathcal{U}(t,t_{0})\ket{\psi (t_{0})}=\mathrm{e}^{\frac{H\left(t-t_{0}\right)}{i\hbar}}\ket{\psi (t_{0})}$. In general,
\begin{equation}
	\mathcal{U}(t,t_{0})=\mathbb{1}+\sum_{n=1}^{\infty}\left(\frac{i}{\hbar}\right)^{n}\int_{t_{0}}^{t}dt_{1}\int_{t_{0}}^{t_{1}}dt_{2}\cdots\int_{t_{0}}^{t_{n-1}}dt_{n}H(t_{1})H(t_{2})\cdots H(t_{n}).\label{eq:soltot}
\end{equation}
\par Another question that can be made, regarding the postulate, is: which system is, indeed, closed? The closeness of a system is an approximation that can be made for a large range of physical systems, yielding great results. It is not assumed that these systems does not interact at all, but that its interaction are not relevant in its quantum mechanical description. Naturally, this approximation is not necessarily true for any system. In order to account for such interaction terms, the open quantum system theory is considered in many cases (see section \ref{sec:QT} for a brief review). 
\par The next postulate establishes the collapse of a quantum state, after a measurement process.

\begin{postuladoq}
	Quantum measurements are described by a collection $\left\{M_{m}\right\}$ of measurement operators. These are operators acting on the state space of the system being measured. The index $m$ refers to the m-th measurement outcome that may occur in the experiment. If the state of the quantum system is $\ket{\psi}$ immediately before the measurement then the probability that result $m$ occurs is given by 
	\begin{equation}
		p(m)=\bra{\psi}M_{m}^{\dagger}M_{m}\ket{\psi}=\mathrm{Tr}\left(M_{m}^{\dagger}M_{m}\ket{\psi}\bra{\psi}\right),
	\end{equation}
	where $\mathrm{Tr}$ denotes the trace operation, and the state of the system after the measurement is
	\begin{equation}
		\frac{M_{m}\ket{\psi}}{\sqrt{\bra{\psi}M_{m}^{\dagger}M_{m}\ket{\psi}}}.
	\end{equation}
	The measurement operators satisfy the completeness equation,
	\begin{equation}
	\sum_{m}M_{m}^{\dagger}M_{m}=\mathbb{1}.
	\end{equation}
	\label{postqt3}
\end{postuladoq}
A special case of measurement is a projective measurement.
\begin{definicoes}
	A projective measurement is described by an observable $M$, an Hermitian operator on the state space of the system being observed. The observable has a spectral decomposition, 
	\begin{equation}
		M=\sum_{m}mP_{m}
	\end{equation}
	where $P_{m}$ is the projector onto the eigenspace of $M$ with eigenvalue $m$. The possible outcomes of the measurement correspond to the eigenvalues, $m$, of the observable. Upon measuring the state $\ket{\psi}$, the probability of getting result $m$ is given by
	\begin{equation}
		p(m)=\bra{\psi}P_{m}\ket{\psi}=\mathrm{Tr}\left(P_{m}\ket{\psi}\bra{\psi}\right).
	\end{equation}
	Given that outcome $m$ occurred, the state of the quantum system immediately after the measurement is
	\begin{equation}
		\frac{P_{m}\ket{\psi}}{\sqrt{p(m)}}.
	\end{equation}
	\label{def:observable}
\end{definicoes}
\par It is important to remark that the observable spectrum defines a basis of the Hilbert space, i.e., the set of eigenkets $S_{e}$, of which all elements are such that $M\ket{m_{j}}=m_{j}\ket{m_{j}}$. This set defines a basis for the state space defined at postulate \ref{postqt1}. The definition above also enables a physical interpretation of the observables: since they are Hermitian, their eigenvalues are real and can be associated with measurable quantities. The expectation value of an observable $M$, when the system state is $\ket{\psi}$ is defined as 
\begin{equation}
	\langle M\rangle=\bra{\psi}M\ket{\psi}=\mathrm{Tr}\left(M\ket{\psi}\bra{\psi}\right).
\end{equation}
Defining the operator
\begin{equation}
 \Delta M\equiv M-\bra{\psi}M\ket{\psi},\label{eq:DELTAM}
\end{equation}
the variance of the observable is given by
\begin{equation}
\langle \left(\Delta M\right)^{2}\rangle=\left\langle \left(M^{2}-2M\langle M\rangle+\langle M\rangle^{2}\right)\right\rangle=\langle M^{2}\rangle-\langle M\rangle^{2}.\label{eq:standardeviationfund}
\end{equation}
Such quantities will be central in the discussions established in the following chapter. From the above definition, the Heisenberg uncertainty can be stated with a proper statistical interpretation: for any two observables $A$ and $B$ acting on a Hilbert space related to a quantum system, it follows that \cite{sakurai2017modern}
\begin{equation}
	\langle \left(\Delta A\right)^{2}\rangle\langle \left(\Delta B\right)^{2}\rangle\geq \frac{1}{4}\left|\langle \left[A,B\right]\rangle\right|^{2},\label{eq:uncert}
\end{equation}
where $\left[A,B\right]=AB-BA$ is the commutator. This inequality establishes that for two non-commutating operators, it is not possible for one to attain absolute determination for both quantities, i.e. the variance of $A$ and $B$ can not assume zero simultaneously.
\par The following axiom establishes how composite systems are considered under QM mathematical framework.
\begin{postuladoq}
	The state space of a composite physical system is the tensor product of the state spaces of the component physical systems. Moreover, if there are systems numbered $1$ through $n$, and the $i$-th system is prepared in the state $\ket{\psi_{i}}$, then the joint state of the total system is $\ket{\psi_{1}}\otimes\ket{\psi_{2}}\otimes\cdot \cdot\cdot \ket{\psi_{n}}$.
	\label{postqt4}
\end{postuladoq}
The ideas related with this postulate are essential for discussion considering quantum open systems. Also, the association of the composite system with a tensor product of the state spaces is crucial for the treatment of properties related with Information Theory as for instance the entanglement \cite{nielsen2010quantum}. 
\section{\textbf{Heisenberg Picture}}
\label{sec: Heisenberg}
\par Suppose that a system is described by the state vector $\ket{\psi_{0}}$ at the instant $t=0$. After a time $t$, the state will be represented by $\ket{\psi(t)}=\mathcal{U}(t,0)\ket{\psi_{0}}\equiv\mathcal{U}_{t}\ket{\psi_{0}}$, where $\mathcal{U}(t,0)\equiv\mathcal{U}_{t}$ is an unitary operator representing the time evolution of the state. As mentioned above, the expectation value of an time-independent observable $A$, considering the state $\ket{\psi(t)}$, is just $\langle A\rangle=\bra{\psi(t)}A\ket{\psi(t)}$. However, such expression can be written in two equivalent ways
\begin{equation}
	\langle A\rangle=\bra{\psi(t)}A\ket{\psi(t)}=\underbrace{\bra{\psi_{0}}\mathcal{U}^{\dagger}}_{\bra{\psi(t)}}A\underbrace{\mathcal{U}\ket{\psi_{0}}}_{\ket{\psi(t)}}=\bra{\psi_{0}}\underbrace{\mathcal{U}^{\dagger}A\mathcal{U}}_{A_{H}}\ket{\psi_{0}}=\bra{\psi_{0}}A_{H}\ket{\psi_{0}}.\label{eq:scheisen}
\end{equation} 
There are two possible interpretations for the equalities above:
\begin{itemize}
	\item \emph{(Schrödinger Picture)} The state of the system evolves to $\ket{\psi(t)}=\mathcal{U}_{t}\ket{\psi_{0}}$ and the operator is unchanged; 
	\item \emph{(Heisenberg Picture)} The operator evolves as $A\rightarrow A_{H}=\mathcal{U}_{t}^{\dagger}A\mathcal{U}_{t}$ and the state vector remains as $\ket{\psi_{0}}$.
\end{itemize}
Note that the sub-index $H$ is used to describe that the operator is in the Heisenberg picture. 
\par In cases where the Schrödinger operator $A(t)$ depends explicitly on time, one has the following rules.
\begin{itemize}
	\item \emph{(Schrödinger Picture)} The state of the system evolves to $\ket{\psi(t)}=\mathcal{U}_{t}\ket{\psi_{0}}$ while the operator keeps its form $A(t)$ for all times; 
	\item \emph{(Heisenberg Picture)} The operator assumes a new time dependency $A(0)\rightarrow A_{H}(t)=\mathcal{U}_{t}^{\dagger}A(t)\mathcal{U}_{t}$ and the state vector remains as $\ket{\psi_{0}}$.
\end{itemize}
\par Given a Heisenberg operator $A_{H}$ one may find its Schrödinger counterpart via the transformation:
\begin{equation}
\mathcal{U}A_{H}\mathcal{U}_{t}^{\dagger}=\underbrace{\mathcal{U}\mathcal{U}_{t}^{\dagger}}_{\mathbb{1}}A_{s}\underbrace{\mathcal{U}\mathcal{U}_{t}^{\dagger}}_{\mathbb{1}}=A_{s}.\label{eq:heischron}
\end{equation}
In the present discussion, the subindex $S$ is attached to denote Schrodinger operators, although this notation is suppressed, in latter discussions, for the sake of simplicity. This formula will show to be particularly interesting for the present work.
\par An equation of motion for the operators can be deduced under the Heisenberg picture, considering the Schrödinger equation. Consider the general case in which Schrödinger picture operator $A_{s}\equiv A_{s}(t)$ is explicitly time-dependent. It follows, considering Eq. \eqref{eq:schrun} and the unitarity of $\mathcal{U}_{t}$ that
\begin{equation}
	\frac{dA_{H}}{dt}=\frac{\left[A_{H},H_{H}\right]}{i\hbar}+\mathcal{U}_{t}^{\dagger}\frac{\partial A_{s}}{\partial t}\mathcal{U}_{t}.\label{eq:heisengeral}
\end{equation}
In the special case in which an operator $A_{s}$ is time-independent,
\begin{equation}
	\frac{dA_{H}}{dt}=\frac{\left[A_{H},H\right]}{i\hbar}\label{eq:Heiseqm}
\end{equation} 
which is the so-called Heisenberg equation of motion \cite{sakurai2017modern}. 
\par From the above discussion, it is therefore manageable to determine the time-derivative operator of a Heisenberg operator, by considering Eq. \eqref{eq:heisengeral}. For instance, if one considers one-dimensional position operator, $X_{H}$, it is possible to obtain the speed operator by substituting $A\rightarrow X$ in Eq. \eqref{eq:heisengeral}, so that
\begin{equation}
V_{H}\equiv\frac{dX_{H}}{dt}\equiv\dot{X}_{H}=\frac{\left[X_{H},H_{H}\right]}{i\hbar}+\mathcal{U}_{t}^{\dagger}\frac{\partial X_{s}}{\partial t}\mathcal{U}_{t}.\label{eq:derx}
\end{equation}
where $V_{H}$ is the speed operator in the Heisenberg picture. A tricky question (but necessary for future discussions) can be made: how can the speed operator be computed in Schrödinger picture? This question can be answered by substituting $A\rightarrow V$ in Eq. \eqref{eq:heischron}, resulting in
\begin{equation}
	V_{s}=\mathcal{U}V_{H}\mathcal{U}_{t}^{\dagger}.\label{eq:conversaoheischron}
\end{equation}
By considering Eq. \eqref{eq:derx},
 \begin{equation}
 \begin{array}{rl}
 V_{s}&\displaystyle=\mathcal{U}V_{H}\mathcal{U}_{t}^{\dagger}=\mathcal{U}\left(\frac{\left[X_{H},H_{H}\right]}{i\hbar}+\mathcal{U}_{t}^{\dagger}\frac{\partial X_{s}}{\partial t}\mathcal{U}_{t}\right)\mathcal{U}_{t}^{\dagger}\\
 &\displaystyle=\frac{\left[X_{s},H_{s}\right]}{i\hbar}+\frac{\partial X_{s}}{\partial t}.\label{eq:conversaoheischron3}
 \end{array}
 \end{equation}
Note that the results above are not restricted to the speed only: considering instead of $X$, a operator $A$ acting on the same space, it can be shown via similar arguments that
\begin{equation}
	\left(\dot{A}_{H}\right)_{s}=\frac{\left[A_{s},H_{s}\right]}{i\hbar}+\frac{\partial A_{s}}{\partial t},
\end{equation}
where $\left(\dot{A}_{H}\right)_{s}$ is\footnote{This observable is not written as $\dot{A}_{s}$ (as is done for $V$), in order to avoid confusion with the meaning of the dot, associated with time derivative.} \emph{the Schrödinger operator associated with the Heisenberg operator} $\dot{A}_{H}\equiv\frac{d A_{H}}{dt}$. Notice that \eqref{eq:conversaoheischron3} is different from the time derivative of the Schrödinger operator, which can be zero in general. In section \ref{sec:Ehren}, the Heisenberg picture is adopted as starting point for the derivation of an important result connecting classical and quantum mechanics, namely, the Ehrenfest theorem. Before doing so, it is convenient to have a look at a special class of states. 
\section{\textbf{Gaussian states}}
The position and momentum observables $X,P$ acting in a Hilbert space $\mathcal{H}$ are considered. Since position eigenkets $\left\{\ket{x}\right\}$ define a basis for the state space, then any system in such a space can be described as \cite{sakurai2017modern,lebedev2003functional}
\begin{equation}
	\ket{\psi}=\int_{x=-\infty}^{\infty}\ket{x}\langle x|\psi\rangle\equiv\int_{x=-\infty}^{\infty}\ket{x}\Psi(x),
\end{equation}
where $\Psi(x)=\langle x|\psi\rangle$ is the so-called wave function. The Gaussian state is obtained by taking the wave function as
\begin{equation}
	\Psi(x)=\frac{\mathrm{e}^{-\frac{\left(x-x_{0}\right)^{2}}{4\Delta_{x,0}^{2}}}\mathrm{e}^{\frac{ip_{0}x}{\hbar}}}{\left(2\pi \Delta_{x,0}^{2}\right)^{\frac{1}{4}}}\label{eq:Gaussian}
\end{equation}
or, considering the momentum bases, 
\begin{equation}
\langle p|\psi_{0} \rangle=\left(\frac{2\Delta_{x,0}^{2}}{\hbar^{2}\pi}\right)^{\frac{1}{4}}\exp\left[-\frac{\Delta_{x,0}^{2}\left(p-p_{0}\right)^{2}}{\hbar^{2}}-\frac{i\left(p-p_{0}\right)x_{0}}{\hbar}\right].\label{eq:Gaussianp}
\end{equation}
The following properties can be obtained for such state:
\begin{eqnarray}
\left\langle X\right\rangle=x_{0}, & &\left\langle\left(\Delta X\right)^{2}\right\rangle=\Delta_{x,0}^{2},\\
\left\langle P\right\rangle=p_{0}, & & \left\langle\left(\Delta P\right)^{2}\right\rangle=\frac{\hbar^{2}}{4\Delta_{x,0}^{2}}.
\end{eqnarray}
If the variance in position is such that, for instance, $\mathcal{O}\left(\Delta_{x,0}^{2}\right)\approx 10^{-17}$, then $\mathcal{O}\left(\Delta_{p,0}^{2}\right)\approx 10^{-17}$. Therefore, the Gaussian packet can offer a good approximation for treating particles.
\section{\textbf{Ehrenfest Theorem}}
\label{sec:Ehren}
Consider a system described by the Hamiltonian 
\begin{equation}
	H=\frac{P^{2}}{2m}+\mathcal{V}(X),\label{eq:hsistiso}
\end{equation}
where $X$ and $P$ are the position and momentum operator in the Schrödinger picture\footnote{For simplicity, the subindex $S$ will be suppressed for Schrödinger operator, unless some ambiguity may occur.} and $\mathcal{V}(X)$ is a position dependent potential. The Heisenberg position operator is then given by $X_{H}=\mathcal{U}_{t}^{\dagger}X\mathcal{U}_{t}$, where $\mathcal{U}_{t}=\exp(-iHt/\hbar)$. The velocity and acceleration are obtained from the Heisenberg equation of motion:
\begin{equation}
\frac{dX_{H}}{dt}\equiv \dot{X}_{H}=\frac{\left[X_{H},H_{H}\right]}{i\hbar}=\frac{P_{H}}{m},\label{eq:ehr}
\end{equation} 
\begin{equation}
\frac{d^{2}X_{H}}{dt^{2}}\equiv \ddot{X}_{H}=\frac{\left[\dot{X}_{H},H_{H}\right]}{i\hbar}=\frac{\left[\frac{P_{H}}{m},H_{H}\right]}{i\hbar}=\frac{-\partial_{X_{H}}\mathcal{V}(X_{H})}{m}.\label{eq:aceehren}
\end{equation} 
From Eq. \eqref{eq:aceehren}, it follows that
\begin{equation}
m\ddot{X}_{H}=-\partial_{X_{H}}\mathcal{V}(X_{H}),\label{eq:newtonoperatoreq}
\end{equation} 
which can be viewed as an operator form of Newton's second law\footnote{Considering non-relativistic regimes and that the mass of the system does not change.}. In this sense one may \emph{define} $-\partial_{X_{H}}\mathcal{V}(X_{H})$ as the \emph{resultant force observable}. Note that this definition makes sense under the hypothesis that the system is described by the Hamiltonian \eqref{eq:hsistiso}. For more general Hamiltonians, the resultant force will naturally be different: the Caldirola-Kanai system to be treated in chapter \ref{cap:kanaicaldirola} has a different type of Hamiltonian. However a large class of problems refers to Hamiltonians of the form \eqref{eq:hsistiso}. As a matter of fact, such Hamiltonian was considered by Ehrenfest in his paper where the present theorem was first derived\footnote{In his formulation, it was considered the wave mechanics formalism, instead.} \cite{ehrenfest1927bemerkung}. In the present discussion only such class of problems shall be treated. 
\par The Ehrenfest theorem appears when one takes the expectation value on both sides of Eq. \eqref{eq:newtonoperatoreq}:
\begin{equation}
m\frac{d^{2}}{dt^{2}}\langle X\rangle=-\left\langle\partial_{X_{H}}\mathcal{V}(X_{H})\right\rangle.\label{eq:ehren}
\end{equation}
It is worth noticing at this point that this is precisely the same form that would be obtained in the Liouvillian formalism, with the pertinent adaptations. Following this equation, it is frequently concluded \cite{ballentine1994inadequacy,ehrenfest1927bemerkung} that the center of the distribution of the state will have its motion described in the same way as a classical system if:
\begin{itemize}
	\item the system is localized, i.e. if $\langle \left(\Delta X\right)^{2}\rangle$ is relatively small or
	\item the potential $\mathcal{V}$ has only polynomial terms with degree smaller than three.
\end{itemize}
To see how this expression compares with classical results, consider a scenario in which the initial conditions of a heavy particle are not known with full certainty or measurements on this particle are conducted with low resolution. The description of the mean position of this particle at an arbitrary instant can be written as 
\begin{equation}
\langle x\rangle^{2}=\langle x^{2}\rangle-\sigma_{x}^{2}
\end{equation} 
which implies 
\begin{equation}
\langle x\rangle=\sqrt{\langle x^{2}\rangle}\left(1-\frac{\sigma_{x}^{2}}{\langle x^{2}\rangle}\right)^{\frac{1}{2}}
\end{equation} 
Then, if $\frac{\sigma_{x}^{2}}{\langle x^{2}\rangle}\ll 1$, then there is no appreciable difference between $\langle x\rangle^{2}$ and $\langle x^{2}\rangle$ or, in general, between $\langle x\rangle^{n}$ and $\langle x^{n}\rangle$, so that $\langle f(x)\rangle\approx f\left(\langle x\rangle\right)$. It follows that there will be no significant difference between \eqref{eq:newtonoperatoreq} and 
\begin{equation}
m\frac{d^{2}}{dt^{2}}\langle x(t_{i})\rangle =-\left.\frac{\partial }{\partial x'} \mathcal{V}(\langle x(t_{i})\rangle)\right|_{x'=\langle x(t_{i})\rangle}.\label{eq:ehrenncl}
\end{equation}
On the other hand, when the uncertainty $\sigma_{x}$ is significant, then the form \eqref{eq:newtonoperatoreq} is the correct one in both the quantum and Liouvillian formalism, and none of these results will agree with the one predicted via single Newtonian trajectory. It is important to highlight that this discussion is not new in the literature\cite{angelo2003aspectos}. As a matter of fact, a similar discussion is treated at \cite{ballentine1994inadequacy}, where the author writes: \emph{"the centroid of a classical ensemble need not follow a classical trajectory if the width of the probability distribution is not negligible."}. 
\par The discussion above was made considering the position operator. However,the same rationale can be applied to the speed operator, $V_{H}=\dot{X}_{H}=\frac{P_{H}}{m}$. Suppose a task is given to an experimentalist which consists of determining the kinetic energy of a particle. By measuring the speed of the particle, she can only access $\langle V\rangle$ and $\langle \left(\Delta V\right)^{2}\rangle$, and then compute $\frac{m}{2}\langle V\rangle^{2}$. This, however, does not agree in general with statistical predictions, may them be quantum or Liouvillian, which provide $\frac{m}{2}\langle V^{2}\rangle$. Good agreement will emerge only when the uncertainty $\langle \left(\Delta V\right)^{2}\rangle=\frac{\langle \left(\Delta P\right)^{2}\rangle}{m}$ is sufficiently small so that $\langle V^{2}\rangle\approx\langle V\rangle^{2}$ (Notice that the essential difference between the quantum and the Liouvillian descriptions is the existence of the uncertainty principle in the former). 
\par The question then arises as to whether the correct form for the kinetic energy in the interface Thermo-Statistical Physics is $\frac{m}{2}\langle V^{2}\rangle$ or $\frac{m}{2}\langle V\rangle^{2}$. If the position and the speed of a macro system, as for instance a piston, are given by $\langle X\rangle$ and $\langle V\rangle$, which are macroscopically accessible quantities, then one might say the work and the kinetic energy would be given with regards to $d\langle X\rangle$ and $\frac{m}{2}\langle V\rangle^{2}$. The discussion around this treatment is postponed to the next chapter, since it is connected to the definition of work to be proposed there. 
\section{\textbf{Quantum statistical physics}}
\label{sec:quantumstat}
\par In QM, there is an intrinsic randomness associated with the collapse of the state. One can not known \emph{a priori} for which eigenstate of the measured observable the state vector will reduce to upon a measurement. On the other hand, precise preparation of a state $\ket{\psi_{0}}$ implies, via Schrodinger's equation, full determinism for the future state  $\ket{\psi(t)}$ of the system. Even with total information about $\ket{\psi(t)}$ one can not predict, however, the outcome of the collapse. This is a fundamental ignorance associated with the belief (expressed by Quantum Mechanics) that nature is irreducibly random. The uncertainties underlying the Liouvillian formalism is of a purely subjective nature: the classical paradigm imposes that a particle certainly occupies a given point in phase space (determinism), but one is operationally ignorant about it. There is no essential indefiniteness here. A formalism proposed by von Neumann \cite{schwabl2006statistical,von2010proof} can couple statistics assumptions within quantum systems, which shall be regarded next.
\subsection{Density operator}
\label{subsec:density}
\par Consider that an experimentalist possesses a large number of copies of a quantum system such that each copy is in a given state of the set $\left\{\ket{\psi_{i}}\right\}$. The relative probability of occurrence of $\ket{\psi_{i}}$ is $w_{i}$, with  $\sum_{i} w_{i}=1$. What will be the mean value of an observable $A$ in such system? If one evaluates the mean value considering only an specific state $\ket{\psi_{j}}$, the mean value shall be just $\bra{\psi_{j}}A\ket{\psi_{j}}$. However, the experimentalist in the present case does not actually know for which copy the measurement was performed and, therefore, to evaluate the mean value of $A$ she has to consider the ensemble average
\begin{equation}
	\langle A\rangle =\sum_{i} w_{i} \bra{\psi_{i}}A\ket{\psi_{i}}.
\end{equation}
Note that, by considering $\rho=\sum_{i}w_{i}\ket{\psi_{i}}\bra{\psi_{i}}$, one has
\begin{equation}
\displaystyle\sum_{i} w_{i} \bra{\psi_{i}}A\ket{\psi_{i}}=\sum_{i} w_{i} \mathrm{Tr}\left( \ket{\psi_{i}}\bra{\psi_{i}}A\right)=\mathrm{Tr}\left(\sum_{i} w_{i}\ket{\psi_{i}}\bra{\psi_{i}}A\right)=\mathrm{Tr}\left(\rho A\right). 
	\end{equation}
The descriptor $\rho$, which is commonly called density operator or density matrix, can be viewed as a generalization of the pure-state description of a state, since it is a convex combination of pure states. In this capacity, $\rho$ describes a mixed ensemble or a mixed state. In particular, for $w_{i} = \delta_{ij}$ one recovers the pure-state formalism, as $\rho = \ket{\psi_{j}}\bra{\psi_{j}}$.
\par Apart from the fact that $\rho$ is an operator acting on a vector space, because it is used to describe a collection of distinct pure states, and thus takes $w_{i}$ as a probability distribution, it can be thought as having a close analogy with the classical-statistical distribution $\rho_{cl}$ discussed in the Liouvillian formalism. In this sense, it is clear that $\rho$ can encode two "flavors" of uncertainty: a classical one, associated with the probability distribution $w_{i}$, and a quantum one, related to the intrinsic uncertainty of each element $\ket{\psi_{i}}$ of the ensemble.
\par Any $\rho$ considered for the description of quantum system must satisfy the properties:
\begin{itemize}
	\item Normalization: $\mathrm{Tr}\rho=1$;
	\item Hermiticity: $\rho=\rho^{\dagger}$;
	\item Positivity: for any state vector $\ket{\psi}$, $\bra{\psi}\rho\ket{\psi}\geq 0$;
	\item Purity: $\mathrm{Tr}\rho^{2}$ is defined as the \emph{purity} of $\rho$; it satisfies $\mathrm{Tr}\rho^{2}\leq 1$, where $\mathrm{Tr}\rho^{2}= 1$ if and only if the state is pure, i.e. there is a state $\ket{\psi}$ such that $\rho=\ket{\psi}\bra{\psi}$.
\end{itemize}
\par As for the vector-state case, some postulates follow. Here it is made a few remarks concerning only some consequences of the postulates, since the grounding in which they are founded is analogous to the pure state case\footnote{For more details, the reader is referred to \cite{nielsen2010quantum}.}. 
\begin{postuladoq2}
	Associated with any isolated physical system is a complex vector space with inner product (that is, a Hilbert space) known as the state space of the system. The system is completely described by its density operator, which is a positive operator $\rho$ with trace one, acting on the state space of the system. If a quantum system is in the state $\rho_{i}$ with probability $p_{i}$, then the density operator for the system is $\sum_{i}p_{i}\rho_{i}$.
\end{postuladoq2}
\begin{postuladoq2}
	The evolution of a closed quantum system is described by a unitary transformation. That is, the state $\rho$ of the system at time $t_{1}$ is related to the state $\rho^{'}$ of the system at time $t_{2}$ by a unitary operator $\mathcal{U}$ which depends only on	the times $t_{1}$ and $t_{2}$,
	\begin{equation}
	\rho^{'}=\mathcal{U}\rho\mathcal{U}^{\dagger}.\label{eq:evolrho}
	\end{equation}
	\label{post:2}
\end{postuladoq2}
\begin{postuladoq2}
	Quantum measurements are described by a collection $\left\{M_{m}\right\}$ of measurement operators. These are operators acting on the state space of the system being measured. The index $m$ refers to the measurement outcomes that may occur in the experiment. If the state of the quantum system is $\rho$ immediately before the measurement then the probability that the outcome $m$ occurs is given by 
	\begin{equation}
	p(m)=\mathrm{Tr}\left(M_{m}^{\dagger}M_{m}\rho\right)
	\end{equation}
	and the state of the system after the measurement is
	\begin{equation}
	\frac{M_{m}\rho M_{m}^{\dagger} }{\mathrm{Tr}\left(M_{m}^{\dagger}M_{m}\rho\right)}
	\end{equation}
	The measurement operators satisfy the completeness equation,
	\begin{equation}
	\sum_{m}M_{m}^{\dagger}M_{m}=\mathbb{1}.
	\end{equation}
\end{postuladoq2}
\begin{postuladoq2}
	The state space of a composite physical system is the tensor product of the state spaces of the component physical systems. Moreover, if there are systems numbered $1$ through $n$, and the $i$-th system is prepared in the state $\rho_{i}$, then the joint state of the total system is $\rho_{1}\otimes\rho_{2}\otimes\cdot \cdot\cdot \rho_{n}$.
\end{postuladoq2}
Postulate \ref{postqt2} can be restated in a different form, i.e., the density operator dynamics is described by the Liouville-von Neumann equation
\begin{equation}
\frac{\partial \rho}{\partial t}=\frac{\left[H,\rho\right]}{i\hbar}=-\frac{\left[\rho,H\right]}{i\hbar}.\label{eq:vonneumaneq}
\end{equation} 
As a matter of fact, it can be proved that if the unitary operator considered in Eq. \eqref{eq:evolrho} satisfies the Schrödinger equation (Eq. \eqref{eq:schrun}), then Eq. \eqref{eq:vonneumaneq} results. Note that this dynamical equation resembles Eq. \eqref{eq:poissonliou}, obtained under the scope of the Liouville theorem. In fact, \eqref{eq:vonneumaneq} emerges from \eqref{eq:poissonliou} via the usual quantization rule $\left\{\rho,\mathcal{H}\right\}\rightarrow\frac{\left[\rho,H\right]}{i\hbar}$. 
\par In analogy with the description of a pure state at the beginning of the chapter, the definition \ref{def:observable} is renewed from postulate \ref{postqt3}.
\begin{definicoes}
	A projective measurement is described by an observable $M$, an Hermitian operator on the state space of the system being observed. The observable has a spectral decomposition, 
	\begin{equation}
	M=\sum_{m}mP_{m},
	\end{equation}
	where $P_{m}$ is the projector onto the eigenspace of $M$ with eigenvalue $m$. The possible outcomes of the measurement correspond to the eigenvalues, $m$, of the observable. Upon measuring the state $\rho$, the probability of getting result $m$ is given by
	\begin{equation}
	p(m)=\mathrm{Tr}\left(P_{m}^{\dagger}P_{m}\rho\right).
	\end{equation}
	Given that the outcome $m$ occurred, the state of the quantum system immediately after the measurement is
	\begin{equation}
	\frac{P_{m}\rho P_{m}^{\dagger} }{\mathrm{Tr}\left(P_{m}^{\dagger}P_{m}\rho\right)}.
	\end{equation}
	\label{def:observable2}
\end{definicoes}
 The mean value of an observable $A$ is defined analogously to the pure state case $\langle A\rangle\equiv \mathrm{Tr}\left(\rho A\right)$. Note that, the Heisenberg picture can also be established in this context: considering that $\rho_{0}$ evolves to $\rho(t)=\mathcal{U}_{t}\rho_{0}\mathcal{U}_{t}^{\dagger}$, after an time interval $t$, then, by the cyclic property of the trace, one has
\begin{equation}
	\langle A_{S}(t)\rangle_{t}=\mathrm{Tr}\left(\rho(t) A_{S}(t)\right)=\mathrm{Tr}\left(\mathcal{U}_{t}\rho_{0}\mathcal{U}_{t}^{\dagger} A_{S}(t)\right)=\mathrm{Tr}\left(\rho_{0}\mathcal{U}_{t}^{\dagger} A_{S}(t)\mathcal{U}_{t}\right)=\mathrm{Tr}\left(\rho_{0}A_{H}(t)\right),\label{eq:heisen2}
\end{equation}
where $A_{H}(t)=\mathcal{U}_{t}^{\dagger} A_{S}(t)\mathcal{U}_{t}$ is the Heisenberg operator associated with the explicitly time-dependent Schrödinger operator $A_S(t)$.
\par Postulate \ref{postqt4} enables one to define a density operator for a composite system. However, as the system evolves, it is possible that the system can not, after a time $t$, be in a separable state of the form treated in the postulate. The system can, for instance, be \emph{entangled}, a concept defined as follows.
\begin{definicoes}
	A bipartite\footnote{ At the present work, it will be considered only bipartite systems in the discussions. The extension for more than two systems is considered in \cite{nielsen2010quantum}.} density operator $\rho$, that acts on the state space $\mathcal{H}=\mathcal{H}_{A}\otimes \mathcal{H}_{B}$, is said to be separable if and only if it can be written as the convex sum
	\begin{equation}
		\begin{array}{ccc}
		\rho=\sum_{i}w_{i}\rho_{A}^{i}\otimes\rho_{B}^{i}, & w_{i}\geq 0, & \sum_{i}w_{i}=1.\label{eq:nentang}
		\end{array}
	\end{equation}
	Otherwise, $\rho$ is said to be entangled. 
\end{definicoes}
Entanglement is a quantum resource that has far reaching consequences \cite{nielsen2010quantum,bilobran2015measure,bell2001einstein,einstein1935can}. A particularly important one manifests itself in the dynamics of reduced states.
\begin{definicoes}
	Consider a bipartite density operator $\rho$ that acts on the state space $\mathcal{H}=\mathcal{H}_{\mathcal{A}}\otimes \mathcal{H}_{\mathcal{B}}$.
	The reduced density operator for system $\mathcal{A}$ is defined by
	\begin{equation}
		\rho_{\mathcal{A}}\equiv \mathrm{Tr}_{\mathcal{B}}\left(\rho\right),
	\end{equation}
	where $ \mathrm{Tr}_{\mathcal{B}}$ is a map of operators known as the partial trace over system $\mathcal{B}$. The partial trace is defined by 
	\begin{equation}
		\mathrm{Tr}_{\mathcal{B}}\left(\ket{a_{1}}\bra{a_{2}}\otimes \ket{b_{1}}\bra{b_{2}}\right)\equiv \ket{a_{1}}\bra{a_{2}}\mathrm{Tr}\left( \ket{b_{1}}\bra{b_{2}}\right)=\ket{a_{1}}\bra{a_{2}}\langle b_{2}|b_{1}\rangle
	\end{equation}
	where $\ket{a_{1}}$ and $\ket{a_{2}}$ are any two vectors in the state space of $\mathcal{A}$, and $\ket{b_{1}}$ and $\ket{b_{2}}$ are any	two vectors in the state space of $\mathcal{B}$.
	\label{def:partialtrace}
\end{definicoes}
Therefore, if the parts of a system are not entangled, as described in Eq. \eqref{eq:nentang}, the local density operator of the subsystem $\mathcal{A}$ will be simply $\sum_{i}w_{i}\rho_{\mathcal{A}}^{i}$.  However, if the subsystems $\mathcal{A}$ and $\mathcal{B}$ are entangled, such result does not hold. Entanglement can be generated via physical interactions, which lead to non-unitary local (reduced) dynamics whereas the global one is unitary. The effects of non-unitary evolution are highlighted at subsection \ref{subsec:interacting}, where the Thermodynamics point of view is employed. 
\subsection{Entropy} 
\par Entropy is a fundamental concept for both Information Theory and Thermodynamics. In the Classical Thermodynamics context, entropy can be defined via the number of possible microstates (Gibbs' approach) or the volume of the energy shell (Boltzmann's approach). On the other hand, in classical information theory, the concept of information is quantified in terms of the Shannon entropy \cite{nielsen2010quantum}, whose form matches the Gibbsian entropy \cite{reichl2016modern}. In the scope of QM it is common to define the \emph{von Neumann entropy} as
\begin{equation}
	S_{v}\left(\rho\right)=-\mathrm{Tr}\left(\rho\ln \rho\right).
\end{equation} 
The von Neumann entropy is frequently considered in the scope of quantum information theory, since it has special properties that enables one to quantify entanglement, mutual information, discord, among others \cite{nielsen2010quantum,bilobran2015measure}. Some of these properties are listed below.
\begin{itemize}
	\item Purity: $0 \leq S_{v}(\rho) \leq \ln d$, where $S=0$ iff $\rho$ is pure and $d$ is the dimension of the Hilbert space on which $\rho$ acts;
	\item Invariance: $S_{v}$ is invariant under unitary transformations, i.e. given an unitary transformation $\mathcal{U}$, $S_{v}(\mathcal{U}\rho\mathcal{U}^{\dagger})=S_{v}(\rho)$;  
	\item Concavity: Given sets $\left\{\alpha_{i} \right\}$ and $\left\{\rho^{i}\right\}$ such that $\sum_{i}\alpha_{i}=1$ and $\rho^{i}$ are density operators, then $S_{v}\left(\sum_{i}\alpha_{i}\rho^{i}\right)\geq\sum_{i} \alpha_{i} S_{v}\left(\rho^{i}\right)$;
	\item Subadditivity: $S_{v}(\rho)\leq S_{v}(\rho_{\mathcal{A}})+S_{v}(\rho_{\mathcal{B}})$ (equality holding for $\rho = \rho_{\mathcal{A}} \otimes \rho_{\mathcal{B}}$).
\end{itemize}
The von Neumann entropy is also frequently considered in quantum statistics, under a thermodynamical point of view. Such approach shall be considered next, along with others concepts.

\section{\textbf{Thermodynamics perspective}} 
Some concepts often used in scenarios involving the quantum statistical description of the thermodynamics of macroscopic systems at equilibrium are succinctly reviewed in this section.
\begin{itemize}
	\item \emph{Statistical equilibrium}: The considerations, under statistical physics are the same as those considered classically; some authors consider that $\frac{\partial \rho }{\partial t}=0$ defines that the corresponding distribution $\rho$ is in statistical equilibrium.
	\item \emph{Equiprobability a priori and micro-canonical ensemble}: It is also established in analogy with the classical case: the microstates compatible with the macrostate of a physical system, when in thermodynamic equilibrium, are equiprobable. Such hypothesis is explicitly considered in order to define the micro-canonical ensemble for quantum systems. 
	\item \emph{Ergodic hypothesis}: It has been considered also in the scope of QM, and it was proved true for a class of problems \cite{gemmer20044,von2010proof}.    
	\item \emph{Entropy}: It is frequently given by the entropy of von Neumann in textbooks. In the context of quantum information theory, other entropic formulas are considered, such as the linear entropy, the Tsallis entropy, and the Rènyi entropy \cite{nielsen2010quantum,costa2013bayes}. It is also assumed that the process will tend to an equilibrium macrostate, where the entropy is maximum. By considering the von Neumann proposition, it can be proved that, for a fixed value of energy, the distribution of a closed system at equilibrium that yields the greatest entropy will actually be the micro-canonical ensemble one \cite{schwabl2006statistical}. 
	\item \emph{Canonical Ensemble} Defined similarly to the classical case, $\rho_{qc}=\frac{\mathrm{e}^{-\frac{H}{kT}}}{\mathcal{Z}}$ where $\mathcal{Z}=\mathrm{Tr}\left(\mathrm{e}^{-\frac{H}{kT}}\right)$ is the partition function and $H$ is the Hamiltonian of the system. 
\end{itemize} 
It is important to remark that the above concepts differ slightly from quantum systems to classical ones when considering macroscopic systems at equilibrium. However, the quantum description applies to problems that have no correspondence in classical physics, as for instance, ideal paramagnetic Spin-$\frac{1}{2}$ systems \cite{schwabl2006statistical,salinas1997introduccao}. 
\par Recently, there has been a renewed interest in  describing the thermodynamics of out-of-equilibrium few-particle systems. This is the main task of an emergent field sometimes referred to as Quantum Thermodynamics (QTh). In the following, the main aspects related with these subject are treated, with focus on the characterization of work and heat.

\section{\textbf{QTh: work and heat}}
\label{sec:QT}
\par QTh is an area that has recently gained the attention of the scientific community \cite{gemmer20044}. There have been great efforts to extend thermodynamical concepts to the domain of quantum mechanical systems. However, only few concepts can be considered as well established, with some supporting experimental results. Due to the great extent and the yet incomplete character of the subject, only the definitions of work, in its many perspectives, will be discussed here\footnote{The reader is referred to a review of QTh given in \cite{gemmer20044}, for more details and reference on the area.}. Before treating specifically some definitions, it is given a general description of the systems discussed/studied in the QTh domain. 
\subsection{Interacting systems evolution}
\label{subsec:interacting}

\par Consider a composite system described by a density operator $\rho$ acting on a Hilbert space $\mathcal{H}=\mathcal{H}_{\mathcal{S}}\otimes\mathcal{H}_{\mathcal{E}}$. Roughly, a system $\mathcal{S}$ is said to be open if there is an interaction between the elements of the system subspace $\mathcal{H}_{\mathcal{S}}$ with those of the environment subspace $\mathcal{H}_{\mathcal{E}}$\footnote{It is emphasized, as before, that it can be considered more than two systems. However, it simplifies the discussion here to consider only two.}. The Hamiltonian for the composite system $\mathcal{S}+\mathcal{E}$ is written, in general form, as
\begin{equation}
	H=H_{\mathcal{S}}\otimes \mathbb{1}_{\mathcal{E}}+\mathbb{1}_{\mathcal{S}}\otimes H_{\mathcal{E}}+H_{I},\label{eq:hamitongeral}
\end{equation} 
where $H_{\mathcal{S}}$ is Hamiltonian operator acting on $\mathcal{H}_{\mathcal{S}}$, $H_{\mathcal{E}}$ acts on $\mathcal{H}_{\mathcal{E}}$, $H_{I}$ the coupling term, acting on $\mathcal{H}$, and $\mathbb{1}_{\mathcal{S}}$ and $\mathbb{1}_{\mathcal{E}}$ are the identity operators on spaces  $\mathcal{H}_{\mathcal{S}}$ and $\mathcal{H}_{\mathcal{E}}$, respectively. Here, all operators are considered to be time-independent.
\par Frequently, one is interested in describing how the system $\mathcal{S}$ evolves in time under the coupling with the environment $\mathcal{E}$. To this end, two models are often employed. One, by considering an effective non-autonomous Hamiltonian for $\mathcal{S}$ and, two, by solving the composite dynamics and then discarding (via partial trace) the environment.
\par When the system is modeled via a non-autonomous dynamics, an effective time-dependent Hamiltonian
\begin{equation}
H_{ef}\equiv H_{ef}(a(t)).\label{eq:effectiveH}
\end{equation}
is adopted for governing the system dynamics. This operator explicitly depends on time through a control function $a(t)$, which encodes all the environment influence over $\mathcal{S}$\footnote{This was already considered in the classical case. See section \ref{subsec:externalparam} for more details.}. In this case, the system is considered to  undergone a unitary evolution given by
\begin{equation}
i\hbar \frac{\partial \mathcal{U}}{\partial t}=H_{ef}(a(t))\mathcal{U}.\label{eq:schrun2}
\end{equation}
Because the system is effectively closed, no entanglement can be described with the environment and, therefore, the entropy (purity of the state) is conserved. 
\par An alternative model consists of constructing the solution for the composite state $\rho$ and then to obtain the reduced state $\rho_{\mathcal{S}}$ via the partial trace. However, it may be difficult to get an exact solution due to huge number of degrees of freedom in the environment and the quantum correlations generated during the time evolution of the system. In order to avoid these technical difficulties, physical considerations such as, for example, weak interaction and no-memory effects (Born and Markov approximations, respectively) can be made, leading to derivation of master equations of the form
\begin{equation}
\frac{\partial \rho_{\mathcal{S}}}{\partial t}=\frac{\left[H_{u},\rho_{\mathcal{S}}\right]}{i\hbar}+\mathcal{L}(\rho_{\mathcal{S}}). \label{eq:lindblad}
\end{equation}
This equation determines the time evolution of the system state $\rho_\mathcal{S}$ in terms of an effective Hamiltonian $H_{u}$ on $\mathcal{H}_\mathcal{S}$ and a Liouvillian superoperator $\mathcal{L}$. The first term on right-hand side implies a unitary dynamics whereas the second accounts for non-unitary environmental effects, such as dissipation and decoherence, which invariably come with an irreversible production of entropy \cite{louisell1973quantum,gardiner2004quantum,angelo2003aspectos}.
\par The master-equation formalism provided great advance in different fields of physics as for instance in studies involving Brownian quantum motion \cite{gardiner2004quantum,wen2004quantum}, interaction of particles with Lasers \cite{louisell1973quantum}, and spin-boson models \cite{de2017dynamics}. However, as the coupling with the environment becomes stronger, those approximations no longer apply and non-Markovian approaches are necessary \cite{tan2011non,tu2008non,zhang2012general,allahverdyan2005work}.  
\par Formalisms based on master equations are frequently considered for autonomous systems, in the sense that the system as a whole is considered without a driven mechanism, as in the case of Eq. \eqref{eq:effectiveH}. However, it is important to remark that the master equation formalism can also be adopted for non-autonomous quantum systems. For instance, consider the case in which the system $\mathcal{S}$ is described by the effective Hamiltonian \eqref{eq:effectiveH}, however it is desired to describe the dynamics of a sub-system $\mathcal{S}'$ of $\mathcal{S}$, such that $\mathcal{S}=\mathcal{S}'+\mathcal{R}$. This is the case, for example, of a system $\mathcal{S}'$ interacting with heat reservoirs $\mathcal{R}$; the environment $\mathcal{E}$ then controls the interaction by changing $a(t)$ in $H_{eff}(a)$. As a matter of fact, the definitions of work and heat given by Alicki were formalized regarding indirectly a system with this features, as will be described next.
\subsection{Alicki's approach}
\label{subsec:alicki} 
Alicki regarded an open system $\mathcal{S}'$ interacting with $N$ heat reservoirs systems $\mathcal{R}$. He then considered the master equation formalism for describing the evolution of $\rho_{\mathcal{S}'}$ such that
\begin{equation}
\frac{\partial \rho_{\mathcal{S}'}}{\partial t}=\frac{\left[H_{u},\rho_{\mathcal{S}'}\right]}{i\hbar}+\mathcal{L}(\rho_{\mathcal{S}'}), \label{eq:lindblad2}
\end{equation}
with the unitary part governed by $H_{u}=H_{\mathcal{S}'}+h_{t}$, where $H_{\mathcal{S}'}$ is the free Hamiltonian (as defined in Eq. \eqref{eq:hamitongeral}) of the open system $\mathcal{S}'$ and $h_{t}$ is an explicitly time-dependent self-joint operator representing the effect of changing some externally controllable condition in $\mathcal{E}$. In other words, $h_{t}$ represents the alteration in the energy of the system $\mathcal{S}'$, given some changes of the external conditions related with $\mathcal{E}$. The energy of the system $\mathcal{S}'$ is then written as
\begin{equation}
	E_{\mathcal{S}'}=\mathrm{Tr}\left(\rho_{\mathcal{S}'}(t)H_{u}(t)\right)\label{eq:Ealicki}
\end{equation} 
and its time derivative can be decomposed as
\begin{equation}
	\frac{\partial }{\partial t}E_{\mathcal{S}'}=\mathrm{Tr}\left(\frac{\partial \rho_{\mathcal{S}'}(t)}{\partial t}H_{u}(t)\right)+\mathrm{Tr}\left(\rho_{\mathcal{S}'}(t)\frac{\partial H_{u}(t)}{\partial t}\right).
\end{equation}
Alicki then introduces \cite{alicki1979quantum} the following identifications: 
\begin{equation}
		\dot{\mathcal{Q}}_{ak}\equiv \mathrm{Tr}\left(\frac{\partial \rho_{\mathcal{S}'}(t)}{\partial t}H_{u}(t)\right)\qquad\text{and}\qquad \dot{\mathcal{W}}_{ak}\equiv\mathrm{Tr}\left(\rho_{\mathcal{S}'}(t)\frac{\partial H_{u}(t)}{\partial t}\right),\label{eq:Alickdot}
\end{equation}
so that 
\begin{equation}
	\Delta E_{\mathcal{S}'}(t_{f}-t_{i})=\int_{t_{i}}^{t_{f}}\frac{\partial }{\partial t}E_{\mathcal{S}'}\,dt=\int_{t_{i}}^{t_{f}}\dot{\mathcal{Q}}_{ak}dt+\int_{t_{i}}^{t_{f}}\dot{\mathcal{W}}_{ak}dt= \mathcal{Q}_{ak}(t_{f}-t_{i})+ \mathcal{W}_{ak}(t_{f}-t_{i}).\label{eq:alickifund}
\end{equation}
Alicki interpreted $ \mathcal{W}_{ak}(t_{f}-t_{i})$ as \emph{the total work done on the system} $\mathcal{S}'$ and $ \mathcal{Q}_{ak}(t_{f}-t_{i})$ as \emph{the total heat that enters the system} $\mathcal{S}'$. Alicki's approach is predominant in the context of non-autonomous quantum systems. The model is mainly justified by the fact that the resulting equation is a clear expression of the first law of thermodynamics. Furthermore, Alicki proved that the efficiency of a cycle performed by an open system coupled to two Markovian reservoirs $\mathcal{R}_{1}$ and $\mathcal{R}_{2}$ at distinct temperatures $T_{1}$ and $T_{2}$, must be less than or equal to a Carnot one, i.e.
\begin{equation}
	 -\frac{\mathcal{W}_{ak}(t_{f}-t_{i})}{\mathcal{Q}_{ak}^{(1)}(t_{f}-t_{i})}\leq \frac{T_{1}-T_{2}}{T_{1}},
\end{equation}
where $\mathcal{Q}_{ak}^{(1)}$ is the heat flow from $\mathcal{R}_{1}$. This is an important result supporting Alicki's approach.
\par It is important to emphasize the essential role played by the choice of the system being analyzed. Had one considered the analysis of heat and work flow from a system composed, for instance, of $\mathcal{S}'$ and one of the $N$ reservoirs, then the energy of this new system, say $\mathcal{S}''$, would be computed by considering $H_{\mathcal{S}''}$ instead of $H_{\mathcal{S}'}$ in Eq. \eqref{eq:Ealicki}. Although it is not always trivial to define the form of $H_{\mathcal{S}^{''}}$ of these two interacting system, once such task is accomplished, it is possible, in principle, to determine the work and heat flow from Eqs. \eqref{eq:Alickdot} and \eqref{eq:alickifund}, regarding the point of view of Alicki. This key feature is revisited in chapter \ref{cap:kanaicaldirola}. 
\par As mentioned  in \cite{weimer2008local}, Alicki's proposal has received some criticisms: \emph{"(...) it is not obvious how to apply this definition to processes involving an internal transfer of work and heat" and "the microscopic foundation\footnote{In ref. \cite{weimer2008local}, the differential forms of Eq. \eqref{eq:Alickdot} were used, namely, $d\mathcal{W}_{ak}=\frac{\partial \mathcal{W}_{ak}}{\partial t}dt$ and $d\mathcal{Q}_{ak}=\frac{\partial \mathcal{Q}_{ak}}{\partial t}dt$. The form stated here was adopted in order to place the quoting into context.} of Eq. \eqref{eq:Alickdot} is rather unclear: As thermodynamic behavior may occur even in small quantum systems, it should, in principle, be possible to obtain $d\mathcal{W}_{ak}$ and $d\mathcal{Q}_{ak}$ even there"}. Other definitions for work has also been considered, frequently considering autonomous quantum systems, some of them considered next.
\subsection{Other approaches}
\par In the context of autonomous quantum machines, the mechanism that makes the external control in the non-autonomous case is explicitly included as a quantum system with a proper state space. Consider, for example, the system depicted in Fig. \ref{fig:autonomous}. 
\begin{figure}[!h]
	\begin{center}
		\includegraphics[angle=0, scale=0.5]{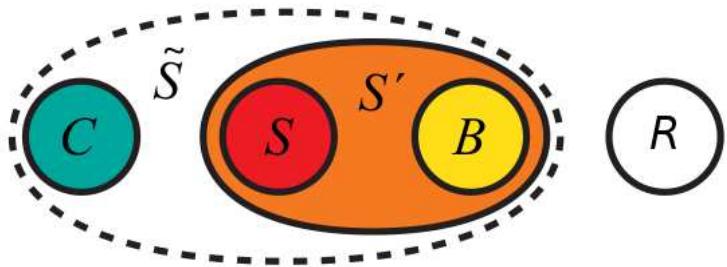}
		\caption{A system $S$ and an environment composed of a control $C$, a bath $B$ and a energy reservoir $R$. The reservoir $R$ provides energy for the system $S$, given the interaction between the latter with the bath $B$. Adapted from Ref. \cite{aaberg2018fully}.}
		\label{fig:autonomous}
	\end{center}
\end{figure}
\par The environment $\mathcal{E}$ is now composed of a reservoir $R$ that supplies energy to the system $\mathcal{S}\rightarrow S$, a thermal bath $B$, and a controller $C$ that governs the coupling between $B$ and $S$. By controlling the dynamic between the bath and the system $S$, it is possible to obtain the amount of energy that flows between the reservoir $R$ and the system $S$. As energy flows into (out of) the system, it is considered that work is performed on (by) the system. 
\par An analogous methodology has been adopted by other authors, however considering distinct mechanisms that provide energy to the system, and calling it as work or heat.
\par There are other approaches considered in the literature: 
\begin{itemize}
	\item Tonner and Mahler \cite{tonner2005autonomous} described work as $d\mathcal{W}=\zeta\cdot dL$, where $\zeta$ is an intensive variable and $L$ represents a mean value of the position of an oscillator;
	\item Weimer et al. \cite{weimer2008local} described work as the internal energy variation $\Delta U$, when satisfied the condition $\Delta S =0$, i.e. $\delta \mathcal{W}=\Delta U\iff \Delta S =0$;
	\item Departing from the closed-dynamics description of a quantum dipole  strongly coupled to a single-photon pulse, Valente et al. \cite{valente2017quantum} derived a master equation, from which they extracted a time-dependent Hamiltonian $H_{u}^{'}(t)$ for the dipole. Then, without any need for evoking self-joint operators $h_{t}$ to describe the effective open dynamics of the dipole, as in Alicki's proposal, the authors defined work and heat in the form 
	\begin{equation}
	\dot{\mathcal{Q}}_{ak}\equiv \mathrm{Tr}\left(\frac{\partial \rho_{\mathcal{S}'}(t)}{\partial t}H_{u}^{'}(t)\right)\qquad\text{and}\qquad \dot{\mathcal{W}}_{ak}\equiv\mathrm{Tr}\left(\rho_{\mathcal{S}'}(t)\frac{\partial H_{u}^{'}(t)}{\partial t}\right).\label{eq:valente}
	\end{equation}
	\item Bochkov and Kuzovlev \cite{bochkov1977general} defined the change in the Hamiltonian $H_{ef}$ of a system $\mathcal{S}$ as 
	\begin{equation}
	\mathcal{W}_{0}=H_{ef}(t)-H_{ef}(0)\equiv\int_{0}^{\tau}dt\lambda_{t}\dot{Q}_{t}
	\end{equation} 
	where $\lambda_{t}$ stands for the resulting external classical force acting on the system conjugated with the coordinate $Q_{t}$. The difference between the Hamiltonian at different instants of time $\mathcal{W}_{0}$ was interpreted as the work done on the system, during the interval $[0,t]$ \cite{campisi2011colloquium}. Considering these definitions, important fluctuation relations were obtained.
\end{itemize}  
\par Although the above definitions are not equivalent, they follow from the same line of reasoning: they are viewed as analogies with \emph{Classical Thermodynamics} perspective of work. Now, given that only few-particle quantum systems are under concern, one might wonder why not adopting a purely mechanical perspective for work, the one based on force and displacement, as in classical mechanics. The answer is immediate: because quantum mechanics precludes the notion of trajectory. In the next chapter, it is proposed a quantum mechanical form for work, which is still in analogy with Classical Mechanics although it does not demand the notion of trajectory.

%% file: descricaodomodelo_fisica.tex
\chapter{A NEW DEFINITION OF QUANTUM WORK}
\label{cap:definicaotrab}
\par It was discussed throughout the previous chapters, how work was defined in the scope of Classical Mechanics, Physical Statistics and Classical or Quantum Thermodynamics. The approach regarded under a Classical Mechanics perspective was not exported to Quantum Thermodynamics because it was argued that trajectories are not defined within Quantum Mechanics. As a result, a different path was adopted by the scientific community, in which work was defined in analogy with Classical Thermodynamics or Statistical Physics arguments. In this chapter, it is presented a different approach for work, within Quantum Mechanics domain, defined as an analogy with the Classical Mechanics perspective. It is avoided, however, the treatment of trajectories. 
\par The chapter is divided into two main parts: first, the new definition is established and the scope in which it is applied is formalized; then some properties of the new definition is approached at section \ref{sec:propw}. 
\section{\textbf{Formal definition}}
\par Work $\mathcal{W}_{cl}$ was established, under a Classical Mechanics perspective, in chapter \ref{cap:fundamentacao} as 
\begin{equation}
\mathcal{W}_{cl}= \int_{c}\underline{f}\cdot d\underline{x}\label{eq:wclnossa}.
\end{equation}
In the same chapter, it was mentioned that for a sufficiently continuous trajectories, Equation \eqref{eq:wclnossa} could be written as
\begin{equation}
\mathcal{W}_{cl}(c(t))= \int_{c(t)}\underline{f}\cdot \frac{d\underline{x}}{dt}dt= \int_{c(t)}\underline{f}\cdot \underline{v}dt\label{eq:wclnossa2}.
\end{equation}
where $\underline{v}$ is the velocity field. Unlike Eq. \eqref{eq:wclnossa}, the above formula does not employ the element of trajectory $d\underline{x}$; In fact, given the velocity and the force acting on a system, the path which the system pass through is implicitly taken into account in the time dependence of $\underline{v}$ and $\underline{f}$. Now, since there is a quantum counterpart for the velocity field, it is in principle possible to pursue a quantum mechanical work in analogy with Eq. \eqref{eq:wclnossa2}. To construct such definition, it is assumed that the system to be analyzed is non-relativistic and is represented by a unitarily evolving density operator $\rho(t)=\mathcal{U}_{t} \rho(0) \mathcal{U}_{t}^{\dagger}$ acting on $\mathcal{H}=\mathcal{H}_{\mathcal{S}}\otimes \mathcal{H}_{\mathcal{E}}$, where $\mathcal{H}_{\mathcal{E}}$ can be a multipartite space in general. In other words, the joint system $\mathcal{S}+\mathcal{E}$ is closed and the reduced density matrix is given by $\rho_{\mathcal{S}}(t) = Tr_{\mathcal{E}} \rho(t)$. 
\par While $\mathcal{E}$ is kept as generic as possible, $\mathcal{S}$ is considered to be a particle of constant mass $m$ moving, for simplicity, in one dimension. In this context, the pertinent observable for $\mathcal{S}$ is the position operator $X^{\mathcal{S}}$. In order to make the deductions to be presented in a simpler way, the superindex and subindex $\mathcal{S}$ will be suppressed ($X^{\mathcal{S}}\rightarrow X$ and $\rho_{\mathcal{S}}(t)\rightarrow\rho(t)$), unless some ambiguity may occur.
\par In search for a quantum counterpart for Eq. \eqref{eq:wclnossa2}, it is natural to employ the Heisenberg picture, through which one finds the Heisenberg position operator
\begin{equation}
		X_{H}(t)=\mathcal{U}_{t}^{\dagger}X\mathcal{U}_{t}
\end{equation} 
for any time $t$. It then follows from Eq. \eqref{eq:heisengeral} the time derivatives $\dot{X}_{H}$ and $\ddot{X}_{H}$, which are hereafter regarded as the quantum analogues of the classical velocity $v$ and acceleration $a$, respectively. On the other hand, the observable to be adopted as the quantum analogues of the force field acting on $\mathcal{S}$ has to be carefully defined. As in the Hamiltonian formulation of Classical Mechanics, here the notion of force will be taken to derive from the  model of interaction via the evaluation of the acceleration in the Heisenberg picture. Following the example given in Ref. \cite{sakurai2017modern}, one might consider a particle of mass $m$ and charge $q$. In terms of the time-independent scalar and vector potentials $\phi(\underline{X}_{H})$ and $\underline{A}_{H}\left(\underline{X}_{H}\right)$, the electric and magnetic fields read 
\begin{equation}
	\underline{E}_{H}=-\nabla_{\underline{X}_{H}}\phi,\,\,\,\,\,\,\underline{B}_{H}=-\nabla_{\underline{X}_{H}}\times \underline{A}_{H}
\end{equation}
From the Hamiltonian operator
\begin{equation}
	H=\frac{1}{2m}\left(\underline{P}-\frac{q}{c}\underline{A}\right)+q\phi,\label{eq:lorentz}
\end{equation}
one can compute the Heisenberg equations $\underline{\dot{X}}=\frac{[\underline{X},H]}{i\hbar}$ and $\underline{\dot{P}}=\frac{[\underline{P},H]}{i\hbar}$. Taking the time derivative of the former equation and using the latter, one finally obtains
\begin{equation}
	m\underline{\ddot{X}}_{H}=q\left[\underline{E}_{H}+\frac{1}{2c}\left(\underline{\dot{X}}_{H}\times \underline{B}_{H}-\underline{B}_{H}\times \underline{\dot{X}}_{H}\right)\right].\label{eq:lorentz2}
\end{equation}
The term in the right-hand side is then identified as the Lorentz force, in the quantum mechanical domain. Another example was explored in the context of Ehrenfest theorem, reproduced in section \ref{sec:Ehren}, where the forces that acts on the system were represented by
\begin{equation}
	-\partial_{X_{H}}\mathcal{V}(X_{H}).
\end{equation}
This is called a force operator in direct analogy with Newton's second law. The common characteristics of both example are:
\begin{itemize}
	\item The operators related with force have the expected physical dimension of force and;
	\item The quantum version of Newton's second law holds: the sum of the forces equals $m\underline{\ddot{X}}_{H}$.
\end{itemize}
Based on these aspects, the following concept is introduced.
 \begin{definicoes}
 	The resultant force operator $F_{H}^{R}$ associated with a particle of mass $m$ whose dynamics is governed by a Hamiltonian $H$ is given, in Heisenberg's picture, by 
 	\begin{equation}
 	F_{H}^{R}=\sum_{i=1}^{n}F_{H}^{i}=m\ddot{X}_{H}.\label{eq:newtonquantico}
 	\end{equation}  
 	\label{prin:1}
 \end{definicoes}
The notion of a particular basic force $F_{H}^{i}$ emerges by direct inspection of the parcels composing the acceleration, as in Eq. \eqref{eq:lorentz2}.
\par Given the definition of a quantum force operator, it is now possible to introduce a quantum mechanical definition of work. 
\begin{definicoes}
	The resultant quantum work imparted by the resultant quantum force $F_{H}^{R}$ (defined as in \eqref{eq:newtonquantico}) during a time interval $t_{f}-t_{i}$ on a particle of mass $m$ is given by
	\begin{equation}	
	\mathcal{W}_{q}^{R}\left(t_{f}-t_{i}\right)= \int_{t_{i}}^{t_{f}}\frac{1}{2}\left\langle\left\{F_{H}^{R},\dot{X}_{H}\right\}\right\rangle dt,\label{eq:wqmnossa}
	\end{equation}
	where $\left\langle\left\{F_{H}^{R},\dot{X}_{H}\right\}\right\rangle=\mathrm{Tr}\left(\rho(t_{i})\left\{F_{H}^{R},\dot{X}_{H}\right\}\right)$, 
	 $\left\{\,\,\,,\,\, \right\}$ is the anticommutator 
	\begin{equation}
	\left\{F_{H}^{R},\dot{X}_{H}\right\} =F_{H}^{R}\dot{X}_{H}+\dot{X}_{H} F_{H}^{R}\label{anticommu}
	\end{equation}
	and $\rho(t_{i})$ is the initial density operator acting on $\mathcal{H}_{\mathcal{S}}\otimes \mathcal{H}_{\mathcal{E}}$.
	\label{def:wr}
\end{definicoes}
\par The factor $\frac{1}{2}$ and the anticommutator $\left\{\,\,\,,\,\, \right\}$ were employed to make the work real. In fact, although $F_{H}^{R}$ and $\dot{X}_{H}$ are Hermitian operators, their product may not generally be. This can be checked via the expression
\begin{equation}
	\left(F_{H}^{R} \dot{X}_{H}\right)^{\dagger}=\dot{X}_{H}^{\dagger}  F_{H}^{R\,\dagger}=\dot{X}_{H} F_{H}^{R},
\end{equation}  
which does not necessarily equals $F_{H}^{R} \dot{X}_{H}$. Therefore, had one defined work as 
\begin{equation}
	\int_{t_{i}}^{t_{f}}\left\langle F_{H}^{R}  \dot{X}_{H}\right\rangle dt
\end{equation}
then complex values could possibly be obtained, an undesirable feature. In addition, notice that the work \eqref{eq:wqmnossa} is given in terms of an expectation value induced by the initial preparation $\rho(t_{i})$ of the system. It follows, therefore, that the work defined as in \eqref{eq:wqmnossa} is ensured to be a real function of time.
\par It is worth mentioning that the work \eqref{eq:wqmnossa} applies regardless of the characteristics of the coupling with the environment. Indeed, no assumption was made with respect to the strength of the interaction, the number of degrees of freedom in the environment, the amount of correlations in the initial state, or the like. Most importantly, given that its construction was guided by a mechanical rationale, it is expected for definition \eqref{eq:wqmnossa} to exhibit a well defined classical limit.
\section{\textbf{Properties of the new definition}}
\label{sec:propw} 
In order to compute $m\ddot{X}_{H}$, it can be discriminated different types of forces. One may consider, for instance, a case in which the Heisenberg acceleration $\ddot{X}_{H}$ is computed, given a velocity operator $\dot{X}_{H}$, from Eq. \eqref{eq:Heiseqm}, resulting in $m\ddot{X}_{H}=m\frac{d\dot{X}_{H}}{dt}=\frac{m}{i\hbar}[\dot{X}_{H},H]$. If the Hamiltonian operator can be written as  
\begin{equation}
	H=\sum_{i} \mathcal{V}_{i},\label{eq:hamiltonianmulti}
\end{equation}
 where each term $\mathcal{V}_{i}$ represents a different kind of interaction/energy, then
\begin{equation}
	\frac{m}{i\hbar}[\dot{X}_{H},\mathcal{V}_{i}]\label{eq:forceofeach}
\end{equation}
 may be related with a force applied on the system of mass $m$, associated with the interaction $\mathcal{V}_{i}$. The construction of an Hamiltonian in the form \eqref{eq:hamiltonianmulti} is actually considered in the examples given in the previous section. In the case of the Hamiltonian \eqref{eq:lorentz}, describing Classical Electrodynamics interactions, the terms $\frac{1}{2m}\left(\underline{P}-\frac{q}{c}\underline{A}\right)$ are related with the magnetic interaction and $q\phi$ with the electric interaction. Consequently, the Hamiltonian \eqref{eq:lorentz} could be written in the form \eqref{eq:hamiltonianmulti}, regarding $\mathcal{V}_{1}=\frac{1}{2m}\left(\underline{P}-\frac{q}{c}\underline{A}\right)$ and $\mathcal{V}_{2}=q\phi$,  describing the magnetic and electric interactions, respectively. As a result, it can be verified that the resultant force can be splitted into two terms
 \begin{equation}
 	\frac{m}{i\hbar}[\underline{\dot{X}}_{H},\mathcal{V}_{1}]=\frac{q}{2c}\left(\underline{\dot{X}}_{H}\times \underline{B}_{H}-\underline{B}_{H}\times \underline{\dot{X}}_{H}\right)
 \end{equation}
 and
  \begin{equation}
 \frac{m}{i\hbar}[\underline{\dot{X}}_{H},\mathcal{V}_{2}]=q\underline{E}_{H},
 \end{equation}
so that result \eqref{eq:lorentz2} follows by adding both interaction forces. In the case of the Hamiltonian considered in the Ehrenfest theorem, 
\begin{equation}
	H=\frac{P^2}{2m}+\mathcal{V}(X),
\end{equation} 
$\mathcal{V}(X)$ could be regarded as a central potential or a sum of two or more types of central interaction, i.e. $\mathcal{V}(X)=\mathcal{V}_{1}(X)+\mathcal{V}_{2}(X)+\cdots$ . Therefore, it is in principle possible to write the Hamiltonian in the form \eqref{eq:hamiltonianmulti}, where $\mathcal{V}_{1}$ corresponds to the kinetic energy term $\frac{P^2}{2m}$ and  each $\mathcal{V}_{i}$ ($i>1$) corresponds to a different type of central interaction, with a force identified by \eqref{eq:forceofeach}. 
\par From a general perspective, terms like $\mathcal{V}_{i}$ are already introduced in the classical Hamiltonian function, which, in principle, carries out the physical justification for introducing these fundamental interaction terms. In situations where many different force operators can be identified, it is sensible to speak of the work of a particular force.
\begin{definicoes}
	From definitions \ref{def:wr} and \ref{prin:1}, the resultant work can be divided as
	\begin{equation}	
	\mathcal{W}_{q}^{R}\left(t_{f}-t_{i}\right)= \int_{t_{i}}^{t_{f}}\frac{1}{2}\left\langle\left\{F_{H}^{R},\dot{X}_{H}\right\}\right\rangle dt=\sum_{i}\int_{t_{i}}^{t_{f}}\frac{1}{2}\left\langle\left\{F_{H}^{i},\dot{X}_{H}\right\}\right\rangle dt=\sum_{i}\mathcal{W}_{q}^{i},
	\end{equation}
	where 
	\begin{equation}
		\mathcal{W}_{q}^{i}=\int_{t_{i}}^{t_{f}}\frac{1}{2}\left\langle\left\{F_{H}^{i},\dot{X}_{H}\right\}\right\rangle dt
	\end{equation}
	 is defined as the quantum work imparted by a particular quantum force $F_{H}^{i}$ during a time interval $t_{f}-t_{i}$ on a particle of mass $m$.
	\label{def:tudo}
\end{definicoes}
\par It is also possible to define the \emph{kinetic energy} of $\mathcal{S}$:
\begin{definicoes}
	The kinetic energy of a particle of mass $m$ is defined as
	\begin{equation}
	K_{q}(t)=\frac{m}{2}\left\langle\dot{X}_{H}(t)\dot{X}_{H}(t)\right\rangle\equiv\frac{m}{2}\left\langle\left(\dot{X}_{H}\right)^2\right\rangle.
	\end{equation}
	\label{def:kq}
\end{definicoes}
\par Definitions \ref{def:wr} and \ref{def:kq} allow one to derive the following result.
\begin{teoremas}
	The resultant quantum work $\mathcal{W}_{q}^{R}$ (Definition \ref{def:wr}) done on a particle of mass $m$ in a time interval $t_{f}-t_{i}$ equals the variation in its kinetic energy $K_{q}$ (Definition \ref{def:kq}). Formally,
	\begin{equation}
	\mathcal{W}_{q}^{R}\left(t_{f}-t_{i}\right)=K_{q}(t_{f})-K_{q}(t_{i}).\label{eq:teotrabq1}
	\end{equation}
	\label{th:kinetiworq}
\end{teoremas}
\begin{proof}
	Since the quantum state is time-independent in the Heisenberg picture, one has
	\begin{equation}
		\frac{d}{dt}\left\langle\left(\dot{X}_{H}\right)^{2}\right\rangle=\left\langle\frac{d}{dt}\left(\dot{X}_{H}\right)^{2}\right\rangle=\left\langle\ddot{X}_{H}\dot{X}_{H}+\dot{X}_{H}\ddot{X}_{H}\right\rangle=\left\langle\left\{\ddot{X}_{H},\dot{X}_{H}\right\}\right\rangle.
	\end{equation}  
	Importing this result to Eq. \eqref{eq:wqmnossa} and regarding definition \ref{prin:1}, gives 
	\begin{equation}
	\begin{array}{rl}
	\displaystyle \mathcal{W}_{q}^{R}\left(t_{f}-t_{i}\right)&\displaystyle= \int_{t_{i}}^{t_{f}}\frac{1}{2}\left\langle\left\{F_{H}^{R},\dot{X}_{H}\right\}\right\rangle dt=\frac{m}{2} \int_{t_{i}}^{t_{f}}\left\langle\left\{\ddot{X}_{H},\dot{X}_{H}\right\}\right\rangle dt=\frac{m}{2} \int_{t_{i}}^{t_{f}}\frac{d}{dt}\left\langle\left(\dot{X}_{H}\right)^{2}\right\rangle dt\\
	&\displaystyle=\left.\frac{m}{2}\left\langle\left(\dot{X}_{H}\right)^{2}\right\rangle\right|_{t_{i}}^{t_{f}}.\\
	\end{array}
	\end{equation}
	Via definition \ref{def:kq}, the result \eqref{eq:teotrabq1} follows.
\end{proof}
Theorem \ref{th:kinetiworq} may be viewed as a quantum analogue of the classical result established in Eq. \eqref{eq:teotrabk}.
\subsection{Thermodynamical point of view}
In chapter \ref{cap:fundamentacao}, the notion of work was discussed within the scopes of Classical Thermodynamics and Statistical Physics. As put by Chandler \cite{chandler1987introduction} work has the form $\delta \mathcal{W}=\underline{f}_{G}\cdot \underline{dX}_{G}$, where $\underline{f}_{G}$ is the applied "force", $\underline{X}_{G}$ stands for a mechanical extensive variable and the $G$ subindex emphasizes that these variables are general and therefore can assume different forms \cite{chandler1987introduction}. A qualitative point of view was implicitly defended by Callen, differing heat (as an energy transfer \emph{"via the hidden atomic modes of motion"}) from work (as a transfer \emph{"that happen to be macroscopically observable"}). Both perspectives propose the idea that work is a macroscopically measurable form of energy transfer. On the other hand, heat is associated with the microscopic "hidden" motion. Meanwhile, by looking at Eq. \eqref{eq:wqmnossa}, it is not clear whether the proposed form for work $\mathcal{W}_{q}$ can be linked with a macroscopically observable measure. With these questioning in mind, the following discussion is conducted. 
\begin{definicoes}
	The resultant centroid work $\mathcal{W}_{c}^{R}$ of the force operator $F_{H}^{R}$ applied on $\mathcal{S}$ in a interval $t_{f}-t_{i}$ is defined as:
	\begin{equation}	
	\mathcal{W}_{c}^{R}\left(t_{f}-t_{i}\right)= \int_{t_{i}}^{t_{f}}\left\langle F_{H}^{R}\right\rangle \langle\dot{X}_{H}\rangle dt.\label{eq:wqmnossac}
	\end{equation}
	Furthermore, the centroid work $\mathcal{W}_{c}^{i}$ imparted by a particular quantum force $F_{H}^{i}$ during a time interval $t_{f}-t_{i}$ on a particle of mass $m$ is
	\begin{equation}
	\mathcal{W}_{c}^{i}=\int_{t_{i}}^{t_{f}}\langle F_{H}^{i}\rangle \langle\dot{X}_{H}\rangle dt, 
	\end{equation}
	such that 
	\begin{equation}
	\mathcal{W}_{c}^{R}=\sum_{i}\mathcal{W}_{c}^{i}.
	\end{equation}
	\label{def:tudoc}
\end{definicoes}
The idea addressed by definition \ref{def:tudoc} is that the "macroscopically observable" quantities mentioned before can be considered to be the mean force $\langle F_{H}\rangle$ and the mean velocity $\langle\dot{X}_{H}\rangle$, when concerning a macroscopic system. One might note that the quantum work reduces to the centroid work as $\frac{1}{2}\langle \left\{F_{H},\dot{X}_{H}\right\}\rangle \approx \langle F_{H}\rangle\langle\dot{X}_{H}\rangle$, that is, when some correlations can be neglected. In this sense, the centroid work is expected to be closer to the classical mechanical work. To better appreciate this point, consider the following examples.
\par A spring-block system in a frictionless table, mentioned at section \ref{sec:classical}, is depicted in Fig. \ref{fig:slidingblock}.
\begin{figure}[!h]
	\begin{center}
		\includegraphics[angle=0, scale=0.4]{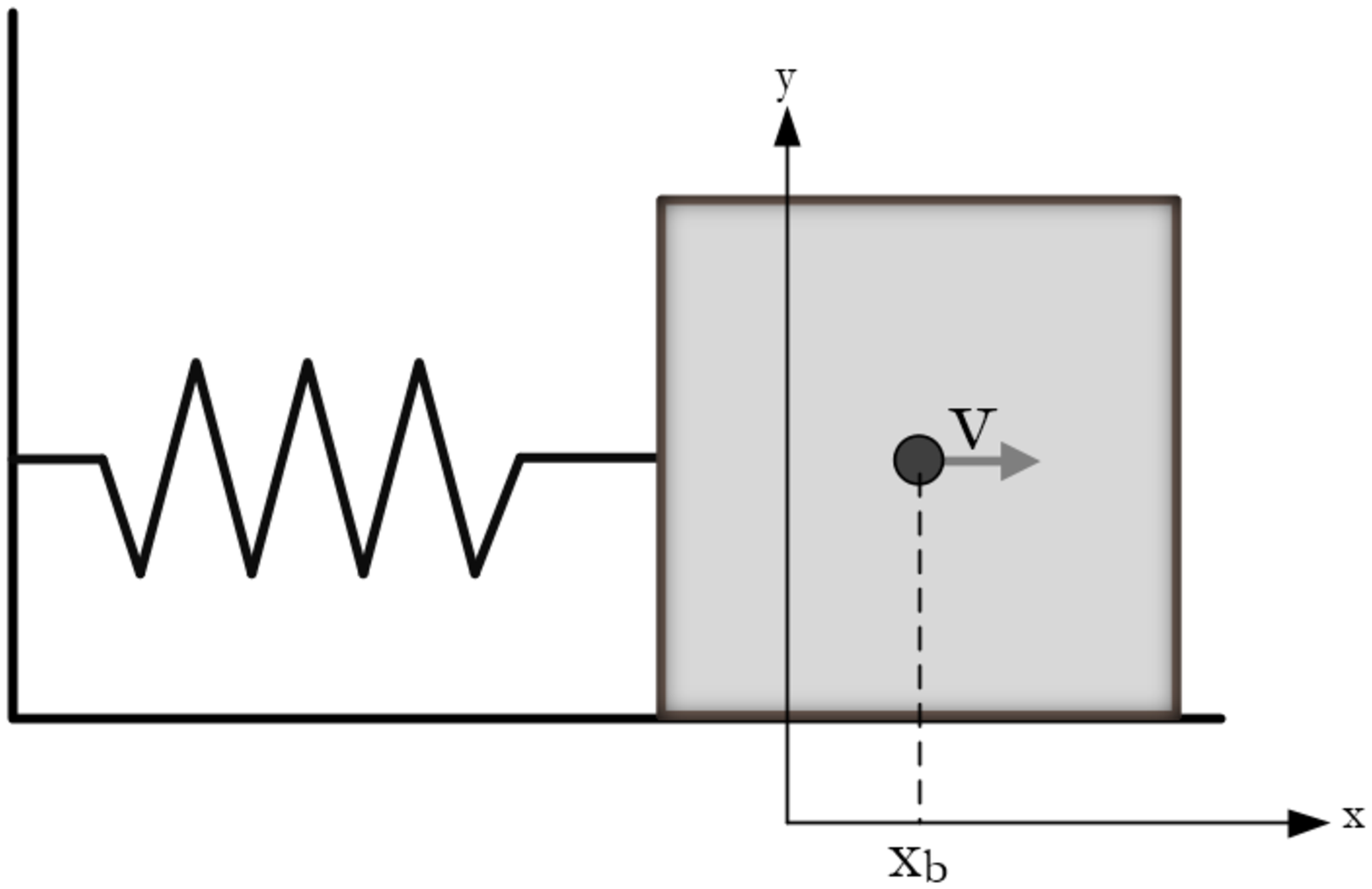}
		\caption{Sliding block on a frictionless table, revisited. The velocity of the block is $\dot{x}_{b}=\langle \dot{X}\rangle$ and the spring acts on the block with mean force $ -kx_{b}=\langle-kX\rangle=\langle F\rangle$.}
		\label{fig:slidingblock}
	\end{center}
\end{figure}
\par By continuously monitoring the block mean distance from the equilibrium position $\langle X\rangle\equiv x_{b}$ in a interval $(t_{i},t_{f})$, the mean speed $\langle \dot{X}\rangle=\frac{d}{dt}\langle X\rangle=\dot{x}_{b}$ can be computed, and, from Hooke's law, the mean force $\langle F\rangle=\langle-kX\rangle= -kx_{b}$ can be inferred. As a result, the centroid work can be evaluated by Eq. \eqref{eq:wqmnossac} as
\begin{equation}	
\mathcal{W}_{c}\left(t_{f}-t_{i}\right)= \int_{t_{i}}^{t_{f}}\left\langle F\right\rangle \langle\dot{X}\rangle dt=\int_{t_{i}}^{t_{f}}-kx_{b} \frac{d}{dt}x_{b} dt=-k\int_{x_{b}\left(t_{i}\right)}^{x_{b}\left(t_{f}\right)}x_{b} dx_{b} =-k\left. \frac{x_{b}^{2}}{2}\right|_{x_{b}\left(t_{i}\right)}^{x_{b}\left(t_{f}\right)}.
\end{equation}
 Similarly, from a thermodynamic perspective, a movable piston confining a gas is represented in Fig. \ref{fig:movablepiston}.
\begin{figure}[!h]
	\begin{center}
		\includegraphics[angle=0, scale=0.15]{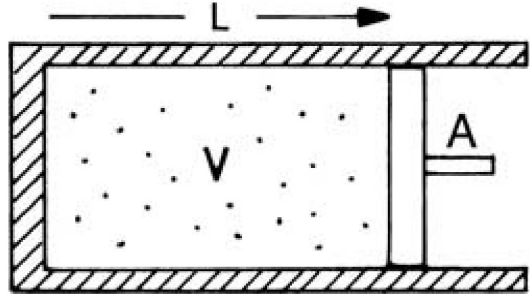}
		\caption{A movable piston with instantaneous position $\langle X\rangle=L $ confining a gas. The mean force applied by the gas with pressure $P$ on the piston is $PA$, where $A$ is the cross-sectional area.  Font: \cite{schwabl2006statistical}.}
		\label{fig:movablepiston}
	\end{center}
\end{figure}
\par The volume $V=LA$ occupied by the gas changes with the piston instantaneous position $\langle X\rangle=L $, where $A$ is the cross-sectional area. The mean force $\langle F\rangle$ applied by a gas with pressure $P$ on a piston with cross-sectional area $A$ can be approximated by $PA$ \cite{schwabl2006statistical}. Then, by measuring the mean velocity of the piston as $\langle \dot{X}\rangle=\frac{dL}{dt}$, the centroid work can be evaluated from Eq. \eqref{eq:wqmnossac}, resulting in
\begin{equation}	
\mathcal{W}_{c}\left(t_{f}-t_{i}\right)= \int_{t_{i}}^{t_{f}}\left\langle F\right\rangle \langle\dot{X}\rangle dt=\int_{t_{i}}^{t_{f}}PA \frac{dL}{dt} dt=\int_{L\left(t_{i}\right)}^{L\left(t_{f}\right)}PA dL=\int_{V\left(t_{i}\right)}^{V\left(t_{f}\right)}P dV,
\end{equation}
which recalls the frequently mentioned form for work applied on a gas $\delta W=-PdV$. Here, the sign is not negative since the gas is exerting work on the piston.  
\par It can be concluded, from the above special systems, that the centroid work can correctly reproduce the classical thermodynamical work. Of course, this is not to say that the purpose of definition \ref{def:tudoc} is to provide a general form of work to be applied in the description of macroscopic systems; \emph{it is just stated that once the macroscopically measurable properties are defined to be the mean force and mean position, then the centroid work is likely to satisfy some basic aspects of Classical Thermodynamics}. 
\par In the following, it is deduced a result that evidences the physical differences between $\mathcal{W}_{q}$ and $\mathcal{W}_{c}$.
\begin{teoremas}
	The resultant centroid work $\mathcal{W}_{c}^{R}$ done on a particle of mass $m$ during a time interval $t_{f}-t_{i}$ equals the variation in the centroid kinetic energy $K_{c}\equiv \frac{m}{2}\left\langle\dot{X}_{H}\right\rangle^{2}$. Formally,
	\begin{equation}
	\mathcal{W}_{c}^{R}\left(t_{f}-t_{i}\right)=\frac{m}{2}\left(\left\langle\dot{X}_{H}(t_{f})\right\rangle^{2}-\left\langle\dot{X}_{H}(t_{i})\right\rangle^{2}\right).\label{eq:teotrabq1c}
	\end{equation}
	\label{th:kinetiworqc}
\end{teoremas}
\begin{proof}
		It was employed similar arguments as the ones given for Theorem \ref{th:kinetiworq}. Considering $\frac{d}{dt}\left(\langle\dot{X}_{H}\rangle^{2}\right)=2\langle\ddot{X}_{H}\rangle\langle\dot{X}_{H}\rangle$, results in
	\begin{equation}
	\mathcal{W}_{c}^{R}\left(t_{f}-t_{i}\right)= \int_{t_{i}}^{t_{f}}\left\langle F_{H}^{R}\right\rangle \langle\dot{X}_{H}\rangle dt=m \int_{t_{i}}^{t_{f}}\langle\ddot{X}_{H}\rangle \langle\dot{X}_{H}\rangle dt=\frac{m}{2} \int_{t_{i}}^{t_{f}}\frac{d}{dt}\left(\langle\dot{X}_{H}\rangle^{2}\right) dt.\label{eq:kineticcele2}
	\end{equation}
	which implies the result \eqref{eq:teotrabq1c}.
\end{proof}
Now the difference between $\mathcal{W}_{q}^{R}$ and $\mathcal{W}_{c}^{R}$ can be made noticeable. From the variance definition for speed,
\begin{equation}
\left\langle \left(\Delta \dot{X}_{H}\right)^{2}\right\rangle=\left\langle \left(\dot{X}_{H}\right)^{2}\right\rangle-\left\langle \dot{X}_{H}\right\rangle^{2}\label{eq:squarmeandeviation}
\end{equation}
\begin{equation}
\left\langle \left(\dot{X}_{H}\right)^{2}\right\rangle=\left\langle \dot{X}_{H}\right\rangle^{2}+\left\langle \left(\Delta \dot{X}_{H}\right)^{2}\right\rangle.
\end{equation} 
Multiplying by  $\frac{m}{2}$ and taking variations in time lead to
\begin{equation}
\underbrace{\left.\frac{m}{2}\left\langle \left(\dot{X}_{H}\right)^{2}\right\rangle\right|_{t_{i}}^{t_{f}}}_{\mathcal{W}_{q}^{R}\left(t_{f}-t_{i}\right)}=\underbrace{\left.\frac{m}{2}\left\langle \dot{X}_{H}\right\rangle^{2}\right|_{t_{i}}^{t_{f}}}_{\mathcal{W}_{c}^{R}\left(t_{f}-t_{i}\right)}+\underbrace{\left.\frac{m}{2}\left\langle \left(\Delta \dot{X}_{H}\right)^{2}\right\rangle\right|_{t_{i}}^{t_{f}}}_{\mathcal{W}_{th}^{R}\left(t_{f}-t_{i}\right)}.\label{eq:centralprotrab}
\end{equation}
The extra term $\mathcal{W}_{th}^{R}\left(t_{f}-t_{i}\right)$ will be called \emph{thermal work} (for reasons to be explained). With the identifications given above, one sees that the resultant quantum work $\mathcal{W}_{q}^{R}\left(t_{f}-t_{i}\right)$ is the centroid work added with a quantum fluctuation term $\mathcal{W}_{th}^{R}\left(t_{f}-t_{i}\right)$, that is,
\begin{equation}
\mathcal{W}_{q}^{R}\left(t_{f}-t_{i}\right)= \mathcal{W}_{c}^{R}\left(t_{f}-t_{i}\right)+\mathcal{W}_{th}^{R}\left(t_{f}-t_{i}\right)
\end{equation}
Being directly proportional to the speed variance, it is clear that the thermal work $\mathcal{W}_{th}^{R}\left(t_{f}-t_{i}\right)$ $\equiv\left.\frac{m}{2}\left\langle \left(\Delta \dot{X}_{H}\right)^{2}\right\rangle\right|_{t_{i}}^{t_{f}}$ derives from intrinsic randomness. Now, recalling Callen's perspective, according to which heat is to be associated with some form of inaccessible random motion, then the conceptual link of $\mathcal{W}_{th}^{R}$ with heat becomes almost inescapable. However, up until a more deep examination be conducted, this connection will be put aside and the term thermal work will be used in reference to the ideas underlying the so-called thermal energy, as discussed in section \ref{sec:classical}.

\subsection{A Thermodynamics first law perspective}
\par Whenever the system of interest is a single point mass with no internal structure whatsoever, then there is no other form of energy that can be stored exclusively in this system but kinetic energy (see Sec. \ref{sec:classical} for a related discussion). If this particle is not isolated from the environment, then there will be some coupling energy, but arguably this energy is not confined to the particle; it actually is shared by the particle and the environment through the interacting potential. In this case, such interacting energy cannot be named "internal" to the particle. It follows from this that as far as a single particle of mass $m$ is elected as the system of interest, then its internal energy $U_i$ necessarily has to assume the form
	\begin{equation}
	U_{i}(t)=K_{q}(t)=\frac{m}{2}\left\langle\dot{X}_{H}(t)\dot{X}_{H}(t)\right\rangle.\label{eq:energiainternak}
	\end{equation}
Plugging this relation into Eq. \eqref{eq:centralprotrab} yields
\begin{equation}
\left.U_{i}\right|_{t_{i}}^{t_{f}}=\mathcal{W}_{c}^{R}\left(t_{f}-t_{i}\right)+\mathcal{W}_{th}^{R}\left(t_{f}-t_{i}\right),\label{eq:centralprotrab2}
\end{equation}
which in a more compact form, reads
\begin{equation}
	\Delta U_{i}= \mathcal{W}_{c}^{R}\left(t_{f}-t_{i}\right)+\mathcal{W}_{th}^{R}(t_{f}-t_{i}),
\end{equation}
where $\Delta \left(\cdot\right)=\left.\left(\cdot\right) \right|_{t_{i}}^{t_{f}}$. The centroid work $\mathcal{W}_{c}^{R}$, as argued in the previous subsection, can be viewed as a form of energy transfer that is encompassed by the notions of work usually adopted by Classical Thermodynamics. The first law of Thermodynamics, described in chapter \ref{cap:fundamentacao}, was written as: \emph{"the heat flux to a system in any process (at constant mole numbers) is simply the difference in internal energy between the final and initial states, diminished by the work done in that process"}. If the work done here, under a Thermodynamics point of view, is regarded as $ \mathcal{W}_{c}^{R}$ and the internal energy is given as in Eq. \eqref{eq:energiainternak}, then should not one definitely identify $ \mathcal{W}_{th}^{R}$ with heat?   This question will be left open for future studies, since this dissertation has no further arguments ensuring that this can indeed be the case. In particular, the eventual connection of $\mathcal{W}_{th}^{R}$ with entropy and its adequacy for many-particle systems are still to be investigated. In what follows other implications associated with the thermal work are pointed out.
\subsection{Centroid distribution work}
One of the topics covered in section \ref{sec:Ehren} was the description of the conditions under which the centroid of a (quantum or classical) distribution has its dynamics described in the same way as a Newtonian particle. That is, the conditions that allow for the approximation
  \begin{equation}
  m\left\langle \ddot{X}_{H}\right\rangle=-\left\langle\partial_{X_{H}} \mathcal{V}(X_{H})\right\rangle\approx -\partial_{\left\langle X_{H}\right\rangle} \mathcal{V}(\left\langle X_{H}\right\rangle)\label{eq:ehren2}.
  \end{equation}
Here the interest is in determining the necessary conditions for one to regard the total work realized on the particle as equal to the kinetic energy changes evaluated for the center of the distribution (the centroid). In order words, how accurate would a description of the total work done on a particle during an interval time $t_{f}-t_{i}$ in terms of $\frac{m}{2}\left.\left\langle\dot{X}_{H}\right\rangle^{2}\right|_{t_{i}}^{t_{f}}$ be? From the results above, this question refers to how close $\mathcal{W}_{q}^{R}$ and $\mathcal{W}_{c}^{R}$ are. To answer this question, one multiplies Eq. \eqref{eq:squarmeandeviation} by $\frac{m}{2}$ and organizes the result as
\begin{equation}
\frac{m}{2}\left\langle \dot{X}_{H}\right\rangle^{2}=\frac{m}{2}\left\langle \left(\dot{X}_{H}\right)^{2}\right\rangle\left(1-\frac{\frac{m}{2}\left\langle \left(\Delta \dot{X}_{H}\right)^{2}\right\rangle}{\frac{m}{2}\left\langle \left(\dot{X}_{H}\right)^{2}\right\rangle}\right).
\end{equation}
Therefore, if 
\begin{equation}
	\frac{\left\langle \left(\Delta \dot{X}_{H}\right)^{2}\right\rangle}{\left\langle \left(\dot{X}_{H}\right)^{2}\right\rangle}\ll 1 
\end{equation}
through the entire interval $t_{f}-t_{i}$, then  
\begin{equation}
	\frac{m}{2}\left\langle \dot{X}_{H}\right\rangle^{2}\approx\frac{m}{2}\left\langle \left(\dot{X}_{H}\right)^{2}\right\rangle.
\end{equation}
This means that by measuring the centroid kinetic energy $\frac{m}{2}\left\langle\dot{X}_{H}\right\rangle^{2}$ one approximately gets the total work imparted on the particle. In addition, one has
\begin{equation}
\mathcal{W}_{th}(t_{f}-t_{i})=\left.\frac{m}{2}\left\langle \left(\Delta \dot{X}_{H}\right)^{2}\right\rangle\right|_{t_{i}}^{t_{f}}\approx 0.
\end{equation}
This will be the case for instance in a situation where a heavy particle initially prepared in a Gaussian state with uncertainty $\Delta P = \frac{\Delta V}{m} $ and mean momentum $mV_{0}$ is submitted to an impulsive force that increases the mean momentum to $mV_{1}$ without appreciably changing the momentum uncertainty. In this case, the thermal work has no significant contribution to the quantum work.
\par In the next chapter, the theoretical framework developed here will be applied to a case study involving a dissipative model.

 %

%% file: osciladordissipativo.tex
\chapter{CASE STUDY: A DISSIPATIVE MODEL}
\label{cap:kanaicaldirola}
The dynamics of an open quantum system, as mentioned in section \ref{sec:QT}, may be described considering the system as non-autonomous or autonomous. In the former case the influences of the environment are modeled via a time-dependent Hamiltonian. One of the first efforts aiming at adding dissipation for quantum systems, along this line of reasoning, were put forward by Caldirola \cite{caldirola1941forze} and Kanai \cite{kanai1948quantization}, independently. By considering the Hamiltonian operator
\begin{equation}
	H=\mathrm{e}^{-2\lambda t}\frac{P^{2}}{2m_{0}}+k_{0}\mathrm{e}^{2\lambda t}\frac{X^{2}}{2},\label{eq:CKFirst}
\end{equation}
with the real positive constants $\lambda$, $m_{0}$ and $k_{0}$, it was possible to describe the dynamics of a dissipative system as for instance a damped oscillator. This model has been applied, however, with other interpretations and in different areas \cite{brown1991quantum,schuch1999effective, schuch1997nonunitary}. Recently, it was shown that the Hamiltonian \eqref{eq:CKFirst} can be derived from the traditional formalism of open quantum systems \cite{sun1995exact}.  Here this model will be applied to analyze in a simple effective way, the energetic flows in a dissipative system, with particular emphasis to the notion of quantum work (as defined in the previous chapter), Alicki's approach for heat and work, and the classical notion of work. Hereafter, the Caldirola-Kanai model, will be referred to simply as CK, for simplicity. 
\par First it is analyzed the Classical analogue of the CK model, in which case the Hamilton equations of motion are shown to reproduce the Newtonian equation typical of a damped oscillator. Next, similar equations of motion are derived by submitting the Hamiltonian \eqref{eq:CKFirst} to the Heisenberg picture. With that, the aforementioned comparison among the various approaches are conducted. Then, a statistical treatment is considered, where it is verified for which cases the presently-introduced quantum work reproduces the classical limit. Finally, it is analyzed some special effects regarding quantum superposition.

\section{\textbf{Classical system}}
\label{sec:purecl}
The classical Hamiltonian, analogous to the quantum operator is described as
\begin{equation}
	H_{cl}=\frac{p^{2}}{2 m_{0}}\mathrm{e}^{-2\lambda t}+k_{0}\frac{x^{2}}{2 }\mathrm{e}^{2\lambda t}.
\end{equation}
From the Hamilton equations \cite{goldstein2011classical}
\begin{equation}
	\begin{array}{lcr}
		\dot{p}=-\frac{\partial}{\partial x}H_{cl}=-k_{0}x\mathrm{e}^{2\lambda t}& \text{and}& \dot{x}=\frac{\partial}{\partial p}H_{cl}=\frac{p}{ m_{0}}\mathrm{e}^{-2\lambda t},\label{eq:hamiltonclassica}
	\end{array}
\end{equation}
the equation of motion of a damped oscillator follows:
\begin{equation}
\begin{array}{ccc}
\displaystyle\ddot{x}=\frac{\dot{p}}{ m_{0}}\mathrm{e}^{-2\lambda t}-2\lambda\frac{p}{ m_{0}}\mathrm{e}^{-2\lambda t}&\Rightarrow&\displaystyle\ddot{x}+2\lambda\dot{x}+\omega x=0\label{eq:motioneq}
\end{array}
\end{equation}
being $\omega^{2}=\frac{k_{0}}{m_{0}}$ the natural frequency of the harmonic oscillator. The system, therefore, can be regarded as the oscillator depicted in Figures \ref{fig:slidingblockresultsa} or \ref{fig:slidingblockresultsb}, where a drag force $F_{d}=-2\lambda m_{0}\dot{x}$, with constant drag $2\lambda m_{0}$ coefficient, is considered to be acting on the system, $x$ is the position of the block, and $\dot{x}$ is its velocity. 
\begin{figure}[!h]
	\begin{center}
		\begin{subfigure}[t]{0.4\textwidth}
			\includegraphics[angle=0, scale=0.3]{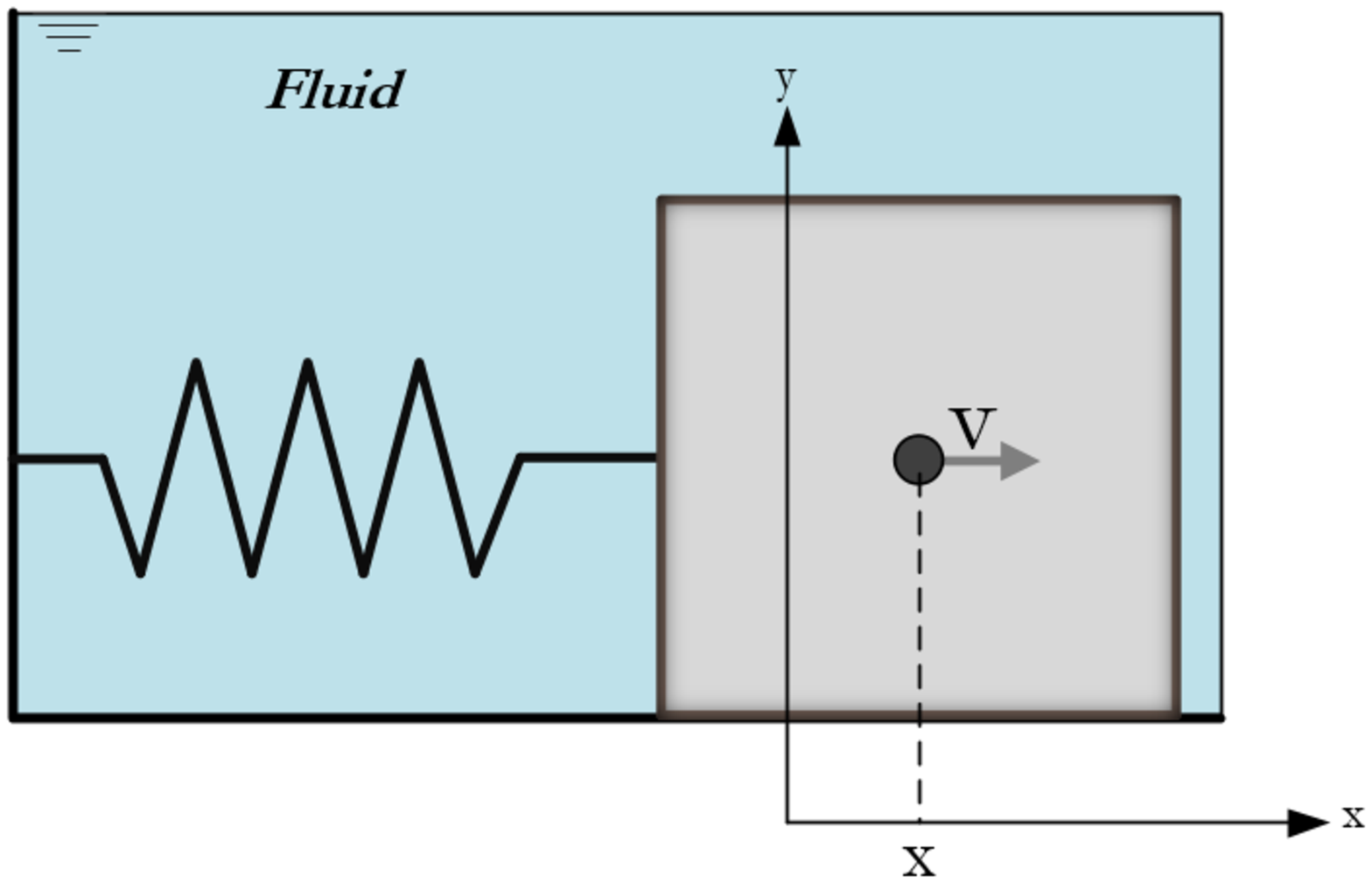} 
			\caption{Sliding block immersed.}
			\label{fig:slidingblockresultsa}
		\end{subfigure}
		\begin{subfigure}[t]{0.4\textwidth}
			\includegraphics[angle=0, scale=0.3]{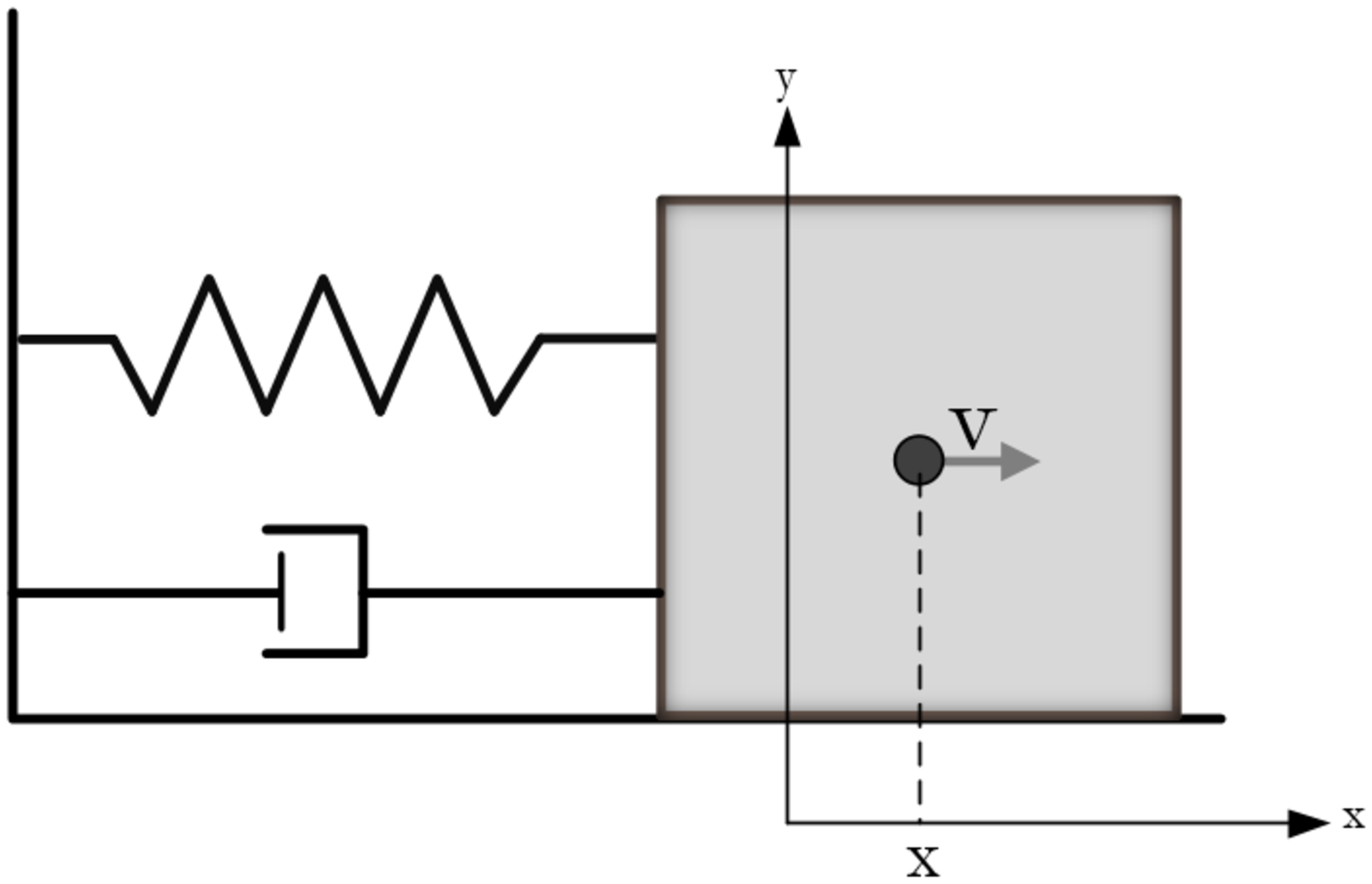}
			\caption{Viscous element added to the spring-block model.}
			\label{fig:slidingblockresultsb}
		\end{subfigure}
		\caption{Two forms of represent the system described by the equation of motion \eqref{eq:motioneq}: (a) the drag force $F_{d}=-2\lambda m_{0}\dot{x}$ is applied on the block of mass $m_{0}$ as it moves immersed on a fluid or (b) the viscous element, with viscosity coefficient $-2\lambda m_{0}$, exerts the same drag force.}
		\label{fig:slidingblockresults}
	\end{center}
\end{figure} 
\par The solutions are
\begin{equation}
x(\tau)=x_{0}\mathrm{e}^{-\tau}\left[\cosh\left(\zeta\tau\right)+\left(1+\frac{p_{0}}{x_{0}m_{0}\lambda}\right)\frac{\sinh\left(\zeta\tau\right)}{\zeta}\right]\label{eq:posicaocl1}
\end{equation}
and
\begin{equation}
p(\tau)=p_{0}\mathrm{e}^{\tau}\left[\cosh\left(\zeta\tau\right)-\left(1+\frac{k_{0}x_{0}}{p_{0}\lambda}\right)\frac{\sinh\left(\zeta\tau\right)}{\zeta}\right],\label{eq:ptau}
\end{equation}
where $\tau=\lambda t$ is a dimensionless time, $\zeta=\sqrt{1-\frac{\omega^{2}}{\lambda^{2}}}$, and $x_{0}$ and $p_{0}$ are the initial conditions. Notice that the canonical momentum $p$ differs from the mechanical momentum, since from Eq. \eqref{eq:hamiltonclassica} one has 
\begin{equation}
	m_{0}\dot{x}=p\mathrm{e}^{-2\lambda t}.\label{eq:ptau2}
\end{equation}
 Denoting the block velocity as $\dot{x}=v$, one finds from Eqs. \eqref{eq:ptau} and \eqref{eq:ptau2}
\begin{equation}
v(\tau)=\mathrm{e}^{-\tau}\left[\frac{p_{0}}{m_{0}}\cosh\left(\zeta\tau\right)-\left(\frac{p_{0}}{m_{0}}+\frac{k_{0}x_{0}}{m_{0}\lambda}\right)\frac{\sinh\left(\zeta\tau\right)}{\zeta}\right].\label{eq:velocity1ckcl}
\end{equation}
It follows that the classical kinetic energy can be written as
\begin{equation}
K_{cl}(\tau)=\frac{m_{0} v^{2}(\tau)}{2}=\frac{m_{0}}{2}\mathrm{e}^{-2\tau}\left[\frac{p_{0}}{m_{0}}\cosh\left(\zeta\tau\right)-\left(\frac{p_{0}}{m_{0}}+\frac{k_{0}x_{0}}{m_{0}\lambda}\right)\frac{\sinh\left(\zeta\tau\right)}{\zeta}\right]^{2}
\end{equation}
or, expanding the square,
\begin{equation}
\begin{array}{rl}
	\displaystyle K_{cl}(\tau)&\displaystyle=\frac{1}{2m_{0}}\mathrm{e}^{-2\tau}\left[p_{0}^{2}\cosh^{2}\left(\zeta\tau\right)-2p_{0}\cosh\left(\zeta\tau\right)\left(p_{0}+\frac{k_{0}x_{0}}{\lambda}\right)\frac{\sinh\left(\zeta\tau\right)}{\zeta}+\right.\\
	&\displaystyle+\left.\left(p_{0}+\frac{k_{0}x_{0}}{\lambda}\right)^{2}\frac{\sinh^{2}\left(\zeta\tau\right)}{\zeta^{2}}\right].\\
\end{array}
\label{eq:kineticccccc}
\end{equation}
Notice that the energy-work theorem fully applies to this effective model, for, given that the resultant force reads $m_{0}\ddot{x}$, one has
\begin{equation}
\mathcal{W}_{cl}(\tau-0)=	\int_{x\left(0\right)}^{x\left(\tau\right)}m_{0}\ddot{x}dx=m_{0}\int_{0}^{\tau}\ddot{x}\dot{x}d\tau^{'}=\frac{m_{0}}{2}\left.\dot{x}^{2}(\tau^{'})\right|_{0}^{\tau}\equiv\left.\frac{m_{0} v^{2}(\tau^{'})}{2}\right|_{0}^{\tau}.\label{eq:wcl2222}
\end{equation}
It is important to remark that $\mathcal{W}_{cl}(\tau)$ corresponds to the work done \emph{on the block} in the interval $\tau$. The work done on the spring or the environment are not explicitly encompassed by $\mathcal{W}_{cl}$, although the work done \emph{by} these systems on the block is. Throughout the present chapter, it will be adopted a framework within which the block is taken as the system of interest, that is, the one of which the energy flow is analyzed. In particular, the classical work \eqref{eq:wcl2222} will be compared with the results provided by the quantum and Liouvillian formalisms. The following scaling will prove convenient to make the desirable comparisons among the models.
\subsection{Dimensionless scheme}
The initial energy $E_{0}$ of the spring-block system can be written as 
\begin{equation}
E_{0}\equiv\frac{m_{0}\omega^{2} x_{0}^{2}}{2}+\frac{ p_{0}^{2}}{2m_{0}}.
\end{equation}
The total energy can be divided into two parcels, which can be written 
\begin{equation}
	\frac{m_{0}\omega^{2} x_{0}^{2}}{2}\equiv\varepsilon E_{0},\qquad\frac{ p_{0}^{2}}{2m_{0}}\equiv\left(1-\varepsilon\right)E_{0},
\end{equation}
with $0\leq\varepsilon\leq 1$. The parameter $\varepsilon$ gives the relative weight of the elastic energy in the total energy of the system. After some manipulations, Eq. \eqref{eq:kineticccccc} can be written as
\begin{equation}
\begin{array}{rl}
\displaystyle K_{cl}(\tau)&\displaystyle=E_{0}\mathrm{e}^{-2\tau}\left(\left(1-\varepsilon\right)\cosh^{2}(\zeta\tau)-\left(2\frac{\omega \sqrt{\varepsilon-\varepsilon^{2}}}{\lambda\zeta}+2\frac{1-\varepsilon}{\zeta}\right)\sinh(\zeta \tau)\cosh(\zeta\tau)+\right.\\
&\displaystyle\left.	+\left(\omega^{2}\frac{\varepsilon}{\lambda^{2}\zeta^{2}}+2\omega\frac{\sqrt{\varepsilon-\varepsilon^{2}}}{\lambda\zeta^{2}}+\frac{1-\varepsilon}{\zeta^{2}}\right)\sinh^{2}(\zeta \tau)\right).\\
\end{array}
\end{equation}
In order to have a dimensionless expression, both sides are divided by the initial kinetic energy $K_{cl}(0)\equiv K_{0}=\frac{ p_{0}^{2}}{2m_{0}}=\left(1-\varepsilon\right)E_{0}$, resulting in 
\begin{equation}
\begin{array}{rl}
\displaystyle \frac{K_{cl}(\tau)}{K_{0}}& \displaystyle =\frac{\mathrm{e}^{-2\tau}}{1-\varepsilon}\left(\left(1-\varepsilon\right)\cosh^{2}(\zeta\tau)-\left(2\frac{\omega \sqrt{\varepsilon-\varepsilon^{2}}}{\lambda\zeta}+2\frac{1-\varepsilon}{\zeta}\right)\sinh(\zeta \tau)\cosh(\zeta\tau)+\right.\\
&\displaystyle\left.	+\left(\omega^{2}\frac{\varepsilon}{\lambda^{2}\zeta^{2}}+2\omega\frac{\sqrt{\varepsilon-\varepsilon^{2}}}{\lambda\zeta^{2}}+\frac{1-\varepsilon}{\zeta^{2}}\right)\sinh^{2}(\zeta \tau)\right).\label{eq:classikal2}
\end{array}
\end{equation}
Consequently the dimensionless classical work, in an interval $t$, reads
\begin{equation}
\frac{\mathcal{W}_{cl}(\tau)}{K_{0}}=\frac{\left.K_{cl}(\tau^{'})\right|_{0}^{\tau}}{K_{0}}.\label{eq:classicalw2}
\end{equation}
Recalling that
\begin{equation}
\zeta= \sqrt{1-\frac{\omega^{2}}{\lambda^{2}}},\qquad \tau=\lambda t=\frac{\lambda}{\omega}\omega t,
\end{equation} 
it is possible to determine the value of the dimensionless work in a dimensionless interval $\omega t$, in terms of $\varepsilon$ and $\frac{\omega}{\lambda}$. The same scheme is used to write the position as
\begin{equation}
\displaystyle\frac{x(\tau)}{x_{0}}=\mathrm{e}^{-\tau}\left[\cosh (\zeta \tau)+\left(\frac{1}{\zeta}+\frac{\omega}{\lambda}\frac{\sqrt{1-\varepsilon }}{\zeta\sqrt{\varepsilon }}\right)\sinh(\zeta \tau)\right].
\end{equation}
This scaling is not well defined when $\varepsilon=\frac{m_{0}\omega^{2} x_{0}^{2}}{2E_{0}}=0$, since in this case $x_{0}=0$. In such situation, one can use a different scaling. From Eq. \eqref{eq:posicaocl1} one has
\begin{equation}
x(\tau)=\mathrm{e}^{-\tau}\frac{p_{0}}{m_{0}\lambda}\frac{\sinh\left(\zeta\tau\right)}{\zeta},
\end{equation}
which allows one to write
\begin{equation}
\frac{x(\tau)}{x_{m}}=\mathrm{e}^{-\tau}\frac{\sinh\left(\zeta\tau\right)}{\zeta},
\end{equation}
where $x_{m}\equiv \frac{p_{0}}{m_{0}\lambda}$ emerges as the convenient scale factor. 
\par The time-dependent terms appearing in the kinetic energy \eqref{eq:classikal2}, namely, $\sinh^{2}\left(\zeta\tau\right)$, $\sinh\left(\zeta\tau\right)\cosh\left(\zeta\tau\right)$ and $\cosh^{2}\left(\zeta\tau\right)$, will be recurrent in the results to be derived later. In other to further simplify the analysis and compare the results, Eq. \eqref{eq:classikal2} is rewritten as the inner product
\begin{equation}
\displaystyle \frac{K_{cl}(\tau)}{K_{0}}=\mathrm{e}^{-2\tau}\underline{\alpha}_{cl}\cdot\underline{\Gamma}(t),\label{eq:classikal}
\end{equation}
where
\begin{equation}
	\underline{\alpha}_{cl}=\frac{1}{1-\varepsilon}\left[
	\begin{array}{c}
	1-\varepsilon\\\displaystyle
	-\frac{2}{\zeta}\left(\frac{\omega \sqrt{\varepsilon-\varepsilon^{2}}}{\lambda}+\left(1-\varepsilon\right)\right)\\
	\displaystyle\left(\frac{\omega^{2}\varepsilon}{\lambda^{2}\zeta^{2}}+2\omega\frac{\sqrt{\varepsilon-\varepsilon^{2}}}{\lambda\zeta^{2}}+\frac{1-\varepsilon}{\zeta^{2}}\right)
	\end{array}
	\right], \qquad \underline{\Gamma}(t)=\left[
	\begin{array}{c}
	\displaystyle\cosh^{2}(\zeta\tau)\\
	\displaystyle\sinh(\zeta \tau)\cosh(\zeta\tau)\\
	\displaystyle\sinh^{2}(\zeta\tau)
	\end{array}
	\right].\label{eq:alphacldef}
\end{equation}
In this form, all the time-dependent terms are written in the compact form of a vector $\mathrm{e}^{-2\tau}\underline{\Gamma}(t)$ and the classical work \eqref{eq:classicalw2} can be written in a simpler form
\begin{equation}
\frac{\mathcal{W}_{cl}(\tau)}{K_{0}}=\underline{\alpha}_{cl}\cdot\left(\mathrm{e}^{-2\tau}\underline{\Gamma}(t)-\underline{\Gamma}(0)\right).\label{eq:classicalw}
\end{equation}
\par It is now convenient to analyze two main scenarios: an underdamped regime, for which  $\frac{\omega}{\lambda}> 1$, and an overdamped one, where $\frac{\omega}{\lambda}<1$. In the former case, $\zeta^{2}=1-\frac{\omega^{2}}{\lambda^{2}}<0$. Therefore, 
\begin{equation}
	\zeta=i\sqrt{\frac{\omega^{2}}{\lambda^{2}}-1}\label{eq:zeta}
\end{equation} 
where $i^{2}=-1$ and
\begin{align}
\underline{\alpha}_{cl}=\frac{1}{1-\varepsilon}&\left[
\begin{array}{c}
1-\varepsilon\\\displaystyle
\frac{2i}{\sqrt{\frac{\omega^{2}}{\lambda^{2}}-1}}\left(\frac{\omega \sqrt{\varepsilon-\varepsilon^{2}}}{\lambda}+\left(1-\varepsilon\right)\right)\\
\displaystyle\left(\omega^{2}\frac{\varepsilon}{\lambda^{2}\zeta^{2}}+2\omega\frac{\sqrt{\varepsilon-\varepsilon^{2}}}{\lambda\zeta^{2}}+\frac{1-\varepsilon}{\zeta^{2}}\right)
\end{array}
\right],\\
\underline{\Gamma}(t)=&\left[
\begin{array}{c}
\displaystyle\cos^{2}(\sqrt{\omega^{2}-\lambda^{2}}t)\\
\displaystyle i\sin(\sqrt{\omega^{2}-\lambda^{2}}t)\cos(\sqrt{\omega^{2}-\lambda^{2}}t)\\
\displaystyle-\sin^{2}(\sqrt{\omega^{2}-\lambda^{2}}t)
\end{array}
\right].
\end{align}
$\underline{\Gamma}(t)$ has a periodical behavior, although $\mathrm{e}^{-2\tau}\underline{\Gamma}(t)$ has not: the oscillation amplitude will decay at an exponential time rate. This behavior is depicted in Fig. \ref{fig:positionclassicaluod}, where the underdamped regime is considered with  $\frac{\omega}{\lambda}=10$ and $\varepsilon=0$.
\begin{figure}[h]
	\begin{center}
		\begin{subfigure}[h]{1\textwidth}
			\centering
			\includegraphics[angle=0, scale=0.55]{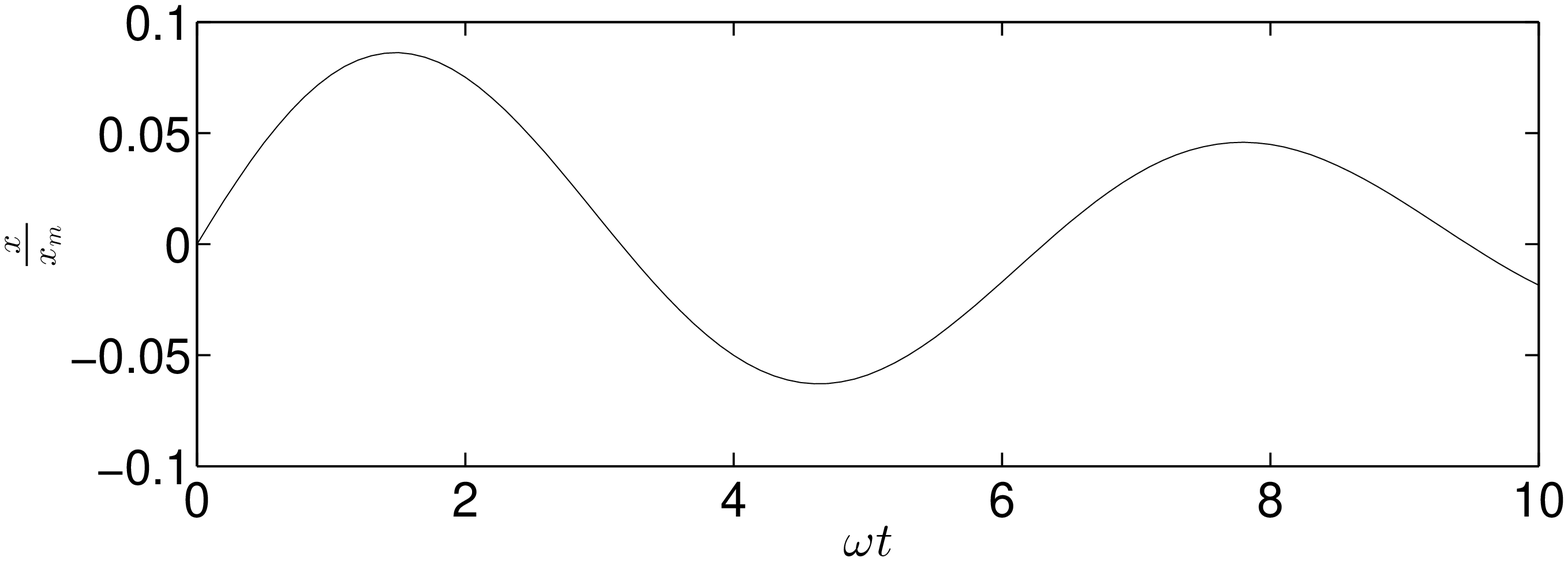} 
			\caption{ }
			\label{fig:positionclassicalud}
		\end{subfigure}
		\begin{subfigure}[h]{1\textwidth}
			\centering
			\includegraphics[angle=0, scale=0.55]{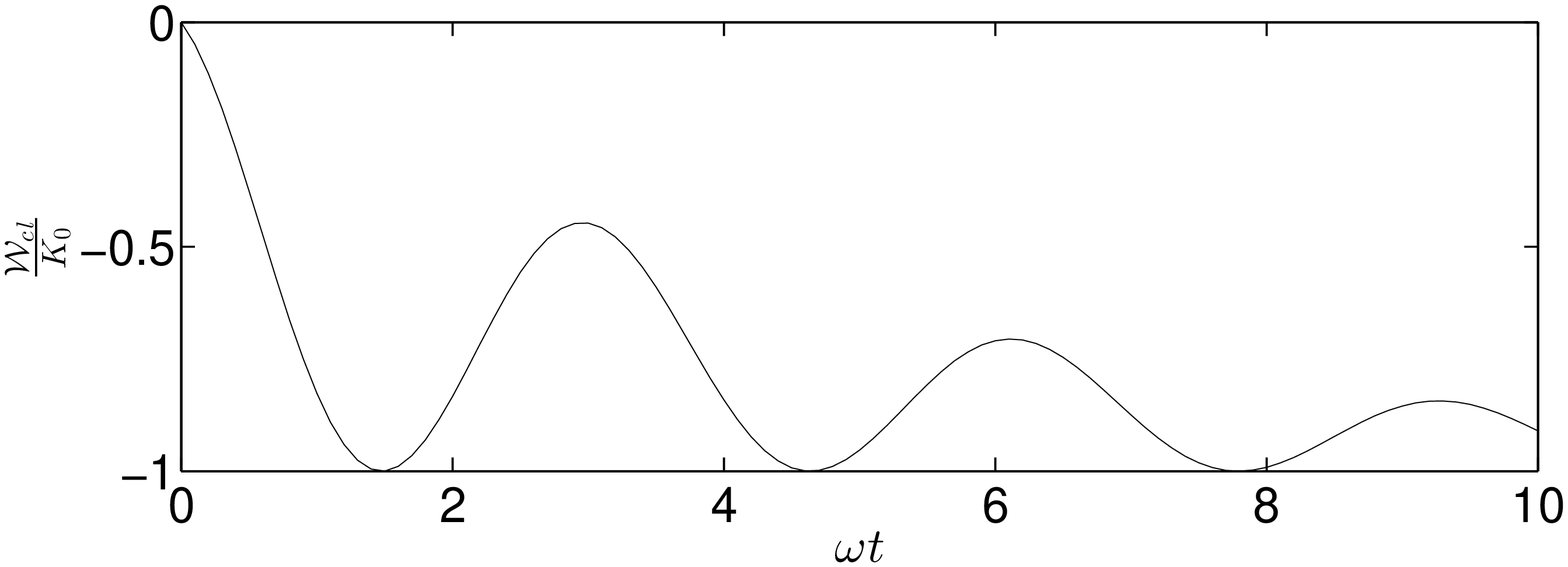}
			\caption{ }
			\label{fig:workclassicalud}
		\end{subfigure}
		\caption{Block dimensionless (a) position $\frac{x(\tau)}{x_{m}}$ and (b) classical work $\frac{\mathcal{W}_{cl}}{K_{0}}$ as a function of the dimensionless time $\omega t$ within the time interval $[0,10]$, for an underdamped regime ($\frac{\omega}{\lambda}=10$). The block is initially in the equilibrium position, with no potential energy stored ($\varepsilon=0$).}
		\label{fig:positionclassicaluod}
	\end{center}
\end{figure}
\par Because it has been chosen that the initial energy of the system is purely kinetic ($\varepsilon=0$, an initial condition that will be employed throughout this chapter), one finds that the kinetic energy decreases, so that $\mathcal{W}_{cl}(t)=K_{cl}(t)-K_{0}<0$. It can be stated thus that "the environment is doing work on the system, extracting its internal energy". The average decay, as expected, assumes an exponential form and some oscillations can be observed. 
When the regime considered is overdamped, then $\frac{\omega}{\lambda}<1$ and $\zeta$ has no imaginary part:
\begin{equation}
	\zeta=\sqrt{1-\frac{\omega^{2}}{\lambda^{2}}}.
\end{equation}
The time dependence is therefore described by 
\begin{equation}
\mathrm{e}^{-2\tau}\underline{\Gamma}(t)=\mathrm{e}^{-2\tau}\left[
\begin{array}{c}
\displaystyle\cosh^{2}(\sqrt{1-\frac{\omega^{2}}{\lambda^{2}}}\lambda t)\\
\displaystyle\sinh(\sqrt{1-\frac{\omega^{2}}{\lambda^{2}}}\lambda t)\\
\displaystyle\sinh^{2}(\sqrt{1-\frac{\omega^{2}}{\lambda^{2}}}\lambda t)
\end{array}
\right].
\end{equation}
Since $\sqrt{1-\frac{\omega^{2}}{\lambda^{2}}}\leq 1$ and recalling that $\tau=\lambda t$, then $\mathrm{e}^{\sqrt{1-\frac{\omega^{2}}{\lambda^{2}}}2\lambda t}\leq \mathrm{e}^{2\tau}$. As a consequence, the terms of the hyperbolic functions increases slower than $\mathrm{e}^{-2\tau}$ decreases. Therefore, $\mathrm{e}^{-2\tau}\underline{\Gamma}(t)$ will decrease exponentially. Such considerations are depicted in Fig. \ref{fig:workclassicaluod}, where it is regarded a case in which $\frac{\omega}{\lambda}=0.1$ and $\varepsilon=0$.
\begin{figure}[h]
	\begin{center}
		\begin{subfigure}[h]{1\textwidth}
			\centering
			\includegraphics[angle=0, scale=0.55]{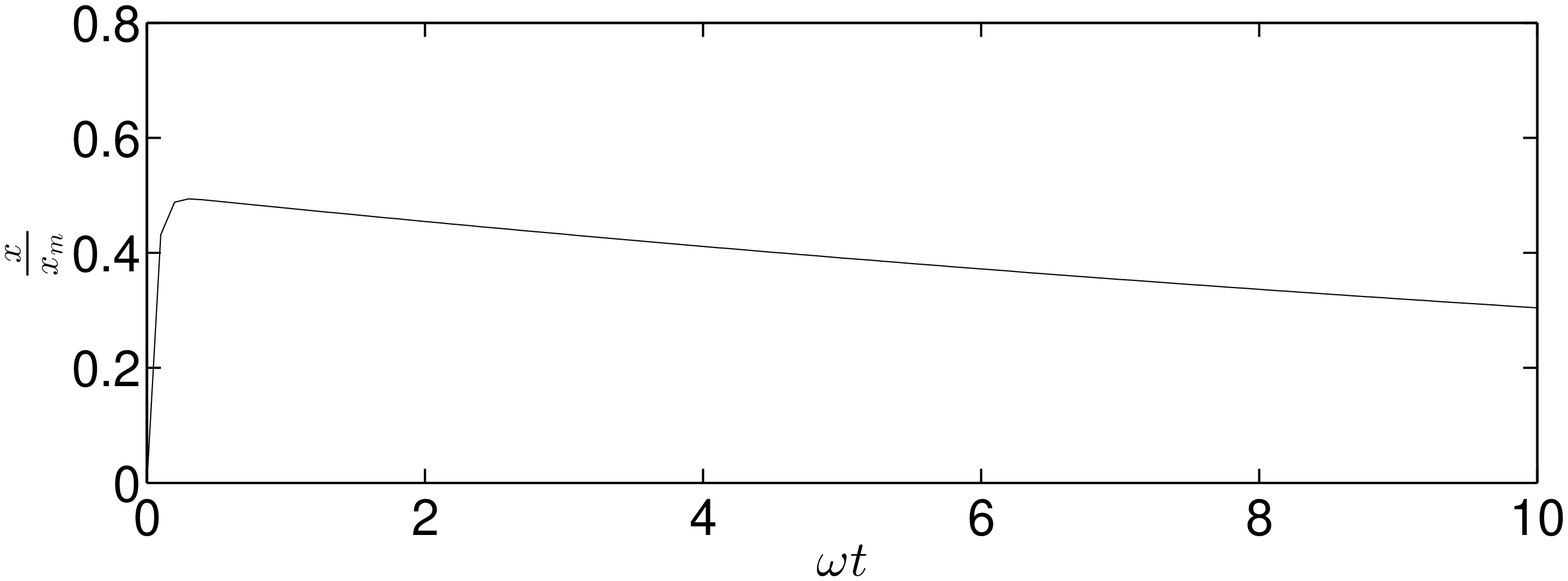} 
			\caption{Position.}
			\label{fig:positionclassicalod}
		\end{subfigure}
		\begin{subfigure}[h]{1\textwidth}
			\centering
			\includegraphics[angle=0, scale=0.55]{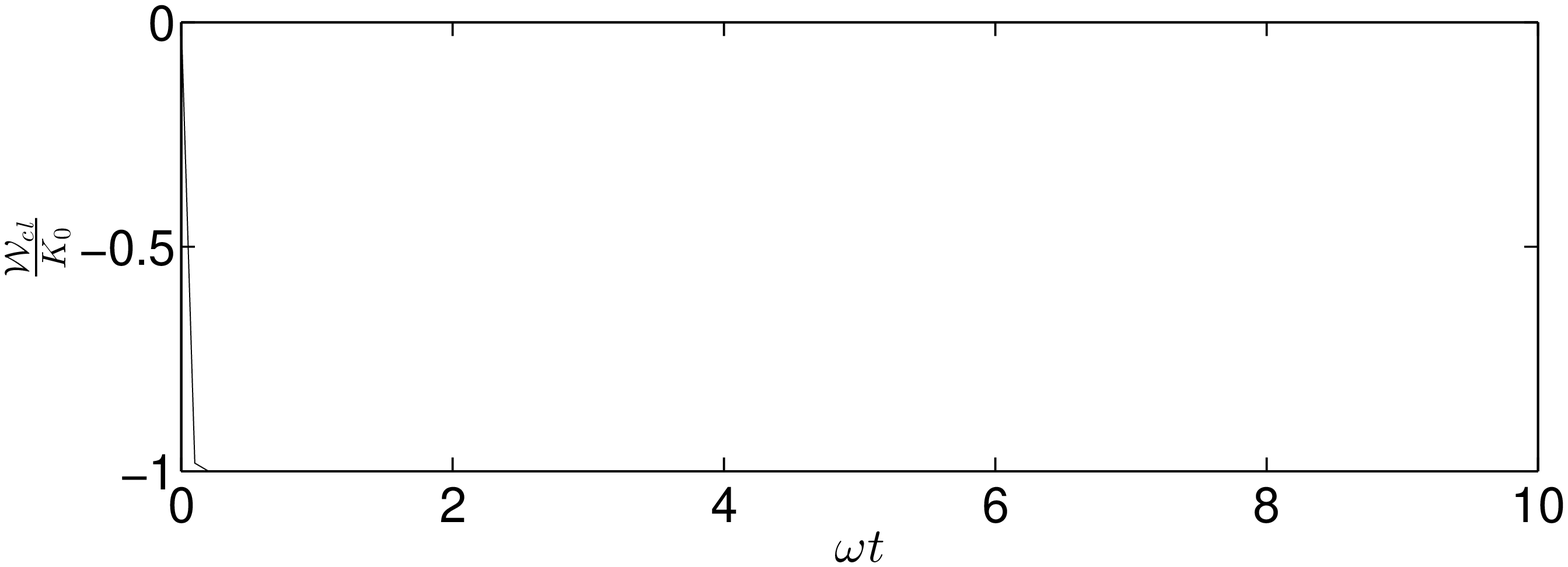}
			\caption{Work done by the environment.}
			\label{fig:workclassicalod}
		\end{subfigure}
		\caption{Block dimensionless position $\frac{x(\tau)}{x_{m}}$ and classical work $\frac{\mathcal{W}_{cl}}{K_{0}}$ as a function of the dimensionless time $\omega t$ within the time interval $[0,10\omega t]$, for an overdamped regime ($\frac{\omega}{\lambda}=0.1$). The block is initially in the equilibrium position, with no potential energy stored ($\varepsilon=0$).}
		\label{fig:workclassicaluod}
	\end{center}
\end{figure} 
The kinetic energy undergoes an exponential decay, without oscillations, as expected. The figures \ref{fig:positionclassicaluod} and \ref{fig:workclassicaluod} evidence the features that are expected for the damped oscillators, showing that the CK classic model can suitably describe such a system. In what follows, two specializations are considered for the model.
\subsection{Purely conservative or dissipative forces}
\par From the CK Hamiltonian, it is immediately seen that, for $\lambda=0$, one has a simple harmonic oscillator, in which case only a conservative elastic force acts on the block. No environment influence remains at all. The equation \eqref{eq:classikal} then reduces to
\begin{equation}
\frac{K_{cl}(\tau)}{K_{0}}=\underline{\alpha}_{cl}^{\lambda=0}\cdot\underline{\Gamma}^{\lambda=0}(t),
\end{equation}
where 
\begin{equation}
\underline{\alpha}_{cl}^{\lambda=0}=\frac{1}{1-\varepsilon}\left[
\begin{array}{c}
1-\varepsilon\\\displaystyle
-2\sqrt{\varepsilon-\varepsilon^{2}}\\
\varepsilon
\end{array}
\right], \qquad 
 \underline{\Gamma}^{\lambda=0}(t)=\left[
\begin{array}{c}
\displaystyle\cos^{2}(\omega t)\\
\displaystyle \sin(\omega t)\cos(\omega t)\\
\displaystyle-\sin^{2}(\omega t)
\end{array}
\right].\label{eq:alphacll0}
\end{equation}
 The sinusoidal form of the equations reveals the usual scenario: the spring performs work on the block, initially extracting its kinetic energy and storing it in potential form and then eventually delivering it back to the block, in a cyclic way. On the other hand, by setting $k_{0}=0$ one "turns off" the spring, which leads the CK model to emulate a purely dissipative drag force. In this case, 
\begin{equation}
\frac{K_{cl}(\tau)}{K_{0}}=\mathrm{e}^{-4\tau} \qquad \text{and} \qquad \frac{\mathcal{W}_{cl}(\tau)}{K_{0}}= \mathrm{e}^{-4\tau}-1.
\end{equation}
As expected, one finds that the environment constantly extracts the system energy, which thus asymptoticly drops to zero.  
\par It should be clear from the above that the CK classical model correctly simulate many interesting dynamical systems, as for instance a harmonic oscillator, a damped oscillator, and a continuously decelerated moving body. Next, the quantum counterpart of this model is investigated.
\section{\textbf{Quantum system}}
\label{sec:Quantumsystems}
The CK Hamiltonian operator, as described in the introduction of the chapter, reads
\begin{equation}
H=\frac{P^{2}}{2 m_{0}}\mathrm{e}^{-2\lambda t}+k_{0}\frac{X^{2}}{2 }\mathrm{e}^{2\lambda t}.\label{eq:Hamiltonianoquantuco}
\end{equation}
Within the Heisenberg picture (see Eq. \eqref{eq:heisengeral}), the velocity operator is computed as
\begin{equation}
\displaystyle V_{H}\equiv \dot{X}_{H}\displaystyle=\frac{\left[X_{H},H_{H}\right]}{i\hbar}=\frac{P_{H}}{m_{0}}\mathrm{e}^{-2\lambda t}
\label{eq:vhcalc}
\end{equation}
which, in the Schrödinger picture reduces to
\begin{equation}
V=\mathcal{U}V_{H}\mathcal{U}_{t}^{\dagger}=\frac{P}{m_{0}}\mathrm{e}^{-2\lambda t}.\label{eq:conversaoheischron2}
\end{equation}
The acceleration operator turns out to be
\begin{equation}
\begin{array}{rl}
\displaystyle A_{H}\equiv \ddot{X}_{H}&\displaystyle=\frac{\left[V_{H},H_{H}\right]}{i\hbar}+\mathcal{U}_{t}^{\dagger}\frac{\partial V}{\partial t}\mathcal{U}_{t}\\
&\displaystyle=\frac{\left[\frac{P_{H}}{m_{0}}\mathrm{e}^{-2\lambda t},k_{0}\frac{X_{H}^{2}}{2 }\mathrm{e}^{2\lambda t}\right]}{i\hbar}+\frac{P_{H}}{m_{0}}\frac{\partial \mathrm{e}^{-2\lambda t}}{\partial t}\\
&\displaystyle=-\omega^{2}X_{H}-2\lambda V_{H}\\
\end{array}
\end{equation}
yielding, via Eq. \eqref{eq:vhcalc},
\begin{equation}
\ddot{X}_{H}+2\lambda\dot{X}_{H}+\omega^{2}X_{H}=0.\label{eq:motioneq2}
\end{equation} 
The mathematical structure of this result is identical to that of a classical damped oscillator (see Eq. \eqref{eq:motioneq}). However, quantum mechanics is not about operators only; it is necessary to take quantum states into account. In other words, the Eq. \eqref{eq:motioneq2} \emph{per se} does not reveal what can be predicted for the system. For one to proceed with the calculation of the quantum work, it is adopted the Gaussian wave function
\begin{equation}
\displaystyle\langle x|\psi_{0}\rangle=\Psi(x,0)=\frac{\mathrm{e}^{-\frac{\left(x-x_{0}\right)^{2}}{4\Delta_{x,0}^{2}}}\mathrm{e}^{\frac{ip_{0}x}{\hbar}}}{\left(2\pi \Delta_{x,0}^{2}\right)^{\frac{1}{4}}},\label{eq:Gaussianresult}
\end{equation}
where $x_{0}$ and $p_{0}$, are to be identified with the classical initial conditions. This state is very convenient because it allows for (i) analytical computations (as will be seen) and (ii) a continuous interpolation between the regimes of high and low spatial delocalization (which is regulated by the uncertainty $\Delta_{x,0}$). In addition, because this is a minimum-uncertainty state ($\Delta_{p,0}\Delta_{x,0}=\frac{\hbar}{2}$), one may expect, in light of the Ehrenfest theorem, to have a good agreement between quantum and classical results at short times.
\par Now, standard techniques of quantum mechanics are employed. Using the CK Hamiltonian one can compute, via a method based on Lie algebras \cite{iesus}, the time propagator $\mathcal{U}_{t}$ and then the time-evolved vector state $\ket{\psi(t)}=\mathcal{U}_{t}\ket{\psi_{0}}$, where $\ket{\psi_{0}}=\int_{-{\infty}}^{\infty}\langle x|\psi_{0}\rangle\ket{x}dx$. From this, the evolved wave function $\Psi(x,t)=\langle x|\psi(t)\rangle$ is determined and expectations values are computed. The details of this lengthy calculation are given in Appendix \ref{cap:kanaicaldirolasol}. The main results are summarized below.
\begin{equation}
\Psi(x,t)= A(t)\exp\left[-\frac{\left(x-\mathrm{e}^{-\frac{c_{0}}{2}}\left(x_{0}-c_{-}p_{0}\right)\right)^{2}}{\mathrm{e}^{-c_{0}}\left(4\Delta_{x,0}^{2}-2i\hbar c_{-}\right)}+i\frac{c_{-}}{2\hbar}p_{0}^{2}+i\frac{e^{\frac{c_{0}}{2}}}{\hbar}p_{0}x-i\theta +i\frac{c_{+}}{2\hbar}x^{2}\right]\label{eq:funondaresults}
\end{equation}
where
\begin{align}
	&A(t)=\left(\frac{1}{2\pi \Delta_{x,t}^{2}}\right)^\frac{1}{4},\\
	&\Delta_{x,t}^{2}=\mathrm{e}^{-c_{0}}\Delta_{x,0}^{2}\left(1+\frac{\hbar^{2}c_{-}^{2}}{4\Delta_{x,0}^{4}}\right),\\
	&\theta(t)=\frac{1}{2}\arctan\left[-\frac{\hbar c_{-}}{2\Delta_{x,0}^{2}}\right].
\end{align}
The following expectation values and variances have been determined:
\begin{align}
	&\displaystyle\langle X\rangle_{q}=\displaystyle\frac{\mathrm{e}^{-\tau}}{\zeta\lambda m_{0}}\left[\left(\zeta \cosh (\zeta \tau)+\sinh(\zeta \tau)\right)\lambda m_{0}x_{0}+p_{0}\sinh(\zeta \tau)\right],\label{eq:posimeanres}\\
	&\displaystyle\langle X^{2}\rangle_{q}=\displaystyle\mathrm{e}^{-2\tau}\left[k_{4} \cosh^{2} (\zeta \tau)+k_{5}\sinh(\zeta \tau)\cosh (\zeta \tau)+k_{6}\sinh^{2}(\zeta \tau)\right],\label{eq:melhorQ2r}\\
	&\begin{array}{ll}
	\displaystyle\langle \left(\Delta X\right)^{2}\rangle_{q}&\displaystyle=\mathrm{e}^{-2\tau}\left(\Delta_{x,0}^{2} \cosh^{2} (\zeta \tau)+\frac{2\Delta_{x,0}^{2}}{\zeta} \sinh(\zeta \tau)\cosh (\zeta \tau)\right.+\\
	&\displaystyle+\left.\left(\frac{\Delta_{x,0}^{2}}{\zeta^{2}}+\frac{\Delta_{p,0}^{2}}{m_{0}^{2}\zeta^{2}\lambda^{2}}\right)\sinh^{2}(\zeta \tau)\right),
	\end{array}\\
	&\begin{array}{ll}
	\langle	 P\rangle_{q} &\displaystyle=\mathrm{e}^{\tau}\left(p_{0}^{2}\cosh^{2}(\zeta\tau)-\left(2\frac{k_{0}x_{0}p_{0}}{\lambda\zeta}+2\frac{p_{0}^{2}}{\zeta}\right)\sinh(\zeta \tau)\cosh(\zeta\tau)\right. +\\
	&\displaystyle+\left.\left(\frac{k_{0}^{2}x_{0}^{2}}{\lambda^{2}\zeta^{2}}+2\frac{p_{0}}{\zeta}\frac{k_{0}x_{0}}{\lambda\zeta}+\frac{p_{0}^{2}}{\zeta^{2}}\right)\sinh^{2}(\zeta \tau)\right)^{\frac{1}{2}},
	\end{array}\\
	&\displaystyle\langle P^{2}\rangle_{q}=\mathrm{e}^{2\tau}\left[k_{1}\cosh^{2}(\zeta\tau)-k_{2}\sinh(\zeta \tau)\cosh(\zeta\tau)+k_{3}\sinh^{2}(\zeta \tau)\right],\label{eq:melhorP2r}
	\end{align}
	\begin{align}
	&\begin{array}{rl}
	\displaystyle\langle\left(\Delta P\right)^{2} \rangle_{q}&\displaystyle=\mathrm{e}^{2\tau}\Delta_{x,0}^{2}\frac{k_{0}^{2}}{\lambda^{2}}\frac{\sinh^{2}(\zeta \tau)}{\zeta^{2}}+\\
	&\displaystyle+\mathrm{e}^{2\tau}\left(\frac{\hbar^{2}}{4\Delta_{x,0}^{2}}\cosh^{2}(\zeta\tau)-2\frac{\hbar^{2}}{4\Delta_{x,0}^{2}\zeta}\sinh(\zeta \tau)\cosh(\zeta\tau)+\frac{\hbar^{2}}{4\Delta_{x,0}^{2}\zeta^{2}}\sinh^{2}(\zeta \tau)\right).
	\end{array}\label{eq:desvmomimeanres}
\end{align}
where
\begin{align}
&\displaystyle k_{1}=p_{0}^{2}+\frac{\hbar^{2}}{4\Delta_{x,0}^{2}},\label{eq:k1r}\\
&\displaystyle k_{2}=2\frac{k_{0}x_{0}p_{0}}{\lambda\zeta}+2\frac{p_{0}^{2}}{\zeta}+2\frac{\hbar^{2}}{\Delta_{q,0}^{2}\zeta},\label{eq:k2r}\\
&\displaystyle k_{3}=\frac{k_{0}^{2}x_{0}^{2}}{\lambda^{2}\zeta^{2}}+2\frac{p_{0}}{\zeta}\frac{k_{0}x_{0}}{\lambda\zeta}+\frac{p_{0}^{2}}{\zeta^{2}}+\frac{\Delta_{x,0}^{2}k_{0}^{2}}{\lambda^{2}\zeta^{2}}+\frac{\hbar^{2}}{4\Delta_{x,0}^{2}\zeta^{2}},\label{eq:k3r}\\
&\displaystyle k_{4}=x_{0}^{2}+\Delta_{x,0}^{2},\label{eq:k4r}\\
&\displaystyle k_{5}=\frac{2\zeta\lambda^{2} m_{0}^{2}x_{0}^{2}+2\zeta\lambda m_{0}x_{0}p_{0}}{\zeta^{2}\lambda^{2} m_{0}^{2}}+\frac{2\Delta_{x,0}^{2}}{\zeta},\label{eq:k5r}\\
&\displaystyle k_{6}=\frac{p_{0}^{2} +2\lambda m_{0}x_{0}p_{0}+\lambda^{2} m_{0}^{2}x_{0}^{2}}{\zeta^{2}\lambda^{2} m_{0}^{2}}+\frac{\Delta_{x,0}^{2}}{\zeta^{2}}+\frac{\hbar^{2}}{4m_{0}^{2}\zeta^{2}\lambda^{2}\Delta_{x,0}^{2}}.\label{eq:k6r}
\end{align}
The subindex $q$ added to the formulas above is intended to distinguish the quantum results presented here from the classical-statistical ones derived in section \ref{sec:distribuicoes}. Now, the scaling procedure can be applied to the quantum results.
\subsection{Scaling the quantum results}
\label{subsec:dimensionless}
As for the classical system, it is defined an initial energy related with the mean values of position and momentum at $t=0$:
  \begin{equation}
  E_{0}\equiv \frac{m_{0}\omega^{2} x_{0}^{2}}{2}+\frac{ p_{0}^{2}}{2m_{0}}\label{eq:E02}
  \end{equation}
  and
  \begin{equation}
  \frac{m_{0}\omega^{2} x_{0}^{2}}{2}\equiv\varepsilon E_{0}\qquad \qquad\frac{ p_{0}^{2}}{2m_{0}}\equiv\left(1-\varepsilon\right)E_{0}
  \end{equation}
  with $0\leq\varepsilon\leq 1$. However, for quantum systems there are quantum fluctuations terms, which are proposed to be treated as follows:
  \begin{equation}
  	e_{0}\equiv \frac{m_{0}\omega^{2} \Delta_{x,0}^{2}}{2}+\frac{ \Delta_{p,0}^{2}}{2m_{0}},\label{eq:e02}
  \end{equation}
  where 
  \begin{equation}
  \frac{m_{0}\omega^{2} \Delta_{x,0}^{2}}{2}=\varepsilon_{\Delta} e_{0},\qquad\frac{ \Delta_{p,0}^{2}}{2m_{0}}=\left(1-\varepsilon_{\Delta}\right)e_{0},\label{eq:flutu1}
  \end{equation}
  where $0\leq\varepsilon_{\Delta}\leq 1$. The new parameters $\varepsilon_{\Delta}$ and $e_{0}$ are introduced in analogy with $\varepsilon$ and $E_{0}$, respectively. While $e_{0}$ is the amount of energy associated with the quantum fluctuations (quantum uncertainties), $\varepsilon_{\Delta}$ gives the weight of the elastic-energy fluctuation in the $e_{0}$. It is also convenient to introduce the parameter
  \begin{equation}
\vartheta=\frac{e_{0}}{E_{0}},\label{eq:flutu2}
\end{equation}
which measures the relevance of the fluctuation relative to mean energy of the system.
Since $\frac{\hbar}{2}=\Delta_{p,0}\Delta_{x,0}$, it follows from \eqref{eq:flutu1} that
  \begin{equation}
  	\displaystyle\frac{m_{0}\omega^{2} \Delta_{x,0}^{2}}{2}\frac{ \Delta_{p,0}^{2}}{2m_{0}}=\frac{\hbar^{2}}{16}\omega^{2}=\varepsilon_{\Delta}\left(1-\varepsilon_{\Delta}\right)e_{0}^{2}\qquad\Rightarrow\qquad
  	e_{0}=\frac{\hbar\omega}{4\sqrt{\varepsilon_{\Delta}\left(1-\varepsilon_{\Delta}\right)}}
  \end{equation}
  which yields, via Eq. \eqref{eq:flutu2},
  \begin{equation}
  \vartheta=\frac{\hbar\omega}{4E_{0}\sqrt{\varepsilon_{\Delta}\left(1-\varepsilon_{\Delta}\right)}}.\label{eq:vartheta}
  \end{equation}
The uncertainty principle precludes $\Delta_{q,0}$ and $\Delta_{p,0}$ to be simultaneously zero, so that $e_{0}$ and $\vartheta$ can never be strictly zero.  Most importantly, $\vartheta$ is found to play a very interesting role here. Disregarding the limiting case where $\Delta_{x,0}\rightarrow 0$ and $\Delta_{p,0}\rightarrow \infty$ (or vice-versa), which gives $\varepsilon_{\Delta}\rightarrow 0$ or $1$, for any intermediary $\varepsilon_{\Delta}$ the magnitude of  $\vartheta$ will be dominated by $\frac{\hbar \omega}{E_{0}}$. This ratio quantifies the relevance of the quantum of energy relatively to the mean energy. In this capacity, the relation $\vartheta=\frac{e_{0}}{E_{0}}\propto \frac{\hbar \omega}{E_{0}}$ shows that this parameter behaves as a \emph{semiclassical variable}, for the fact of it being sufficiently small directly implies that quantum fluctuations are negligible, which is a signature of the classical limit.
\par The quantities given in Eqs. \eqref{eq:posimeanres}-\eqref{eq:desvmomimeanres} can then be rewritten considering the scheme of Eqs. \eqref{eq:E02}-\eqref{eq:flutu2}. After some algebraic manipulations, it is obtained that
\begin{align}
&\begin{array}{lll}
&\displaystyle\frac{\langle X\rangle_{q}}{x_{0}}&\displaystyle=\mathrm{e}^{-\tau}\left[\cosh (\zeta \tau)+\left(\frac{1}{\zeta}+\frac{\omega}{\lambda}\frac{\sqrt{1-\varepsilon }}{\zeta\sqrt{\varepsilon }}\right)\sinh(\zeta \tau)\right],
\end{array}\\
&\begin{array}{lll}
&\displaystyle\frac{\langle X^{2}\rangle_{q}}{x_{0}^{2}}&\displaystyle=\mathrm{e}^{-2\tau}\left[\cosh (\zeta \tau)+\left(\frac{1}{\zeta}+\frac{\omega}{\lambda}\frac{\sqrt{1-\varepsilon }}{\zeta\sqrt{\varepsilon }}\right)\sinh(\zeta \tau)\right]^{2}+\\
& &\displaystyle+\frac{\varepsilon_{\Delta}^{2}\vartheta^{2}}{\varepsilon^{2}}\mathrm{e}^{-2\tau}\left[ \cosh^{2} (\zeta \tau)+\frac{2}{\zeta} \sinh(\zeta \tau)\cosh (\zeta \tau)\right.+\\
& &\displaystyle\left.+\frac{1}{\zeta^{2}}\left(1+\frac{\omega^{2}}{\lambda^{2}}\frac{1-\varepsilon_{\Delta}}{\varepsilon_{\Delta}}\right)\sinh^{2}(\zeta \tau)\right],\label{eq:ResultsX2}
\end{array}\\
&\begin{array}{rl}
\displaystyle\frac{\langle \left(\Delta X\right)^{2}\rangle_{q}}{x_{0}^{2}}&\displaystyle=\frac{\varepsilon_{\Delta}^{2}\vartheta^{2}}{\varepsilon^{2}}\mathrm{e}^{-2\tau}\left[ \cosh^{2} (\zeta \tau)+\frac{2}{\zeta}\sinh(\zeta \tau)\cosh (\zeta \tau)\right.+\\
&\displaystyle\left.+\frac{1}{\zeta^{2}}\left(1+\frac{\omega^{2}}{\lambda^{2}}\frac{1-\varepsilon_{\Delta}}{\varepsilon_{\Delta}}\right)\sinh^{2}(\zeta \tau)\right],
\end{array}\\
&\begin{array}{rl}
\displaystyle\frac{\langle P^{2}\rangle_{q}}{p_{0}^{2}}&\displaystyle=\frac{\vartheta\mathrm{e}^{2\tau}}{
	1-\varepsilon}\left[\left(\frac{\omega^{2}}{\lambda^{2}}\frac{\varepsilon_{\Delta}}{\zeta^{2}}+\frac{1-\varepsilon_{\Delta}}{\zeta^{2}}\right)\sinh^{2}(\zeta \tau)+\left(1-\varepsilon_{\Delta}\right)\cosh^{2}(\zeta\tau)\right.+\\
& \displaystyle\left.-2\frac{1-\varepsilon_{\Delta}}{\zeta}\sinh(\zeta \tau)\cosh(\zeta\tau)\right]+\\
&\displaystyle+\frac{\mathrm{e}^{2\tau}}{
	1-\varepsilon}\left[\left(1-\varepsilon\right)\cosh^{2}(\zeta\tau)-\left(2\frac{\omega \sqrt{\varepsilon-\varepsilon^{2}}}{\lambda\zeta}+2\frac{1-\varepsilon}{\zeta}\right)\sinh(\zeta \tau)\cosh(\zeta\tau)+\right.\\
&\displaystyle\left.	+\left(\omega^{2}\frac{\varepsilon}{\lambda^{2}\zeta^{2}}+2\omega\frac{\sqrt{\varepsilon-\varepsilon^{2}}}{\lambda\zeta^{2}}+\frac{1-\varepsilon}{\zeta^{2}}\right)\sinh^{2}(\zeta \tau)\right].\label{eq:ResultsP2}
\end{array}
\end{align}
\begin{align}
&\begin{array}{rl}
\displaystyle\frac{\langle P\rangle_{q}}{p_{0}}&\displaystyle=\frac{\mathrm{e}^{\tau}}{
	\left(1-\varepsilon\right)}\left(\left(1-\varepsilon\right)\cosh^{2}(\zeta\tau)-\left(2\frac{\omega \sqrt{\varepsilon-\varepsilon^{2}}}{\lambda\zeta}+2\frac{1-\varepsilon}{\zeta}\right)\sinh(\zeta \tau)\cosh(\zeta\tau)+\right.\\
&\displaystyle\left.	+\left(\omega^{2}\frac{\varepsilon}{\lambda^{2}\zeta^{2}}+2\omega\frac{\sqrt{\varepsilon-\varepsilon^{2}}}{\lambda\zeta^{2}}+\frac{1-\varepsilon}{\zeta^{2}}\right)\sinh^{2}(\zeta \tau)\right)^{\frac{1}{2}},\label{eq:ResultsP}
\end{array}\\
&\begin{array}{rl}
\displaystyle
\frac{\langle\left(\Delta P\right)^{2}\rangle_{q}}{p_{0}^{2}}&\displaystyle=\frac{\vartheta\mathrm{e}^{2\tau}}{
	1-\varepsilon}\left[\left(\frac{\omega^{2}}{\lambda^{2}}\frac{\varepsilon_{\Delta}}{\zeta^{2}}+\frac{1-\varepsilon_{\Delta}}{\zeta^{2}}\right)\sinh^{2}(\zeta \tau)+\left(1-\varepsilon_{\Delta}\right)\cosh^{2}(\zeta\tau)+\right.\\
&\displaystyle\left.-2\frac{1-\varepsilon_{\Delta}}{\zeta}\sinh(\zeta \tau)\cosh(\zeta\tau)\right].\label{eq:ResultsDP}
\end{array}
\end{align}
Notice that the variances vanish with $\vartheta$, as expected for a classical limit. Before proceeding with the calculation of the quantum work introduced in this dissertation, in the next subsection, results are presented for the current notions of work and heat as defined by Alicki. 
\subsection{Alicki's results for the system block}
\label{subsec:alik}
\par The definition of work employed for the classical system was considered to be evaluated regarding the particle of mass $m_{0}$ only. That is, the spring and the agent that performs drag force were not regarded as part of the system. Proceeding in the same way here, the only form of energy that can be stored exclusively in the system (block) is kinetic. Following the notation given in chapter \ref{cap:fundamentacao2}, the system $\mathcal{S}'$ being regarded here as the block and the system Hamiltonian $H_{u}$ assumes the form 
\begin{equation}
	H_{u}\equiv  \frac{m_{0}}{2}V^{2}
\end{equation}
where $V$ is the velocity operator, given in the Schrödinger picture according to the prescription given by Eq. \eqref{eq:conversaoheischron2}. The internal energy is then given by
\begin{equation}
E_{\mathcal{S}'}\equiv E=\mathrm{Tr}\left(\rho(t) H_{u}\right)\equiv\mathrm{Tr}\left(\rho(t) \frac{m_{0}}{2}V^{2}\right)=\frac{m_{0}}{2}\langle V^{2}\rangle_{q}=K_{q}(t)=\mathrm{Tr}\left(\rho(0) \frac{m_{0}}{2}\dot{X}_{H}^{2}(t)\right).\label{eq:EEquivK}
\end{equation}
The velocity operator $V$ can be computed from Eq.\eqref{eq:conversaoheischron2}, resulting in
\begin{equation}
	\frac{m_{0}}{2}V^{2}=\frac{\mathrm{e}^{-4\tau}}{2m_{0}}P^{2}.\label{eq:kinemoment}
\end{equation}
Inserting this result in Eq. \eqref{eq:EEquivK} gives 
\begin{equation}
E=\frac{m_{0}}{2}\langle V^{2}\rangle_{q}=\frac{\mathrm{e}^{-4\tau}}{2m_{0}}\langle P^{2}\rangle_{q}.
\end{equation}
Recalling the definition $K_{0}=\left(1-\varepsilon\right)=\frac{p_{0}^{2}}{2m_{0}}$, one can write
\begin{equation}
	\frac{\langle P^{2}\rangle_{q}}{p_{0}^{2}}=	\frac{\frac{\langle P^{2}\rangle_{q}}{2m_{0}}}{\frac{p_{0}^{2}}{2m_{0}}}=\frac{\frac{\langle P^{2}\rangle_{q}}{2m_{0}}}{K_{0}}
\end{equation}
which implies  
\begin{equation}
\frac{\frac{m}{2}\langle V^{2}\rangle_{q}}{K_{0}}=\mathrm{e}^{-4\tau}\frac{\langle P^{2}\rangle_{q}}{p_{0}^{2}}.
\end{equation}
 Using Eq. \eqref{eq:ResultsP2}, one then obtains
\begin{equation}
\displaystyle\frac{\frac{m_{0}}{2}\langle V^{2}\rangle_{q}}{K_{0}}=\mathrm{e}^{-2\tau}\left[\left(\underline{\alpha}_{cl}+\vartheta\underline{\beta}\right)\cdot \underline{\Gamma}(t)\right]\label{eq:Resultsv2}
\end{equation}
where 
\begin{equation}
\underline{\beta}=\frac{1}{1-\varepsilon}\left[
\begin{array}{c}
\displaystyle 1-\varepsilon_{\Delta}\\
\displaystyle -2\frac{1-\varepsilon_{\Delta}}{\zeta}\\
\displaystyle \frac{\omega^{2}}{\lambda^{2}}\frac{\varepsilon_{\Delta}}{\zeta^{2}}+\frac{1-\varepsilon_{\Delta}}{\zeta^{2}}
\end{array}
\right],\label{eq:betadef}
\end{equation}
and $\underline{\alpha}_{cl}$ and $\underline{\Gamma}(t)$ are given in Eq. \ref{eq:alphacldef}, rewritten here for convenience
\begin{equation}
\underline{\alpha}_{cl}=\frac{1}{1-\varepsilon}\left[
\begin{array}{c}
1-\varepsilon\\\displaystyle
-\frac{2}{\zeta}\left(\frac{\omega \sqrt{\varepsilon-\varepsilon^{2}}}{\lambda}+\left(1-\varepsilon\right)\right)\\
\displaystyle\left(\frac{\omega^{2}\varepsilon}{\lambda^{2}\zeta^{2}}+2\omega\frac{\sqrt{\varepsilon-\varepsilon^{2}}}{\lambda\zeta^{2}}+\frac{1-\varepsilon}{\zeta^{2}}\right)
\end{array}
\right], \qquad \underline{\Gamma}(t)=\left[
\begin{array}{c}
\displaystyle\cosh^{2}(\zeta\tau)\\
\displaystyle\sinh(\zeta \tau)\cosh(\zeta\tau)\\
\displaystyle\sinh^{2}(\zeta\tau)
\end{array}
\right].\label{eq:alphacldef2}
\end{equation}
From Eq. \eqref{eq:Resultsv2}, it can be seen that as $\vartheta\approx 0$, $\frac{\frac{m_{0}}{2}\langle V^{2}\rangle_{q}}{K_{0}}\approx\mathrm{e}^{-2\tau}\left(\underline{\alpha}_{cl}\cdot \underline{\Gamma}(t)\right)$ and, from the classical definition \eqref{eq:classikal}, $\frac{\frac{m_{0}}{2}\langle V^{2}\rangle_{q}}{K_{0}}\approx\frac{K_{cl}(\tau)}{K_{0}}$. In other words, the quantum kinetic energy $\frac{\frac{m_{0}}{2}\langle V^{2}\rangle_{q}}{K_{0}}$ approximates the classical one $\frac{K_{cl}(\tau)}{K_{0}}$, when $\vartheta$ becomes sufficiently small, showing again that the limit $\vartheta\rightarrow 0$ can consistently be regarded as a classical limit. Now, Alicki's work can be written as
 \begin{equation}
 \mathcal{W}_{ak}\left(t\right) =\int_{0}^{t} \dot{\mathcal{W}}_{a}(t^{'})dt^{'}\label{eq:trabalhopotenciak}
 \end{equation}
 \begin{equation}
 \dot{\mathcal{W}}_{ak}(t^{'})=\mathrm{Tr}\left(\rho(t^{'}) \frac{\partial H_{u}}{\partial t^{'}} \right)\equiv \mathrm{Tr}\left[\rho(t^{'}) \frac{\partial \left(\frac{m_{0}}{2}V^{2}(t^{'})\right)}{\partial t^{'}} \right].
 \end{equation}
  From Eq. \eqref{eq:kinemoment}, and recalling that $\tau=\lambda t$,
\begin{equation}
\dot{\mathcal{W}}_{ak}(t^{'})=\mathrm{Tr}\left[\rho(t^{'})\frac{\partial}{\partial t^{'}}\left(\frac{\mathrm{e}^{-4\lambda t^{'}}}{2m_{0}}P^{2}\right)\right]=-4\lambda\frac{\mathrm{e}^{-4\lambda t^{'}}}{2m_{0}}\mathrm{Tr}\left[\rho(t^{'})P^{2}\right]=-4\lambda \frac{m_{0}}{2}\langle V^{2}\rangle_{q} \label{eq:alickidwdt2}
\end{equation}
Dividing both sides by $K_{0}$, and performing the integral
\begin{equation}
\int_{0}^{t}\frac{\dot{\mathcal{W}}_{ak}(t^{'})}{K_{0}}dt'=\frac{\mathcal{W}_{ak}\left(t\right)}{K_{0}}=-4\lambda \int_{0}^{t}\frac{\frac{m_{0}}{2}\langle V^{2}\rangle_{q}}{K_{0}}dt',\label{eq:trabalhototal}
\end{equation}
one obtains by use of Eq. \eqref{eq:Resultsv2} that 
\begin{equation}
\displaystyle\frac{\mathcal{W}_{ak}\left(t\right)}{K_{0}}=\mathrm{e}^{-2\tau}\left[\left(\underline{\alpha}_{a}+\vartheta\underline{\beta}_{a}\right)\cdot \underline{\Gamma}(t)\right]-\left[\left(\underline{\alpha}_{a}+\vartheta\underline{\beta}_{a}\right)\cdot \underline{\Gamma}(0)\right]\label{eq:Walk}
\end{equation}
where 
\begin{equation}
\underline{\alpha}_{a}=-a_{1}\left[
\begin{array}{c}
\displaystyle\frac{\zeta a_{6}+a_{5}+a_{7}}{\zeta^{2}-1}-a_{5}+a_{7}\\
\displaystyle \frac{2}{\zeta^{2}-1}\left(\zeta\left(a_{5}+a_{7}\right)+a_{6}\right)\\
\displaystyle\frac{\zeta a_{6}+a_{5}+a_{7}}{\zeta^{2}-1}+a_{5}-a_{7}
\end{array}
\right], \qquad 
\underline{\beta}_{a}=-a_{1}\left[
\begin{array}{c}
\displaystyle\frac{\zeta a_{4}+a_{2}+a_{3}}{\zeta^{2}-1}-a_{3}+a_{2}\\
\displaystyle \frac{2}{\zeta^{2}-1}\left(\zeta\left(a_{2}+a_{3}\right)+a_{4}\right)\\
\displaystyle\frac{\zeta a_{4}+a_{2}+a_{3}}{\zeta^{2}-1}+a_{3}-a_{2}
\end{array}
\right],
\end{equation}
with the constants 
\begin{equation}
\begin{array}{ll}
a_{1}=\frac{1}{	1-\varepsilon},\quad a_{2}=\left(\frac{\omega^{2}}{\lambda^{2}}\frac{\varepsilon_{\Delta}}{\zeta^{2}}+\frac{1-\varepsilon_{\Delta}}{\zeta^{2}}\right),& a_{3}=1-\varepsilon_{\Delta},\quad a_{4}=-2\frac{1-\varepsilon_{\Delta}}{\zeta},\quad a_{5}=\left(1-\varepsilon\right), \\
a_{6}=-\left(2\frac{\omega \sqrt{\varepsilon-\varepsilon^{2}}}{\lambda\zeta}+2\frac{1-\varepsilon}{\zeta}\right), & a_{7}=\left(\omega^{2}\frac{\varepsilon}{\lambda^{2}\zeta^{2}}+2\omega\frac{\sqrt{\varepsilon-\varepsilon^{2}}}{\lambda\zeta^{2}}+\frac{1-\varepsilon}{\zeta^{2}}\right).\\
\end{array}
\end{equation}
Heat can be evaluated thus as
\begin{equation}
\frac{\mathcal{Q}_{ak}(t)}{K_{0}}=\frac{\left. E\right|_{0}^{t}}{K_{0}}-\frac{ \mathcal{W}_{ak}(t)}{K_{0}}\equiv\frac{\left. \frac{m_{0}}{2}\langle V^{2}\rangle_{q}\right|_{0}^{t}}{K_{0}}-\frac{ \mathcal{W}_{ak}(t)}{K_{0}}.\label{eq:Qalk}
\end{equation}
The expression resulting by use of  Eqs. \eqref{eq:Walk} and \eqref{eq:Resultsv2} is not really enlightening and will be omitted. In any case, heat can be computed from the kinetic energy and work.
\par In order to discriminate the time evolution of each term and comparing them with the classical results, Eqs. \eqref{eq:Walk}, \eqref{eq:Qalk} and \eqref{eq:Resultsv2} will be computed for the set of parameters $\varepsilon=0$, $\varepsilon_{\Delta}=0.5$ and $\vartheta=0.1$ for two distinct cases: the overdamped oscillator (OO), for which $\frac{\omega}{\lambda}=0.1$, and the underdamped oscillator (UO), with  $\frac{\omega}{\lambda}=10$. The interpretation associated with these parameters is as follows:
\begin{itemize}
	\item $\varepsilon=0$ implies that the initial energy $E_{0}$ (see Eq. \eqref{eq:E02}) is purely kinetic, since $1-\varepsilon=1=\frac{K_{0}}{E_{0}}$;
	\item the energy associated with the fluctuations terms $e_{0}=\frac{m_{0}\omega^{2} \Delta_{x,0}^{2}}{2}+\frac{ \Delta_{p,0}^{2}}{2m_{0}}$, defined in Eq. \eqref{eq:e02}, is equally distributed in potential and kinetic terms, for $\varepsilon_{\Delta}=\frac{m_{0}\omega^{2} \Delta_{x,0}^{2}}{2e_{0}}=0.5=1-\varepsilon_{\Delta}=\frac{ \Delta_{p,0}^{2}}{2m_{0}}$;
	\item $\vartheta=\frac{e_{0}}{E_{0}}=0.1$ indicates that the fluctuation in energy is about $10\%$ of the mean energy. 
\end{itemize} 
The heat $\frac{\mathcal{Q}_{ak}(t)}{K_{0}}$ and the work $\frac{\mathcal{W}_{ak}(t)}{K_{0}}$ as defined in the present section are compared with the classical work $\frac{\mathcal{W}_{cl}(t)}{K_{0}}$ in Figures \ref{fig:alickikvsclassicood} and \ref{fig:alickikvsclassicoud}, for the UO and OO cases, respectively.
\begin{figure}[h]
	\begin{center}
		\begin{subfigure}[h]{1\textwidth}
			\includegraphics[angle=0, scale=0.6]{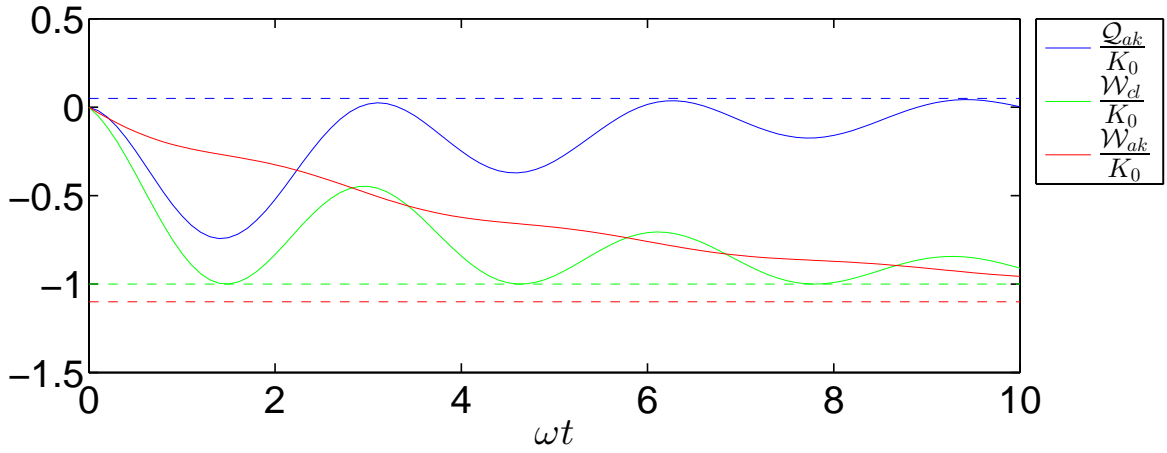} 
			\caption{Underdamped regime: $\frac{\omega}{\lambda}=10$.}
			\label{fig:alickikvsclassicood}
		\end{subfigure}
		\begin{subfigure}[h]{1\textwidth}
			\includegraphics[angle=0, scale=0.6]{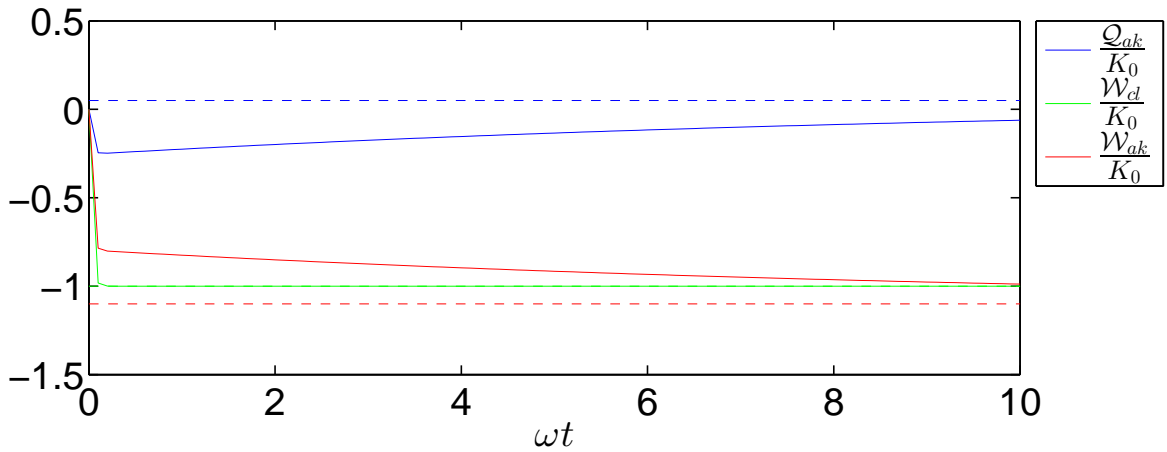}
			\caption{Overdamped regime: $\frac{\omega}{\lambda}=0.1$.}
			\label{fig:alickikvsclassicoud}
		\end{subfigure}
		\caption{Heat (blue line) and work (red line) according to Alicki's approach and the classical result for work (green line) for (a) the UO and (b) the OO, as a function of the dimensionless time $\omega t$. The corresponding asymptotic limits ($t\to\infty$) are represented by dashed lines with the same color code. }
		\label{fig:alickikvsclassico}
	\end{center}
\end{figure} 
\par It can be seen that the work does not match its classical counterpart. In principle, there should be some agreement, as $\vartheta=0.1$ corresponds to a semiclassical regime. One might suspect that this regime is still "too quantum", so that no classical correspondence could be supported. To check this thesis, the strictly classical limit $\vartheta=0$ is considered for the UO with the same set of parameters. The results are shown in Fig. \ref{fig:alickikvsclassicov0}.
\begin{figure}[h]
	\begin{center}
		\includegraphics[angle=0, scale=0.6]{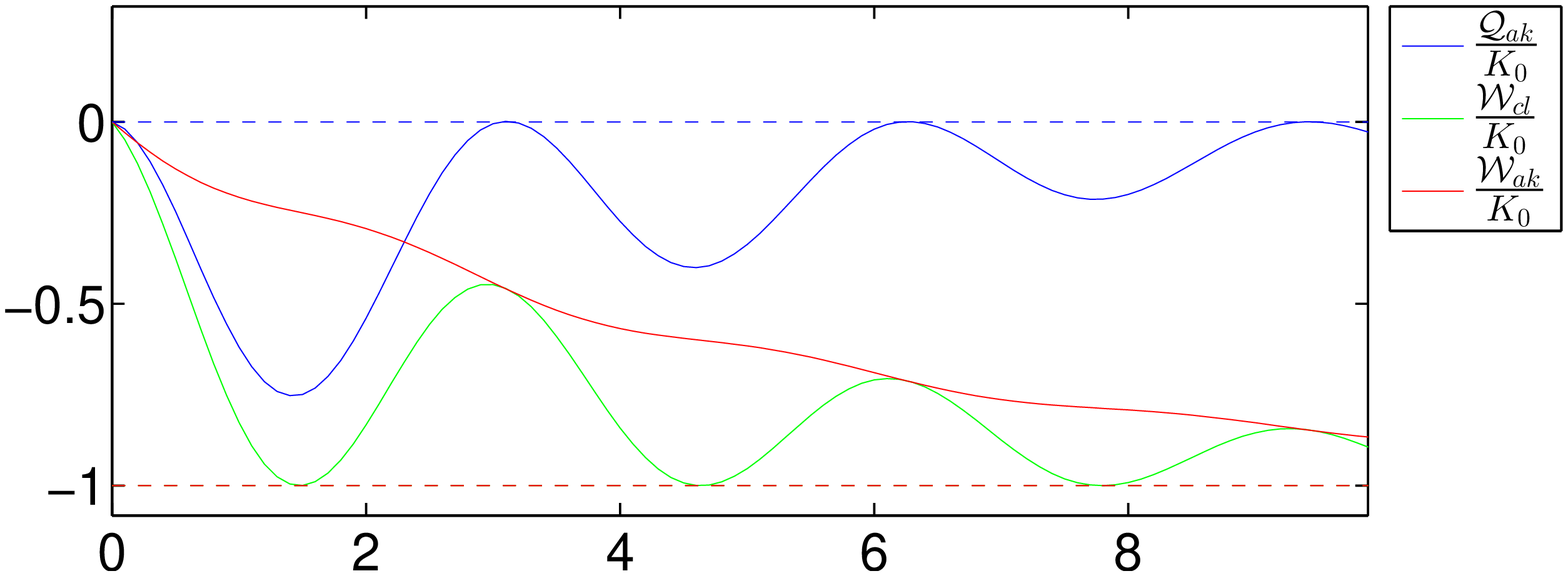}
		\caption{Heat (blue line) and work (red line) according to Alicki's approach with $\vartheta=0$ and the classical result for work (green line) for the UO ($\frac{\omega}{\lambda}=10$) as a function of the dimensionless time $\omega t$. The corresponding asymptotic limits ($t\to\infty$) are represented by dashed lines with the same color code.}
		\label{fig:alickikvsclassicov0}
	\end{center}
\end{figure} 
\par One sees no substancial coincidence between the classical work and its quantum counterpart, except at very specific instants. Also, as can be seen from the definition of heat $\mathcal{Q}_{ak}$, even if the state is very localized, there will still be heat flowing from the system. It is not clear, however, from a classical thermodynamic viewpoint, what is the nature of the heat. In order to analyze this point, a completely dissipative system is considered, where $k_{0}=0$, that is, the spring is turned off and the block continuously decelerates under the action of a drag force. From algebraic manipulations in Eqs. \eqref{eq:Walk}, \eqref{eq:Resultsv2} and \eqref{eq:Qalk}, it is obtained that
\begin{equation}
\frac{\mathcal{W}_{ak}\left(t\right)}{K_{0}}=\left.\frac{\frac{m_{0}}{2}\langle V^{2}\rangle_{q}}{K_{0}}\right|_{0}^{t},
\end{equation}
\begin{equation}
\frac{\frac{m_{0}}{2}\langle V^{2}(t)\rangle_{q}}{K_{0}}=\frac{\vartheta\left(1-\varepsilon_{\Delta}\right)\mathrm{e}^{-4\tau}}{
	1-\varepsilon}+\mathrm{e}^{-4\tau}\label{eq:kk00}
\end{equation}
and
\begin{equation}
\frac{\mathcal{Q}_{ak}\left(t\right)}{K_{0}}=0.
\end{equation}
In words, in the case where there is only dissipative terms, the kinetic energy decays exponentially and there is no heat flux for any set of parameters $\varepsilon$, $\varepsilon_{\Delta}$, $\lambda$. However, from a thermodynamical perspective, dissipation is frequently connected with heat flux, or more generally, thermal effects. 
\par It is also interesting to look at the limiting case where  $\lambda=0$. As already shown in the classical case, the Hamiltonian in Eq. \eqref{eq:Hamiltonianoquantuco} reduces to that of a harmonic oscillator. Then it follows that 
\begin{equation}
\frac{ \mathcal{W}_{ak}\left(t\right)}{K_{0}}=0
\end{equation}
and
\begin{equation}
\frac{ \mathcal{Q}_{ak}\left(t\right)}{K_{0}}=\left.\frac{\frac{m_{0}}{2}\langle V^{2}\rangle_{q}}{K_{0}}\right|_{0}^{t}
\end{equation}
with 
\begin{equation}
\frac{\frac{m_{0}}{2}\langle V^{2}\rangle_{q}}{K_{0}}=\left(\underline{\alpha}_{cl}^{\lambda=0}+\vartheta\underline{\beta}^{\lambda=0}\right)\cdot\underline{\Gamma}^{\lambda=0}(t).\label{eq:Kl0}
\end{equation}
where $\underline{\alpha}_{cl}^{\lambda=0}$ and $\underline{\Gamma}^{\lambda=0}(t)$ are rewritten here, from Eq. \eqref{eq:alphacll0}, for convenience 
\begin{equation}
\underline{\alpha}_{cl}^{\lambda=0}=\frac{1}{1-\varepsilon}\left[
\begin{array}{c}
1-\varepsilon\\\displaystyle
-2\sqrt{\varepsilon-\varepsilon^{2}}\\
\varepsilon
\end{array}
\right], \qquad 
\underline{\Gamma}^{\lambda=0}(t)=\left[
\begin{array}{c}
\displaystyle\cos^{2}(\omega t)\\
\displaystyle \sin(\omega t)\cos(\omega t)\\
\displaystyle-\sin^{2}(\omega t)
\end{array}
\right]\label{eq:alphgl0}
\end{equation}
and 
\begin{equation}
\underline{\beta}^{\lambda=0}=\frac{1}{1-\varepsilon}\left[
\begin{array}{c}
1-\varepsilon_{\Delta}\\\displaystyle
0\\
\varepsilon_{\Delta}
\end{array}
\right],\label{eq:betal0}
\end{equation}
As a result, for a harmonic oscillator, the only energy flux is in form of heat, for any set of parameters. Once again, this does not match with the expected behavior for work and heat from a mechanical or thermodynamical perspective, \emph{viz.} that for an harmonic oscillator, the spring will perform work on the system (block) and no energy will leak from the system in any form, including heat. 
\par To sum up, it has been shown that Alicki's proposal for work and heat fails in producing results that minimally agree with the classical intuition associated with the mechanics of a single particle system. In what follows, it is analyzed the definition of work proposed in this dissertation.
\subsection{The quantum mechanical work}
\label{subsec:nossa}
\par The quantum work\footnote{In order to simplify the discussion, it will be called in this way instead of \emph{resultant} quantum work. In addition, the super index $R$ was also suppressed. The same considerations were regarded for the centroid and thermal work.} $\mathcal{W}_{q}(t)$, introduced in the previous chapter, can be written just as the difference in the kinetic energy, i.e., 
\begin{equation}
	\frac{\mathcal{W}_{q}(t)}{K_{0}}=\left.\frac{\frac{m_{0}}{2}\langle V^{2}\rangle_{q}}{K_{0}}\right|_{0}^{t}
\end{equation}
which, from the expression \eqref{eq:Resultsv2}, can be written as 
\begin{equation}
\displaystyle\frac{\mathcal{W}_{q}(t)}{K_{0}}=\left(\underline{\alpha}_{cl}+\vartheta\underline{\beta}\right)\cdot\left(\mathrm{e}^{-2\tau}\underline{\Gamma}(t)-\underline{\Gamma}(0)\right) \label{eq:Resultsv2x}
\end{equation}
where $\underline{\alpha}_{cl}$, $\underline{\Gamma}(t)$ and $\underline{\beta}$ are defined in Eqs. \eqref{eq:alphacldef2} and \eqref{eq:betadef}.
\par The centroid work  $\mathcal{W}_{c}$ and thermal work $\mathcal{W}_{th}$ can be evaluated as follows: since, from Eq. \eqref{eq:conversaoheischron2}, $V=\frac{\mathrm{e}^{-2\lambda t}}{m_{0}}P$, then 
\begin{equation}
	\frac{\frac{m_{0}}{2}\langle V\rangle_{q}^{2}}{K_{0}}=\frac{\frac{m_{0}}{2}\frac{\mathrm{e}^{-4\lambda t}}{m_{0}^{2}}\langle P\rangle_{q}^{2}}{\frac{p_{0}^{2}}{2m_{0}}}=\mathrm{e}^{-4\lambda t}\frac{\langle P\rangle_{q}^{2}}{p_{0}^{2}} \label{eq:Ktop2}
\end{equation}
which, from \eqref{eq:ResultsP}, gives
\begin{equation}
\displaystyle\frac{\frac{m_{0}}{2}\langle V\rangle_{q}^{2}}{K_{0}}=\mathrm{e}^{-2\tau}\left(\underline{\alpha}_{cl}\cdot \underline{\Gamma}(t)\right)\label{eq:MMV2}
\end{equation}
By comparing Eqs. \eqref{eq:Resultsv2x} and \eqref{eq:MMV2}, one obtains
\begin{equation}
\displaystyle
\frac{\frac{m_{0}}{2}\langle\left(\Delta V\right)^{2}\rangle_{q}}{K_{0}}=\vartheta\mathrm{e}^{-2\tau}\left(\underline{\beta}\cdot \underline{\Gamma}(t)\right)\label{eq:MMDV2}
\end{equation}
From Eqs. \eqref{eq:MMV2} and \eqref{eq:MMDV2} and the definitions given in the previous chapter, the centroid $\mathcal{W}_{c}(t)$ and thermal work $\mathcal{W}_{th}(t)$ assumes
\begin{equation}
	\frac{\mathcal{W}_{c}(t)}{K_{0}}=\underline{\alpha}_{cl}\cdot \left(\mathrm{e}^{-2\tau}\underline{\Gamma}(t)-\underline{\Gamma}(0)\right)\label{eq:centroidw}
\end{equation}
and
\begin{equation}
\frac{\mathcal{W}_{th}(t)}{K_{0}}=\vartheta\underline{\beta}\cdot \left(\mathrm{e}^{-2\tau}\underline{\Gamma}(t)-\underline{\Gamma}(0)\right).\label{eq:thermalw}
\end{equation}
By direct inspection of Eqs. \eqref{eq:centroidw} and \eqref{eq:classicalw}, it immediately follows that the centroid work coincides with the classical work. However, the quantum work of the system $\mathcal{W}_{q}$ will not necessarily follows the same path as $\mathcal{W}_{c}$, since, it must be added to it the thermal work terms. On the other hand, as $\vartheta\rightarrow 0$, then, from Eq.\eqref{eq:Resultsv2x}, \eqref{eq:centroidw} and \eqref{eq:thermalw}, $\mathcal{W}_{q}\rightarrow \mathcal{W}_{c}$ and $\mathcal{W}_{th}\rightarrow 0$; in words, as the quantum fluctuations, quantified by the parameter $\vartheta=\frac{\hbar\omega}{E_{0}4\sqrt{\varepsilon_{\Delta}\left(1-\varepsilon_{\Delta}\right)}}$, becomes negligible, then the quantum work $\mathcal{W}_{q}(t)$ is equal to the classical.
\par  For an illustration of the behavior of each term mentioned above, the UO and the OO cases are retaken. The equations are evaluated for the time interval $[0,10\omega t]$, and the results are presented in Fig. \ref{fig:nossakvsclassicood} and \ref{fig:nossakvsclassicoud}, for the underdamped and overdamped oscillators, respectively.
\begin{figure}[h!]
	\begin{center}
		\begin{subfigure}[h]{1\textwidth}
			\includegraphics[angle=0, scale=0.6]{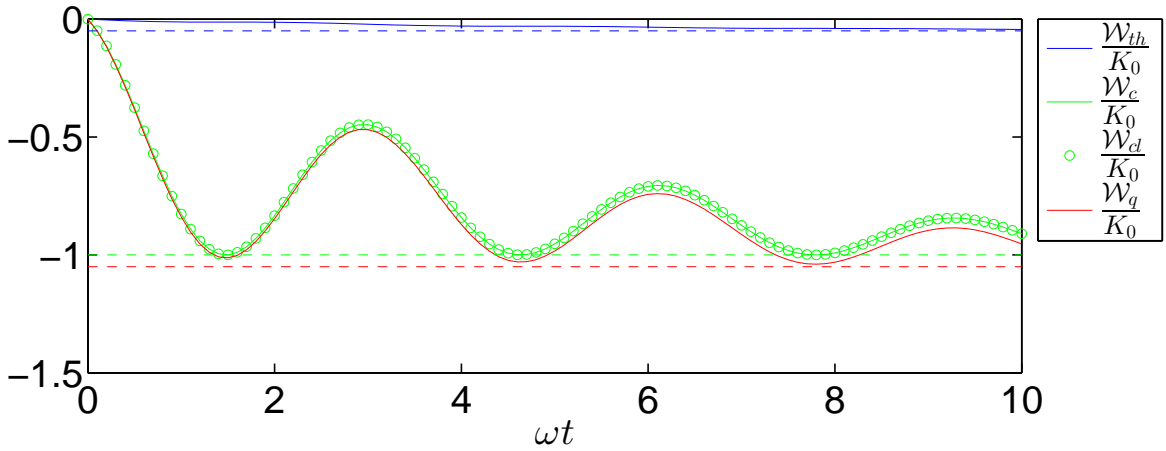} 
			\caption{Underdamped regime $\frac{\omega}{\lambda}=10$.}
			\label{fig:nossakvsclassicood}
		\end{subfigure}
		\begin{subfigure}[h]{1\textwidth}
			\includegraphics[angle=0, scale=0.6]{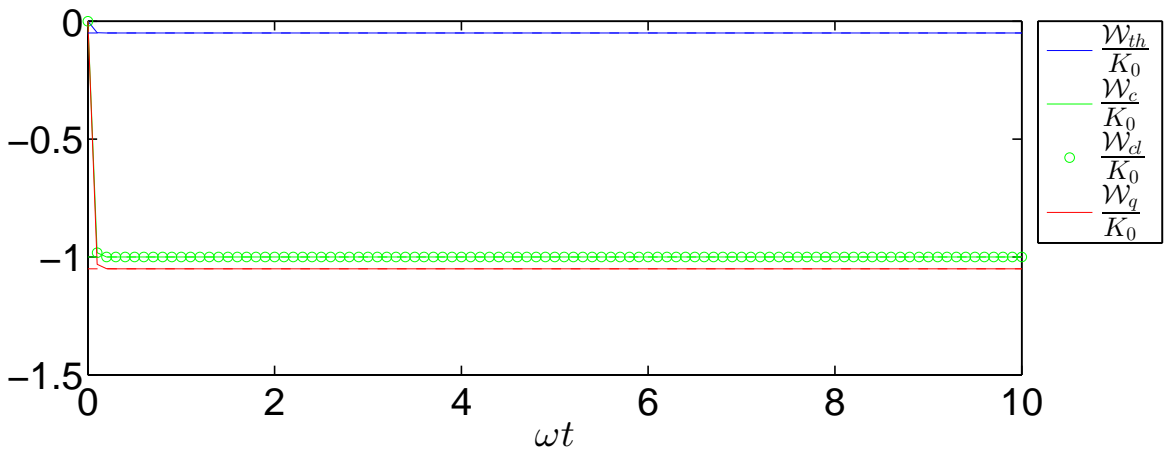}
			\caption{Overdamped regime $\frac{\omega}{\lambda}=0.1$.}
			\label{fig:nossakvsclassicoud}
		\end{subfigure}
		\caption{Centroid work $\mathcal{W}_{c}$ (green line), thermal work $\mathcal{W}_{th}$ (blue line) and quantum work $\mathcal{W}_{q}$ (red line) according to the proposed approach and the classical result for work (green circles) for the UO and OO as a function of the dimensionless time $\omega t$. The corresponding asymptotic limits ($t\to\infty$) are represented by dashed lines with the same color code.}
		\label{fig:nossakvsclassico}
	\end{center}
\end{figure} 
\par The centroid work $\mathcal{W}_{c}$ (green line) has the same behavior as the classical work, as it is expected from the analytical expressions. The thermal work $\mathcal{W}_{th}$ (blue line) decays slowly and consists of a small contribution to the quantum work $\mathcal{W}_{q}$ (red line). Most importantly, it is seen that $\mathcal{W}_{q}$ closely follows the classical work, a desirable behavior since the value $\vartheta = 0.1$ implies a semiclassical regime. 
\par A direct comparison with Alicki's results can be done for $k_{0}=0$ in which case the above formulas give 
\begin{equation}
\frac{\mathcal{W}_{q}(t)}{K_{0}}=-\left(1+\frac{\vartheta\left(1-\varepsilon_{\Delta}\right)}{
		1-\varepsilon}\right)\left(1-\mathrm{e}^{-4\tau}\right),
\end{equation}
\begin{equation}
\frac{\mathcal{W}_{c}(t)}{K_{0}}=-\left(1-\mathrm{e}^{-4\tau}\right).
\end{equation}
\begin{equation}
\frac{\mathcal{W}_{th}(t)}{K_{0}}=-\frac{\vartheta\left(1-\varepsilon_{\Delta}\right)}{
	1-\varepsilon}\left(1-\mathrm{e}^{-4\tau}\right),
\end{equation}
If $\mathcal{W}_{th}$ is to be regarded as a microscopic formulation of heat, then one may say from the result above that heat will be negative for all times, which is in accordance with the fact that energy is dissipated from the system. On the other hand, when no dissipation takes place ($\lambda=0$) one has, from Eq. \ref{eq:Kl0}
\begin{equation}
\frac{\mathcal{W}_{q}(t)}{K_{0}}=\left(\underline{\alpha}_{cl}^{\lambda=0}+\vartheta\underline{\beta}^{\lambda=0}\right)\cdot\left(\underline{\Gamma}^{\lambda=0}(t)-\underline{\Gamma}^{\lambda=0}(0)\right),
\end{equation}
\begin{equation}
\frac{\mathcal{W}_{c}(t)}{K_{0}}=\underline{\alpha}_{cl}^{\lambda=0}\cdot\left(\underline{\Gamma}^{\lambda=0}(t)-\underline{\Gamma}^{\lambda=0}(0)\right)
\end{equation}
and
\begin{equation}
\frac{\mathcal{W}_{th}(t)}{K_{0}}=\vartheta\underline{\beta}^{\lambda=0}\cdot\left(\underline{\Gamma}^{\lambda=0}(t)-\underline{\Gamma}^{\lambda=0}(0)\right),\label{eq:t}
\end{equation}
where $\underline{\alpha}_{cl}^{\lambda=0}$, $\underline{\Gamma}^{\lambda=0}(t)$ and $\underline{\beta}^{\lambda=0}$ are defined in Eqs. \eqref{eq:alphgl0} and \eqref{eq:betal0}.
As stated above the approach matches with the classical case if $\vartheta=0$ and can consider effects of quantum fluctuations in cases in which $\vartheta$ is not negligible.
\par An interesting question then arises from the above discussion. In the quantum domain the thermal work emerges from quantum fluctuations (irreducible quantum uncertainties). Is it possible to find some analogy for this term in classical-statistical mechanics, where fluctuations of a subjective nature (operational uncertainties) take place? This point is addressed in the next section.
\section{\textbf{Classical and quantum fluctuations}}
\label{sec:distribuicoes}
In Chapter \ref{cap:fundamentacao}, some forms have been discussed through which statistical elements can be incorporated to the framework of Classical Mechanics. As discussed in that occasion, one could conceive, in particular, a Gaussian  probability distribution with center at $(x_{0}^{*},p_{0}^{*})$ and variances $\sigma_{x,0}^{2}$ and $\sigma_{p,0}^{2}$ describing an ensemble of classical systems submitted to the same dynamics, with each element of the ensemble with a particular initial condition. The probability density for the occurrence of an initial position $(x_{0},p_{0})$ in such ensemble is given by
\begin{equation}
\rho_{gcl}(t,x_{0},p_{0})=\frac{1}{2\pi \sigma_{p,0}\sigma_{x,0}}\mathrm{e}^{-\frac{1}{2}\left(\left(\frac{x_{0}-x_{0}^{*}}{\sigma_{x,0}}\right)^{2}+\left(\frac{p_{0}-p_{0}^{*}}{\sigma_{p,0}}\right)^{2}\right)}=\mathcal{G}\left[ x_{0}^{*},p_{0}^{*},\sigma_{x,0},\sigma_{p,0}\right].
\end{equation}
As discussed in chapter \ref{cap:fundamentacao}, each element of the ensemble follows a classical trajectory and averages can be computed in the whole ensemble.  In Appendix \ref{cap:Work_classical}, details can be found for the computation of the mean velocity and the underlying variances in the context of the classical CK model. Using the classical velocity \eqref{eq:velocity1ckcl},  
\begin{equation}
v(t,p_{0},x_{0})=\mathrm{e}^{-\tau}\left[\frac{p_{0}}{m_{0}}\cosh\left(\zeta\tau\right)-\left(\frac{p_{0}}{m_{0}}+\frac{k_{0}x_{0}}{m_{0}\lambda}\right)\frac{\sinh\left(\zeta\tau\right)}{\zeta}\right],
\end{equation}
one obtains the statistical (Liouvillian) quantities:
\begin{equation}
\begin{array}{rl}
\langle v(t) \rangle_{gcl}^{2}=\mathrm{e}^{-2\tau}\left(\underline{\alpha}^{*}\cdot \underline{\Gamma}(t)\right)
\end{array}
\end{equation}
and
\begin{equation}
\langle v(t)^{2} \rangle_{gcl}=\mathrm{e}^{-2\tau}\left[\left(\underline{\alpha}^{*}+\underline{\beta}^{*}\right)\cdot \underline{\Gamma}(t)\right],
\end{equation}
where
\begin{equation}
\underline{\alpha}^{*}=\frac{1}{m_{0}^{2}}\left[
\begin{array}{c}
p_{0}^{*2}\\\displaystyle
-2\left(\frac{p_{0}^{*2}}{\zeta}+\frac{k_{0}p_{0}^{*}x_{0}^{*}}{\zeta\lambda}\right)\\
\displaystyle\left(\frac{p_{0}^{*2}}{\zeta^{2}}+2p_{0}^{*}\frac{k_{0}x_{0}^{*}}{\lambda\zeta^{2}}+\frac{k_{0}^{2}x_{0}^{*2}}{\lambda^{2}\zeta^{2}}\right)
\end{array}
\right],\qquad 
\underline{\beta}^{*}=\frac{1}{m_{0}^{2}}\left[
\begin{array}{c}
\sigma_{p,0}^{2}\\\displaystyle
-2\left(\frac{\sigma_{p,0}^{2}}{\zeta}\right)\\
\displaystyle\left(\frac{\sigma_{p,0}^{2}}{\zeta^{2}}+\frac{k_{0}^{2}\sigma_{x,0}^{2}}{\lambda^{2}\zeta^{2}}\right)
\end{array}
\right]
\end{equation}
and $\underline{\Gamma}(t)$ is given in Eq. \eqref{eq:alphacldef}. Since the goal is to compare the present classical-statistical results with the quantum ones, the same scaling process and set of parameters are adopted here, with pertinent notational adaptations:
 \begin{equation}
E_{0}^{*}\equiv \frac{m_{0}\omega^{2} x_{0}^{*2}}{2}+\frac{ p_{0}^{*2}}{2m_{0}},\label{eq:E02l}
\end{equation}
\begin{equation}
\frac{m_{0}\omega^{2} x_{0}^{*2}}{2}\equiv\varepsilon^{*} E_{0}^{*},\qquad K_{0}^{*}=\frac{ p_{0}^{*2}}{2m_{0}}\equiv\left(1-\varepsilon^{*}\right)E_{0}^{*},
\end{equation}
\begin{equation}
e_{0}^{*}\equiv \frac{m_{0}\omega^{2} \sigma_{x,0}^{2}}{2}+\frac{ \Delta_{p,0}^{2}}{2m_{0}},\label{eq:e02l}
\end{equation}
where 
\begin{equation}
\frac{m_{0}\omega^{2} \sigma_{x,0}^{2}}{2}=\varepsilon_{\Delta}^{*} e_{0}^{*},\qquad\frac{ \sigma_{p,0}^{2}}{2m_{0}}=\left(1-\varepsilon_{\Delta}^{*}\right)e_{0}^{*}\label{eq:flutu1l}
\end{equation}
with $0\leq\varepsilon^{*},\varepsilon_{\Delta}^{*}\leq 1$ and
\begin{equation}
\vartheta^{*}=\frac{e_{0}^{*}}{E_{0}^{*}}.\label{eq:flutu2l}
\end{equation}
After some tedious algebraic manipulations, it can be shown that  
\begin{equation}
\displaystyle\frac{\frac{m_{0}}{2}\langle v^{2}\rangle_{gcl}}{K_{0}}=\mathrm{e}^{-2\tau}\left[\left(\underline{\alpha}_{gcl}+\vartheta^{*}\underline{\beta}_{gcl}\right)\cdot \underline{\Gamma}(t)\right],\label{eq:Resultsv22}
\end{equation}
\begin{equation}
\displaystyle\frac{\frac{m_{0}}{2}\langle v\rangle_{gcl}^{2}}{K_{0}}=\mathrm{e}^{-2\tau}\left(\underline{\alpha}_{gcl}\cdot \underline{\Gamma}(t)\right)\label{eq:MMV22}
\end{equation}
and 
\begin{equation}
\displaystyle
\frac{\frac{m_{0}}{2}\langle\left(\sigma_{v}\right)^{2}\rangle_{gcl}}{K_{0}^{*}}=\vartheta^{*}\mathrm{e}^{-2\tau}\left(\underline{\beta}_{gcl}\cdot \underline{\Gamma}(t)\right),\label{eq:MMDV22}
\end{equation}
with 
\begin{equation}
\underline{\alpha}_{gcl}=\frac{1}{1-\varepsilon^{*}}\left[
\begin{array}{c}
1-\varepsilon^{*}\\\displaystyle
-2\left(\frac{\omega \sqrt{\varepsilon^{*}-\varepsilon^{*2}}}{\lambda\zeta}+\frac{1-\varepsilon^{*}}{\zeta}\right)\\
\displaystyle\left(\omega^{2}\frac{\varepsilon^{*}}{\lambda^{2}\zeta^{2}}+2\omega\frac{\sqrt{\varepsilon^{*}-\varepsilon^{*2}}}{\lambda\zeta^{2}}+\frac{1-\varepsilon^{*}}{\zeta^{2}}\right)
\end{array}
\right]
\end{equation}
and
\begin{equation}
\underline{\beta}_{gcl}=\frac{1}{1-\varepsilon^{*}}\left[
\begin{array}{c}
\displaystyle 1-\varepsilon_{\Delta}^{*}\\
\displaystyle -2\frac{1-\varepsilon_{\Delta}^{*}}{\zeta}\\
\displaystyle\left(\frac{\omega^{2}}{\lambda^{2}}\frac{\varepsilon_{\Delta}^{*}}{\zeta^{2}}+\frac{1-\varepsilon_{\Delta}^{*}}{\zeta^{2}}\right)
\end{array}
\right].
\end{equation}
The comparison of these results with Eqs. \eqref{eq:Resultsv2}, \eqref{eq:MMV2} and \eqref{eq:MMDV2} can be done via the identifications $K_{0}^{*}= K_{0}$, $\varepsilon^{*}= \varepsilon$, $\varepsilon_{\Delta}^{*}=\varepsilon_{\Delta}$. With that, one shows that the expressions deduced for $\frac{\frac{m_{0}}{2}\langle v^{2}\rangle_{gcl}}{K_{0}}$, $\frac{\frac{m_{0}}{2}\langle v\rangle_{gcl}^{2}}{K_{0}}$ and $\frac{\frac{m_{0}}{2}\langle\left(\sigma_{v}\right)^{2}\rangle_{gcl}}{K_{0}}$ are identical to that obtained for $\frac{\frac{m_{0}}{2}\langle V^{2} \rangle_{q}}{K_{0}}$, $\frac{\frac{m_{0}}{2}\langle V \rangle_{q}^{2}}{K_{0}}$ and $\frac{\frac{m_{0}}{2}\langle \left(\Delta V\right)^{2} \rangle_{q}}{K_{0}}$, respectively. In other words, the classical-statistical formalism with a Gaussian distribution $\rho_{gcl}$ perfectly emulates the quantum results for a Gaussian pure density operator $\ket{\psi(t)}\bra{\psi(t)}$:
\emph{similar expressions were obtained for the Gaussian quantum state $\ket{\psi(t)}$ and the statistical classical distribution $\rho_{gcl}$}. From the work-energy relation \eqref{eq:wcl2222}, one can compute work as
\begin{equation}
\mathcal{W}_{gcl}(t)=\left.\frac{m_{0} \langle v^{2}(t^{'})\rangle_{gcl}}{2}\right|_{0}^{t}.
\end{equation} 
Then, making the pertinent identifications $K_{0}^{*}= K_{0}$, $\varepsilon^{*}= \varepsilon$ and $\varepsilon_{\Delta}^{*}=\varepsilon_{\Delta}$ one finds a full correspondence between quantum and classical-statistical works:
\begin{equation}
\frac{\mathcal{W}_{q}(t)}{K_{0}}=\left.\frac{\frac{m_{0}}{2}\langle V^{2}\rangle_{q}}{K_{0}}\right|_{0}^{t}\Leftrightarrow \frac{\mathcal{W}_{gcl}(t)}{K_{0}}=\frac{\left.\frac{m_{0} \langle v^{2}(t^{'})\rangle_{gcl}}{2}\right|_{0}^{t}}{K_{0}},\label{eq:v2q2}
\end{equation}
with the correspondence applying also for the centroid and thermal parcels,
\begin{equation}
\frac{	\mathcal{W}_{c}(t)}{K_{0}}\Leftrightarrow \left.\frac{\frac{m_{0}}{2}\langle v\rangle_{gcl}^{2}}{K_{0}^{*}}\right|_{0}^{t}
\end{equation}
and
\begin{equation}
\frac{\mathcal{W}_{th}(t)}{K_{0}}\Leftrightarrow \left.\frac{\frac{m_{0}}{2}\langle\left(\sigma_{v}\right)^{2}\rangle_{gcl}^{2}}{K_{0}^{*}}\right|_{0}^{t}.
\end{equation}
Aiming at further appreciating the connections between the quantum and the classical-statistical formalism, superposition and mixture states are analyzed next.
\subsection{Mixed and superposition states}
\label{subsec:quanticadist}
As mentioned in subsection \ref{subsec:subjective} a mixed probability distribution can be written as
\begin{equation}
\rho_{mgcl}(t,x_{0},p_{0})=\frac{1}{2}\left(\mathcal{G}\left[ x_{0}^{*},p_{0}^{*},\sigma_{x,0},\sigma_{p,0}\right]+\mathcal{G}\left[ -x_{0}^{*},-p_{0}^{*},\sigma_{x,0},\sigma_{p,0}\right]\right)
\end{equation}
with centers at $\pm(x_{0}^{*},p_{0}^{*})$ and uncertainties $\sigma_{x,0}$ and $\sigma_{p,0}$. In Appendix \ref{cap:Work_classical}, expressions have been derived for the mean velocity and its underlying variance in the context of the classical CK model. From Eqs. \eqref{eq:vmcl2} and \eqref{eq:v2mcl} it follows that
\begin{equation}
\langle v(t) \rangle_{mgcl}^{2}=0\label{eq:vmcl22}
\end{equation}
and
\begin{equation}
\langle v(t)^{2} \rangle_{mgcl}=\langle v(t)^{2} \rangle_{gcl}.
\label{eq:v2mcl2}
\end{equation}
In other words, the same mean value for the square velocity is found for both distributions $\rho_{mgcl}$ and $\rho_{gcl}$. Therefore, the expression \eqref{eq:Resultsv22} can be used for describing $\frac{\frac{m_{0}}{2}\langle v(t)^{2} \rangle_{mgcl}}{K_{0}^{*}}$, and even more, it can be written in a dimensionless form, as done for the quantum analogous $\frac{\frac{m_{0}}{2}\langle V^{2} \rangle_{q}}{K_{0}}$. As a result, from Eq.\eqref{eq:v2mcl2} and given the correspondence described in Eq. \eqref{eq:v2q2}, the resultant work $\frac{\mathcal{W}_{mgcl}(t)}{K_{0}}$, related with the mixed Gaussian state, is found to be in full correspondence with $\frac{\mathcal{W}_{q}}{K_{0}}$, that is
\begin{equation}
\frac{\mathcal{W}_{mgcl}(t)}{K_{0}}=\left.\frac{\frac{m_{0} \langle v^{2}(t^{'})\rangle_{gcl}}{2}}{K_{0}}\right|_{0}^{t}\Leftrightarrow \frac{\mathcal{W}_{q}(t)}{K_{0}}=\left.\frac{\frac{m_{0}}{2}\langle V^{2}\rangle_{q}}{K_{0}}\right|_{0}^{t}.\label{eq:wmgcl}
\end{equation}
In words, had one considered the dimensionless expression for work, of a quantum Gaussian state, then the same expression would apply for the classical mixed state. However, the velocity mean value $\langle v(t) \rangle_{mgcl}$ is null, thus differing from the expressions obtained for the quantum state $\rho(t)$ and for $\rho_{gcl}$. To better appreciate the quantum-classical connection, the following density quantum operator is considered:
\begin{equation}
\rho^{\mu}(0)=\frac{\ket{\psi_{0}}\bra{\psi_{0}}+\ket{\psi_{0}^{-}}\bra{\psi_{0}^{-}}+\mathrm{e}^{-\mu}\left(\ket{\psi_{0}^{-}}\bra{\psi_{0}}+\ket{\psi_{0}}\bra{\psi_{0}^{-}}\right)}{N_{\mu}},\label{eq:statot2}
\end{equation}
where $\mu\geq 0$ is a real non-negative parameter and $\ket{\psi_{0}}$ is such that 
\begin{equation}
\langle x |\psi_{0}\rangle\equiv\Psi(x,0)=\left(\frac{1}{2\pi\Delta_{x,0}^{2}}\right)^{\frac{1}{4}}\exp \left[\frac{-\left(x-x_{0}\right)^{2}}{4\Delta_{x,0}^{2}}+i\frac{p_{0}x}{\hbar}\right],\label{eq:estnor2}
\end{equation}
with a symmetrical counterpart
\begin{equation}
\langle x |\psi_{0}^{-}\rangle\equiv\Psi^{-}(x,0)=\left(\frac{1}{2\pi\Delta_{x,0}^{2}}\right)^{\frac{1}{4}}\exp \left[\frac{-\left(x+x_{0}\right)^{2}}{4\Delta_{x,0}^{2}}-i\frac{p_{0}x}{\hbar}\right].\label{eq:estrev2}
\end{equation}
$N_{\mu}$ is the normalization factor, given, as usual as 
\begin{equation}
\begin{array}{rl}
N_{\mu}&=\displaystyle \mathrm{Tr}\left[\ket{\psi_{0}}\bra{\psi_{0}}+\ket{\psi_{0}^{-}}\bra{\psi_{0}^{-}}+\mathrm{e}^{-\mu}\left(\ket{\psi_{0}^{-}}\bra{\psi_{0}}+\ket{\psi_{0}}\bra{\psi_{0}^{-}}\right)\right],\\
&\displaystyle =2\left(1+\mathrm{e}^{-\left[\mu+\frac{x_{0}^{2}}{2\Delta_{x,0}^{2}}+\frac{p_{0}^{2}4\Delta_{x,0}^{2}}{2\hbar^{2}}\right]}\right) \\
\end{array}
\end{equation}
\par The parameter $\mu$ was introduced to allow one to address two important limits: 
\begin{itemize}
	\item When $\mu\rightarrow 0$, then the initial state is described as 
	\begin{equation}
		\rho^{0}(0)=\frac{\ket{\psi_{0}}\bra{\psi_{0}}+\ket{\psi_{0}^{-}}\bra{\psi_{0}^{-}}+\ket{\psi_{0}^{-}}\bra{\psi_{0}}+\ket{\psi_{0}}\bra{\psi_{0}^{-}}}{N_{0}}=\frac{\left(\ket{\psi_{0}}+\ket{\psi_{0}^{-}}\right)\left(\bra{\psi_{0}}+\bra{\psi_{0}^{-}}\right)}{N_{0}}
	\end{equation}
	which is the density matrix operator for a superposition state $\frac{\ket{\psi_{0}}+\ket{\psi_{0}^{-}}}{\sqrt{N_{0}}}$;
	\item In the limit  $\mu\rightarrow \infty$, then
	\begin{equation}
	\rho^{\infty}(0)=\frac{\ket{\psi_{0}}\bra{\psi_{0}}+\ket{\psi_{0}^{-}}\bra{\psi_{0}^{-}}}{N_{\infty}},
	\end{equation}
	which is the mixture state analogue to $\rho_{mgcl}$.    
\end{itemize}
    Therefore, $\mu$ is a control parameter that interpolates these extrema. 
    \par The calculations related to the time evolution of $\rho$ and relevant expectation values are presented in Appendix \ref{cap:Work_quantum}. Among the results obtained, the ones described by Eqs.\eqref{eq:MM2} and \eqref{eq:MMM2} are essential for the present discussion:
    \begin{equation}
    \langle P \rangle_{q}^{\mu}=\mathrm{Tr}\left(\rho^{\mu}(t)P\right)=0,\label{eq:MM22}
    \end{equation}
    \begin{equation}
    \begin{array}{rl}
    \langle P^{2} \rangle_{q}^{\mu}&=\mathrm{Tr}\left(\rho^{\mu}(t)P^{2}\right)=\bra{\psi(t)} P^{2}\ket{\psi(t)}+\\
    &-\left(\frac{p_{0}^{2}}{\frac{\hbar^{2}}{4\Delta_{x,0}^{2}}}+\frac{x_{0}^{2}}{\frac{4\Delta_{x,0}^{2}}{4}}\right)\frac{\bra{\psi(t)} \left(\Delta P\right)^{2}\ket{\psi(t)}}{1+\exp \left[\mu+\frac{1}{2}\left(\frac{p_{0}^{2}}{\frac{\hbar^{2}}{4\Delta_{x,0}^{2}}}+\frac{x_{0}^{2}}{\frac{4\Delta_{x,0}^{2}}{4}}\right)\right]}.
    \end{array}
    \label{eq:MMM22}
    \end{equation}
 Since $\bra{\psi(t)} P^{2}\ket{\psi(t)}=\mathrm{Tr}\left(\rho(t)P^{2}\right)=\langle P^{2}\rangle_{q}$, $\bra{\psi(t)} \left(\Delta P\right)^{2}\ket{\psi(t)}=\langle \left(\Delta P\right)^{2}\rangle_{q}$ and
\begin{equation}
V=\frac{P}{m_{0}}\mathrm{e}^{-2\lambda t},
\end{equation} 
then
\begin{equation}
\frac{m_{0}}{2}V^{2}=\frac{P^{2}}{2m_{0}}\mathrm{e}^{-4\lambda t}
\end{equation} 
and
\begin{equation}
\frac{m_{0}}{2}\left(\Delta V\right)^{2}=\frac{\left(\Delta P\right)^{2}}{2m_{0}}\mathrm{e}^{-4\lambda t}
\end{equation} 
yielding, from Eq. \eqref{eq:MM22},
\begin{equation}
\frac{m_{0}}{2}\left(\langle V \rangle_{q}^{\mu}\right)^{2}=0.\label{eq:MM3}
\end{equation}
From \eqref{eq:MMM22}, one finds
\begin{equation}
\frac{m_{0}}{2}\langle V^{2} \rangle_{q}^{\mu}=\frac{m_{0}}{2}\langle V^{2} \rangle_{q}-\left(\frac{p_{0}^{2}}{\frac{\hbar^{2}}{4\Delta_{x,0}^{2}}}+\frac{x_{0}^{2}}{\frac{4\Delta_{x,0}^{2}}{4}}\right)\frac{\frac{m_{0}}{2}\langle \left(\Delta V\right)^{2} \rangle_{q}}{1+\exp \left[\mu+\frac{1}{2}\left(\frac{p_{0}^{2}}{\frac{\hbar^{2}}{4\Delta_{x,0}^{2}}}+\frac{x_{0}^{2}}{\frac{4\Delta_{x,0}^{2}}{4}}\right)\right]}.\label{eq:MMM23}
\end{equation}
Dividing both sides of Eq. \eqref{eq:MMM23} by $K_{0}=\frac{p_{0}^{2}}{2m_{0}}$ and considering the same parameters defined in subsection \ref{subsec:dimensionless}, it is obtained
\begin{equation}
\frac{\frac{m_{0}}{2}\langle V^{2} \rangle_{q}^{\mu}}{K_{0}}=\frac{\frac{m_{0}}{2}\langle V^{2} \rangle_{q}}{K_{0}}-\frac{1}{\vartheta}\left(\frac{\varepsilon_{\Delta}+\varepsilon-2\varepsilon_{\Delta}\varepsilon}{\varepsilon_{\Delta}\left(1-\varepsilon_{\Delta}\right)}\right)\frac{\frac{\frac{m_{0}}{2}\langle \left(\Delta V\right)^{2} \rangle_{q}}{K_{0}}}{1+\exp \left[\mu+\frac{1}{2\vartheta}\left(\frac{\varepsilon_{\Delta}+\varepsilon-2\varepsilon_{\Delta}\varepsilon}{\varepsilon_{\Delta}\left(1-\varepsilon_{\Delta}\right)}\right)\right]}.\label{eq:auxxxx}
\end{equation}
With the relation (reproduced here from Eq. \eqref{eq:vartheta}) 
\begin{equation}
\vartheta=\frac{\hbar\omega}{E_{0}4\sqrt{\varepsilon_{\Delta}\left(1-\varepsilon_{\Delta}\right)}},\label{eq:vartheta2}
\end{equation}
 the following constant can be defined:
\begin{equation}
	\vartheta^{'}=\left[\frac{1}{\vartheta}\left(\frac{\varepsilon_{\Delta}+\varepsilon-2\varepsilon_{\Delta}\varepsilon}{\varepsilon_{\Delta}\left(1-\varepsilon_{\Delta}\right)}\right)\right]^{-1}=\vartheta\frac{\varepsilon_{\Delta}\left(1-\varepsilon_{\Delta}\right)}{\varepsilon_{\Delta}+\varepsilon-2\varepsilon_{\Delta}\varepsilon}=\frac{\hbar\omega}{4E_{0}}\frac{\sqrt{\varepsilon_{\Delta}\left(1-\varepsilon_{\Delta}\right)}}{\varepsilon_{\Delta}+\varepsilon-2\varepsilon_{\Delta}\varepsilon}.\label{eq:varthetal}
\end{equation}
Eq. \eqref{eq:auxxxx} can then be rewritten as
\begin{equation}
\frac{\frac{m_{0}}{2}\langle V^{2} \rangle_{q}^{\mu}}{K_{0}}=\frac{\frac{m_{0}}{2}\langle V^{2} \rangle_{q}}{K_{0}}-\frac{1}{\vartheta^{'}}\frac{\frac{\frac{m_{0}}{2}\langle  \left(\Delta V\right)^{2} \rangle_{q}}{K_{0}}}{1+\exp \left[\mu+\frac{1}{2\vartheta^{'}}\right]}. \label{eq:docarai}
\end{equation}
Finally, from Eq. \eqref{eq:MM3}, one finds
\begin{equation}
\frac{\frac{m_{0}}{2}\langle \left(\Delta V\right)^{2} \rangle_{q}^{\mu}}{K_{0}}=\frac{\frac{m_{0}}{2}\langle V^{2} \rangle_{q}^{\mu}}{K_{0}}=\frac{\frac{m_{0}}{2}\langle V^{2} \rangle_{q}}{K_{0}}-\frac{1}{\vartheta^{'}}\frac{\frac{\frac{m_{0}}{2}\langle V^{2} \rangle_{q}}{K_{0}}}{1+\exp \left[\mu+\frac{1}{2\vartheta^{'}}\right]}. \label{eq:docarai2}
\end{equation}
\par  As for the mixed classical distribution $\rho_{mgcl}$, the mean velocity will be null and, as a consequence, $\frac{m_{0}}{2}\left(\langle V \rangle_{q}^{\mu}\right)^{2}=0$, for any value of $\mu$. In particular, for  $\mu\rightarrow \infty$ it follows that 
\begin{equation}
	\frac{\frac{m_{0}}{2}\langle \left(\Delta V\right)^{2} \rangle_{q}^{\infty}}{K_{0}}=\frac{\frac{m_{0}}{2}\langle V^{2}\rangle_{q}^{\infty}}{K_{0}}\rightarrow \frac{\frac{m_{0}}{2}\langle V^{2} \rangle_{q}}{K_{0}}.\label{eq:minf}
\end{equation}
Therefore, for the quantum mixed state, $\rho_{q}^{\infty}$, the quantum work $\mathcal{W}_{q}^{\infty}$ can be written as
\begin{equation}
\frac{\mathcal{W}_{q}^{\infty}(t)}{K_{0}}=\left.\frac{\frac{m_{0}}{2}\langle V^{2}\rangle_{q}^{\infty}}{K_{0}}\right|_{0}^{t}=\left.\frac{\frac{m_{0}}{2}\langle V^{2}\rangle_{q}}{K_{0}}\right|_{0}^{t}= \frac{\mathcal{W}_{q}(t)}{K_{0}}.\label{eq:Wqrinf}
\end{equation}
Moreover, as stated explicitly in Eq. \eqref{eq:wmgcl}, the work $\mathcal{W}_{mgcl}$ defined classically for $\rho_{mgcl}$, has the same dimensionless expression as $\mathcal{W}_{q}$, which, from Eq. \eqref{eq:Wqrinf}, implies that 
\begin{equation}
\frac{\mathcal{W}_{mgcl}(t)}{K_{0}}=\left.\frac{\frac{m_{0} \langle v^{2}(t^{'})\rangle_{gcl}}{2}}{K_{0}}\right|_{0}^{t}\Leftrightarrow \frac{\mathcal{W}_{q}^{\infty}(t)}{K_{0}}=\left.\frac{\frac{m_{0}}{2}\langle V^{2}\rangle_{q}^{\infty}}{K_{0}}\right|_{0}^{t}.\label{eq:mixedcqr}
\end{equation}
That is, \emph{ the dimensionless quantum work for an incoherent mixture of Gaussian states equals the classical-statistical work}. Not only the quantum work but the centroid and thermal work have their classical correspondence as
\begin{equation}
\frac{\mathcal{W}_{c}^{\infty}(t)}{K_{0}}=\left.\frac{\frac{m_{0}}{2}\left(\langle V \rangle_{q}^{\infty}\right)^{2}}{K_{0}}\right|_{0}^{t}\Leftrightarrow \left.\frac{\frac{m_{0}}{2}\langle v\rangle_{mgcl}^{2}}{K_{0}}\right|_{0}^{t}\label{eq:mixedcqc}
\end{equation}
and
\begin{equation}
\frac{\mathcal{W}_{th}^{\infty}(t)}{K_{0}}=\left.\frac{\frac{m_{0}}{2}\langle \left(\Delta V\right)^{2} \rangle_{q}^{\mu}}{K_{0}}\right|_{0}^{t}\Leftrightarrow \left.\frac{\frac{m_{0}}{2}\langle\left(\sigma_{v}\right)^{2}\rangle_{mgcl}^{2}}{K_{0}}\right|_{0}^{t}.\label{eq:mixedcqt}
\end{equation}
It is important to remark that, even in cases in which the quantum fluctuations are significant ($\vartheta$ large) the results  \eqref{eq:mixedcqr}, \eqref{eq:mixedcqc}, and \eqref{eq:mixedcqt} are verified. Interestingly, this is an example of scenario where, thanks to the absence of quantum coherence ($\mu\rightarrow \infty$), the classical-statistical formalism can produce a faithful description of microscopic phenomena.
\par The differences between quantum and classical descriptions becomes more apparent when $\mu\rightarrow 0$. In this case, the state is a coherent superposition and
\begin{equation}
\frac{\frac{m_{0}}{2}\langle \left(\Delta V\right)^{2} \rangle_{q}^{0}}{K_{0}}=\frac{\frac{m_{0}}{2}\langle V^{2} \rangle_{q}^{0}}{K_{0}}=\frac{\frac{m_{0}}{2}\langle V^{2} \rangle_{q}}{K_{0}}-\frac{1}{\vartheta^{'}}\frac{\frac{\frac{m_{0}}{2}\langle  \left(\Delta V\right)^{2} \rangle_{q}}{K_{0}}}{1+\exp \left[\frac{1}{2\vartheta^{'}}\right]}. \label{eq:docarai3}
\end{equation} 	
In general, the forms for the work terms are, from \eqref{eq:docarai3} and \eqref{eq:MM3}, written as
\begin{equation}
\frac{\mathcal{W}_{q}^{0}(t)}{K_{0}}=\frac{\mathcal{W}_{q}(t)}{K_{0}}-\frac{1}{\vartheta^{'}}\frac{\frac{\mathcal{W}_{th}(t)}{K_{0}}}{1+\exp \left[\frac{1}{2\vartheta^{'}}\right]},
\end{equation}
\begin{equation}
\mathcal{W}_{c}^{0}(t)=0
\end{equation}
and
\begin{equation}
\frac{\mathcal{W}_{th}^{0}(t)}{K_{0}}=\frac{\mathcal{W}_{q}(t)}{K_{0}}-\frac{1}{\vartheta^{'}}\frac{\frac{\mathcal{W}_{th}(t)}{K_{0}}}{1+\exp \left[\frac{1}{2\vartheta^{'}}\right]}.
\end{equation}
Comparing these expressions with Eqs. \eqref{eq:minf}, \eqref{eq:Wqrinf} and \eqref{eq:MM3}, it can be explicitly verified the strong influence of $\vartheta^{'}$, connected with the quantum fluctuation terms. Since the centroid work will be null for both cases, the differences will be verified for the quantum and thermal work. In order to appreciate the differences implied by quantum coherence, it is considered in Fig. \ref{fig:superpos1} the comparison between the quantum work for an underdamped oscillator for which $\frac{\omega}{\lambda}=10$, $\varepsilon=0$, $\varepsilon_{\Delta}=0.5$ and $\vartheta=1$.
\begin{figure}[h]
	\begin{center}
		\includegraphics[angle=0, scale=0.6]{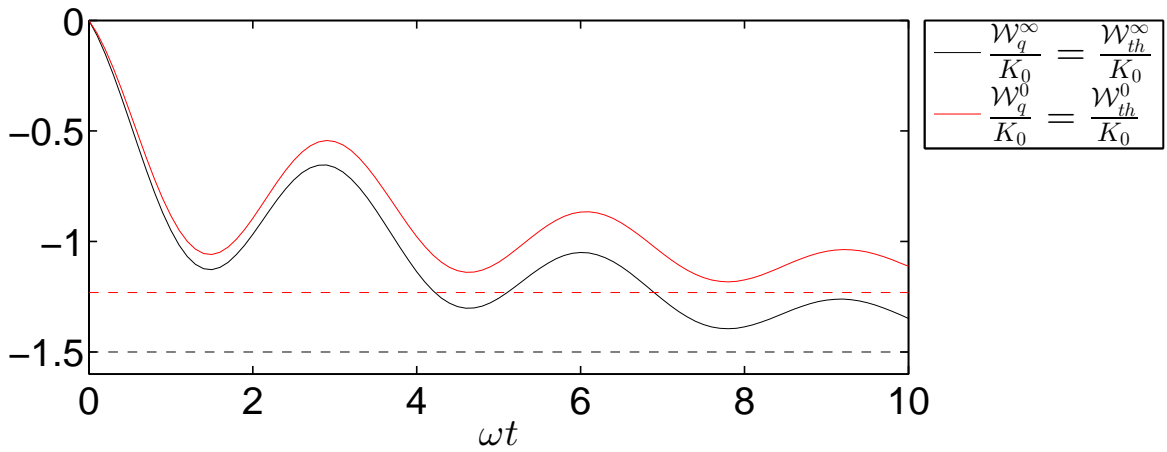} 
		\caption{Evolution of the quantum (and thermal) work for the Gaussian mixture (black line) and the coherent superposition state (red line), for $\frac{\omega}{\lambda}=10$, $\varepsilon=0$, $\varepsilon_{\Delta}=0.5$ and $\vartheta=1$. The corresponding asymptotic limits ($t\to\infty$) are represented by dashed lines with the same color code.}
		\label{fig:superpos1}
	\end{center}
\end{figure}
The energy flowing outside the system will be greater for the mixture than the superposition, as the environment performs work on the system. Furthermore, since the centroid work is null for both cases, $\mathcal{W}_{q}^{0}(t)=\mathcal{W}_{th}^{0}(t)$ and $\mathcal{W}_{q}^{\infty}(t)=\mathcal{W}_{th}^{\infty}(t)$ and the analysis depicted in Fig. \ref{fig:superpos1} also represents a comparison between thermal work, in a mixture and coherent superposition state.
\par Had one connected thermal work with heat, then it could be inferred that more heat flow from a mixture system than from a coherent superposition. That is, it would be possible, in principle, to differentiate both states knowing the heat flow, showing a connection between thermal mechanisms (heat) with information properties (coherence).
\par Another analysis can be made by comparing the superposition coherent state with the purely Gaussian one. First, it is noticed that when the quantum fluctuations terms are negligible ($\vartheta\approx 0$) then quantum coherence does not very much influence the quantum works regarding both states, since
\begin{equation}
\frac{\mathcal{W}_{q}^{0}(t)}{K_{0}}	=\left.\frac{\frac{m_{0}}{2}\langle V^{2}\rangle_{q}^{0}}{K_{0}}\right|_{0}^{t}\approx \left.\frac{\frac{m_{0}}{2}\langle V^{2}\rangle_{q}}{K_{0}}\right|_{0}^{t}=	\frac{\mathcal{W}_{q}(t)}{K_{0}}\approx \frac{\mathcal{W}_{c}(t)}{K_{0}}.
\end{equation}
However, the same can not be asserted for the centroid work, since
\begin{equation}
\frac{\mathcal{W}_{c}^{0}(t)}{K_{0}}=\left.\frac{\frac{m_{0}}{2}\left(\langle V\rangle_{q}^{0}\right)^{2}}{K_{0}}\right|_{0}^{t}=0\neq \left.\frac{\frac{m_{0}}{2}\langle V^{2}\rangle_{q}}{K_{0}}\right|_{0}^{t}.
\end{equation}
Furthermore, the thermal work will also be different for the superposition state in the case in which $\vartheta\approx 0$, since
\begin{equation}
\frac{\mathcal{W}_{th}^{0}(t)}{K_{0}}=\left.\frac{\frac{m_{0}}{2}\langle\left(\Delta  V\right)^{2}\rangle_{q}^{0}}{K_{0}}\right|_{0}^{t}=\left.\frac{\frac{m_{0}}{2}\langle  V^{2}\rangle_{q}^{0}}{K_{0}}\right|_{0}^{t}\approx  \left.\frac{\frac{m_{0}}{2}\langle V^{2}\rangle_{q}}{K_{0}}\right|_{0}^{t}=	\frac{\mathcal{W}_{q}(t)}{K_{0}}.
\end{equation}
In order to verify the differences between a superposition and a purely Gaussian quantum states, it is considered in Figs \ref{fig:superpos11} and \ref{fig:superpos2} the comparison between quantum, thermal and centroid work for an underdamped oscillator for which $\frac{\omega}{\lambda}=10$, $\varepsilon=0$, $\varepsilon_{\Delta}=0.5$ and $\vartheta=1$.
\begin{figure}[h]
	\begin{center}
		\begin{subfigure}[h]{1\textwidth}
			\includegraphics[angle=0, scale=0.6]{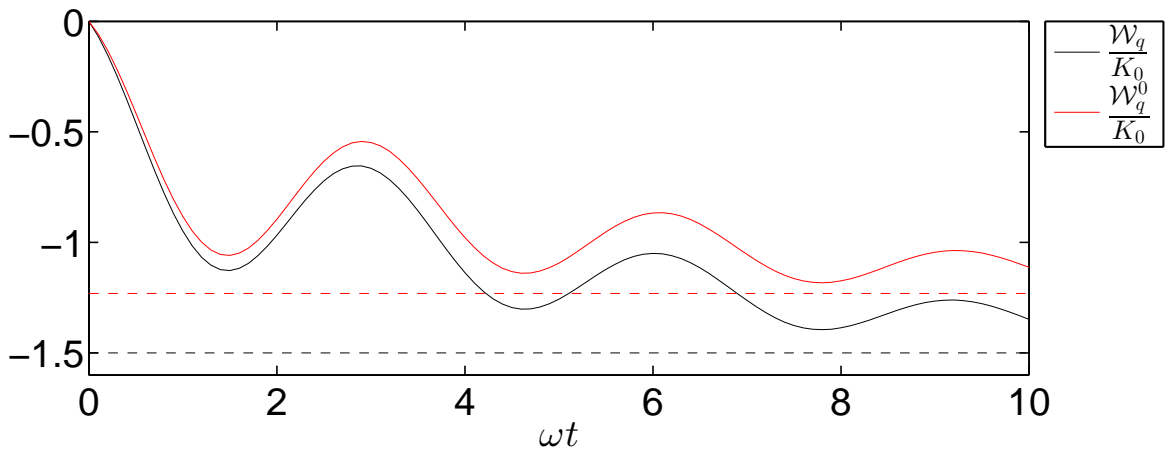} 
			\caption{Comparison between the quantum work regarding: a Gaussian state (black line) and a superposition coherent state (red line).}
			\label{fig:superpos11}
		\end{subfigure}
		\begin{subfigure}[h]{1\textwidth}
			\includegraphics[angle=0, scale=0.6]{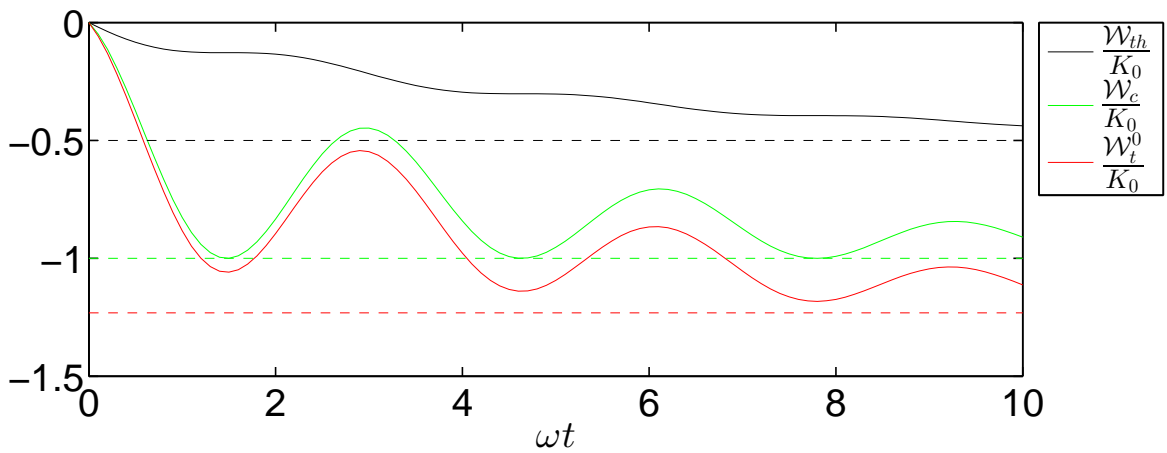}
			\caption{Comparison between the centroid work of a Gaussian state (green line), the thermal work regarding a Gaussian state (black line), and the thermal work for a superposition coherent state (red line).}
			\label{fig:superpos2}
		\end{subfigure}
		\caption{Evolution of the thermal and quantum work for the Gaussian and the coherent superposition state, for $\frac{\omega}{\lambda}=10$, $\varepsilon=0$, $\varepsilon_{\Delta}=0.5$ and $\vartheta=1$. The corresponding asymptotic limits ($t\to\infty$) are represented by dashed lines with the same color code.}
		\label{fig:superpos}
	\end{center}
\end{figure}
\par The results show that quantum coherence diminishes the loss of energy of the system (block) to the environment (spring + air) in comparison with the Gaussian state. However, it is worth noticing that, since $\mathcal{W}_{c}^{0}(t)=0$ for all $t$, that is, the centroid work is null, then an important conceptual difference emerges when comparing the energy loss for a superposition of Gaussian states with the one for a simple Gaussian state. In the former case, the energy loss is exclusively thermal, while in the latter it is just minimally thermal. Had one considered thermal work as heat and centroid work as the classic perspective of thermodynamical work, then one might conclude that  \emph{heat is the central form of energy transfer that takes place when a system is in a symmetric quantum superposition}.
\par It can be concluded, from the results of the present chapter, that for a dissipative system, as described by the CK model, the quantum work here introduced is in quite good agreement with the classical work for Gaussian states. Additionally, work has been computed also via the classical-statistical (Liouvillian) formalism and its adequacy has been verified for mixture states. In presence of quantum coherence, the essential role of thermal work has been identified. The commonly used Alicki's definition of heat and work was put into test, and it was verified that it differs significantly from the classical work in the strictly classical regime. In what follows, some concluding remarks are made.

%% file: consideracoesfisica.tex
\chapter{FINAL REMARKS}
\label{cap:consid}

The main objective of the present dissertation was to propose and analyze a new quantum mechanical definition of work. To accomplish this task, this work was constructed as follows.
\begin{enumerate}[label=\roman*]
	\item A review of the main aspects related with the definition of work and heat within the scope of Classical Mechanics, Thermodynamics and Statistical Physics was made;
	\item The Quantum Mechanics formalism was described and considered as basis for the definition proposed;
	\item The current approaches, regarding the definition of work and heat within the realm of Quantum Thermodynamics, were reviewed;
	\item The main aspects related with Alicki's definition of heat and work were described and some criticisms presented;
	\item The framework for the application of the new definition of work ($W_{q}^{R}$) was established and some general results obtained;
	\item The proposed quantum work $W_{q}^{R}$ was separated into centroid and thermal terms; the role of latter term as heat was discussed;
	\item For the Caldirola-Kanai model, Gaussian states were used to compare $W_{q}^{R}$ with both Alicki's proposal and the classical work; 
	\item The results obtained for $W_{q}^{R}$ were found to be in quite good agreement with the classical and classical-statistical formulas in (semi)classical regime. In quantum regimes, physical interpretations were proposed;
	\item Alicki's approach was shown not to suitably reproduce the classical limit for the model under investigation;
	\item Aspects related with work and heat were discussed for a scenario involving quantum coherence, in which case classical mimics do not apply.
\end{enumerate}
\par From a general perspective, the goal of this dissertation has been fully achieved, since a new notion of quantum mechanical work was (i) formally introduced, (ii) analytically computed in a case study, and (iii) compared with well established approaches showing, in particular, some conceptual advantages.
\par With regards to the presently-introduced notion of quantum mechanical work, many questions are left open for future works. In particular, the following research lines can be enumerated:
\begin{enumerate}[label=\roman*]
	\item Application of the present notion of work in many-particle systems or for the center of mass of a given part (subsystem) of such systems.
	\item Thorough investigation of the thermal work meaning. It is important to clarify what conceptual connections exist (if any) between thermal work and the thermodynamical notions of heat, temperature, and entropy.
	\item Applications to spin systems, where the mechanical notion of work is subtle.
	\item Investigation of the emergence of thermodynamical irreversibility from a microscopic quantum mechanical substratum. 
	\item Related notions such as thermal equilibrium, time averaging, ensemble averaging, and information discard, in connection with quantum work and heat, need to be understood from this fundamental viewpoint.
\end{enumerate}

%% file: apendice_caldirolakanai.tex
\chapter{CALDIROLA-KANAI MODEL: SOLUTION AND EXPECTATION MEAN VALUES}
\label{cap:kanaicaldirolasol}
\par The Caldirola-Kanai model was used for testing the definition of work presented at the this dissertation and in comparison with Alicki's definition. The results of this analysis are presented  in chapter \ref{cap:kanaicaldirola}, using expressions described in this appendix. The structure of the appendix is organized as follows: to begin with, it is deduced the form of a Gaussian state after a time evolution under a Caldirola-Kanai model; then expressions for mean values of some observables are evaluated. The present discussion is strongly based on the arguments and deductions made in Ref. \cite{iesus}.
\section{\textbf{Time evolution}}
The Hamiltonian operator
\begin{equation}
	H=\frac{P^{2}}{2m}\mathrm{e}^{-2\lambda t}+\frac{1}{2}\mathrm{e}^{2\lambda t}k_{0}X^{2},
\end{equation}
 with $[X,P]=i\hbar$, describes the quantum dynamics of the Caldirola-Kanai model. It can also be written as
\begin{equation}
H=\frac{P^{2}}{2m(t)}+\frac{1}{2}k(t)X^{2}\label{eq:HCKt}
\end{equation}
where
\begin{equation}
\begin{array}{lr}
m=m_{0}\mathrm{e}^{2\lambda t}& k=k_{0}\mathrm{e}^{2\lambda t}.
\end{array}
\end{equation}
The Hamiltonian does not commute with itself at different times, due to its dependence with time.  Consequently, it is not a simple task to solve Eq. \eqref{eq:schrun} directly. However, it has been shown \cite{iesus,cheng1988evolution} that for a Hamiltonian of the form
\begin{equation}
	H=\sum_{k} a_{k}(t)J_{k},\label{eq:HJ}
\end{equation}
where the operator $J_{k}$ form a closed Lie algebra such that $[J_{i},J_{k}]\propto J_l$, the unitary time-evolution operator $\mathcal{U}_{t}$ can be written as
\begin{equation}
	\mathcal{U}_{t}=\exp \left(\sum_{k} b_{k}(t)J_{k}\right)
\end{equation}
or
\begin{equation}
\mathcal{U}_{t}=\prod_{k}\exp \left( c_{k}(t)J_{k}\right).\label{eq:Ut}
\end{equation}
If the operators and associated functions assume 
\begin{equation}
\begin{array}{lll}
	J_{+}=\frac{1}{2\hbar}X^{2}, & 	J_{-}=\frac{1}{2\hbar}P^{2}, & J_{0}=\frac{i}{4\hbar}\left(PX+XP\right),\\
	a_{+}=\hbar k(t),& a_{-}=\frac{\hbar}{m(t)}, & a_{0}=0,
\end{array}
\end{equation}
then the Hamiltonian \eqref{eq:HCKt} can be written as \eqref{eq:HJ}. Since the operators $J_{+}$, $J_{-}$ and $J_{0}$ form a closed Lie algebra, then it is possible to obtain the operator $\mathcal{U}_{t}$ in the form \eqref{eq:Ut}, with
\begin{equation}
\begin{array}{lcr}
c_{+}(t)=m(t)\frac{\dot{u}(t)}{u(t)},& c_{-}(t)=-u(0)^{2}\int_{0}^{t}\frac{dt^{'}}{m(t^{'})u^{2}(t^{'})}\frac{\hbar}{m(t)},& c_{0}=-\ln\frac{u^{2}(t)}{u^{2}(0)},
\end{array}
\end{equation}
and boundary conditions $c_{+}(0)=c_{-}(0)=c_{0}(0)=0$ and
\begin{equation}
	\ddot{u}+\frac{\dot{m}(t)}{m(t)}\dot{u}+\frac{k(t)}{m(t)}u=0,\label{eq:uevolv}
\end{equation}
such that $u(0)= 1$ and $\dot{u}(0)=0$.
\par The time evolved wave function is then given by
\begin{equation}
	\Psi (x,t)=\bra{x}\exp\left[ic_{+}(t)J_{+}\right]\exp\left[c_{0}(t)J_{0}\right]\exp\left[ic_{-}(t)J_{-}\right]\ket{\psi_{0}},
\end{equation}
where a Gaussian state $\ket{\psi_{0}}$ is employed. Since $J_{+}=\frac{1}{2\hbar}X^{2}$, then
\begin{equation}
\Psi (x,t)=\exp\left[ic_{+}(t)\frac{x^{2}}{2\hbar}\right]\bra{x}\exp\left[c_{0}(t)J_{0}\right]\exp\left[ic_{-}(t)J_{-}\right]\ket{\psi_{0}}.
\end{equation}
It can be shown that
\begin{equation}
	\bra{x}\exp\left[\frac{ic_{0}(t)}{2\hbar}XP\right]\ket{\phi}=\exp\left[\frac{c_{0}(t)}{2}x\partial_{x}\right]\langle x |\phi\rangle,
\end{equation}
 for any state $\ket{\phi}$. Therefore, the wave function can be rewritten as
\begin{equation}
\Psi (x,t)=\exp\left[ic_{+}(t)\frac{x^{2}}{2\hbar}\right]\exp\left[\frac{c_{0}(t)}{4}\right]\exp\left[\frac{c_{0}(t)}{2}x\partial_{x}\right]\underbrace{\bra{x}\exp\left[ic_{-}(t)J_{-}\right]\ket{\psi_{0}}}_{\Phi(x,t)} \label{eq:psinterm}
\end{equation}
where
\begin{equation}
\Phi(x,t)=\int_{-\infty}^{\infty}\bra{x}\exp\left[\frac{ic_{-}(t)}{2\hbar}P^{2}\right]\ket{p}\langle p|\psi_{0}\rangle dp=\frac{1}{\sqrt{2\pi\hbar}} \int_{-\infty}^{\infty}\exp\left[\frac{ic_{-}(t)}{2\hbar}p^{2}+\frac{ipx}{\hbar}\right]\langle p|\psi_{0}\rangle dp\label{eq:auxPhi}
\end{equation}
Because the initial state has the Gaussian form
\begin{equation}
\Psi(x)=\langle q|\psi_{0} \rangle=\frac{\mathrm{e}^{-\frac{\left(x-x_{0}\right)^{2}}{4\Delta_{x,0}^{2}}}\mathrm{e}^{\frac{ip_{0}x}{\hbar}}}{\left(2\pi \Delta_{x,0}^{2}\right)^{\frac{1}{4}}}\label{eq:Gaussianapp}
\end{equation}
its form in the momentum basis reads
\begin{equation}
\langle p|\psi_{0} \rangle=\left(\frac{2\Delta_{x,0}^{2}}{\hbar^{2}\pi}\right)^{\frac{1}{4}}\exp\left[-\frac{\Delta_{q,0}^{2}\left(p-p_{0}\right)^{2}}{\hbar^{2}}-\frac{i\left(p-p_{0}\right)x_{0}}{\hbar}\right].
\end{equation}
The integral in Eq. \eqref{eq:auxPhi} is then performed, resulting in
\begin{equation}
	\Phi(x,t)=\left(\frac{8\Delta_{x,0}^{2}\pi^{-1}}{16\Delta_{x,0}^{4}+4\hbar^{2}c_{-}^{2}}\right)^{\frac{1}{4}}\exp\left[-\frac{\left(x-x_{0}+c_{-}p_{0}\right)^{2}}{4\Delta_{x,0}^{2}-2i\hbar c_{-}}+\frac{ic_{-}}{2\hbar}p_{0}^{2}+\frac{i}{\hbar}p_{0}\left(x-x_{0}\right)-i\theta\right],\label{eq:phinterm}
\end{equation}
where
\begin{equation}
\theta(t)=\frac{1}{2}\arctan\left[-\frac{\hbar c_{-}}{2\Delta_{x,0}^{2}}\right].
\end{equation}
The following relation was proved in \cite{iesus}:
\begin{equation}
	\exp\left[\frac{c_{0}(t)}{2}q\partial_{q}\right]\Phi(q,t)=\Phi\left(\mathrm{e}^{\frac{c_{0}(t)}{2}}q,t\right).\label{eq:auxc0}
\end{equation}
From Eqs. \eqref{eq:auxc0}, \eqref{eq:psinterm} and \eqref{eq:phinterm}, it follows that
\begin{equation}
\Psi(x,t)= A(t)\exp\left[-\frac{\left(x-\mathrm{e}^{-\frac{c_{0}}{2}}\left(x_{0}-c_{-}p_{0}\right)\right)^{2}}{\mathrm{e}^{-c_{0}}\left(4\Delta_{x,0}^{2}-2i\hbar c_{-}\right)}+i\frac{c_{-}}{2\hbar}p_{0}^{2}+i\frac{e^{\frac{c_{0}}{2}}}{\hbar}p_{0}x-i\theta +i\frac{c_{+}}{2\hbar}x^{2}\right],\label{eq:funonda2}
\end{equation}
where
\begin{align}
&A(t)=\left(\frac{1}{\pi 2\Delta_{x,t}^{2}}\right)^\frac{1}{4},\\
&\Delta_{x,t}^{2}=\mathrm{e}^{-c_{0}}\Delta_{x,0}^{2}\left(1+\frac{\hbar^{2}c_{-}^{2}}{4\Delta_{x,0}^{4}}\right),\\
&\theta(t)=\frac{1}{2}\arctan\left[-\frac{\hbar c_{-}}{2\Delta_{x,0}^{2}}\right].
\end{align}
\subsection{General case}
The use of the above results demands the determination of the coefficients $c_{+}, c_{-}$ and $c_{0}$. Then, inserting the expressions 
\begin{equation}
\begin{array}{lr}
m=m_{0}\mathrm{e}^{2\lambda t}& k=k_{0}\mathrm{e}^{2\lambda t}
\end{array}
\end{equation}
in Eq. \eqref{eq:uevolv}, results 
\begin{equation}
	\ddot{u}+2\lambda\dot{u}+\frac{k_{0}}{m_{0}}u=\ddot{u}+2\lambda\dot{u}+\omega^{2} u.
\end{equation}
The solution of this differential equation is 
\begin{equation}
	u(t)=u_{0}\mathrm{e}^{-\tau}\left[\cosh\left(\zeta\tau \right)+\left(1+\frac{u_{0}^{'}}{\lambda u_{0}}\right)\frac{\sinh\left(\zeta \tau\right)}{\zeta}\right]. \label{eq:ut}
\end{equation}
where $\tau=\lambda t$ and $\zeta=\sqrt{1-\frac{\omega^{2}}{\lambda^{2}}}$. Considering the boundary conditions $u(0)= 1$ and $\dot{u}(0)=0$, Eq. \eqref{eq:ut} is rewritten as
\begin{equation}
u(\tau)=\mathrm{e}^{-\tau}\left[\cosh\left(\zeta\tau \right)+\frac{\sinh\left(\zeta \tau\right)}{\zeta}\right].
\end{equation}
Hence, the coefficients are given by
\begin{align}
	& \displaystyle c_{+}(t)=m(t)\frac{\dot{u}(t)}{u(t)}=-\frac{k_{0}}{\lambda}\frac{\mathrm{e}^{2\tau}}{\zeta \coth (\zeta \tau)+1},\label{eq:cmais}\\
	& \displaystyle c_{-}(t)=-u(0)^{2}\int_{0}^{t}\frac{dt^{'}}{m(t^{'})u^{2}(t^{'})}\frac{\hbar}{m(t)} =-\frac{1}{\lambda m_{0}}\frac{1}{\zeta \coth (\zeta \tau)+1},\\
	& \displaystyle c_{0}=\ln\frac{u^{2}(t)}{u^{2}(0)}=2\ln\left[\frac{\mathrm{e}^{\tau}}{\cosh(\zeta\tau)+\frac{\sinh(\zeta \tau)}{\zeta}}\right].\label{eq:c0}
\end{align}
With the wave function determined and the time dependent coefficients established, the expectation values can be computed.
 \section{\textbf{Mean values}}
 The probability distribution can then be written as
\begin{equation}
\left|\Psi(x,t)\right|^{2}= A(t)^{2}\exp\left[-\frac{1}{2\Delta_{x,t}^{2}}\left(x-\mathrm{e}^{-\frac{c_{0}}{2}}\left(x_{0}-c_{-}p_{0}\right)\right)^{2}\right].
\end{equation}
For the mean position one has
\begin{equation}
\langle X\rangle=\int_{-\infty}^{\infty} \Psi^{*}(x,t)x\Psi(x,t)dx=\int_{-\infty}^{\infty}x \left|\Psi(x,t)\right|^{2}dx\label{eq:position1}
\end{equation}
which yields
\begin{equation}
\langle X\rangle=\mathrm{e}^{-\frac{c_{0}}{2}}\left(x_{0}-c_{-}p_{0}\right).
\end{equation}
Similarly, 
\begin{equation}
\langle X^{2}\rangle=\int_{-\infty}^{\infty}x^{2} \left|\Psi(x,t)\right|^{2}dx=\Delta_{x,t}^{2}+\mathrm{e}^{-c_{0}}\left(x_{0}-c_{-}p_{0}\right)^{2}.
\end{equation}
Therefore, the position variance is given by
\begin{equation}
\langle \left(\Delta X\right)^{2}\rangle=\langle X^{2}\rangle-\langle X\rangle^{2}=\Delta_{x,t}^{2}.
\end{equation}
The mean value of momentum is calculated from  
\begin{equation}
\langle P\rangle=-i\hbar \int_{-\infty}^{\infty}\Psi^{*}(x,t)\frac{\partial}{\partial x}\Psi(x,t)dx,
\end{equation}
resulting in
\begin{equation}
\langle P\rangle=e^{-\frac{c_{0}}{2}}\left[\left(e^{c_{0}}-c_{+}c_{-}\right)p_{0}+c_{+}x_{0}\right].
\end{equation}
Proceeding in a similar way, one finds
\begin{equation}
\begin{array}{rl}
 	\langle P^{2}\rangle&=-\hbar^{2} \int_{-\infty}^{\infty}\Psi^{*}(x,t)\frac{\partial^{2}}{\partial x^{2}}\Psi(x,t)dx\\
 	&=\frac{\mathrm{e}^{-c_{0}}\hbar^{2}}{4\Delta_{x,0}^{2}}\left(\mathrm{e}^{c_{0}}-c_{+}c_{-}\right)^{2}+\mathrm{e}^{-c_{0}}\Delta_{x,0}^{2}c_{+}^{2} +\\
 	&+e^{-c_{0}}\left(\left(e^{c_{0}}-c_{+}c_{-}\right)p_{0}+c_{+}x_{0}\right)^{2}.\\\label{eq:MMM}
\end{array} 
\end{equation}
Finally, one obtains
\begin{equation}
\langle\left(\Delta P\right)^{2} \rangle=\langle P^{2}\rangle-\langle P\rangle^{2}=\frac{\mathrm{e}^{-c_{0}}\hbar^{2}}{4\Delta_{x,0}^{2}}\left(\mathrm{e}^{c_{0}}-c_{+}c_{-}\right)^{2}+\mathrm{e}^{-c_{0}}\Delta_{x,0}^{2}c_{+}^{2}.\label{eq:deltap2}
\end{equation}
Next, the associated functions $c_{+}$, $c_{-}$, and $c_{0}$ are inserted in order to furnish the explicit time-dependence for the above results. 
\subsection{Explicit expressions}
\par After some (lengthy) algebraic manipulations of Eqs. \eqref{eq:position1}-\eqref{eq:deltap2} using the time dependent functions described at Eqs. \eqref{eq:cmais}-\eqref{eq:c0}, the following expressions are obtained: 
\begin{itemize}
	\item Mean position 
	\begin{equation}
	\langle X\rangle=\frac{\mathrm{e}^{-\tau}}{\zeta\lambda m_{0}}\left[\left(\zeta \cosh (\zeta \tau)+\sinh(\zeta \tau)\right)\lambda m_{0}x_{0}+p_{0}\sinh(\zeta \tau)\right];
	\end{equation}
	\item Mean square position
	\begin{equation}
	\langle X^{2}\rangle=\mathrm{e}^{-2\tau}\left[k_{4} \cosh^{2} (\zeta \tau)+k_{5}\sinh(\zeta \tau)\cosh (\zeta \tau)+k_{6}\sinh^{2}(\zeta \tau)\right],\label{eq:melhorQ2}
	\end{equation}
	where
	\begin{align}
	&\displaystyle k_{4}=x_{0}^{2}+\Delta_{x,0}^{2},\label{eq:k4}\\
	&\displaystyle k_{5}=\frac{2\zeta\lambda^{2} m_{0}^{2}x_{0}^{2}+2\zeta\lambda m_{0}x_{0}p_{0}}{\zeta^{2}\lambda^{2} m_{0}^{2}}+\frac{2\Delta_{x,0}^{2}}{\zeta},\label{eq:k5}\\
	&\displaystyle k_{6}=\frac{p_{0}^{2} +2\lambda m_{0}x_{0}p_{0}+\lambda^{2} m_{0}^{2}x_{0}^{2}}{\zeta^{2}\lambda^{2} m_{0}^{2}}+\frac{\Delta_{x,0}^{2}}{\zeta^{2}}+\frac{\hbar^{2}}{4m_{0}^{2}\zeta^{2}\lambda^{2}\Delta_{x,0}^{2}};\label{eq:k6}
	\end{align}
	\item Position variance:
	\begin{equation}
	\begin{array}{rl}
	\langle X^{2}\rangle-\langle X\rangle^{2}&=\mathrm{e}^{-2\tau}\left(\Delta_{x,0}^{2} \cosh^{2} (\zeta \tau)+\frac{2\Delta_{x,0}^{2}}{\zeta} \sinh(\zeta \tau)\cosh (\zeta \tau)\right.+\\
	&+\left.\left(\frac{\Delta_{x,0}^{2}}{\zeta^{2}}+\frac{\Delta_{p,0}^{2}}{m_{0}^{2}\zeta^{2}\lambda^{2}}\right)\sinh^{2}(\zeta \tau)\right);
	\end{array}
	\end{equation}
	\item Mean momentum:
	\begin{equation}
	\begin{array}{rl}
	\langle	 P\rangle &=\mathrm{e}^{\tau}\left(p_{0}^{2}\cosh^{2}(\zeta\tau)-\left(2\frac{k_{0}x_{0}p_{0}}{\lambda\zeta}+2\frac{p_{0}^{2}}{\zeta}\right)\sinh(\zeta \tau)\cosh(\zeta\tau)\right. +\\
	&+\left.\left(\frac{k_{0}^{2}x_{0}^{2}}{\lambda^{2}\zeta^{2}}+2\frac{p_{0}}{\zeta}\frac{k_{0}x_{0}}{\lambda\zeta}+\frac{p_{0}^{2}}{\zeta^{2}}\right)\sinh^{2}(\zeta \tau)\right)^{\frac{1}{2}};
	\end{array}
	\end{equation}
	\item  Mean square momentum:
	\begin{equation}
	\langle P^{2}\rangle=\mathrm{e}^{2\tau}\left[k_{1}\cosh^{2}(\zeta\tau)-k_{2}\sinh(\zeta \tau)\cosh(\zeta\tau)+k_{3}\sinh^{2}(\zeta \tau)\right],\label{eq:melhorP2}
	\end{equation}
	where
	\begin{align}
	&\displaystyle k_{1}=p_{0}^{2}+\frac{\hbar^{2}}{4\Delta_{x,0}^{2}},\label{eq:k1}\\
	&\displaystyle k_{2}=2\frac{k_{0}x_{0}p_{0}}{\lambda\zeta}+2\frac{p_{0}^{2}}{\zeta}+2\frac{\hbar^{2}}{\Delta_{q,0}^{2}\zeta},\label{eq:k2}\\ &\displaystyle k_{3}=\frac{k_{0}^{2}x_{0}^{2}}{\lambda^{2}\zeta^{2}}+2\frac{p_{0}}{\zeta}\frac{k_{0}x_{0}}{\lambda\zeta}+\frac{p_{0}^{2}}{\zeta^{2}}+\frac{\Delta_{x,0}^{2}k_{0}^{2}}{\lambda^{2}\zeta^{2}}+\frac{\hbar^{2}}{4\Delta_{x,0}^{2}\zeta^{2}};\label{eq:k3}
	\end{align}
	\item Momentum variance:
	\begin{equation}
	\begin{array}{rl}
	\langle\left(\Delta P\right)^{2} \rangle&=\mathrm{e}^{2\tau}\Delta_{x,0}^{2}\frac{k_{0}^{2}}{\lambda^{2}}\frac{\sinh^{2}(\zeta \tau)}{\zeta^{2}}+\\
	&+\mathrm{e}^{2\tau}\left(\frac{\hbar^{2}}{4\Delta_{x,0}^{2}}\cosh^{2}(\zeta\tau)-2\frac{\hbar^{2}}{4\Delta_{x,0}^{2}\zeta}\sinh(\zeta \tau)\cosh(\zeta\tau)+\frac{\hbar^{2}}{4\Delta_{x,0}^{2}\zeta^{2}}\sinh^{2}(\zeta \tau)\right).
	\end{array}
	\end{equation}
\end{itemize}

%% file: apendice_workclassical.tex
\chapter{CLASSICAL STATISTICAL ANALYSIS}
\label{cap:Work_classical}

The velocity field associated with the classical Caldirola-Kanai model has the form
\begin{equation}
v(t,p_{0},x_{0})=\mathrm{e}^{-\tau}\left[\frac{p_{0}}{m_{0}}\cosh\left(\zeta\tau\right)-\left(\frac{p_{0}}{m_{0}}+\frac{k_{0}x_{0}}{m_{0}\lambda}\right)\frac{\sinh\left(\zeta\tau\right)}{\zeta}\right],
\end{equation}
as stated in section \ref{sec:purecl}. From this expression, it follows that
\begin{equation}
\begin{array}{rl}
v^{2}(t,p_{0},x_{0})&\displaystyle=\frac{1}{m_{0}^{2}}\mathrm{e}^{-2\tau}\left[p_{0}^{2}\cosh^{2}\left(\zeta\tau\right)-2\left(\frac{p_{0}^{2}}{\zeta}+\frac{k_{0}x_{0}p_{0}}{\zeta\lambda}\right)\sinh\left(\zeta\tau\right)\cosh\left(\zeta\tau\right)+\right.\\
&\displaystyle\left.+\left(\frac{p_{0}^{2}}{\zeta^{2}}+2p_{0}\frac{k_{0}x_{0}}{\lambda\zeta^{2}}+\frac{k_{0}^{2}x_{0}^{2}}{\lambda^{2}\zeta^{2}}\right)\sinh^{2}\left(\zeta\tau\right)\right],
\end{array}
\label{eq:energiacineticaclassica2}
\end{equation}
where $x_{0}$ e $p_{0}$ are the initial conditions. Regarding the normalized Gaussian classical distribution, 
\begin{equation}
\rho_{gcl}(t,x_{0},p_{0})=\frac{1}{2\pi \sigma_{p,0}\sigma_{x,0}}\mathrm{e}^{-\frac{1}{2}\left(\left(\frac{x_{0}-x_{0}^{*}}{\sigma_{x,0}}\right)^{2}+\left(\frac{p_{0}-p_{0}^{*}}{\sigma_{p,0}}\right)^{2}\right)}\equiv\mathcal{G}\left[ x_{0}^{*},p_{0}^{*},\sigma_{x,0},\sigma_{p,0}\right],
\end{equation}
the mean values related with the initial conditions are (see subsection \ref{subsec:subjective})
\begin{align}
&\displaystyle\langle p_{0}\rangle_{gcl}=\int_{-\infty}^{\infty}\int_{-\infty}^{\infty}\rho_{gcl}(t,x_{0},p_{0})p_{0}dx_{0}dp_{0}=p_{0}^{*},\\
&\displaystyle\langle p_{0}^{2} \rangle_{gcl}=\int_{-\infty}^{\infty}\int_{-\infty}^{\infty}\rho_{gcl}(t,x_{0},p_{0})p_{0}^{2}dx_{0}dp_{0}=p_{0}^{*2}+\sigma_{p,0}^{2},\\
&\displaystyle\langle x_{0} \rangle_{gcl}=\int_{-\infty}^{\infty}\int_{-\infty}^{\infty}\rho_{gcl}(t,x_{0},p_{0})x_{0}dx_{0}dp_{0}=x_{0}^{*},\\
&\displaystyle\langle x_{0}^{2} \rangle_{gcl}=\int_{-\infty}^{\infty}\int_{-\infty}^{\infty}\rho_{gcl}(t,x_{0},p_{0})x_{0}^{2}dx_{0}dp_{0}=x_{0}^{*2}+\sigma_{x,0}^{2},\\
&\displaystyle\langle x_{0}p_{0}\rangle_{gcl}= \langle p_{0} x_{0}\rangle_{gcl}= \int_{-\infty}^{\infty}\int_{-\infty}^{\infty}\rho_{gcl}(t,x_{0},p_{0})x_{0}p_{0}dx_{0}dp_{0}=x_{0}^{*}p_{0}^{*}.
\end{align}
For a generic function $f(x_{0},p_{0},t)$ one has 
\begin{equation}
\langle f(t)\rangle_{gcl}=\int_{-\infty}^{\infty}\int_{-\infty}^{\infty}\rho_{gcl}(x_{0},p_{0},0)f(x_{0},p_{0},t)dx_{0}dp_{0}.
\end{equation}
Therefore,
\begin{equation}
\begin{array}{rl}
\langle v(t) \rangle_{gcl}&=\int_{-\infty}^{\infty}\int_{-\infty}^{\infty}\rho_{gcl}(x_{0},p_{0},0)v(x_{0},p_{0},t)dx_{0}dp_{0}\\
&=\mathrm{e}^{-\tau}\left[\displaystyle\frac{\langle p_{0}\rangle_{gcl}}{m_{0}}\cosh\left(\zeta\tau\right)-\left(\frac{\langle p_{0}\rangle_{gcl}}{m_{0}}+\frac{k_{0}\langle x_{0}\rangle_{gcl}}{m_{0}\lambda}\right)\frac{\sinh\left(\zeta\tau\right)}{\zeta}\right]\\
&\displaystyle=\mathrm{e}^{-\tau}\left[\frac{p_{0}^{*}}{m_{0}}\cosh\left(\zeta\tau\right)-\left(\frac{p_{0}^{*}}{m_{0}}+\frac{k_{0}x_{0}^{*}}{m_{0}\lambda}\right)\frac{\sinh\left(\zeta\tau\right)}{\zeta}\right]
\end{array}
\end{equation}
and
\begin{equation}
\begin{array}{rl}
\langle v(t) \rangle_{gcl}^{2}&\displaystyle=\frac{1}{m_{0}^{2}}\mathrm{e}^{-2\tau}\left[p_{0}^{*2}\cosh^{2}\left(\zeta\tau\right)-2\left(\frac{p_{0}^{*2}}{\zeta}+\frac{k_{0}p_{0}^{*}x_{0}^{*}}{\zeta\lambda}\right)\sinh\left(\zeta\tau\right)\cosh\left(\zeta\tau\right)\right.+\\
&\displaystyle\left.+\left(\frac{p_{0}^{*2}}{\zeta^{2}}+2p_{0}^{*}\frac{k_{0}x_{0}^{*}}{\lambda\zeta^{2}}+\frac{k_{0}^{2}x_{0}^{*2}}{\lambda^{2}\zeta^{2}}\right)\sinh^{2}\left(\zeta\tau\right)\right].
\end{array}
\end{equation}
Similarly, one shows that
\begin{equation}
\begin{array}{rl}
\langle v(t)^{2} \rangle_{gcl} &\displaystyle=\frac{1}{m_{0}^{2}}\mathrm{e}^{-2\tau}\left[p_{0}^{*2}\cosh^{2}\left(\zeta\tau\right)-2\left(\frac{p_{0}^{*2}}{\zeta}+\frac{k_{0}p_{0}^{*}x_{0}^{*}}{\zeta\lambda}\right)\sinh\left(\zeta\tau\right)\cosh\left(\zeta\tau\right)+\right.\\
&\displaystyle\left.+\left(\frac{p_{0}^{*2}}{\zeta^{2}}+2p_{0}^{*}\frac{k_{0}x_{0}^{*}}{\lambda\zeta^{2}}+\frac{k_{0}^{2}x_{0}^{*2}}{\lambda^{2}\zeta^{2}}\right)\sinh^{2}\left(\zeta\tau\right)\right]+\frac{1}{m_{0}^{2}}\mathrm{e}^{-2\tau}\left[\sigma_{p,0}^{2}\cosh^{2}\left(\zeta\tau\right)\right.+\\
&\displaystyle\left.-2\left(\frac{\sigma_{p,0}^{2}}{\zeta}\right)\sinh\left(\zeta\tau\right)\cosh\left(\zeta\tau\right)+\left(\frac{\sigma_{p,0}^{2}}{\zeta^{2}}+\frac{k_{0}^{2}\sigma_{x,0}^{2}}{\lambda^{2}\zeta^{2}}\right)\sinh^{2}\left(\zeta\tau\right)\right].
\end{array}
\end{equation}
\section{\textbf{Mixed state}}
The classical-statistical counterpart of a quantum mixed state reads
\begin{equation}
\rho_{mgcl}(t,x_{0},p_{0})=\frac{1}{2}\left(\mathcal{G}\left[ x_{0}^{*},p_{0}^{*},\sigma_{x,0},\sigma_{p,0}\right]+\mathcal{G}\left[- x_{0}^{*},-p_{0}^{*},\sigma_{x,0},\sigma_{p,0}\right]\right)
\end{equation}
such that
\begin{align}
&\int_{-\infty}^{\infty}\int_{-\infty}^{\infty}\rho_{mgcl}(t,x_{0},p_{0})dx_{0}dp_{0}=1,\\
&\langle p_{0}\rangle_{mgcl}=\int_{-\infty}^{\infty}\int_{-\infty}^{\infty}\rho_{mgcl}(t,x_{0},p_{0})p_{0}dx_{0}dp_{0}=0,\\
&\langle p_{0}^{2}\rangle_{mgcl}=\int_{-\infty}^{\infty}\int_{-\infty}^{\infty}\rho_{mgcl}(t,x_{0},p_{0})p_{0}^{2}dx_{0}dp_{0}=p_{0}^{*2}+\sigma_{p,0}^{2},\\
&\langle x_{0}\rangle_{mgcl}=\int_{-\infty}^{\infty}\int_{-\infty}^{\infty}\rho_{mgcl}(t,x_{0},p_{0})x_{0}dx_{0}dp_{0}=0,\\
&\langle x_{0}^{2}\rangle_{mgcl}=\int_{-\infty}^{\infty}\int_{-\infty}^{\infty}\rho_{mgcl}(t,x_{0},p_{0})x_{0}^{2}dx_{0}dp_{0}=x_{0}^{*2}+\sigma_{x,0}^{2},\\
&\langle x_{0}p_{0}\rangle_{mgcl}=\langle p_{0}x_{0}\rangle_{mgcl}=\int_{-\infty}^{\infty}\int_{-\infty}^{\infty}\rho_{mgcl}(t,x_{0},p_{0})x_{0}p_{0}dx_{0}dp_{0}=x_{0}^{*}p_{0}^{*}.
\end{align}
It follows that 
\begin{equation}
\begin{array}{rl}
\langle v(t) \rangle_{mgcl}&\displaystyle=\int_{-\infty}^{\infty}\int_{-\infty}^{\infty}\rho_{mgcl}(x_{0},p_{0},0)v(x_{0},p_{0},t)dx_{0}dp_{0}\\
&\displaystyle=\mathrm{e}^{-\tau}\left[\displaystyle\frac{\langle p_{0}\rangle_{mgcl}}{m_{0}}\cosh\left(\zeta\tau\right)-\left(\frac{\langle p_{0}\rangle_{mgcl}}{m_{0}}+\frac{k_{0}\langle x_{0}\rangle_{mgcl}}{m_{0}\lambda}\right)\frac{\sinh\left(\zeta\tau\right)}{\zeta}\right]=0\\
\end{array}
\end{equation}
and 
\begin{equation}
\langle v(t) \rangle_{mgcl}^{2}=0.\label{eq:vmcl2}
\end{equation}
Since $v^{2}(t)$ depends only on terms like $x_{0}^{2}$, $p_{0}^{2}$, and $x_{0}p_{0}$, and the mean values of these terms are the same as the ones given by the pure Gaussian case above, then 
\begin{equation}
\begin{array}{rl}
\langle v(t)^{2} \rangle_{mgcl}=\langle v(t)^{2} \rangle_{gcl} &\displaystyle=\frac{1}{m_{0}^{2}}\mathrm{e}^{-2\tau}\left[p_{0}^{*2}\cosh^{2}\left(\zeta\tau\right)\right.+\\
&\displaystyle\left.-2\left(\frac{p_{0}^{*2}}{\zeta}+\frac{k_{0}p_{0}^{*}x_{0}^{*}}{\zeta\lambda}\right)\sinh\left(\zeta\tau\right)\cosh\left(\zeta\tau\right)+\right.\\
&\displaystyle\left.+\left(\frac{p_{0}^{*2}}{\zeta^{2}}+2p_{0}^{*}\frac{k_{0}x_{0}^{*}}{\lambda\zeta^{2}}+\frac{k_{0}^{2}x_{0}^{*2}}{\lambda^{2}\zeta^{2}}\right)\sinh^{2}\left(\zeta\tau\right)\right]+\\
&\displaystyle+\frac{1}{m_{0}^{2}}\mathrm{e}^{-2\tau}\left[\sigma_{p,0}^{2}\cosh^{2}\left(\zeta\tau\right)\right.+\\
&\displaystyle\left.-2\left(\frac{\sigma_{p,0}^{2}}{\zeta}\right)\sinh\left(\zeta\tau\right)\cosh\left(\zeta\tau\right)+\left(\frac{\sigma_{p,0}^{2}}{\zeta^{2}}+\frac{k_{0}^{2}\sigma_{x,0}^{2}}{\lambda^{2}\zeta^{2}}\right)\sinh^{2}\left(\zeta\tau\right)\right].
\end{array}
\label{eq:v2mcl}
\end{equation}

%% file: apendice_workquantum.tex
\chapter{QUANTUM MIXED STATES ANALYSIS}
\label{cap:Work_quantum}
The quantum state considered in subsection \ref{subsec:quanticadist} is described by the density operator 
\begin{equation}
\rho^{\mu}(0)=\frac{\ket{\psi_{0}}\bra{\psi_{0}}+\ket{\psi_{0}^{-}}\bra{\psi_{0}^{-}}+\mathrm{e}^{-\mu}\left(\ket{\psi_{0}^{-}}\bra{\psi_{0}}+\ket{\psi_{0}}\bra{\psi_{0}^{-}}\right)}{N_{\mu}},\label{eq:statot}
\end{equation}
where
\begin{equation}
\langle x |\psi_{0}\rangle\equiv\Psi(x,0)=\left(\frac{1}{2\pi\Delta_{x,0}^{2}}\right)^{\frac{1}{4}}\exp \left[\frac{-\left(x-x_{0}\right)^{2}}{4\Delta_{x,0}^{2}}+i\frac{p_{0}x}{\hbar}\right],\label{eq:estnor}
\end{equation}
\begin{equation}
\langle x |\psi_{0}^{-}\rangle\equiv\Psi^{-}(x,0)=\left(\frac{1}{2\pi\Delta_{x,0}^{2}}\right)^{\frac{1}{4}}\exp \left[\frac{-\left(x+x_{0}\right)^{2}}{4\Delta_{x,0}^{2}}-i\frac{p_{0}x}{\hbar}\right],\label{eq:estrev}
\end{equation}
and $N_{\mu}$ is the normalization factor, given by 
\begin{equation}
N_{\mu}=\mathrm{Tr}\left[\ket{\psi_{0}}\bra{\psi_{0}}+\ket{\psi_{0}^{-}}\bra{\psi_{0}^{-}}+\mathrm{e}^{-\mu}\left(\ket{\psi_{0}^{-}}\bra{\psi_{0}}+\ket{\psi_{0}}\bra{\psi_{0}^{-}}\right)\right].
\end{equation}
The evolved Gaussian state, evaluated in Appendix \ref{cap:kanaicaldirolasol}, is given by 
\begin{equation}
\langle x |\psi(t)\rangle= A\exp\left[-\frac{\left(x-\mathrm{e}^{-\frac{c_{0}}{2}}\left(x_{0}-c_{-}p_{0}\right)\right)^{2}}{\mathrm{e}^{-c_{0}}\left(4\Delta_{x,0}^{2}-2i\hbar c_{-}\right)}+i\frac{c_{-}}{2\hbar}p_{0}^{2}+i\frac{e^{\frac{c_{0}}{2}}}{\hbar}p_{0}x-i\theta +i\frac{c_{+}}{2\hbar}x^{2}\right],\label{eq:estevol}
\end{equation}
where
\begin{align}
&A\equiv A(t)=\left(\frac{1}{2\pi \Delta_{x,t}^{2}}\right)^\frac{1}{4},\\
&\Delta_{x,t}^{2}=\mathrm{e}^{-c_{0}}\Delta_{x,0}^{2}\left(1+\frac{\hbar^{2}c_{-}^{2}}{4\Delta_{x,0}^{4}}\right),\\
&\theta(t)=\frac{1}{2}\arctan\left[-\frac{\hbar c_{-}}{2\Delta_{x,0}^{2}}\right].
\end{align}
The expression \eqref{eq:estrev} reduces to \eqref{eq:estnor}, via the replacements $x_{0}\rightarrow -x_{0}$ $p_{0}\rightarrow -p_{0}$. Furthermore, the derivation of the solution \eqref{eq:estevol} does not make any restriction to the signs of  $x_{0}$ and $p_{0}$.  Therefore, the evolution of state $\langle x|\psi_{0}^{-}\rangle$ can be obtained from Equation \eqref{eq:estevol} via $x_{0}\rightarrow -x_{0}$ $p_{0}\rightarrow -p_{0}$, which yields
\begin{equation}
\langle x |\psi^{-}(t)\rangle= A\exp\left[-\frac{\left(x+\mathrm{e}^{-\frac{c_{0}}{2}}\left(x_{0}-c_{-}p_{0}\right)\right)^{2}}{\mathrm{e}^{-c_{0}}\left(4\Delta_{x,0}^{2}-2i\hbar c_{-}\right)}+i\frac{c_{-}}{2\hbar}p_{0}^{2}-i\frac{e^{\frac{c_{0}}{2}}}{\hbar}p_{0}x-i\theta +i\frac{c_{+}}{2\hbar}x^{2}\right].
\end{equation}
Therefore, the state \eqref{eq:statot} evolves as
\begin{equation}
\begin{array}{rl}
\rho^{\mu}(t)&=\mathcal{U}_{t}\frac{\ket{\psi_{0}}\bra{\psi_{0}}+\ket{\psi_{0}^{-}}\bra{\psi_{0}^{-}}+\mathrm{e}^{-\mu}\left(\ket{\psi_{0}^{-}}\bra{\psi_{0}}+\ket{\psi_{0}}\bra{\psi_{0}^{-}}\right)}{N_{\mu}}\mathcal{U}_{t}^{\dagger}\\
&=\frac{\ket{\psi(t)}\bra{\psi(t)}+\ket{\psi^{-}(t)}\bra{\psi^{-}(t)}+\mathrm{e}^{-\mu}\left(\ket{\psi^{-}(t)}\bra{\psi(t)}+\ket{\psi(t)}\bra{\psi^{-}(t)}\right)}{N_{\mu}}.
\end{array}
\end{equation}
It is important to remark that $N_{\mu}$ does not evolve in time. 
\par The objective here is to compute the expectation values of the operators $X$, $P$, $X^{2}$, and $P^{2}$ for any time $t$. To this end, some general properties are firstly considered. Given two state vectors $\ket{\alpha}$ and $\ket{\beta}$, and an observable $A$, it follows that
\begin{equation}
Tr\left[\ket{\alpha}\bra{\beta}A\right]=\int_{-\infty}^{\infty}\langle x|\alpha\rangle \bra{\beta}A\ket{x} dx=\bra{\beta}A\left(\int_{-\infty}^{\infty}\ket{x}\bra{x} dx\right)\ket{\alpha}=\bra{\beta}A\ket{\alpha}.\label{eq:traco}
\end{equation}
Since $A$ is hermitian, then
\begin{equation}
\bra{\beta}A\ket{\alpha}=\bra{\alpha}A^{\dagger}\ket{\beta}^{*}=\bra{\alpha}A\ket{\beta}^{*}.\label{eq:complexconj}
\end{equation} 
For an observable acting on the state space of $\ket{\psi(t)}$ and $\ket{\psi^{-}(t)}$, it follows that,
\begin{equation}
\begin{array}{ll}
\mathrm{Tr}\left(\rho^{\mu}(t)A\right)&\displaystyle=\frac{1}{N_{\mu}}\mathrm{Tr}\left[\ket{\psi(t)}\bra{\psi(t)}+\ket{\psi^{-}(t)}\bra{\psi^{-}(t)}+\right.\\
&\left.+\mathrm{e}^{-\mu}\left(\ket{\psi^{-}(t)}\bra{\psi(t)}+\ket{\psi(t)}\bra{\psi^{-}(t)}\right)A\right]\\
&\displaystyle=\frac{1}{N_{\mu}}\left[\bra{\psi(t)}A\ket{\psi(t)}+\bra{\psi^{-}(t)}A\ket{\psi^{-}(t)}+\right.\\
&\displaystyle\left.+\mathrm{e}^{-\mu}\left(\bra{\psi(t)}A\ket{\psi^{-}(t)}+\bra{\psi^{-}(t)}A\ket{\psi(t)}\right)\right]\\
&\displaystyle=\frac{1}{N_{\mu}}\left[\bra{\psi(t)}A\ket{\psi(t)}+\bra{\psi^{-}(t)}A\ket{\psi^{-}(t)}+\right.\\
&\displaystyle\left.+\mathrm{e}^{-\mu}\left(\bra{\psi^{-}(t)}A\ket{\psi(t)}^{*}+\bra{\psi^{-}(t)}A\ket{\psi(t)}\right)\right].\\
\end{array}
\end{equation} 
Therefore, from the expressions of $\bra{\psi^{-}(t)}A\ket{\psi(t)}$, $\bra{\psi(t)}A\ket{\psi(t)}$, and $\bra{\psi^{-}(t)}A\ket{\psi^{-}(t)}$, one finds
\begin{equation}
	\langle A \rangle_{q}^{\mu}\equiv\mathrm{Tr}\left(\rho^{\mu}(t)A\right).
\end{equation}
In order to obtain the expressions for $\langle X \rangle_{q}^{\mu}$, $\langle X^{2} \rangle_{q}^{\mu}$, $\langle P \rangle_{q}^{\mu}$ and $\langle P^{2} \rangle_{q}^{\mu}$, some remarks are in order: 
\begin{itemize}
	\item The derivation of $\bra{\psi(t)}A\ket{\psi(t)}$, with $A=X,X^{2},P,P^{2}$, is summarized in Appendix \ref{cap:kanaicaldirolasol};
	\item The expressions for the symmetrically reversed Gaussian $\bra{\psi^{-}(t)}A\ket{\psi^{-}(t)}$, with $A=X,X^{2},P,P^{2}$, can be obtained considering the same as for $\ket{\psi(t)}$, however, taking $x_{0}\rightarrow -x_{0}$ $p_{0}\rightarrow -p_{0}$;
	\item The terms $\bra{\psi^{-}(t)}A\ket{\psi(t)}$ and $\bra{\psi(t)}A\ket{\psi^{-}(t)}$, with $A=X,X^{2},P,P^{2}$, can be obtained considering Equation \eqref{eq:complexconj} and following the relations
	\begin{equation}
	\bra{\psi^{-}(t)}X\ket{\psi(t)}=\int_{-\infty}^{\infty}\langle x|\psi(t)\rangle\langle \psi^{-}(t)|x\rangle xdx,
	\end{equation}
	\begin{equation}
	\bra{\psi^{-}(t)}X^{2}\ket{\psi(t)}=\int_{-\infty}^{\infty}\langle x|\psi(t)\rangle\langle \psi^{-}(t)|x\rangle x^{2}dx,
	\end{equation}
	\begin{equation}
	\begin{array}{rl}
	\bra{\psi^{-}(t)}P\ket{\psi(t)}&\displaystyle=\mathrm{Tr}\left(P\ket{\psi(t)}\bra{\psi^{-}(t)}\right)=\int_{-\infty}^{\infty}\bra{q} P\ket{\psi(t)}\langle \psi^{-}(t)|x\rangle dx\\
	&\displaystyle=-i\hbar\int_{-\infty}^{\infty}\left(\frac{\partial}{\partial x}\langle x|\psi(t)\rangle\right)\langle \psi^{-}(t)|x\rangle dx,
	\end{array}
	\end{equation}
	and 
	\begin{equation}
	\bra{\psi^{-}(t)}P^{2}\ket{\psi(t)}=\left(-i\hbar\right)^{2}\int_{-\infty}^{\infty}\left(\frac{\partial^{2}}{\partial x^{2}}\langle x|\psi(t)\rangle\right)\langle \psi^{-}(t)|x\rangle dx.
	\end{equation}
\end{itemize}
The resulting expressions will be summarized next.
\section{\textbf{Normalization factor}}
Since
\begin{equation}
\langle \psi^{-}(t)|\psi^{-}(t)\rangle=\langle \psi(t)|\psi(t)\rangle=1
\end{equation}
and
\begin{equation}
\langle \psi(t)|\psi^{-}(t)\rangle=\langle \psi^{-}(t)|\psi(t)\rangle=\int_{-\infty}^{\infty}\Psi^{-}(x,t)\Psi(x,t)dx=\exp -\left[\frac{x_{0}^{2}}{2\Delta_{x,0}^{2}}+2\frac{p_{0}^{2}\Delta_{x,0}^{2}}{\hbar^{2}}\right]
\end{equation}
then
\begin{equation}
\begin{array}{rl}
N_{\mu}&=\mathrm{Tr}\left[\ket{\psi_{0}}\bra{\psi_{0}}+\ket{\psi_{0}^{-}}\bra{\psi_{0}^{-}}+\mathrm{e}^{-\mu}\left(\ket{\psi_{0}^{-}}\bra{\psi_{0}}+\ket{\psi_{0}}\bra{\psi_{0}^{-}}\right)\right]\\
&=\mathrm{Tr}\left[\ket{\psi(t)}\bra{\psi(t)}+\ket{\psi^{-}(t)}\bra{\psi^{-}(t)}+\mathrm{e}^{-\mu}\left(\ket{\psi^{-}(t)}\bra{\psi(t)}+\ket{\psi(t)}\bra{\psi^{-}(t)}\right)\right]\\
&=2\left(1+\mathrm{e}^{-\left[\mu+\frac{x_{0}^{2}}{2\Delta_{x,0}^{2}}+\frac{p_{0}^{2}4\Delta_{x,0}^{2}}{2\hbar^{2}}\right]}\right).
\end{array}
\end{equation}
\section{\textbf{Mean position}}
Given that
\begin{equation}
\bra{\psi(t)} X\ket{\psi(t)}=\mathrm{e}^{-\frac{c_{0}(t)}{2}}\left(x_{0}-c_{-}p_{0}\right),
\end{equation}
\begin{equation}
\bra{\psi^{-}(t)} X\ket{\psi^{-}(t)}=-\mathrm{e}^{-\frac{c_{0}(t)}{2}}\left(x_{0}-c_{-}p_{0}\right),
\end{equation}
\begin{equation}
\bra{\psi^{-}(t)}X\ket{\psi(t)}=\frac{i\exp \left[-\frac{c_{0}}{2}-\frac{x_{0}^{2}}{2\Delta_{x,0}^{2}}-\frac{p_{0}^{2}4\Delta_{x,0}^{2}}{2\hbar^{2}}\right]\left(16p_{0}\Delta_{x,0}^{4}+4c_{-}x_{0}\hbar^{2}\right)}{32\Delta_{x,0}^{4}\hbar^{2}},
\end{equation}
and
\begin{equation}
\bra{\psi(t)}X\ket{\psi^{-}(t)}=-\frac{i\exp \left[-\frac{c_{0}}{2}-\frac{x_{0}^{2}}{2\Delta_{x,0}^{2}}-\frac{p_{0}^{2}4\Delta_{x,0}^{2}}{2\hbar^{2}}\right]\left(16p_{0}\Delta_{x,0}^{4}+4c_{-}x_{0}\hbar^{2}\right)}{32\Delta_{x,0}^{4}\hbar^{2}},
\end{equation}
then
\begin{equation}
\langle X \rangle_{q}^{\mu}=\mathrm{Tr}\left(\rho^{\mu}(t)X\right)=0.
\end{equation}
\section{\textbf{Mean square position}}
From
\begin{equation}
\bra{\psi(t)} X^{2}\ket{\psi(t)}=\Delta_{x,t}^{2}+\mathrm{e}^{-c_{0}(t)}\left(x_{0}-c_{-}p_{0}\right)^{2},
\end{equation}
\begin{equation}
\bra{\psi^{-}(t)} X^{2}\ket{\psi^{-}(t)}=\Delta_{x,t}^{2}+\mathrm{e}^{-c_{0}(t)}\left(x_{0}-c_{-}p_{0}\right)^{2},
\end{equation}
\begin{equation}
\bra{\psi^{-}(t)}X^{2}\ket{\psi(t)}=\exp \left[-2\frac{x_{0}^{2}}{4\Delta_{x,0}^{2}}-2\frac{p_{0}^{2}4\Delta_{x,0}^{2}}{4\hbar^{2}}\right]\left(\Delta_{x,t}^{2}-\mathrm{e}^{-c_{0}}\left(2\frac{c_{-}\hbar x_{0}}{4\Delta_{x,0}^{2}}+\frac{p_{0}4\Delta_{x,0}^{2}}{2\hbar}\right)^{2}\right),
\end{equation}
and
\begin{equation}
\bra{\psi(t)}X^{2}\ket{\psi^{-}(t)}=\exp \left[-2\frac{x_{0}^{2}}{4\Delta_{x,0}^{2}}-2\frac{p_{0}^{2}4\Delta_{x,0}^{2}}{4\hbar^{2}}\right]\left(\Delta_{x,t}^{2}-\mathrm{e}^{-c_{0}}\left(c_{-}\frac{ x_{0}}{\Delta_{x,0}}\frac{\hbar}{2\Delta_{x,0}}+\frac{p_{0}}{\frac{\hbar}{2\Delta_{x,0}}}\Delta_{x,0}\right)^{2}\right)
\end{equation}
it follows that
\begin{equation}
\begin{array}{rl}
\langle X^{2} \rangle_{q}^{\mu}&=\mathrm{Tr}\left(\rho^{\mu}(t)X^{2}\right)\\
&\displaystyle=\Delta_{x,t}^{2}+\\
&\displaystyle+\frac{\mathrm{e}^{-c_{0}(t)}\left[\left(x_{0}-c_{-}p_{0}\right)^{2}-\exp \left[-\mu-2\frac{x_{0}^{2}}{4\Delta_{x,0}^{2}}-2\frac{p_{0}^{2}4\Delta_{x,0}^{2}}{4\hbar^{2}}\right]\left(c_{-}\frac{ x_{0}}{\Delta_{x,0}}\frac{\hbar}{2\Delta_{x,0}}+\frac{p_{0}}{\frac{\hbar}{2\Delta_{x,0}}}\Delta_{x,0}\right)^{2}\right]}{2\left(1+\exp -\left[\mu+2\frac{x_{0}^{2}}{4\Delta_{x,0}^{2}}+2\frac{p_{0}^{2}4\Delta_{x,0}^{2}}{4\hbar^{2}}\right]\right)}
\end{array}
\end{equation}
\section{\textbf{Mean momentum}}
Considering
\begin{equation}
\bra{\psi(t)} P\ket{\psi(t)}=e^{-\frac{c_{0}(t)}{2}}\left[\left(e^{c_{0}(t)}-c_{+}(t)c_{-}\right)p_{0}+c_{+}(t)x_{0}\right],
\end{equation}
\begin{equation}
\bra{\psi^{-}(t)} P\ket{\psi^{-}(t)}=-e^{-\frac{c_{0}(t)}{2}}\left[\left(e^{c_{0}(t)}-c_{+}(t)c_{-}\right)p_{0}+c_{+}(t)x_{0}\right],
\end{equation}
\begin{equation}
\bra{\psi^{-}(t)}P\ket{\psi(t)}=i\exp\left[-\frac{c_{0}}{2}-\frac{2x_{0}^{2}}{4\Delta_{x,0}^{2}}-\frac{   2p_{0}^{2}4\Delta_{x,0}^{2}}{4\hbar^{2}}\right],
\end{equation}
and
\begin{equation}
\bra{\psi(t)}P\ket{\psi^{-}(t)}=-i\exp\left[-\frac{c_{0}}{2}-\frac{2x_{0}^{2}}{4\Delta_{x,0}^{2}}-\frac{   2p_{0}^{2}\Delta_{q,0}^{2}}{4\hbar^{2}}\right],
\end{equation}
it follows that
\begin{equation}
\langle P \rangle_{q}^{\mu}=\mathrm{Tr}\left(\rho^{\mu}(t)P\right)=0.\label{eq:MM2}
\end{equation}
\section{\textbf{Mean square momentum}}
From
\begin{equation}
\begin{array}{rl}
\bra{\psi(t)} P^{2}\ket{\psi(t)}&\displaystyle=\frac{\mathrm{e}^{-c_{0}(t)}\hbar^{2}}{4\Delta_{x,0}^{2}}\left(\mathrm{e}^{c_{0}(t)}-c_{+}(t)c_{-}(t)\right)^{2}+\frac{\mathrm{e}^{-c_{0}(t)}4\Delta_{x,0}^{2}c_{+}^{2}(t)}{4} +\\
&\displaystyle+e^{-c_{0}(t)}\left(\left(e^{c_{0}(t)}-c_{+}(t)c_{-}\right)p_{0}+c_{+}(t)x_{0}\right)^{2},
\end{array}
\end{equation} 
\begin{equation}
\begin{array}{rl}
\bra{\psi^{-}(t)} P^{2}\ket{\psi^{-}(t)}&\displaystyle=\frac{\mathrm{e}^{-c_{0}(t)}\hbar^{2}}{4\Delta_{x,0}^{2}}\left(\mathrm{e}^{c_{0}(t)}-c_{+}(t)c_{-}(t)\right)^{2}+\frac{\mathrm{e}^{-c_{0}(t)}4\Delta_{x,0}^{2}c_{+}^{2}(t)}{4}+\\
&\displaystyle+e^{-c_{0}(t)}\left(\left(e^{c_{0}(t)}-c_{+}(t)c_{-}\right)p_{0}+c_{+}(t)x_{0}\right)^{2},
\end{array}
\end{equation}
\begin{equation}
\begin{array}{rl}
\bra{\psi^{-}(t)}P^{2}\ket{\psi(t)}&\displaystyle=\exp \left[-c_{0}-2\frac{x_{0}^{2}}{4\Delta_{x,0}^{2}}-2\frac{p_{0}^{2}4\Delta_{x,0}^{2}}{4\hbar^{2}}\right]\left[\frac{\hbar^{2}}{ 4\Delta_{x,0}^{2}}\left(\mathrm{e}^{c_{0}}-c_{+}c_{-}\right)^{2}+\frac{4\Delta_{x,0}^{2}c_{+}^{2}}{4 }\right.+\\
&\displaystyle\left.-\left(\left(\mathrm{e}^{c_{0}}-c_{+}c_{-}\right)\frac{\hbar}{2\Delta_{x,0}}\frac{x_{0}}{\Delta_{x,0}}-c_{+}\frac{p_{0}}{\frac{\hbar}{2\Delta_{x,0}}}\Delta_{x,0}\right)^{2}\right],
\end{array}
\end{equation}
and
\begin{equation}
\begin{array}{rl}
\bra{\psi(t)}P^{2}\ket{\psi^{-}(t)}&=\exp \left[-c_{0}-2\frac{x_{0}^{2}}{4\Delta_{x,0}^{2}}-2\frac{p_{0}^{2}4\Delta_{x,0}^{2}}{4\hbar^{2}}\right]\left[\frac{\hbar^{2}}{ 4\Delta_{x,0}^{2}}\left(\mathrm{e}^{c_{0}}-c_{+}c_{-}\right)^{2}+\frac{4\Delta_{x,0}^{2}c_{+}^{2}}{4 }\right.+\\
&\left.-\left(\left(\mathrm{e}^{c_{0}}-c_{+}c_{-}\right)\frac{\hbar}{2\Delta_{x,0}}\frac{x_{0}}{\frac{4\Delta_{x,0}}{2}}-c_{+}\frac{p_{0}}{\frac{\hbar}{2\Delta_{x,0}}}\Delta_{x,0}\right)^{2}\right],
\end{array}
\end{equation}
one shows, after some algebraic manipulation regarding Eqs. \eqref{eq:MMM} and \eqref{eq:deltap2}, that

\begin{equation}
\langle P^{2} \rangle_{q}^{\mu}=\mathrm{Tr}\left(\rho^{\mu}(t)P^{2}\right)=\bra{\psi(t)} P^{2}\ket{\psi(t)}-\left(\frac{p_{0}^{2}}{\frac{\hbar^{2}}{4\Delta_{x,0}^{2}}}+\frac{x_{0}^{2}}{\frac{4\Delta_{x,0}^{2}}{4}}\right)\frac{\bra{\psi(t)} \Delta P^{2}\ket{\psi(t)}}{1+\exp \left[\mu+\frac{1}{2}\left(\frac{p_{0}^{2}}{\frac{\hbar^{2}}{4\Delta_{x,0}^{2}}}+\frac{x_{0}^{2}}{\frac{4\Delta_{x,0}^{2}}{4}}\right)\right]}.\label{eq:MMM2}
\end{equation}
\section{\textbf{Summarizing the results}}
The expectation values obtained for the state \eqref{eq:statot}, are summarized as
\begin{equation}
\langle X \rangle_{q}^{\mu}=\mathrm{Tr}\left(\rho^{\mu}(t)X\right)=0,
\end{equation}
\begin{equation}
\langle P \rangle_{q}^{\mu}=\mathrm{Tr}\left(\rho^{\mu}(t)P\right)=0,
\end{equation}
\begin{equation}
\begin{array}{rl}
\langle X^{2} \rangle_{q}^{\mu}&=\mathrm{Tr}\left(\rho^{\mu}(t)X^{2}\right)\\
&\displaystyle=\Delta_{x,t}^{2}+\\
&\displaystyle+\frac{\mathrm{e}^{-c_{0}(t)}\left[\left(x_{0}-c_{-}p_{0}\right)^{2}-\exp \left[-\mu-2\frac{x_{0}^{2}}{4\Delta_{x,0}^{2}}-2\frac{p_{0}^{2}4\Delta_{x,0}^{2}}{4\hbar^{2}}\right]\left(c_{-}\frac{ x_{0}}{\Delta_{x,0}}\frac{\hbar}{2\Delta_{x,0}}+\frac{p_{0}}{\frac{\hbar}{2\Delta_{x,0}}}\Delta_{x,0}\right)^{2}\right]}{2\left(1+\exp -\left[\mu+2\frac{x_{0}^{2}}{4\Delta_{x,0}^{2}}+2\frac{p_{0}^{2}4\Delta_{x,0}^{2}}{4\hbar^{2}}\right]\right)}
\end{array}
\end{equation}
and
\begin{equation}
\langle P^{2} \rangle_{q}^{\mu}=\mathrm{Tr}\left(\rho^{\mu}(t)P^{2}\right)=\bra{\psi(t)} P^{2}\ket{\psi(t)}-\left(\frac{p_{0}^{2}}{\frac{\hbar^{2}}{4\Delta_{x,0}^{2}}}+\frac{x_{0}^{2}}{\Delta_{x,0}^{2}}\right)\frac{\bra{\psi(t)} \Delta P^{2}\ket{\psi(t)}}{1+\exp \left[\mu+\frac{1}{2}\left(\frac{p_{0}^{2}}{\frac{\hbar^{2}}{4\Delta_{x,0}^{2}}}+\frac{x_{0}^{2}}{\Delta_{x,0}^{2}}\right)\right]}.
\end{equation}

%% file: dissertacao_fisica.bbl
\begin{thebibliography}{10}

\bibitem{laplace2012pierre}
Pierre-Simon Laplace.
\newblock {\em Pierre-Simon Laplace Philosophical Essay on Probabilities:
  Translated from the fifth French edition of 1825 With Notes by the
  Translator}, volume~13.
\newblock Springer Science \& Business Media, 2012.

\bibitem{arnol2013mathematical}
Vladimir~Igorevich Arnol'd.
\newblock {\em Mathematical methods of classical mechanics}, volume~60.
\newblock Springer Science \& Business Media, 2013.

\bibitem{goldstein2011classical}
Herbert Goldstein.
\newblock {\em Classical mechanics}.
\newblock Pearson Education India, 2011.

\bibitem{kleppner2013introduction}
Daniel Kleppner and Robert Kolenkow.
\newblock {\em An introduction to mechanics}.
\newblock Cambridge University Press, 2013.

\bibitem{schwabl2006statistical}
F~Schwabl.
\newblock {\em Statistical mechanics}.
\newblock Springer, 2006.

\bibitem{tolman1938principles}
Richard~Chace Tolman.
\newblock {\em The principles of statistical mechanics}.
\newblock Courier Corporation, 1938.

\bibitem{landau1968statistical}
Lev~Davidovich Landau and Evgenii~M Lifshitz.
\newblock {\em Statistical physics: V. 5: course of theoretical physics}.
\newblock Pergamon press, 1968.

\bibitem{callen1985thermodynamics}
Herbert~B Callen.
\newblock {\em Thermodynamics and an Introduction to Thermostatistics}.
\newblock Wiley, 1985.

\bibitem{moran2010fundamentals}
Michael~J Moran, Howard~N Shapiro, Daisie~D Boettner, and Margaret~B Bailey.
\newblock {\em Fundamentals of engineering thermodynamics}.
\newblock John Wiley \& Sons, 2010.

\bibitem{gurtin2010mechanics}
Morton~E Gurtin, Eliot Fried, and Lallit Anand.
\newblock {\em The mechanics and thermodynamics of continua}.
\newblock Cambridge University Press, 2010.

\bibitem{gemmer20044}
Jochen Gemmer, Mathias Michel, and G{\"u}nter Mahler.
\newblock {\em Quantum Thermodynamics}.
\newblock Springer, 2004.

\bibitem{truesdell2004non}
Clifford Truesdell and Walter Noll.
\newblock {\em The non-linear field theories of mechanics}.
\newblock Springer, 2004.

\bibitem{silva2015elastoplasticity}
Pinto Silva, Thales~Augusto Barbosa, Hilbeth~Parente Azikri~de Deus, et~al.
\newblock Elastoplasticity 2d problems: Numerical applications of the tikhonov
  regularization method.
\newblock {\em Applied Mechanics \& Materials}, 751, 2015.

\bibitem{EUeHPAD}
T.~A. B.~P. Silva, Hilbeth~P. Azikri~de Deus, and C.~O.~R. Negr{\~a}o.
\newblock A numerical approach on new constitutive model for thixotropic
  substances.
\newblock {\em Applied Mechanics and Materials}, 751(1):95--101, 2014.

\bibitem{macosko1994rheology}
Christopher~W Macosko and Ronald~G Larson.
\newblock {\em Rheology: principles, measurements, and applications}.
\newblock VCH New York, 1994.

\bibitem{arnol1968ergodic}
Vladimir~Igorevich Arnol'd and Andr{\'e} Avez.
\newblock {\em Ergodic problems of classical mechanics}.
\newblock WA Benjamin, 1968.

\bibitem{kubo1957statistical}
Ryogo Kubo.
\newblock Statistical-mechanical theory of irreversible processes. i. general
  theory and simple applications to magnetic and conduction problems.
\newblock {\em Journal of the Physical Society of Japan}, 12(6):570--586, 1957.

\bibitem{reif2009fundamentals}
Frederick Reif.
\newblock {\em Fundamentals of statistical and thermal physics}.
\newblock McGraw-Hill, 1965.

\bibitem{chandler1987introduction}
David Chandler.
\newblock {\em Introduction to modern statistical mechanics}.
\newblock 1987.

\bibitem{schrodinger1989statistical}
Erwin Schr{\"o}dinger.
\newblock {\em Statistical thermodynamics}.
\newblock Unabridged Dover, 1989.

\bibitem{jarzynski1997nonequilibrium}
Christopher Jarzynski.
\newblock Nonequilibrium equality for free energy differences.
\newblock {\em Physical Review Letters}, 78(14):2690, 1997.

\bibitem{campisi2011colloquium}
Michele Campisi, Peter H{\"a}nggi, and Peter Talkner.
\newblock Colloquium: Quantum fluctuation relations: Foundations and
  applications.
\newblock {\em Reviews of Modern Physics}, 83(3):771, 2011.

\bibitem{sakurai2017modern}
Jun~John Sakurai.
\newblock {\em Modern quantum mechanics Revised Edition}.
\newblock Addison-Wesley Publishing Company, 1994.

\bibitem{eisberg1979fisica}
Robert Eisberg and Robert Resnick.
\newblock {\em F{\'\i}sica qu{\^a}ntica}.
\newblock Rio de Janeiro: Editora Campus, 1979.

\bibitem{salinas1997introduccao}
S{\'\i}lvio~RA Salinas.
\newblock {\em Introdu{\c{c}}{\~a}o a f{\'\i}sica estat{\'\i}stica vol. 09}.
\newblock Edusp, 1997.

\bibitem{nielsen2010quantum}
Michael~A Nielsen and Isaac~L Chuang.
\newblock {\em Quantum computation and quantum information}.
\newblock Cambridge university press, 2010.

\bibitem{bilobran2015measure}
ALO Bilobran and RM~Angelo.
\newblock A measure of physical reality.
\newblock {\em EPL (Europhysics Letters)}, 112(4):40005, 2015.

\bibitem{goold2016role}
John Goold, Marcus Huber, Arnau Riera, L{\'\i}dia del Rio, and Paul Skrzypczyk.
\newblock The role of quantum information in thermodynamics—a topical review.
\newblock {\em Journal of Physics A: Mathematical and Theoretical},
  49(14):143001, 2016.

\bibitem{aaberg2018fully}
Johan {\AA}berg.
\newblock Fully quantum fluctuation theorems.
\newblock {\em Physical Review X}, 8(1):011019, 2018.

\bibitem{perarnau2017no}
Mart{\'\i} Perarnau-Llobet, Elisa B{\"a}umer, Karen~V Hovhannisyan, Marcus
  Huber, and Antonio Acin.
\newblock No-go theorem for the characterization of work fluctuations in
  coherent quantum systems.
\newblock {\em Physical Review Letters}, 118(7):070601, 2017.

\bibitem{hanggi2015other}
Peter H{\"a}nggi and Peter Talkner.
\newblock The other qft.
\newblock {\em Nature Physics}, 11(2):108, 2015.

\bibitem{alhambra2016fluctuating}
{\'A}lvaro~M Alhambra, Lluis Masanes, Jonathan Oppenheim, and Christopher
  Perry.
\newblock Fluctuating work: From quantum thermodynamical identities to a second
  law equality.
\newblock {\em Physical Review X}, 6(4):041017, 2016.

\bibitem{jarzynski2017stochastic}
Christopher Jarzynski.
\newblock Stochastic and macroscopic thermodynamics of strongly coupled
  systems.
\newblock {\em Physical Review X}, 7(1):011008, 2017.

\bibitem{valente2017quantum}
D~Valente, F~Brito, R~Ferreira, and T~Werlang.
\newblock Work on a quantum dipole by a single-photon pulse.
\newblock {\em Optics letters}, 43(11):2644--2647, 2018.

\bibitem{youssef2009quantum}
M~Youssef, G~Mahler, and A-SF Obada.
\newblock Quantum optical thermodynamic machines: Lasing as relaxation.
\newblock {\em Physical Review E}, 80(6):061129, 2009.

\bibitem{alipour2016correlations}
S~Alipour, F~Benatti, F~Bakhshinezhad, M~Afsary, S~Marcantoni, and
  AT~Rezakhani.
\newblock Correlations in quantum thermodynamics: Heat, work, and entropy
  production.
\newblock {\em Scientific Reports}, 6:35568, 2016.

\bibitem{popescu2006entanglement}
Sandu Popescu, Anthony~J Short, and Andreas Winter.
\newblock Entanglement and the foundations of statistical mechanics.
\newblock {\em Nature Physics}, 2(11):754, 2006.

\bibitem{tasaki1998quantum}
Hal Tasaki.
\newblock From quantum dynamics to the canonical distribution: general picture
  and a rigorous example.
\newblock {\em Physical Review Letters}, 80(7):1373, 1998.

\bibitem{louisell1973quantum}
William~Henry Louisell and William~H Louisell.
\newblock {\em Quantum statistical properties of radiation}, volume~7.
\newblock Wiley New York, 1973.

\bibitem{gardiner2004quantum}
Crispin Gardiner and Peter Zoller.
\newblock {\em Quantum noise: a handbook of Markovian and non-Markovian quantum
  stochastic methods with applications to quantum optics}, volume~56.
\newblock Springer Science \& Business Media, 2004.

\bibitem{wen2004quantum}
Xiao-Gang Wen.
\newblock {\em Quantum field theory of many-body systems: from the origin of
  sound to an origin of light and electrons}.
\newblock Oxford University Press on Demand, 2004.

\bibitem{karlewski2014time}
Christian Karlewski and Michael Marthaler.
\newblock Time-local master equation connecting the born and markov
  approximations.
\newblock {\em Physical Review B}, 90(10):104302, 2014.

\bibitem{tan2011non}
Hua-Tang Tan and Wei-Min Zhang.
\newblock Non-markovian dynamics of an open quantum system with initial
  system-reservoir correlations: A nanocavity coupled to a coupled-resonator
  optical waveguide.
\newblock {\em Physical Review A}, 83(3):032102, 2011.

\bibitem{tu2008non}
Matisse~WY Tu and Wei-Min Zhang.
\newblock Non-markovian decoherence theory for a double-dot charge qubit.
\newblock {\em Physical Review B}, 78(23):235311, 2008.

\bibitem{zhang2012general}
Wei-Min Zhang, Ping-Yuan Lo, Heng-Na Xiong, Matisse Wei-Yuan Tu, and Franco
  Nori.
\newblock General non-markovian dynamics of open quantum systems.
\newblock {\em Physical Review Letters}, 109(17):170402, 2012.

\bibitem{maruyama2009colloquium}
Koji Maruyama, Franco Nori, and Vlatko Vedral.
\newblock Colloquium: The physics of maxwell’s demon and information.
\newblock {\em Reviews of Modern Physics}, 81(1):1, 2009.

\bibitem{micadei2017reversing}
Kaonan Micadei, John~PS Peterson, Alexandre~M Souza, Roberto~S Sarthour, Ivan~S
  Oliveira, Gabriel~T Landi, Tiago~B Batalh{\~a}o, Roberto~M Serra, and Eric
  Lutz.
\newblock Reversing the thermodynamic arrow of time using quantum correlations.
\newblock {\em arXiv preprint arXiv:1711.03323}, 2017.

\bibitem{kliesch2014locality}
M~Kliesch, C~Gogolin, MJ~Kastoryano, A~Riera, and J~Eisert.
\newblock Locality of temperature.
\newblock {\em Physical Review x}, 4(3):031019, 2014.

\bibitem{rugh1997dynamical}
Hans~Henrik Rugh.
\newblock Dynamical approach to temperature.
\newblock {\em Physical Review Letters}, 78(5):772, 1997.

\bibitem{allahverdyan2005work}
AE~Allahverdyan, R~Serral Gracia, and Th~M Nieuwenhuizen.
\newblock Work extraction in the spin-boson model.
\newblock {\em Physical Review E}, 71(4):046106, 2005.

\bibitem{tonner2005autonomous}
Friedemann Tonner and G{\"u}nter Mahler.
\newblock Autonomous quantum thermodynamic machines.
\newblock {\em Physical Review E}, 72(6):066118, 2005.

\bibitem{lorch2018optimal}
Niels L{\"o}rch, Christoph Bruder, Nicolas Brunner, and Patrick~P Hofer.
\newblock Optimal work extraction from quantum states by photo-assisted cooper
  pair tunneling.
\newblock {\em arXiv preprint arXiv:1802.10572}, 2018.

\bibitem{bender2000quantum}
Carl~M Bender, Dorje~C Brody, and Bernhard~K Meister.
\newblock Quantum mechanical carnot engine.
\newblock {\em Journal of Physics A: Mathematical and General}, 33(24):4427,
  2000.

\bibitem{alicki1979quantum}
Robert Alicki.
\newblock The quantum open system as a model of the heat engine.
\newblock {\em Journal of Physics A: Mathematical and General}, 12(5):L103,
  1979.

\bibitem{ribeiro2016quantum}
Wellington~L Ribeiro, Gabriel~T Landi, and Fernando~L Semi{\~a}o.
\newblock Quantum thermodynamics and work fluctuations with applications to
  magnetic resonance.
\newblock {\em American Journal of Physics}, 84(12):948--957, 2016.

\bibitem{weimer2008local}
Hendrik Weimer, Markus~J Henrich, Florian Rempp, Heiko Schr{\"o}der, and
  G{\"u}nter Mahler.
\newblock Local effective dynamics of quantum systems: A generalized approach
  to work and heat.
\newblock {\em EPL (Europhysics Letters)}, 83(3):30008, 2008.

\bibitem{kanai1948quantization}
E~Kanai.
\newblock On the quantization of the dissipative systems.
\newblock {\em Progress of Theoretical Physics}, 3(4):440--442, 1948.

\bibitem{caldirola1941forze}
Piero Caldirola.
\newblock Forze non conservative nella meccanica quantistica.
\newblock {\em Il Nuovo Cimento (1924-1942)}, 18(9):393--400, 1941.

\bibitem{david1988fundamentals}
Halliday David and Resnick Robert.
\newblock {\em Fundamentals of physics}.
\newblock John Wiley \& Sons Incorporated, 1988.

\bibitem{pathria1996}
R.~K. Pathria.
\newblock {\em Statistical mechanics}.
\newblock Elsevier, 1996.

\bibitem{incropera1999fundamentos}
Frank~P Incropera and David~P DeWitt.
\newblock {\em Fundamentos de transferencia de calor}.
\newblock Pearson Educaci{\'o}n, 1999.

\bibitem{de2011topicos}
Marcus~AM de~Aguiar.
\newblock T{\'o}picos de mec{\^a}nica cl{\'a}ssica.
\newblock 2011.

\bibitem{szczepinski2002error}
W~Szczepinski, Z~Kotulski, and M~Bonnet.
\newblock {\em Error Analysis with Applications in Engineering}.
\newblock Springer.

\bibitem{lebowitz1973modern}
Joel~L Lebowitz and Oliver Penrose.
\newblock Modern ergodic theory.
\newblock {\em Physics Today}, 26(2):23--29, 1973.

\bibitem{lebedev2003functional}
LP~Lebedev, II~Vorovich, and GM~Gladwell.
\newblock {\em Functional Analysis: Applications in Mechanics and Inverse
  Problems}.
\newblock Dordrecht: Springer Netherlands, 2003.

\bibitem{ehrenfest1927bemerkung}
Paul Ehrenfest.
\newblock Bemerkung {\"u}ber die angen{\"a}herte g{\"u}ltigkeit der klassischen
  mechanik innerhalb der quantenmechanik.
\newblock {\em Zeitschrift f{\"u}r Physik A Hadrons and Nuclei},
  45(7):455--457, 1927.

\bibitem{ballentine1994inadequacy}
Leslie~E Ballentine, Yumin Yang, and JP~Zibin.
\newblock Inadequacy of ehrenfest’s theorem to characterize the classical
  regime.
\newblock {\em Physical Review A}, 50(4):2854, 1994.

\bibitem{angelo2003aspectos}
Renato~Moreira Angelo.
\newblock {\em Aspectos qu{\^a}nticos e cl{\'a}ssicos da din{\^a}mica de
  emaranhamento em sistemas hamiltonianos}.
\newblock PhD thesis, Universidade Estadual de Campinas, 2003.

\bibitem{von2010proof}
John von Neumann.
\newblock Proof of the ergodic theorem and the h-theorem in quantum mechanics.
\newblock {\em The European Physical Journal H}, 35(2):201--237, 2010.

\bibitem{bell2001einstein}
John~S Bell.
\newblock On the einstein podolsky rosen paradox.
\newblock In {\em John S Bell On The Foundations Of Quantum Mechanics}, pages
  7--12. World Scientific, 2001.

\bibitem{einstein1935can}
Albert Einstein, Boris Podolsky, and Nathan Rosen.
\newblock Can quantum-mechanical description of physical reality be considered
  complete?
\newblock {\em Physical Review}, 47(10):777, 1935.

\bibitem{reichl2016modern}
Linda~E Reichl.
\newblock {\em A modern course in statistical physics}.
\newblock John Wiley \& Sons, 2016.

\bibitem{costa2013bayes}
ACS Costa and RM~Angelo.
\newblock Bayes' rule, generalized discord, and nonextensive thermodynamics.
\newblock {\em Physical Review A}, 87(3):032109, 2013.

\bibitem{de2017dynamics}
In{\'e}s de~Vega and Daniel Alonso.
\newblock Dynamics of non-markovian open quantum systems.
\newblock {\em Reviews of Modern Physics}, 89(1):015001, 2017.

\bibitem{bochkov1977general}
GN~Bochkov and Yu~E Kuzovlev.
\newblock General theory of thermal fluctuations in nonlinear systems.
\newblock {\em Zh. Eksp. Teor. Fiz}, 72:238--243, 1977.

\bibitem{brown1991quantum}
Lowell~S Brown.
\newblock Quantum motion in a paul trap.
\newblock {\em Physical Review Letters}, 66(5):527, 1991.

\bibitem{schuch1999effective}
Dieter Schuch.
\newblock Effective description of the dissipative interaction between simple
  model-systems and their environment.
\newblock {\em International Journal of Quantum Chemistry}, 72(6):537--547,
  1999.

\bibitem{schuch1997nonunitary}
Dieter Schuch.
\newblock Nonunitary connection between explicitly time-dependent and nonlinear
  approachesfor the description of dissipative quantum systems.
\newblock {\em Physical Review A}, 55(2):935, 1997.

\bibitem{sun1995exact}
Chang-Pu Sun and Li-Hua Yu.
\newblock Exact dynamics of a quantum dissipative system in a constant external
  field.
\newblock {\em Physical Review A}, 51(3):1845, 1995.

\bibitem{iesus}
I{\'e}sus Souza~Freire.
\newblock Modelos din{\^a}micos de medic{\c{c}}{\~a}o qu{\^a}ntica: realidade e
  repouso emergentes.
\newblock Master's thesis, 2018.

\bibitem{cheng1988evolution}
CM~Cheng and PCW Fung.
\newblock The evolution operator technique in solving the schrodinger equation,
  and its application to disentangling exponential operators and solving the
  problem of a mass-varying harmonic oscillator.
\newblock {\em Journal of Physics A: Mathematical and General}, 21(22):4115,
  1988.

\end{thebibliography}
